\documentclass[reqno]{amsart}
\usepackage{fullpage,amsthm,amssymb,amsfonts,graphicx,mathrsfs,textcomp,color}
\pdfoutput=1

\newtheorem{rhp}{Riemann-Hilbert Problem}
\newtheorem{assume}{Assumption}
\newtheorem{lemma}{Lemma}
\newtheorem{proposition}{Proposition}
\newtheorem{theorem}{Theorem}
\newtheorem{corollary}{Corollary}
\newtheorem{definition}{Definition}
\newcommand{\trans}[1]{{#1}^{\ensuremath{\mathsf{T}}}}
\newcommand{\breather}{{\sf{B}}}
\newcommand{\kink}{{\sf{K}}}
\newcommand{\rotational}{{\sf{R}}}
\newcommand{\librational}{{\sf{L}}}
\newcommand{\bo}{\mathcal{O}}
\newcommand{\lo}{\mathfrak{o}}
\numberwithin{equation}{section}
\numberwithin{theorem}{section}
\numberwithin{proposition}{section}
\numberwithin{lemma}{section}
\numberwithin{corollary}{section}
\numberwithin{assume}{section}
\numberwithin{rhp}{section}
\numberwithin{figure}{section}
\numberwithin{definition}{section}
\title{The Sine-Gordon Equation in the Semiclassical Limit:  Dynamics
of Fluxon Condensates}
\author{Robert J. Buckingham}
\address[R. J. Buckingham]{Department of Mathematical Sciences\\ University of Cincinnati\\ PO Box 210025\\ Cincinnati, OH 45221.}
\email{buckinrt@uc.edu}
\urladdr{http://homepages.uc.edu/~buckinrt/}
\author{Peter D. Miller}
\address[P. D. Miller]{Department of Mathematics, University of Michigan\\East Hall\\530 Church St.\\Ann Arbor, MI 48109.}
\email{millerpd@umich.edu}
\urladdr{http://www.math.lsa.umich.edu/~millerpd/}
\date{\today}
\thanks{The authors were partially supported by the National Science Foundation under grant DMS-0807653. Additionally, R. J. Buckingham acknowledges the support of the Charles Phelps Taft Research Foundation.}
\begin{document}
\begin{abstract}
We study the Cauchy problem for the sine-Gordon equation in the semiclassical limit with pure-impulse
initial data of sufficient strength to generate 
both high-frequency rotational motion near the peak of the impulse profile and also high-frequency librational motion in the tails.
We show that for small times independent of
the semiclassical scaling parameter, 
both types of motion are accurately described by explicit formulae involving elliptic functions.  These
formulae
demonstrate consistency with predictions of Whitham's formal modulation theory in both the hyperbolic (modulationally stable) and elliptic (modulationally unstable) cases.
\end{abstract}
\maketitle
\tableofcontents

\section{Introduction}
The sine-Gordon equation in laboratory coordinates
\begin{equation}
\epsilon^2u_{tt}-\epsilon^2u_{xx}+\sin(u)=0,\quad x\in\mathbb{R},\quad t>0,
\label{eq:SG}
\end{equation}
is a model for magnetic flux propagation in long superconducting Josephson junctions \cite{ScottCR76}, but perhaps may be most easily
thought of as the equation describing mechanical vibrations of 
a ``ribbon'' pendulum (the continuum limit of a system of linearly arranged co-axial pendula with nearest-neighbor torsion coupling).  These and other applications are discussed in detail in the review article \cite{BaroneEMS71}.
The correct Cauchy problem for this equation involves determining the solution consistent with the given initial data
\begin{equation}
u(x,0)=F(x)\quad\text{and}\quad \epsilon u_t(x,0)=G(x).
\label{eq:IC}
\end{equation}
Here $F(x)$ and $G(x)$ are independent of the fixed parameter $\epsilon$.  
This Cauchy problem is
globally well-posed \cite{BuckinghamM08}: if $p\ge 1$ and $F$, $F'$,
and $G$ are functions in $L^p(\mathbb{R})$ then there is a unique
solution with $u$, $u_x$, and $u_t$ all lying in
$L^\infty_\mathrm{loc}(\mathbb{R}_+;L^p(\mathbb{R}))$.  Moreover if
the initial data have one more $L^p$ derivative, so that $F''$ and
$G'$ are functions in $L^p(\mathbb{R})$, then this further regularity
is preserved as well so that $u_{xx}$ and $u_{tx}$ lie in
$L^\infty_\mathrm{loc}(\mathbb{R}_+;L^p(\mathbb{R}))$.  These
well-posedness results also hold in a slightly modified form when the
initial displacement has nonzero asymptotic values:  $F(x)\to 2\pi
n_\pm$, $n_\pm\in\mathbb{Z}$, as $x\to\pm\infty$.  In this case the
topological charge $n_+-n_-$ is preserved for all time in the solution $u$.

If suitable initial conditions $F$ and $G$ are fixed, one may therefore in principle construct the unique global solution $u(x,t)$ of \eqref{eq:SG} subject to \eqref{eq:IC} for each positive
$\epsilon$.  Our interest is in the asymptotic behavior of this family of global solutions in the \emph{semiclassical limit} $\epsilon\to 0$.  The well-known elementary excitations
of the sine-Gordon equation include solitons of kink (or antikink) and breather type; these have a width proportional to $\epsilon$ while the length scales in the initial conditions
\eqref{eq:IC} are independent of $\epsilon$.  This suggests that when $\epsilon\ll 1$ the
initial conditions of the system can be viewed as preparing a ``condensate'' whose ultimate breakup will liberate approximately $1/\epsilon$ fundamental particles. 

The decay process will take some time to become complete, and during
the intermediate stages one may expect that some solitons may
partially emerge from the condensate moving with nearly identical
velocities, thus forming a modulated wavetrain.  The simplest models for these
wavetrains are the periodic (modulo $2\pi$) traveling wave exact solutions of
\eqref{eq:SG} of the form 
\begin{equation}
u(x,t)=U\left(\frac{\Phi(x,t)}{\epsilon}\right),\quad U(\zeta+2\pi)=U(\zeta)
\pmod{2\pi},\quad\Phi(x,t)=kx-\omega t,
\end{equation}
where $k$ is the \emph{wavenumber} and $\omega$ the \emph{frequency} of the wavetrain.
With this substitution, the sine-Gordon equation reduces to an ordinary
differential equation that can be integrated once to 
\begin{equation}
\frac{1}{2}(\omega^2-k^2)\left(\frac{dU}{d\zeta}\right)^2 -\cos(U)=\mathcal{E}
\label{eq:integratedODE}
\end{equation}
where $\mathcal{E}$ is an integration constant having the interpretation of \emph{energy}.  There are four types
of solutions subject to the periodicity condition:
\begin{itemize}
\item\emph{Superluminal librational wavetrains} correspond to
  $\omega^2>k^2$ and $|\mathcal{E}|<1$.  From a phase portrait it is
  evident that $U(\zeta+2\pi)=U(\zeta)$ if the nonlinear dispersion
  relation
\begin{equation}
\omega^2-k^2=2\pi^2\left[\int_{-\cos^{-1}(-\mathcal{E})}^{+\cos^{-1}(-\mathcal{E})}
\frac{d\phi}{\sqrt{\cos(\phi)+\mathcal{E}}}\right]^{-2}
\label{eq:NLdisprel1}
\end{equation}
is satisfied, and then $U$ oscillates about a mean value of $U=0\pmod{2\pi}$ with an amplitude strictly less than $\pi$.
\item\emph{Subluminal librational wavetrains} correspond to $\omega^2<k^2$
  and $|\mathcal{E}|<1$.  From a phase portrait it is evident that
  $U(\zeta+2\pi)=U(\zeta)$ if the nonlinear dispersion relation
\begin{equation}
\omega^2-k^2=-2\pi^2\left[\int_{-\cos^{-1}(\mathcal{E})}^{+\cos^{-1}(\mathcal{E})}
\frac{d\phi}{\sqrt{\cos(\phi)-\mathcal{E}}}\right]^{-2}
\label{eq:NLdisprel2}
\end{equation}
is satisfied,
and then $U$ oscillates about a mean value of $U=\pi\pmod{2\pi}$ with an amplitude strictly less than $\pi$.
\item\emph{Superluminal rotational wavetrains} correspond to $\omega^2>k^2$
  and $\mathcal{E}>1$.  Here the phase portrait indicates that
  $U(\zeta+2\pi)=U(\zeta)\pm 2\pi$ if the nonlinear dispersion
  relation
\begin{equation}
\omega^2-k^2=8\pi^2\left[\int_{-\pi}^\pi\frac{d\phi}{\sqrt{\cos(\phi)+\mathcal{E}}}\right]^{-2}
\label{eq:NLdisprel3}
\end{equation}
is satisfied, with $U'(\zeta)$ strictly nonzero being largest in
magnitude when $\zeta=0\pmod{2\pi}$.
\item\emph{Subluminal rotational wavetrains} correspond to $\omega^2<k^2$ and
  $\mathcal{E}<-1$.  Here the phase portrait indicates that
  $U(\zeta+2\pi)=U(\zeta)\pm 2\pi$ if the nonlinear dispersion
  relation
\begin{equation}
\omega^2-k^2=-8\pi^2\left[\int_{-\pi}^\pi\frac{d\phi}
{\sqrt{\cos(\phi)-\mathcal{E}}}\right]^{-2}
\label{eq:NLdisprel4}
\end{equation}
is satisfied, with $U'(\zeta)$ strictly nonzero being largest in
magnitude when $\zeta=\pi\pmod{2\pi}$.
\end{itemize}
In the classical mechanics literature \cite{Goldstein} the term \emph{libration} is used to characterize the
kind of motion in which both position and momentum are periodic functions while the term \emph{rotation} is used to characterize motions in which momentum is periodic but position is not because the momentum has a nonzero average value.  Furthermore, the dichotomy of subluminal waves versus superluminal waves is important because the sine-Gordon equation is strictly hyperbolic with characteristic velocities $v_\mathrm{p}=\pm 1$; thus subluminal waves have phase velocities bounded in magnitude by the characteristic velocity, while superluminal waves move faster than the (unit) characteristic speed.  In the superluminal (respectively, subluminal) case, the energy value of $\mathcal{E}=1$ (respectively, $\mathcal{E}=-1$) corresponds to the separatrix in the phase portrait of the simple pendulum, at which point the period (respectively, wavelength) of the waves tends to infinity.  In this limit, each of the four types of wavetrain degenerates into a train of well-separated kink-type solitons; for rotational waves the kinks all have the same topological charge, while for librational waves the pulses alternate from kink to antikink for zero net charge.

In a body of work beginning with his seminal 1965 paper
\cite{Whitham65}, Whitham developed a nonlinear theory of modulated wavetrains.
The main idea in the current context is that one seeks solutions of
the sine-Gordon equation \eqref{eq:SG} of the approximate form
\begin{equation}
u(x,t)=U\left(\frac{\Phi(x,t)}{\epsilon}\right) + \bo(\epsilon)
\label{eq:modulationansatz}
\end{equation}
over space and time intervals of $\bo(1)$ length, where two essential changes
are made in the leading term:
\begin{enumerate}
\item The parameters $k$, $\omega$, and $\mathcal{E}$ are no longer taken as
constant, but
are allowed to depend on $(x,t)$ as long as the appropriate nonlinear dispersion
relation is satisfied pointwise, and
\item The phase $\Phi(x,t)$ is replaced with a general (nonlinear) function
of $(x,t)$ and the local wavenumber and frequency are derived therefrom by the
relations
\begin{equation}
k(x,t):=\frac{\partial\Phi}{\partial x}\quad\text{and}\quad
\omega(x,t):=-\frac{\partial\Phi}{\partial t}.
\label{eq:komegadefine}
\end{equation}
\end{enumerate}
By consistency, the definition \eqref{eq:komegadefine} imposes that
$k$ and $\omega$ are necessarily linked by 
\begin{equation}
\frac{\partial k}{\partial t} +\frac{\partial\omega}{\partial x}=0,
\label{eq:conswaves}
\end{equation}
an equation that expresses \emph{conservation of waves}.  It then
follows that for the error term in \eqref{eq:modulationansatz} to
remain $\bo(\epsilon)$ formally, one additional partial differential
equation on $k$, $\omega$, and $\mathcal{E}$, only two of which are independent
due to the nonlinear dispersion relation, is required to hold.  This
equation may be derived by many different methods.  Perhaps the most
direct in this context 
is to appeal to an averaged variational principle \cite{Whitham74}.
The sine-Gordon equation \eqref{eq:SG} is the Euler-Lagrange equation for the
variational principle:
\begin{equation}
\frac{\delta}{\delta u}\iint L[u]\,dx\,dt = 0 \;\implies\;
\epsilon^2u_{tt}-\epsilon^2u_{xx}+\sin(u)=0
\end{equation}
where the Lagrangian density is
\begin{equation}
L[u]:=\frac{1}{2}\epsilon^2u_t^2 - \left[\frac{1}{2}\epsilon^2u_x^2-\cos(u)\right]
\end{equation}
having the interpretation of the difference between kinetic and
potential energy densities.  The procedure is to substitute the exact
wavetrain into $L$, using the differential equation
\eqref{eq:integratedODE} satisfied by $U$ to simplify the resulting
expression:
\begin{equation}
L\left[U\left(\frac{kx-\omega t}{\epsilon}\right)\right]=\frac{1}{2}
(\omega^2-k^2)U'(\zeta)^2 +\cos(U(\zeta)) = 
2(\mathcal{E}+\cos(U(\zeta)))-\mathcal{E},
\end{equation}
where for the exact solution $\epsilon \zeta=kx-\omega t$.  This expression is
periodic in $\zeta$ with period $2\pi$, so one may define its period average
as
\begin{equation}
\langle L\rangle := \frac{1}{\pi}\int_{-\pi}^{+\pi}(\mathcal{E}+\cos(U(\zeta)))\,d\zeta - \mathcal{E}.
\end{equation}
An exact expression for $U(\zeta)$ is not necessary; using the
differential equation \eqref{eq:integratedODE}, a shifted version of
$U$ may be used as the integration variable although the details are
slightly different in the four cases; defining integrals
\begin{equation}
I_\librational(\mathcal{E}):=\frac{\sqrt{2}}{\pi}
\int_{-\cos^{-1}(\mathcal{E})}^{+\cos^{-1}(\mathcal{E})}\sqrt{\cos(\phi)-\mathcal{E}}\,
d\phi>0,\quad -1<\mathcal{E}<1,
\label{eq:Ilibrational}
\end{equation}
and
\begin{equation}
I_\rotational(\mathcal{E}):=\frac{1}{\sqrt{2}\pi}
\int_{-\pi}^\pi\sqrt{\cos(\phi)-\mathcal{E}}\,d\phi>0,\quad
\mathcal{E}<-1,
\label{eq:Irotational}
\end{equation}
and noting for future reference that
\begin{equation}
\text{$I_\librational''(\mathcal{E})>0$ for $-1<\mathcal{E}<1$ while
$I_\rotational''(\mathcal{E})<0$ for $\mathcal{E}<-1$,}
\label{eq:IBKdoubleprime}
\end{equation}
the result is that
\begin{equation}
\langle L\rangle = J(\mathcal{E})\mu\sqrt{|\omega^2-k^2|}-\mathcal{E}
\end{equation}
where $\mu=\mathrm{sgn}(\omega^2-k^2)$ distinguishes the superluminal and subluminal
cases, and where 
\begin{equation}
J(\mathcal{E}):=I_\librational(-\mu \mathcal{E}) \quad \text{or}\quad
J(\mathcal{E}):=I_\rotational(-\mu \mathcal{E})
\label{eq:Jdefine}
\end{equation}
depending on whether we are considering librational or rotational wavetrains, 
respectively.
One then substitutes $k=\theta_x$ and $\omega=-\theta_t$ and formulates
the \emph{averaged variational principle}:
\begin{equation}
\frac{\delta}{\delta \mathcal{E}}
\iint \langle L\rangle\,dx\,dt= 0\quad\text{and}\quad
\frac{\delta}{\delta\theta}\iint\langle L\rangle\,dx\,dt = 0.
\end{equation}
The first of these two equations reproduces in each case the
corresponding nonlinear dispersion relation
\eqref{eq:NLdisprel1}--\eqref{eq:NLdisprel4} in the form
\begin{equation}
\frac{\partial\langle L\rangle}{\partial \mathcal{E}}=
J'(\mathcal{E})\mu\sqrt{|\omega^2-k^2|}-1
=0.
\label{eq:NLdisprelcommon}
\end{equation}
The second is a first-order
partial differential equation:
\begin{equation}
\frac{\partial}{\partial t}\left[-\frac{\partial\langle L\rangle}{\partial\omega}\right] + \frac{\partial}{\partial x}
\left[\frac{\partial\langle L\rangle}{\partial k}\right]=0.
\label{eq:WhithamII}
\end{equation}
This equation together with \eqref{eq:conswaves} and the nonlinear
dispersion relation \eqref{eq:NLdisprelcommon} to eliminate one of the
three variables yields a closed system of equations to determine these
fields as functions of $(x,t)$.  From the exposition in
\cite{ScottCR76} one learns to appreciate the utility of taking $\mathcal{E}$
and the \emph{phase velocity} 
\begin{equation}
v_\mathrm{p}:=\frac{\omega}{k} 
\end{equation}
as the two unknowns,
and thus the calculations go as follows.  Clearly, one has
\begin{equation}
-\frac{\partial\langle L\rangle}{\partial \omega}=
-\frac{\omega J(\mathcal{E})}{\sqrt{|\omega^2-k^2|}}\quad\text{and}\quad
\frac{\partial\langle L\rangle}{\partial k}=
-\frac{kJ(\mathcal{E})}{\sqrt{|\omega^2-k^2|}}.
\end{equation}
Using $\omega=v_\mathrm{p}k$ together with the nonlinear dispersion relation
in the form \eqref{eq:NLdisprelcommon} shows that $k$ and $\omega$ may be eliminated in favor of $v_\mathrm{p}$ and $\mathcal{E}$:
\begin{equation}
k=\frac{\sigma\mu}{J'(\mathcal{E})\sqrt{|v_\mathrm{p}^2-1|}}
\quad\text{and}\quad
\omega=v_\mathrm{p}k=\frac{\sigma\mu v_\mathrm{p}}
{J'(\mathcal{E})\sqrt{|v_\mathrm{p}^2-1|}},
\end{equation}
where $\sigma=\pm 1$ is an arbitrary sign whose role is to select different
branches of the dispersion relation.  Therefore, the variational modulation
equation \eqref{eq:WhithamII} becomes
\begin{equation}
\frac{\partial}{\partial t}\left[
\frac{v_\mathrm{p}J(\mathcal{E})}{\sqrt{|v_\mathrm{p}^2-1|}}\right] +
\frac{\partial}{\partial x}\left[
\frac{J(\mathcal{E})}{\sqrt{|v_\mathrm{p}^2-1|}}\right]=0
\end{equation}
and the conservation of waves equation \eqref{eq:conswaves} becomes
\begin{equation}
\frac{\partial}{\partial t}
\left[\frac{1}{J'(\mathcal{E})\sqrt{|v_\mathrm{p}^2-1|}}\right] +
\frac{\partial}{\partial x}
\left[\frac{v_\mathrm{p}}{J'(\mathcal{E})\sqrt{|v_\mathrm{p}^2-1|}}\right]=0
\end{equation}
(the sign $\sigma=\pm 1$ drops out in each case).  An application of the chain
rule puts the system in the form
\begin{equation}
\frac{\partial}{\partial t}\begin{bmatrix}v_\mathrm{p}\\\mathcal{E}\end{bmatrix}
+\frac{1}{\mathcal{V}(v_\mathrm{p},\mathcal{E})}
\begin{bmatrix}
v_\mathrm{p}[J(\mathcal{E})J''(\mathcal{E})+J'(\mathcal{E})^2] & 
(1-v_\mathrm{p}^2)^2J'(\mathcal{E})J''(\mathcal{E})\\
-J(\mathcal{E})J'(\mathcal{E}) & 
v_\mathrm{p}[J(\mathcal{E})J''(\mathcal{E})+J'(\mathcal{E})^2]\end{bmatrix}
\frac{\partial}{\partial x}\begin{bmatrix} v_\mathrm{p}\\
\mathcal{E}\end{bmatrix}=
\mathbf{0}
\label{eq:Whithamsystem_general}
\end{equation}
where
\begin{equation}
\mathcal{V}(v_\mathrm{p},\mathcal{E}):=J(\mathcal{E})J''(\mathcal{E})+
v_\mathrm{p}^2J'(\mathcal{E})^2.
\end{equation}
Thus the dependence on the sign $\mu=\pm 1$ also disappears except
from within the definition \eqref{eq:Jdefine} of $J(\mathcal{E})$.
Actually, this system as written has an apparent singularity if $v_\mathrm{p}$
blows up as can happen in the superluminal cases 
when the wavenumber $k$ vanishes.  In these cases it is better to introduce
the \emph{reciprocal phase velocity} 
\begin{equation}
n_\mathrm{p}:=\frac{1}{v_\mathrm{p}}
\end{equation}
and then in the overlap region where neither $v_\mathrm{p}$ nor $n_\mathrm{p}$ 
vanishes, \eqref{eq:Whithamsystem_general} takes the form
\begin{equation}
\frac{\partial}{\partial t}\begin{bmatrix}n_\mathrm{p}\\\mathcal{E}\end{bmatrix}
+\frac{1}{\mathcal{N}(n_\mathrm{p},\mathcal{E})}
\begin{bmatrix}
n_\mathrm{p}[J(\mathcal{E})J''(\mathcal{E})+J'(\mathcal{E})^2] &
-(1-n_p^2)^2J'(\mathcal{E})J''(\mathcal{E})\\
J(\mathcal{E})J'(\mathcal{E}) & n_\mathrm{p}
[J(\mathcal{E})J''(\mathcal{E})+J'(\mathcal{E})^2]
\end{bmatrix}\frac{\partial}{\partial x}
\begin{bmatrix}n_\mathrm{p}\\\mathcal{E}\end{bmatrix}=\mathbf{0}
\label{eq:Whithamsystem_rewrite}
\end{equation}
where
\begin{equation}
\mathcal{N}(n_\mathrm{p},\mathcal{E}):=n_p^2J(\mathcal{E})J''(\mathcal{E}) +J'(\mathcal{E})^2.
\label{eq:calNdef}
\end{equation}
This latter form of the Whitham modulation equations has no apparent
singularity when $n_\mathrm{p}\to 0$ corresponding to $v_\mathrm{p}\to\infty$.

The characteristic velocities $c=c_j$, $j=0,1$, are the
eigenvalues of the coefficient matrix of $x$-derivatives and therefore
are the roots of the quadratic equation
\begin{equation}
\left(v_\mathrm{p}[J(\mathcal{E})J''(\mathcal{E})+J'(\mathcal{E})^2]-
\mathcal{V}(v_\mathrm{p},\mathcal{E})c\right)^2 +
(1-v_\mathrm{p}^2)^2J(\mathcal{E})J'(\mathcal{E})^2J''(\mathcal{E})=0
\end{equation}
or, equivalently as the coefficient matrices in \eqref{eq:Whithamsystem_general}
and \eqref{eq:Whithamsystem_rewrite} are similar,
\begin{equation}
(n_\mathrm{p}[J(\mathcal{E})J''(\mathcal{E})+J'(\mathcal{E})^2]-\mathcal{N}(n_\mathrm{p},
\mathcal{E})c)^2 + (1-n_\mathrm{p}^2)^2J(\mathcal{E})J'(\mathcal{E})^2
J''(\mathcal{E})=0.
\end{equation}
Since $\mathcal{V}(v_\mathrm{p},\mathcal{E})$ and $\mathcal{N}(n_\mathrm{p},\mathcal{E})$ are 
real the Whitham
modulation system in either form \eqref{eq:Whithamsystem_general} 
or \eqref{eq:Whithamsystem_rewrite} is hyperbolic
(corresponding to real and distinct characteristic velocities) if and
only if $J(\mathcal{E})J''(\mathcal{E})<0$.  Using \eqref{eq:Jdefine}
and \eqref{eq:IBKdoubleprime} then shows that the Whitham systems
governing modulations of rotational waves (both types, superluminal and
subluminal) are hyperbolic, while those governing librational waves 
(again, both types) are elliptic.  The Whitham modulation theory has been generalized
to handle modulated \emph{multiphase} waves \cite{ForestM83,ErcolaniFM84,ErcolaniFMM87}
having any number of $2\pi$-periodic phase variables; however the full implications of the
resulting modulation equations generalizing \eqref{eq:Whithamsystem_rewrite} have apparently only
been understood for waves with a maximum of two phases.

The Whitham modulation theory makes predictions of so-called
``modulational stability'' of wavetrains on the basis of whether the
quasilinear system of modulation equations is hyperbolic
(modulationally stable) or elliptic (modulationally unstable).  One
should think of modulational stability as linear (neutral) stability
of perturbations to the wavetrain that have similar characteristic
wavelengths and periods to the unperturbed exact wavetrain solution.
Thus, hyperbolicity of the modulation equations suggests the absence
of a ``slow'' sideband instability, but does not necessarily rule out
instabilities to perturbations with wavenumbers far from the
unperturbed wavenumber $k$ or even sideband instabilities with
exponential growth rates that are far from the unperturbed frequency
$\omega$.  Among the candidate wavetrain types for stability of the
linearized equation
\begin{equation}
\epsilon^2v_{tt}-\epsilon^2v_{xx}+\cos\left(U\left(\frac{kx-\omega t}{\epsilon}
\right)\right)v=0
\label{eq:linearizedSG}
\end{equation}
(in the sense of a $L^2(\mathbb{R})$ estimate on $v$ and $\epsilon
v_t$ that depends on initial data but is independent of $t$) there are
thus only the subluminal and superluminal rotational wavetrains.  It is not
difficult to believe that the superluminal rotational wavetrains are 
linearly
unstable, since in the limiting case of $k=0$ and $\mathcal{E}\downarrow 1$ one
obtains an orbit homoclinic to the exact constant solution $u(x,t)\equiv (2m+1)\pi$,
$m\in\mathbb{Z}$ which is obviously unstable to small spatially
constant perturbations.  Indeed, such perturbations cause all of the
vertical pendula to ``drop'' simultaneously; the growth rate of the
perturbation is clearly large compared to the zero frequency of the
unperturbed solution explaining why the instability is not captured by
Whitham theory.  The definitive statement in the literature
\cite{Scott69} is that the subluminal rotational wavetrains are indeed 
linearly
stable (and the only stable type among the four types of traveling
waves), although we have not been able to verify the line of argument
in all details (it is not clear to us from the proof that the stable
solutions obtained in the case of subluminal rotational wavetrains form a 
basis of
$L^2(\mathbb{R})$, or that the exponentially growing solutions
obtained in the other three cases are relevant as they appear to be
unbounded in $x$).

In this paper, we will analyze the Cauchy problem for \eqref{eq:SG}
with initial data $F$ and $G$ independent of $\epsilon$, in the
semiclassical limit $\epsilon\to 0$.  We will show that for a general
class of ``pure impulse'' initial data, most of the real $x$-axis is
occupied for small time (independent of $\epsilon$) by modulated
superluminal wavetrains of either rotational or librational types. This result shows the relevance of the Whitham modulation
theory even in some cases when it results in elliptic modulation
equations.
(In a forthcoming paper \cite{BuckinghamMseparatrix} we will show that the
remaining part of the $x$-axis is occupied by more complicated
oscillations that nonetheless have a certain universal form for $t$
small.)

\subsection{Pure impulse initial data for the sine-Gordon equation.  Connection to the Zakharov-Shabat scattering problem}
\label{sec:pure-impulse}
Equation \eqref{eq:SG} is the compatibility condition for the Lax pair
\begin{equation}
4i\epsilon \mathbf{v}_x = \begin{bmatrix}
\displaystyle 4E(w)-\frac{i}{\sqrt{-w}}(1-\cos(u)) & \displaystyle \frac{i}{\sqrt{-w}}\sin(u)-i\epsilon(u_x+u_t)\\
\displaystyle \frac{i}{\sqrt{-w}}\sin(u)+i\epsilon(u_x+u_t) & \displaystyle -4E(w)+\frac{i}{\sqrt{-w}}(1-\cos(u))\end{bmatrix}\mathbf{v},
\label{eq:FT}
\end{equation}
\begin{equation}
4i\epsilon\mathbf{v}_t = \begin{bmatrix}
\displaystyle 4D(w)+\frac{i}{\sqrt{-w}}(1-\cos(u)) & \displaystyle -\frac{i}{\sqrt{-w}}\sin(u)-i\epsilon(u_x+u_t)\\
\displaystyle -\frac{i}{\sqrt{-w}}\sin(u)+i\epsilon(u_x+u_t) & \displaystyle -4D(w)-\frac{i}{\sqrt{-w}}(1-\cos(u))\end{bmatrix}\mathbf{v},
\end{equation}
where
\begin{equation}
E(w):=\frac{i}{4}\left[\sqrt{-w}+\frac{1}{\sqrt{-w}}\right]\quad\text{and}\quad
D(w):=\frac{i}{4}\left[\sqrt{-w}-\frac{1}{\sqrt{-w}}\right],
\label{eq:DE}
\end{equation}
and $w\in\mathbb{C}\setminus\mathbb{R}_+$ is the spectral parameter.
\emph{Here and throughout this paper, the radical refers to the principal branch of the square root.}
Analysis of the Cauchy problem for \eqref{eq:SG} posed with initial
data \eqref{eq:IC} may be carried out in some detail by means of the
inverse scattering transform based on the differential equation
\eqref{eq:FT}, the so-called Faddeev-Takhtajan eigenvalue problem.
A self-contained account of this analysis can be found in
our paper \cite{BuckinghamM08}, where the spectral parameter $z=i\sqrt{-w}$ is 
used; the utility of $w$ is related to an even symmetry of the spectrum
in the $z$-plane.

By \emph{pure impulse} initial data we simply mean data for which the
initial displacement $F(x)$ vanishes identically.  
An elementary observation is that if $F(x)\equiv 0$,
then for $t=0$ the Faddeev-Takhtajan eigenvalue
problem reduces to the Zakharov-Shabat eigenvalue problem:
\begin{equation}
\epsilon\mathbf{v}_x = \begin{bmatrix}
-i\lambda & \psi(x) \\
-\psi(x)^* & i\lambda\end{bmatrix}\mathbf{v},\quad \psi(x):=-\frac{1}{4}G(x),\quad \lambda:=E(w).
\label{eq:ZS}
\end{equation}
This is the eigenvalue problem in the Lax pair for the sine-Gordon 
equation in characteristic coordinates and the cubic focusing nonlinear 
Schr\"odinger equations.  Although when viewed as an eigenvalue problem 
of the form
$\mathcal{L}\mathbf{v}=\lambda\mathbf{v}$ the operator $\mathcal{L}$
is nonselfadjoint, several useful facts are known about the spectrum
of $\mathcal{L}$ in the case relevant here that $\psi(x)$ is real.  It
is easy to see that for real $\psi$ the discrete spectrum comes in
quartets symmetric with respect to the involutions
$\lambda\mapsto\lambda^*$ and $\lambda\mapsto -\lambda$.  Moreover,
Klaus and Shaw \cite{KlausS02} have proved that if $\psi\in
L^1(\mathbb{R})\cap C^1(\mathbb{R})$ is real, of one sign, and has a
single critical point (so that the graph of $|\psi(x)|$ is ``bell-shaped''), then the discrete spectrum is purely imaginary
and nondegenerate.

In the context of pure impulse initial data for which $G$ is a
nonpositive (without loss of generality) function of Klaus-Shaw type, the
necessarily purely imaginary eigenvalues $\lambda$ may be approximated
by a formal WKB method applicable when $\epsilon\ll 1$.  The result of
this analysis is that with the Klaus-Shaw function $G(x)$ one
associates the WKB phase integral
\begin{equation}
\Psi(\lambda):=
\frac{1}{4}\int_{x_-(\lambda)}^{x_+(\lambda)}\sqrt{G(s)^2+16\lambda^2}\,ds, 
\quad 0<y:=-i\lambda<\max_{x\in\mathbb{R}}\left(-\frac{1}{4}G(x)
\right),
\label{eq:WKBphase}
\end{equation}
where $x_-(\lambda)<x_+(\lambda)$ are the two roots of $G(s)^2+16\lambda^2$
when $\lambda$ is as indicated.
Then one defines approximate eigenvalues $\lambda^0_k$ by the
Bohr-Sommerfeld quantization rule
\begin{equation}
\Psi(\lambda^0_k)=\pi\epsilon\left(k+\frac{1}{2}\right),\quad k=0,1,2,\dots,N(\epsilon)-1,
\label{eq:BohrSommerfeld}
\end{equation}
where the asymptotic number of eigenvalues on the positive imaginary axis is
\begin{equation}
N(\epsilon)=\left\lfloor\frac{1}{2}+\frac{1}{4\pi\epsilon}
\|G\|_1\right\rfloor,
\label{eq:numberofeigenvalues}
\end{equation}
where $\|\cdot\|_1$ denotes the standard $L^1(\mathbb{R})$ norm.  To
each simple eigenvalue $\lambda$ of the Zakharov-Shabat eigenvalue
problem \eqref{eq:ZS} there corresponds a \emph{proportionality
  constant} $\gamma$ relating the solution having normalized decaying
asymptotics as $x\to -\infty$ with the solution having normalized
decaying asymptotics as $x\to +\infty$.  If $G$ is a real even
function of $x$, then one can show by symmetry that $\gamma=\pm 1$,
and WKB theory for Klaus-Shaw potentials suggests that the
proportionality constants alternate in sign along the imaginary axis.
Thus for even $G$, to the approximate eigenvalue $\lambda_k^0$ we
associate the approximate proportionality constant $\gamma_k^0$ given
by
\begin{equation}
\gamma_k^0:=(-1)^{k+1},\quad k=0,1,2,\dots,N(\epsilon)-1,\quad
\text{for even $G$.}
\end{equation}
The final ingredient of the scattering data for $G$ in the
Zakharov-Shabat problem \eqref{eq:ZS} is the reflection coefficient
defined for real $\lambda$.  According to the WKB approximation, the
reflection coefficient is small pointwise for $\lambda\neq 0$ when
$\epsilon\ll 1$.  

In this paper, we study even, pure impulse initial data of
Klaus-Shaw type for the sine-Gordon equation \eqref{eq:SG}.
Thus we assume
\begin{assume}
The initial conditions \eqref{eq:IC} for \eqref{eq:SG} satisfy
$F(x)\equiv 0$.
\label{assume:pureimpulse}
\end{assume}
\begin{assume}
  In the initial condition \eqref{eq:IC} for $\epsilon u_t$, the
  function $G(x)$ is a nonpositive function of Klaus-Shaw type, that
  is, $G\in L^1(\mathbb{R})\cap C^1(\mathbb{R})$ and $G$ has a unique
  local (and global) minimum.
\label{assume:KlausShaw}
\end{assume}
\begin{assume}
The function $G$ is even in $x$:  $G(-x)=G(x)$, placing the
unique minimum of $G$ at $x=0$.  
\label{assume:evenness}
\end{assume}
Note that under Assumptions~\ref{assume:KlausShaw} and
\ref{assume:evenness}, $G$ restricted to $\mathbb{R}_+$ has a unique
inverse function $G^{-1}:(G(0),0)\to\mathbb{R}_+$, and in terms of it we
may rewrite the WKB phase integral in the form
\begin{equation}
\Psi(\lambda)=\frac{1}{2}\int_0^{G^{-1}(-v)}\sqrt{G(s)^2-v^2}\,ds,\quad
\lambda=\frac{iv}{4},\quad 0<v<-G(0).
\label{eq:WKBphaserewrite}
\end{equation}
Thus, $\Psi(iv/4)$ defined on $(0,-G(0))$ is an Abel-type integral
transform of the nondecreasing function $G(x)<0$ defined on
$\mathbb{R}_+$.
\begin{proposition}
On its range, the inverse of the transform \eqref{eq:WKBphaserewrite} is
given by the formula
\begin{equation}
G^{-1}(w) = -\frac{4}{\pi}\int_{-w}^{-G(0)}\frac{\varphi(v)\,dv}
{\sqrt{v^2-w^2}},\quad
G(0)<w<0,
\label{eq:inversetransform}
\end{equation}
where
\begin{equation}
\varphi(v):=\frac{d}{dv}\Psi(\lambda),\quad\lambda=\frac{iv}{4}.
\end{equation}
\label{prop:AbelInverse}
\end{proposition}
The proof of this proposition is a rather straightforward application of Fubini's Theorem and is given in Appendix~\ref{app:initialdata}.

We will require that the WKB phase integral have certain analyticity 
properties to be outlined in Proposition~\ref{prop:theta0} below.  We now make an 
assumption on $G$ that will be sufficient to establish 
Proposition~\ref{prop:theta0} and that can easily be checked 
for a given $G$:
\begin{assume}
The function $G$ is strictly increasing and real-analytic at
  each $x>0$, and the positive and real-analytic function
\begin{equation}
\mathscr{G}(m):=\frac{\sqrt{m}\sqrt{G(0)^2-m}}{2G'(G^{-1}(-\sqrt{m}))},\quad 0<m<G(0)^2
\label{eq:hmdef}
\end{equation}
can be analytically continued to neighborhoods of $m=0$ and $m=G(0)^2$, with
$\mathscr{G}(0)>0$ and $\mathscr{G}(G(0)^2)>0$.  
\label{assume:h}
\end{assume}
We point out that the class of functions $\mathscr{G}(m)$ 
satisfying Assumption~\ref{assume:h} obviously parametrizes a
corresponding class of admissible functions $G(x)$ by simply viewing \eqref{eq:hmdef} 
as an equation to be solved for 
$x=G^{-1}$ given $\mathscr{G}$.  The solution is:
\begin{equation}
x=\int_{G^2}^{G(0)^2}\frac{\mathscr{G}(m)\,dm}{m\sqrt{G(0)^2-m}}.
\label{eq:inversefunctionhgeneral}
\end{equation}
For example, the function  $\mathscr{G}(m)\equiv C>0$ clearly satisfies the analyticity and positivity conditions on $\mathscr{G}$ listed in
Assumption~\ref{assume:h}, and in this 
case the integral in \eqref{eq:inversefunctionhgeneral} can be evaluated
in terms of elementary functions and the resulting function $x=G^{-1}(G)$ can
be inverted to yield
\begin{equation}
G=G(0)\,\mathrm{sech}\left(\frac{G(0)}{2C}x\right),
\end{equation}
perhaps the simplest example of an admissible initial condition.

\begin{proposition}
Suppose that Assumptions \ref{assume:KlausShaw}--\ref{assume:h} hold.  Then the function 
$\Psi(\lambda)$ defined by \eqref{eq:WKBphaserewrite} for $\lambda=iv/4$ and
$0<v<-G(0)$ is positive and strictly decreasing (to zero) in $v$.
Furthermore, $\Psi$ is
real-analytic for $0<v<-G(0)$ and 
has an analytic continuation to neighborhoods of $v=0$ and $v=-G(0)$, for
which 
\begin{equation}
\Psi(\lambda)=0 \quad\text{and}\quad
\frac{d}{dv}\Psi(\lambda)<0,\quad\text{for $\lambda=-iG(0)/4$,}
\label{eq:theta0neartop}
\end{equation}
and, for some $\delta>0$,
\begin{equation}
\Psi(\lambda)=\frac{1}{4}\|G\|_1 + i\alpha\lambda +
\sum_{n=1}^\infty \beta_n\lambda^{2n},\quad |\lambda|<\delta
\label{eq:theta0series0}
\end{equation}
where $\alpha>0$ and $\beta_n\in\mathbb{R}$ for all $n$.  
\label{prop:theta0}
\end{proposition}
We provide the proof of this statement in Appendix~\ref{app:initialdata}.
In particular, Proposition~\ref{prop:theta0} guarantees the existence of
a simply-connected open set $\Xi\subset\mathbb{C}$ containing the closed
imaginary interval $0\le -4i\lambda\le -G(0)$ in which $\Psi(\lambda)$
may be considered as a holomorphic function of $\lambda$ whose restriction
to that interval is a real-valued function given by \eqref{eq:WKBphase}.
By Schwartz reflection we therefore will have
\begin{equation}
\Psi(-\lambda^*)=\Psi(\lambda)^*,\quad \lambda\in \Xi,\quad
-\lambda^*\in \Xi,
\label{eq:Schwartztheta0}
\end{equation}
which also shows that without loss of generality we may simply take
$\Xi$ to be symmetric with respect to reflection through the imaginary
axis.  By the strict monotonicity and reality of $\Psi(\lambda)$
for $0\le -4i\lambda\le -G(0)$ it follows from the Cauchy-Riemann
equations that there exists some positive number $\delta_1<\delta$ such that for
$\lambda$ in the open rectangle $D_+:=\{\lambda\in\Xi,\; 0<
\Re\{\lambda\}<\delta_1,\; 0<\Im\{\lambda\}<-G(0)/4\}$ the inequality
$\Im\{\Psi(\lambda)\}>0$ holds, with the strict inequality failing
only as $\lambda$ approaches the imaginary axis.  According to
\eqref{eq:Schwartztheta0}, if $\lambda\in D_-:=-D_+^*$, then
$\Im\{\Psi(\lambda)\}< 0$.

Our analysis will require an assumption about $\epsilon>0$:
\begin{assume}
The small number $\epsilon$ lies in the infinite sequence
\begin{equation}
\epsilon=\epsilon_N:=\frac{\|G\|_1}{4\pi N},\quad
N=1,2,3,\dots.
\label{eq:epsilonNgeneral}
\end{equation}
For such $\epsilon$ we have from \eqref{eq:numberofeigenvalues} that
$N(\epsilon)=N$.
\label{assume:epsilonNgeneral}
\end{assume}
Clearly, according to \eqref{eq:theta0series0},
Assumption~\ref{assume:epsilonNgeneral} implies that 
\begin{equation}
\frac{\Psi(0)}{\epsilon_N}=\pi N, \quad N=1,2,3,\dots.
\end{equation}
Also, according to
WKB theory, another implication of
Assumption~\ref{assume:epsilonNgeneral} is that the reflection
coefficient is \emph{uniformly} small for $\lambda\in\mathbb{R}$,
\emph{i.e.}  the choice \eqref{eq:epsilonNgeneral} makes the reflection
coefficient negligible in a neighborhood of $\lambda=0$.

In fact, our strategy will be to replace the scattering data
corresponding to initial data of the above type with its WKB
approximation, admittedly an \emph{ad hoc} step, and then to carry out
rigorous analysis of the inverse-scattering problem in the limit
$\epsilon\to 0$.  The sequence of exact solutions of \eqref{eq:SG}
generated by the spectral approximation procedure and indexed by
$N=1,2,3,\dots$ is an example of a semiclassical soliton ensemble in
the sense of \cite{Miller02}.  In the context of the sine-Gordon
equation and its application to superconducting Josephson junctions we
will call this sequence $\{u_N(x,t)\}$ of exact solutions the 
\emph{fluxon condensate} associated with the impulse profile $G(x)$.  The
accuracy of this procedure for studying the Cauchy problem is
suggested by the fact that the fluxon condensate recovers the initial
data at $t=0$ to within an error of $\bo(\epsilon)$ (see Corollary~\ref{corr:ICapproximate} below).
Also, there exist special cases for which the fluxon condensate
represents the \emph{exact} solution of the Cauchy problem when
$\epsilon$ lies in the sequence $\epsilon=\epsilon_N$ (see
\eqref{eq:epsilonNgeneral}), giving further justification to the
procedure.

\subsection{Exact solutions.  Impulse threshold for generation of rotational waves}
\label{sec:exactsolutions}
If furthermore we suppose that the sine-Gordon system \eqref{eq:SG} on
$\mathbb{R}$ is set into motion at $t=0$ by pure impulse initial data
of the special form
\begin{equation}
F(x)\equiv 0\quad\text{and}\quad G(x)=-4A\,\mathrm{sech}(x)
\label{eq:specialdata}
\end{equation}
for some $A>0$, the Zakharov-Shabat problem reduces further to a
special case that was studied by Satsuma and Yajima \cite{SatsumaY74}.
From their work it follows that if $\epsilon=\epsilon_N:=A/N$ for
$N=1,2,3,\dots$, the reflection coefficient for \eqref{eq:ZS} vanishes
identically as a function of $\lambda$, and the eigenvalues in the
upper half-plane are the purely imaginary numbers
$\lambda=iA-i(k+\tfrac{1}{2})\epsilon_N$, for $k=0,1,\dots,N-1$.  The
auxiliary scattering data consist of the proportionality constants
linking eigenfunctions with prescribed decay as $x\to -\infty$ with
those having prescribed decay as $x\to +\infty$: these are simply
alternating signs $(-1)^{k+1}$.  It is easily confirmed that these
scattering data correspond exactly to the WKB approximation described
above when $G$ is defined as above; in particular, the phase integral
\eqref{eq:WKBphase} evaluates to
\begin{equation}
\Psi(\lambda)=i\pi\lambda+\pi A,\quad\text{if $G(x)=-4A\,\mathrm{sech}(x)$.}
\label{eq:specialtheta0}
\end{equation}

In such a reflectionless situation, it becomes possible to solve the
inverse-scattering problem by finite-dimensional linear algebra, and
thus we obtain an exact solution $u(x,t)$ of \eqref{eq:SG} with
$\epsilon=\epsilon_N$ and $\epsilon$-independent initial data
\eqref{eq:specialdata} for each positive integer $N$.  We will give
more details about this procedure later (see the final paragraph of 
\S \ref{section-formulation}), 
but for now we discuss the
results of an empirical study of these exact solutions.

Figure~\ref{fig:Ap25} shows plots of the exact solution of the Cauchy problem for
\eqref{eq:SG} with initial conditions given by \eqref{eq:specialdata} with $A=1/4$.
\begin{figure}[h]
\begin{center}
\includegraphics[width=0.3\linewidth]{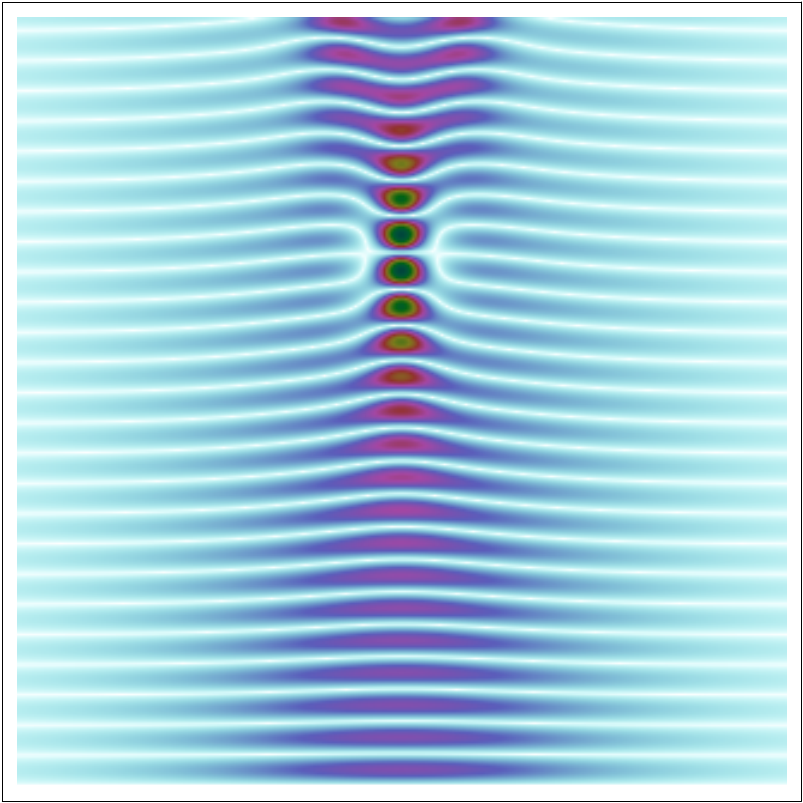}
\includegraphics[width=0.3\linewidth]{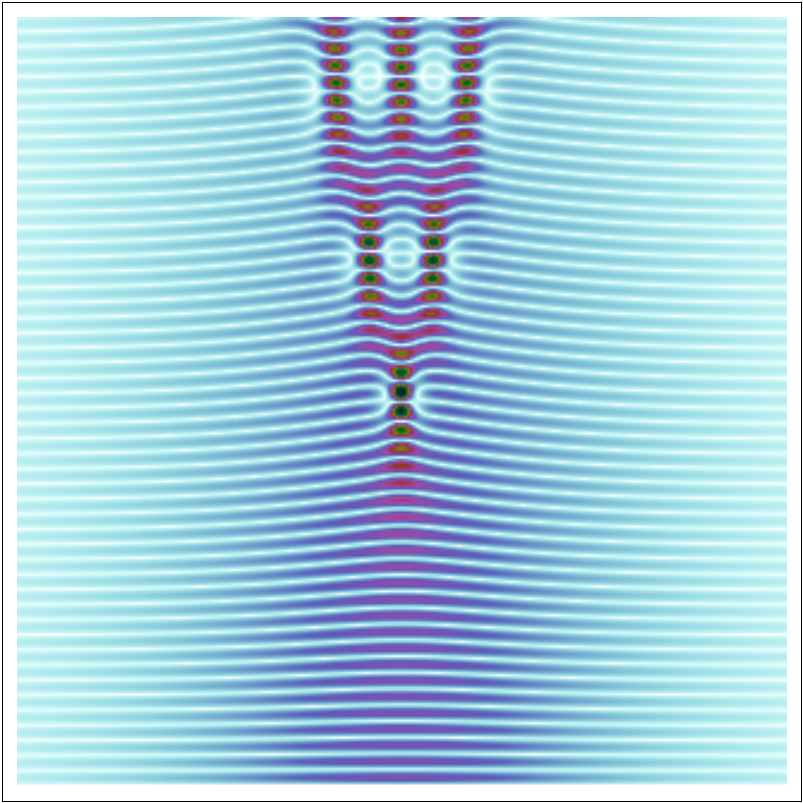}
\includegraphics[width=0.3\linewidth]{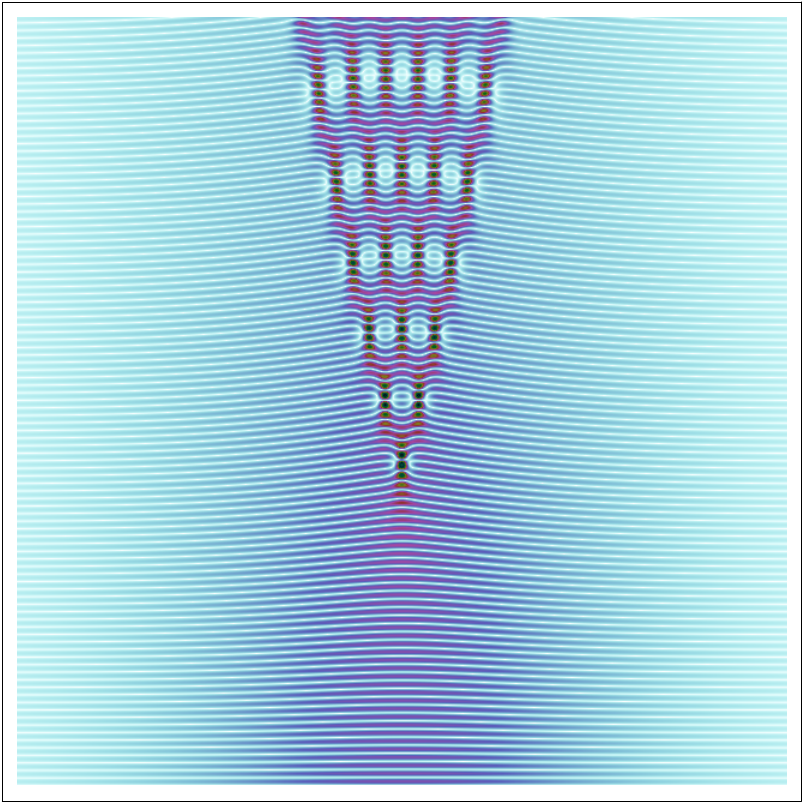}\\
\includegraphics{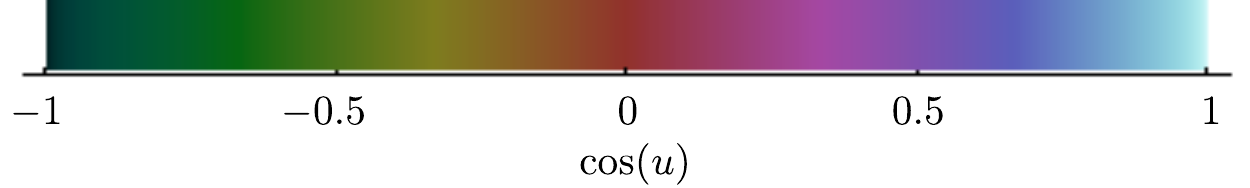}
\end{center}
\caption{\emph{Plots of $\cos(u)$ with $A=1/4$ over the $(x,t)$-plane.  The horizontal axis is $-2.5<x<2.5$ and the vertical axis is $0<t<5$.  Left:  $N=4$.  Center:  $N=8$.  Right:  $N=16$.
}}
\label{fig:Ap25}
\end{figure}
These plots show that the initial impulse sets the pendula into nearly synchronous librational
motion of frequency proportional to $N$ (inversely proportional to $\epsilon$).  However,
an instability seems to appear in the modulational pattern, leading to a kind of focusing of wave
energy near $x=0$.  In a region of the $(x,t)$-plane that seems to become more well-defined as
$N$ increases, the focused waves take on a different character; in particular, the synchrony
of nearby pendula is lost as spatial structures with wavelengths inversely proportional to $N$
spontaneously appear.  Nonetheless, the oscillations present after the focusing event are still
fairly regular and result in an approximately quasiperiodic spatiotemporal pattern.  This type of
dynamics is qualitatively similar to what is known to occur in the semiclassical limit of the
focusing nonlinear Schr\"odinger equation \cite{MillerK98,KamvissisMM03,TovbisVZ04,LyngM07}, another problem that has elliptic modulation equations
as is expected here before the focusing due to the apparent librational motion.

In Figure~\ref{fig:Ap75} we give plots analogous to those in Figure~\ref{fig:Ap25} but now
we increase the amplitude of the impulse by choosing $A=3/4$.
\begin{figure}[h]
\begin{center}
\includegraphics[width=0.3\linewidth]{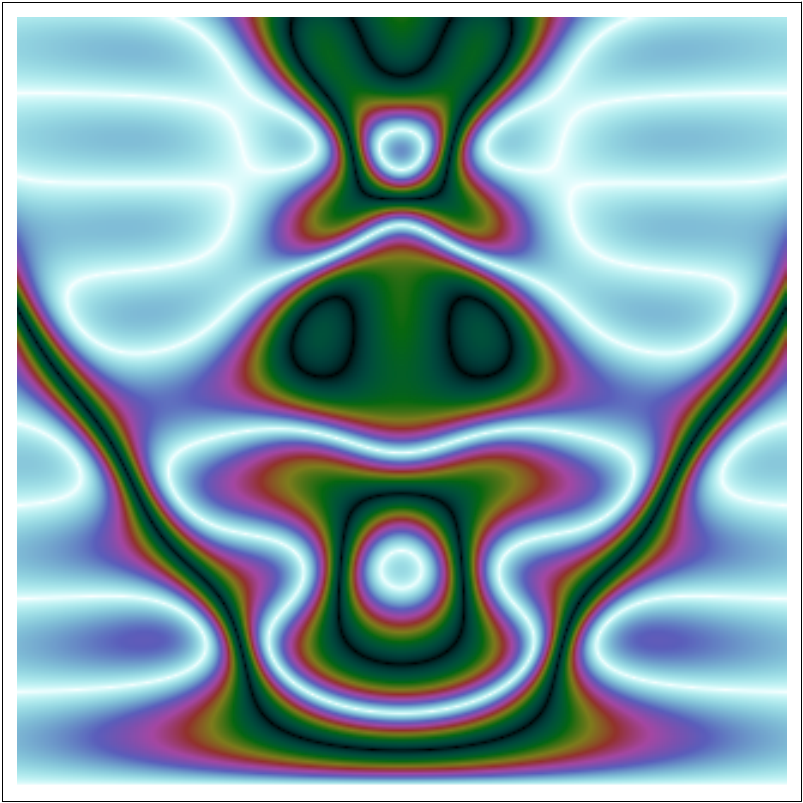}
\includegraphics[width=0.3\linewidth]{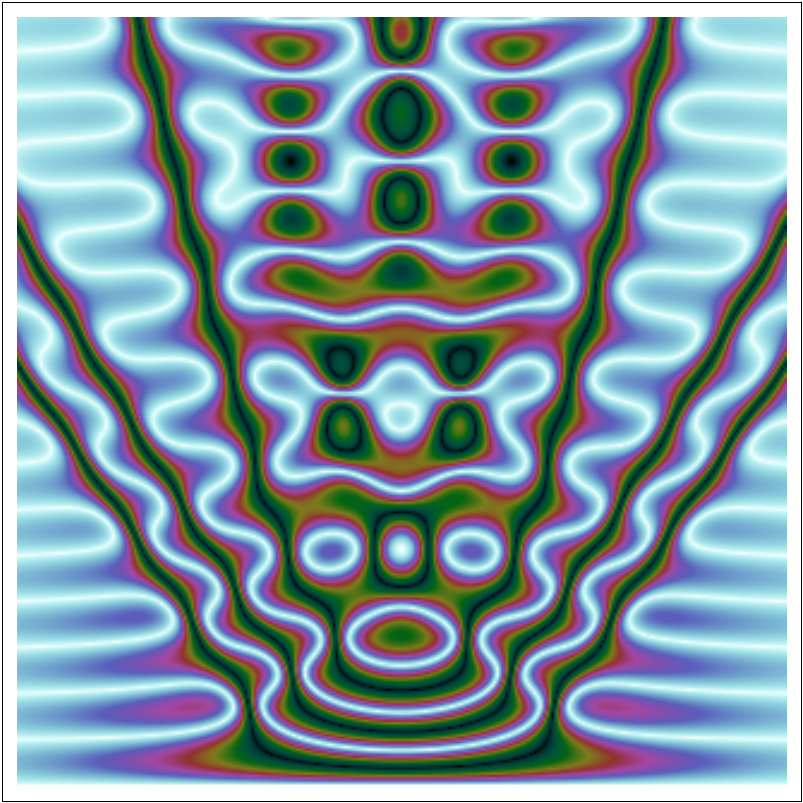}
\includegraphics[width=0.3\linewidth]{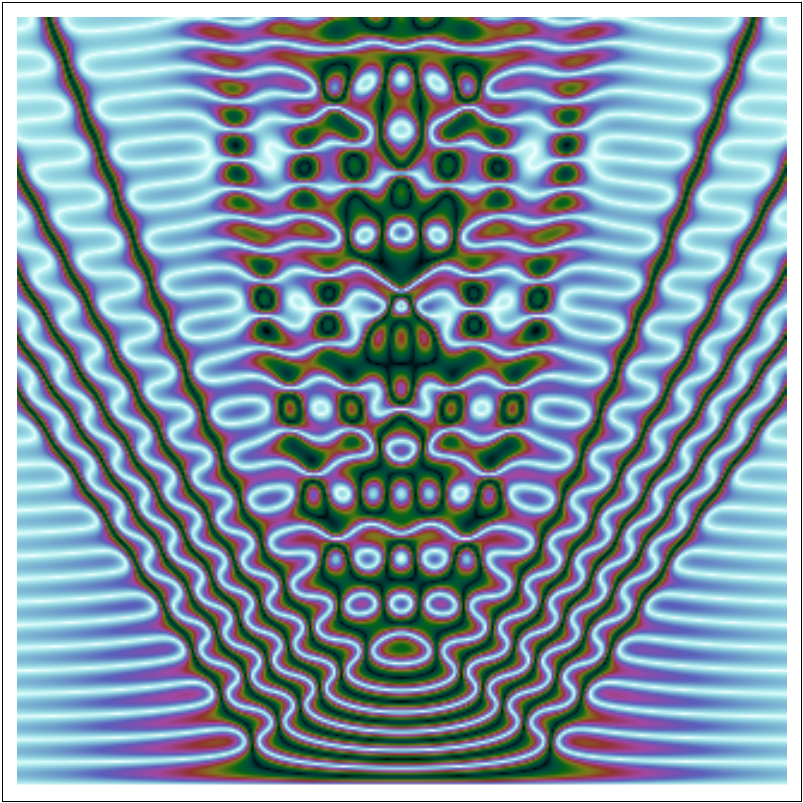}\\
\includegraphics{colormap.pdf}
\end{center}
\caption{\emph{Same as in Figure~\ref{fig:Ap25} but now $A=3/4$.}}
\label{fig:Ap75}
\end{figure}
The dynamics are evidently quite different: one clearly can identify three types of behavior near
$t=0$:
\begin{itemize}
\item Nearly synchronous librational motion of the pendula is apparent for large $|x|$, where the initial impulse is weak.
\item Nearly synchronous rotational motion of the pendula is apparent for small $|x|$, where the initial impulse is strongest.
\item Strongly asynchronous motion of undetermined type is apparent near two transitional values
of $x$.  These transitional points appear to shed kinks and antikinks.
\end{itemize}

These plots therefore suggest that some sort of transition in the dynamics
occurs when the maximum amplitude of the impulse, $A$, exceeds some
threshold value between $A=1/4$ and $A=3/4$.  In our paper
\cite{BuckinghamM08} we observed that in similar families of exact
solutions an analogous transition occurs if the initial data $(F,G)$ when viewed as
a curve (parametrized by $x$) in the phase portrait of the simple pendulum equation 
$\epsilon^2u_{tt} +\sin(u)=0$ crosses the separatrix (and the transition appears to occur near the
specific $x$-values where crossings occur).  Thus one expects that the existence of
$x\in\mathbb{R}$ where
\begin{equation}
(1-\cos(F(x))) + \frac{1}{2}G(x)^2=2
\end{equation}
may lead to a different kind of dynamics in a small time interval near
$t=0$ of length independent of $\epsilon$.  This equation has real solutions when $F(x)$ and $G(x)$ 
are given by \eqref{eq:specialdata} if $A>1/2$.  From the plots in Figure~\ref{fig:Ap75} it is clear
that when $A>1/2$ there is a symmetrical interval around $x=0$ in
which the impulse is sufficiently large to cause rotation of the angle
$u$ outside of the fundamental range $-\pi<u<\pi$, which leads to an
emission of kinks carrying positive and negative topological charge in
opposite directions from the ends of this (shrinking in time)
interval, which bound a triangular region in the $(x,t)$-plane.  
This region therefore
appears to contain a modulated superluminal rotational wavetrain, 
while for $|x|$
sufficiently large one observes what appear to be modulated superluminal
librational waves.

The goal of this paper is to show that this type of structure is
universal for pure-impulse fluxon condensates with sufficient impulse present
at $t=0$.
Thus we further impose
\begin{assume}
The function $G(x)$ satisfies $G(0)<-2$.
\label{assume:rotational}
\end{assume}
By Assumptions~\ref{assume:KlausShaw}, \ref{assume:evenness}, and \ref{assume:rotational}, there exists a positive number $x_\mathrm{crit}>0$ defined by 
\begin{equation}
x_\mathrm{crit}:=G^{-1}(-2),
\label{eq:xcrit}
\end{equation}
and based upon the above heuristic discussion we may expect the dynamics of pure-impulse fluxon condensates to be of a different character
for $|x|<x_\mathrm{crit}$ than for $|x|>x_\mathrm{crit}$.
  
\subsection{Statement of results}
\label{sec:results}
Our results concern the asymptotic behavior, in the limit $N\uparrow\infty$ equivalent to $\epsilon_N\downarrow 0$, of the functions $u_N(x,t)$ making up the fluxon condensate
associated with the pure-impulse initial condition of impulse
profile $G(\cdot)$.  As mentioned above, for fixed $N$, $u_N(x,t)$  is not exactly the solution of
the Cauchy initial-value problem with the corresponding initial data (although it is an exact solution of \eqref{eq:SG}), and
the proper definition will be given below in Definition~\ref{def:condensate}.
The statements below concern two regions of the $(x,t)$-plane depending on $G(\cdot)$
but not on $\epsilon=\epsilon_N$.  The region $S_\librational$ is specified in terms of a continuous time-horizon function
$T_\librational (x)>0$ defined for $|x|>x_\mathrm{crit}$ with $\lim_{|x|\downarrow x_\mathrm{crit}}T_\librational (x)=0$; then $(x,t)\in S_\librational$ if and only if $|x|>x_\mathrm{crit}$ and
$|t|<T_\librational (x)$.  The region $S_\rotational$ is specified in terms of a continuous time-horizon function $T_\rotational (x)>0$ defined for $|x|<x_\mathrm{crit}$ with $\lim_{|x|\uparrow x_\mathrm{crit}}T_\rotational (x)=0$, as well as two curves $t=t_\pm(x)$ with $t_\pm(0)=0$,  $t_+'(x)>0$, and $t_-'(x)<0$.  Then $(x,t)\in S_\rotational$ if and only if $|x|<x_\mathrm{crit}$, $|t|<T_\rotational (x)$, and 
$t\neq t_\pm(x)$.  See Figure~\ref{fig:BKregions}.
\begin{figure}[h]
\begin{center}
\includegraphics{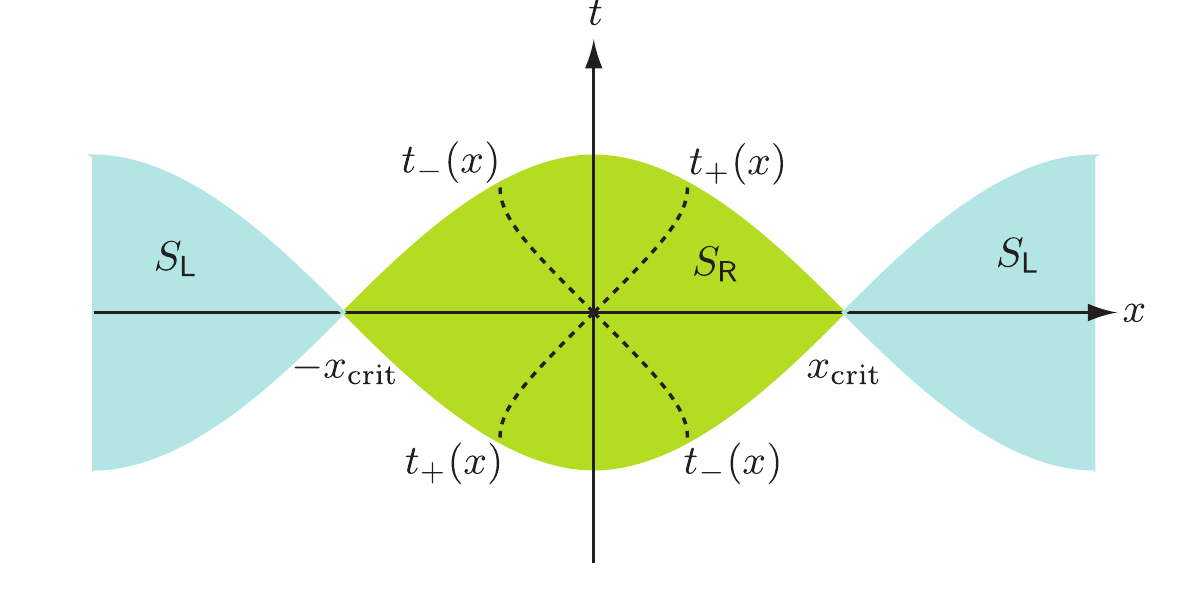}
\end{center}
\caption{\emph{The regions $S_\librational$\and $S_\rotational$.} }
\label{fig:BKregions}
\end{figure}

Let $K(\cdot)$ denote the complete elliptic integral of the first kind:
\begin{equation}
K(m):=\int_0^1\frac{ds}{\sqrt{(1-s^2)(1-ms^2)}},\quad 0<m<1.
\label{eq:ellipticKdef}
\end{equation}
We will prove the following two results under Assumptions~\ref{assume:pureimpulse}--\ref{assume:rotational},
and also under the more technical Assumption~\ref{assume:nodoublegeneral} to be presented shortly.
The asymptotic formulae for the fluxon condensate are in fact quite explicit in terms of the Jacobi elliptic functions $\mathrm{sn}(z;m)$, $\mathrm{cn}(z;m)$, and $\mathrm{dn}(z;m)$, for which we cite
the text by Akhiezer \cite{Akhiezer} as a reference.  
\begin{theorem}[Small-Time Librational Asymptotics]
There exist well-defined differentiable functions $n_\mathrm{p}:S_\librational\to
(-1,1)$ and $\mathcal{E}:S_\librational\to(-1,1)$ satisfying the initial conditions
$n_\mathrm{p}(x,0)=0$ and $\mathcal{E}(x,0)=\tfrac{1}{2}G(x)^2-1$ and
the elliptic Whitham 
system \eqref{eq:Whithamsystem_rewrite} where $J(\mathcal{E}):=I_\librational (-\mathcal{E})$
and $I_\librational(\cdot)$ is given by \eqref{eq:Ilibrational}, as well as the inequality
\begin{equation}
x\frac{\partial n_\mathrm{p}}{\partial t}(x,0)<0,\quad |x|>x_\mathrm{crit}.
\label{eq:Librationalnpineq}
\end{equation}
Defining an elliptic parameter $m=m(x,t)$ by
\begin{equation}
m(x,t)=m_\librational (x,t):=\frac{1+\mathcal{E}(x,t)}{2}\in (0,1),\quad (x,t)\in S_\librational,
\label{eq:mlibrationalE}
\end{equation}
and a real phase $\Phi(x,t)$ by 
\begin{equation}
\Phi(x,t):=-\int_0^t\omega(x,t')\,dt',
\end{equation}
where
\begin{equation}
\omega(x,t):=-\frac{\pi}{2K(m(x,t))}\frac{1}{\sqrt{1-n_\mathrm{p}(x,t)^2}},
\end{equation}
the following asymptotic formulae hold pointwise for $(x,t)\in S_\librational$:
\begin{equation}
\begin{split}
\cos\left(\frac{1}{2}u_N(x,t)\right)&=\mathrm{dn}\left
(\frac{2\Phi(x,t)K(m(x,t))}{\pi\epsilon_N};m(x,t)\right) + \bo(\epsilon_N)\\
\sin\left(\frac{1}{2}u_N(x,t)\right) &= -\sqrt{m(x,t)}\,\mathrm{sn}
\left(\frac{2\Phi(x,t) K(m(x,t))}{\pi\epsilon_N};m(x,t)\right) + 
\bo(\epsilon_N)\\
\epsilon_N \frac{\partial u_N}{\partial t}(x,t)&=
-\frac{4K(m(x,t))}{\pi}\frac{\partial\Phi}{\partial t}
\sqrt{m(x,t)}\,\mathrm{cn}\left(\frac{2\Phi(x,t)K(m(x,t))}{\pi\epsilon_N};m(x,t)
\right) +
\bo(\epsilon_N).
\end{split}
\label{eq:Basymptoticresults}
\end{equation}
Moreover, the error terms are uniform for $(x,t)$ in compact subsets of 
$S_\librational$.  The phase $\Phi(x,t)$ also satisfies $\partial\Phi/\partial x = k(x,t):=\omega(x,t)n_\mathrm{p}(x,t)$.
\label{thm:librational}
\end{theorem}

The accuracy of the asymptotic formulae \eqref{eq:Basymptoticresults} is illustrated in Figure~\ref{fig:librationalaccuracy}.
\begin{figure}[h]
\begin{center}
\includegraphics{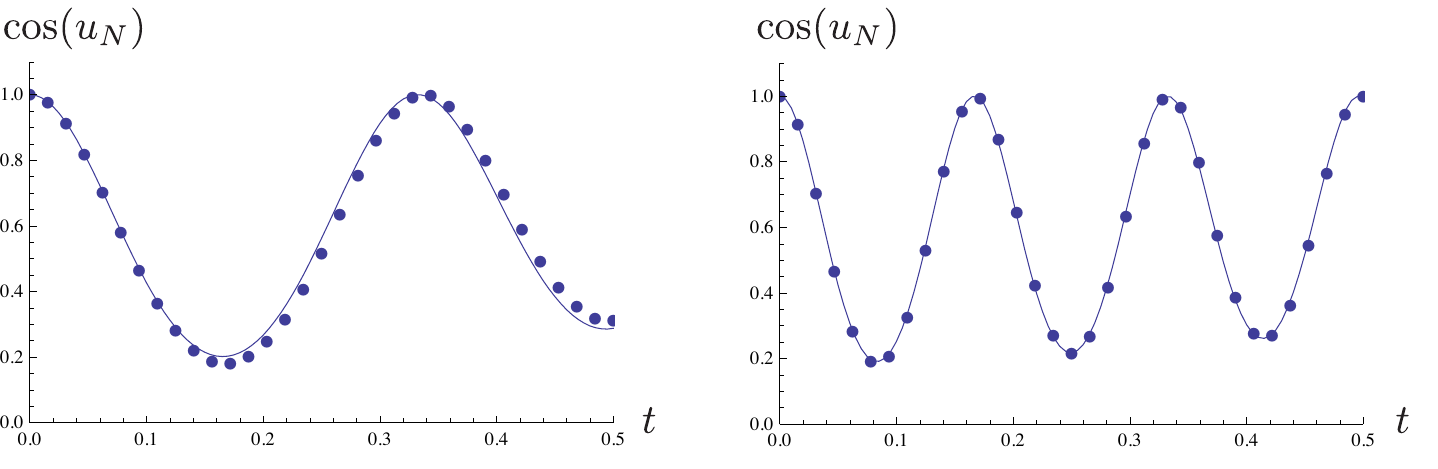}
\end{center}
\caption{\emph{The cosine of the exact solution $u_N(x,t)$ for the special initial data described in
\S\ref{sec:exactsolutions} for $A=3/4$ plotted with points, and the corresponding asymptotic formula plotted with curves, for fixed $x=1.5$ and $0<t<0.5$ so that $(x,t)\in S_\librational$ yielding
librational motion as described by Theorem~\ref{thm:librational}.  Left:  $N=8$, or equivalently
$\epsilon_N=0.09375$.  Right:  $N=16$, or equivalently $\epsilon_N=0.046875$. }}
\label{fig:librationalaccuracy}
\end{figure}

\begin{theorem}[Small-Time Rotational Asymptotics]
There exist well-defined differentiable functions $n_\mathrm{p}:S_\rotational\to(-1,1)$ and $\mathcal{E}:S_\rotational\to(1,+\infty)$ satisfying the initial conditions $n_\mathrm{p}(x,0)=0$ and $\mathcal{E}(x,0)=\tfrac{1}{2}G(x)^2-1$ and the hyperbolic
Whitham system \eqref{eq:Whithamsystem_rewrite} where $J(\mathcal{E}) = I_\rotational(-\mathcal{E})$
and $I_\rotational(\cdot)$ is given by \eqref{eq:Irotational}, as well as the inequality 
\begin{equation}
x\frac{\partial n_\mathrm{p}}{\partial t}(x,0)>0,\quad 0<|x|<x_\mathrm{crit}.
\label{eq:Rotationalnpineq}
\end{equation}
Moreover, the functions $n_\mathrm{p}(x,t)$ and $\mathcal{E}(x,t)$ extend continuously along with
their first partial derivatives to the curves $t=t_\pm(x)$, and they obey the $(x,t)$-independent inequalities
\begin{equation}
0\le \frac{1-n_\mathrm{p}(x,t)}{1+n_\mathrm{p}(x,t)}\left(\mathcal{E}(x,t)+\sqrt{\mathcal{E}(x,t)^2-1}\right)<
\frac{1}{4}\left(G(0)-\sqrt{G(0)^2-4}\right)^2
\label{eq:rotationalineq1}
\end{equation}
and
\begin{equation}
\frac{1-n_\mathrm{p}(x,t)}{1+n_\mathrm{p}(x,t)}\left(\mathcal{E}(x,t)-\sqrt{\mathcal{E}(x,t)^2-1}\right)>\frac{1}{4}\left(G(0)+\sqrt{G(0)^2-4}\right)^2>0.
\label{eq:rotationalineq2}
\end{equation}
Defining an elliptic parameter $m=m(x,t)$ by
\begin{equation}
m(x,t)=m_\rotational(x,t):=\frac{2}{1+\mathcal{E}(x,t)}\in (0,1),\quad (x,t)\in S_\rotational,
\label{eq:mrotationalE}
\end{equation}
and a real phase $\Phi(x,t)$ by 
\begin{equation}
\Phi(x,t):=-\int_0^t\omega(x,t')\,dt',
\end{equation}
 where
\begin{equation}
\omega(x,t):=
-\frac{\pi}{2K(m(x,t))}\frac{\sqrt{\mathcal{E}(x,t)+\sqrt{\mathcal{E}(x,t)^2-1}}+
\sqrt{\mathcal{E}(x,t)-\sqrt{\mathcal{E}(x,t)^2-1}}}{2}\frac{1}{\sqrt{1-n_\mathrm{p}(x,t)^2}},
\end{equation}
the following asymptotic formulae hold pointwise for $(x,t)\in S_\rotational$:
\begin{equation}
\begin{split}
\cos\left(\frac{1}{2}u_N(x,t)\right)&=\mathrm{cn}\left
(\frac{2\Phi(x,t)K(m(x,t))}{\pi\epsilon_N};m(x,t)\right) + \bo(\epsilon_N)\\
\sin\left(\frac{1}{2}u_N(x,t)\right) &= -\mathrm{sn}
\left(\frac{2\Phi(x,t)K(m(x,t))}{\pi\epsilon_N};m(x,t)\right) + \bo(\epsilon_N)\\
\epsilon_N \frac{\partial u_N}{\partial t}(x,t)&=-\frac{4K(m(x,t))}{\pi}
\frac{\partial\Phi}{\partial t}
\,\mathrm{dn}\left(\frac{2\Phi(x,t)K(m(x,t))}{\pi\epsilon_N};m(x,t)\right) +
\bo(\epsilon_N).
\end{split}
\label{eq:Kasymptoticresults}
\end{equation}
Moreover, the error terms are uniform for $(x,t)$ in compact subsets of $S_\rotational$.  The
phase $\Phi(x,t)$ also satisfies $\partial\Phi/\partial x = k(x,t):=\omega(x,t)n_\mathrm{p}(x,t)$.
\label{thm:rotational}
\end{theorem}

The accuracy of the asymptotic formulae \eqref{eq:Kasymptoticresults} is illustrated in Figure~\ref{fig:rotationalaccuracy}.
\begin{figure}[h]
\begin{center}
\includegraphics{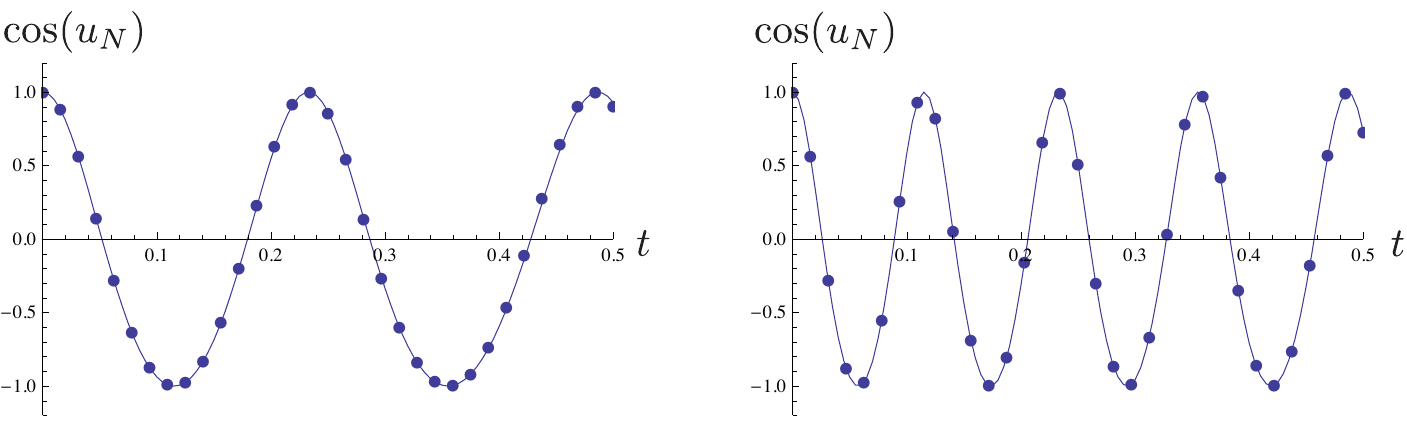}
\end{center}
\caption{\emph{The cosine of the exact solution $u_N(x,t)$ for the special initial data described in
\S\ref{sec:exactsolutions} for $A=3/4$ plotted with points, and the corresponding asymptotic formula plotted with curves, for fixed $x=-0.15625$ and $0<t<0.5$ so that $(x,t)\in S_\rotational$ yielding
rotational motion as described by Theorem~\ref{thm:rotational}.  Left:  $N=8$, or equivalently $\epsilon_N=0.09375$.  Right:  $N=16$, or equivalently $\epsilon_N=0.046875$.  Note that for the given value of $x$, $t_-(x)\approx 0.153$, but that there is no evidence of any change in behavior near this value of $t$ in either the exact solution
or the asymptotic solution.}}
\label{fig:rotationalaccuracy}
\end{figure}

We now make several observations about
these results:
\begin{itemize}
\item The asymptotic formulae \eqref{eq:Basymptoticresults} correspond to a high-frequency
superluminal librational wavetrain that is relatively slowly modulated through the $(x,t)$-dependence of
the quantities $n_\mathrm{p}$ and $\mathcal{E}$.  Indeed, Taylor expansion of these formulae about
a fixed point $(x_0,t_0)\in S_\librational$ shows that if one sets $x=x_0+\epsilon_N\tilde{x}$ and $t=t_0+\epsilon_N\tilde{t}$ and takes the limit $\epsilon_N\downarrow 0$ holding $\tilde{x}$ and $\tilde{t}$ fixed, we recover an exact superluminal librational wavetrain solution of the sine-Gordon equation as a function of $x$ and $t$ characterized by the fixed energy $\mathcal{E}(x_0,t_0)$ and the linear phase
$\Phi(x_0,t_0)+ k(x_0,t_0)(x-x_0) -\omega(x_0,t_0)(t-t_0)$.  It is clear in this case that
the function $u_N(x,t)$ is confined to the range $(-\pi,\pi)$  and that both $u_N$ and its time derivative are periodic functions of the linear phase as is consistent with librational motion.  Similarly, the asymptotic
formulae \eqref{eq:Kasymptoticresults} correspond to a high-frequency and slowly modulated superluminal rotational wavetrain, and in particular $u_N$ is monotonic while its derivative is periodic, as is consistent with rotational motion.
\item We have excluded the curves $t=t_\pm(x)$ from $S_\rotational$ as a matter of complete academic
honesty, as our proof would require a technical modification to extend pointwise asymptotics to these
curves, and to have uniformity for $t\approx t_\pm(x)$ is yet a further matter requiring a double-scaling
limit.  However, it is easy to see that the explicit terms in the asymptotic formulae \eqref{eq:Kasymptoticresults} undergo no phase transition in the vicinity of the curves $t=t_\pm(x)$,
and while we cannot honestly exclude the possibility that there is some phenomenon here to be captured by more detailed analysis, there is no indication of such in the plots shown in \S\ref{sec:exactsolutions}.
\item  As part of our proof, we show that in each case the asymptotic formula for $\epsilon_N\partial u_N/\partial t$ is consistent with those for $\cos(\tfrac{1}{2}u_N)$ and $\sin(\tfrac{1}{2}u_N)$ in the sense that differentiation of the latter with respect to $t$ assuming that the error terms remain subdominant after
differentiation yields the former up to terms of order $\bo(\epsilon_N)$.
\item The inequalities \eqref{eq:Librationalnpineq} and \eqref{eq:Rotationalnpineq} together with the initial condition $n_\mathrm{p}(x,0)=0$ give information about
the direction of motion of the waves.  For example, if $x>x_\mathrm{crit}$ and $t>0$, the fluxon condensate behaves like a train of superluminal 
librational waves propagating rapidly to the left (since the phase velocity $v_\mathrm{p}$ is large and negative), and if $0<x<x_\mathrm{crit}$ and $t>0$
the condensate behaves like a train of superluminal rotational waves propagating rapidly to the right.
\item The fields $n_\mathrm{p}(x,t)$ and $\mathcal{E}(x,t)$ have in each case exactly the interpretation of reciprocal
phase velocity and energy as explained earlier in the context of formal modulation theory.  Although
we conclude that these quantities satisfy the Whitham system \eqref{eq:Whithamsystem_rewrite},
we wish to emphasize that at no point in our proof do we apply techniques from the theory of partial differential equations (\textit{e.g.}, the Cauchy-Kovalevskaya method) to solve the Whitham system with
the specified initial data.  Instead, the fields $n_\mathrm{p}(x,t)$ and $\mathcal{E}(x,t)$ are constructed
by means of the solution of a system of nonlinear algebraic equations via the Implicit Function Theorem.
We first obtain functions $\mathfrak{p}=\mathfrak{p}(x,t)$ and $\mathfrak{q}=\mathfrak{q}(x,t)$ by solving the equations  $M=I=0$ (see Proposition~\ref{prop:tneq0continuegeneral}), or
for $(x,t)\in S_\rotational$ the equations $\hat{M}=\hat{I}=0$ (see Proposition~\ref{prop:origin}),
and then recover $n_\mathrm{p}(x,t)$ and $\mathcal{E}(x,t)$
therefrom via \eqref{eq:phasevelocity} and \eqref{eq:Euv} respectively.
The fact that these functions satisfy the quasilinear partial differential equations \eqref{eq:Whithamsystem_rewrite} is, from the point of view of our methodology, more or less a coincidence.
\item Interestingly, the asymptotics of the fluxon condensate are accurately described by superluminal wavetrains, even though the fundamental particles of the condensate are
all subluminal solitons of kink and breather type.   
\item The semiclassical asymptotics of the sine-Gordon equation are, at least in the case of
a sufficiently strong initial impulse consistent with Assumption~\ref{assume:rotational}, much more complicated
even for small times $t$ (independent of $\epsilon_N$) than in the case of more well-known integrable
partial differential equations like Korteweg-de Vries (see \cite{DeiftVZ97})  and focusing nonlinear
Schr\"odinger (see \cite{KamvissisMM03,TovbisVZ04}).  For these latter problems, there is an initial
stage of the evolution where the asymptotics are described for all $x\in\mathbb{R}$ in terms of 
elementary functions.  Here, we require two different types of formulae both involving higher transcendental functions, and even these do not suffice to describe the dynamics near the transitional points $x=\pm x_\mathrm{crit}$.
\item If, in contrast to Assumption~\ref{assume:rotational} we instead take $-2<G(0)\le 0$, then it will be clear from our proof that the region $S_\librational$ extends to include a full neighborhood of $(x,t)=(0,0)$, and hence becomes connected, including the entire $t=0$ axis.  Then Theorem~\ref{thm:librational} will suffice to describe the semiclassical asymptotics for small time uniformly for
$x\in\mathbb{R}$.  This is consistent with the dynamics pictured in Figure~\ref{fig:Ap25}.
\item  Our results show that the two superluminal types of modulated single-phase wave behavior observed in Figure~\ref{fig:Ap75} for small time and $x$ bounded away from $\pm x_\text{crit}$ are 
universal for the class of initial data we consider.  It is reasonable to conjecture that 
the qualitatively different behavior observed in the space-time plane on the other side of nonlinear caustic curves might also be universal, and that it might correspond to various types of modulated \emph{multiphase} waves.   We note that, with appropriate modifications and additional 
work, the methods described in this paper are capable of handling these other cases.  
Specifically, for $(x,t)$ in these regions the model Riemann-Hilbert problems \\ref{rhp:wOdotlibrational} and \ref{rhp:wOdotrotational} will need to be generalized to ones that are 
solved using the function theory of higher-genus hyperelliptic Riemann surfaces.  The associated Whitham equations
that generalize \eqref{eq:Whithamsystem_rewrite} will also have a correspondingly increased number of dependent variables (see \cite{ForestM83,ErcolaniFM84,ErcolaniFMM87}).
\end{itemize}

The relevance of the fluxon condensate $\{u_N(x,t)\}$ 
to the Cauchy problem for
\eqref{eq:SG} with $\epsilon=\epsilon_N$ and with pure-impulse initial data characterized by the even Klaus-Shaw function $G$
is then the following result:
\begin{corollary}
When $t=0$, the fluxon condensate  $\{u_N(x,t)\}$ associated with the pure-impulse initial condition of impulse profile $G(\cdot)$ satisfies
\begin{equation}
u_N(x,0)=\bo(\epsilon_N)\pmod{4\pi}\quad\text{and}\quad\epsilon_N\frac{\partial u_N}{\partial t}(x,0)=
G(x)+\bo(\epsilon_N)
\end{equation}
where the error estimates are valid pointwise for $x\neq 0$ and $|x|\neq x_\mathrm{crit}$, and
uniformly on compact subsets of the set of pointwise validity.
\label{corr:ICapproximate}
\end{corollary}
In this sense, the fluxon condensate approximates the solution of the Cauchy problem for \eqref{eq:SG} with initial data \eqref{eq:IC} when $\epsilon=\epsilon_N$ and $N$ is large.

\subsection{Outline of the rest of the paper}
\label{sec:outline}
The remaining sections of the paper are devoted to the proofs of 
Theorems \ref{thm:librational} (small-time librational asymptotics) 
and \ref{thm:rotational} (small-time rotational asymptotics).  These are 
proven using 
the well-developed inverse-scattering method (see the brief 
discussion in \S\ref{sec:pure-impulse} and the more detailed exposition 
in our paper \cite{BuckinghamM08}).  The fluxon condensates $\{u_N(x,t)\}$ 
we study are defined by their scattering data, effectively skipping the 
forward-scattering transform.  Thus all of our analysis concerns the 
inverse-scattering transform.  

Our approach is to use the Riemann-Hilbert problem formulation of the 
inverse-scattering transform.  Since the scattering data are reflectionless 
(that is, comprised of only eigenvalues and the corresponding proportionality 
constants), the associated Riemann-Hilbert problem for the 
$2\times 2$ matrix-valued function $\mathbf{J}(w)$ has only poles 
and a 
completely trivial jump on the positive real axis (which could be removed 
at the cost of artificially doubling the number of poles through the transformation $w=z^2$).
This setup and 
parts of the subsequent analysis are similar to an analagous work on 
semiclassical soliton ensembles for the focusing nonlinear 
Schr\"odinger equation \cite{KamvissisMM03}.  

Our first step is to make a local change of variables 
$\mathbf{J}(w)\to\mathbf{M}(w)$ in the Riemann-Hilbert 
problem in \S\ref{sec:interpolation} that removes the poles at the price 
of introducing further jump contours (which are more amenable to analysis).  
Depending on the value of $x$ and $t$, different transformations are 
used in different parts of the spectral $w$-plane.  These different 
transformations are illustrated in Figures 
\ref{fig:Ann_DeltaEmpty}--\ref{fig:Ann_NablaNearMinusOneOverM}.  
Note that only one of these cases ($\Delta=\emptyset$ shown in 
Figure \ref{fig:Ann_DeltaEmpty}) is required to 
analyze solutions with even initial data for small times and $x$ bounded away from 
the origin.  

In \S\ref{sec:g-function} we make another transformation 
$\mathbf{M}(w)\to\mathbf{N}(w)$ involving a $g$-function, 
a standard tool (first introduced in \cite{DeiftVZ97})
for controlling the behavior of jump matrices in the 
Riemann-Hilbert problem.  The details of the construction of the 
$g$-function are contained in \S\ref{sec:construction-of-g}.  
Finding the $g$-function involves identifying parts of the jump 
contours as a \emph{band} $\beta$ and parts as a \emph{gap} $\gamma$.  
There are two possible topological configurations for the band 
$\beta$ (see Figures \ref{fig:betagammaL} and \ref{fig:betagammaR}).  
In case $\librational$ (leading to librational wavetrains), the 
band endpoints are a complex-conjugate pair ($w=\mathfrak{p}\pm i\sqrt{-\mathfrak{q}}$).   
In case $\rotational$ (leading to rotational wavetrains), the 
band endpoints are real ($w=\mathfrak{p}\pm\sqrt{\mathfrak{q}}$).  The band endpoints are 
chosen to enforce certain conditions making the jump matrices easy to analyze 
in the limit $\epsilon_N\to 0$.  The key part of the analysis is 
proving there exists a $g$-function \emph{independent of $\epsilon_N$} 
such that these conditions are satisfied for some fixed small time.
This analysis is particularly delicate in a neighborhood of 
$(x,t)=(0,0)$, as described in \S\ref{continuation-from-x=0}.
In \S\ref{sec:Whitham} we show that 
the band endpoints associated to an admissable $g$-function satisfy the 
Whitham modulation equations \eqref{eq:Whithamsystem_rewrite} in Riemann invariant (diagonal) form.

In \S\ref{sec:opening-a-lens} we open lenses in the Riemann-Hilbert 
problem by making the transformation $\mathbf{N}(w)\to\mathbf{O}(w)$ 
as illustrated in Figures \ref{fig:BcaseOcontour}--\ref{fig:KcaseOMinusOOM}.  
This has the effect of making the jump matrices exponentially 
close to the identity
in $\epsilon_N$ except on the band $\beta$, the positive real axis, and in 
small neighborhoods of the band endpoints.  In \S\ref{sec:global-parametrix} 
we construct the \emph{global parametrix} $\dot{\mathbf{O}}(w)$, 
which is 
an explicit approximation of $\mathbf{O}(w)$ that is valid in the 
whole complex $w$-plane.  In 
\S\ref{sec:error} we prove rigorously that $\dot{\mathbf{O}}(w)$ is an 
$\bo(\epsilon_N)$-approximation of $\mathbf{O}(w)$.  Thus we can 
use $\dot{\mathbf{O}}(w)$ to compute the solution $u_N(x,t)$ to the 
sine-Gordon equation modulo errors of size $\mathcal{O}(\epsilon_N)$.  Finally, 
Appendix 
\ref{app:outer} contains the details of calculating the solutions 
$u_N(x,t)$ from the outer parametrix $\dot{\mathbf{O}}(w)$.
A significant portion of the effort in Appendix \ref{app:outer} 
is devoted to translating more or less standard formulae involving 
Riemann theta functions of genus one into expressions involving 
Jacobi elliptic functions with suitable moduli.  This is necessary to 
establish the simple form of the asymptotic formulae appearing in the statements of Theorems \ref{thm:librational}
and \ref{thm:rotational}.

\subsection{Notation and terminology}
\label{sec:notation}
With the exception of the identity matrix $\mathbb{I}$ and
the three Pauli matrices,
\begin{equation}
\sigma_1:=\begin{bmatrix}0 & 1\\1 & 0\end{bmatrix},\quad
\sigma_2:=\begin{bmatrix}0 & -i\\i & 0\end{bmatrix},\quad
\sigma_3:=\begin{bmatrix}1 & 0 \\ 0 & -1\end{bmatrix},
\label{eq:Pauli}
\end{equation}
we will denote matrices by bold capital letters (\emph{e.g.} $\mathbf{M}$)
and vectors by bold lowercase letters (\emph{e.g.} $\mathbf{v}$).
We will use $\overline{A}$ for the closure of a set $A\subset\mathbb{R}^2$,
and denote complex conjugation with an asterisk.

In what follows, by a planar \emph{arc} we will mean the image of a
continuous and piecewise-smooth one-to-one map $(0,1)\ni t\mapsto
w(t)\in\mathbb{C}\cong\mathbb{R}^2$ with parameter $t$ and
nonvanishing derivative.  By a planar \emph{contour} we will mean the
$\mathbb{R}^2$-closure of a finite union of arcs.  Thus a contour is
always a closed set in the topology of $\mathbb{R}^2$ and may contain
self-intersection points where various arcs meet.  An \emph{oriented
  contour} is a contour $K$ written in a particular way as the
$\mathbb{R}^2$-closure of a finite union, denoted $\vec{K}$, of
pairwise-disjoint arcs each of which is assigned an orientation in the
obvious way according to its parametrization by $t$.  An oriented
contour may include at most a finite number of points at which the
orientation is not properly defined.  These may be self-intersection
points, endpoints, or points dividing an arc into oppositely-oriented
sub-arcs.  

If $K$ is an oriented contour and $f:\mathbb{C}\setminus
K\to \mathbb{C}$ is an analytic function, we denote by $f_+(\xi)$
(respectively $f_-(\xi)$) the boundary value taken by $f(w)$ as
$w\to\xi\in\vec{K}$ from the left (right) according to local
orientation, if it exists.  We use analogous notation for
vector-valued functions (\emph{e.g.} $\mathbf{v}_\pm(\xi)$) and
matrix-valued functions (\emph{e.g.} $\mathbf{M}_\pm(\xi)$).
Given a contour $K$ we define a metric $d_K$ on $\mathbb{C}\setminus K$
as follows:
\begin{equation}
d_K(w_1,w_2):=\mathop{\inf_{P\subset\mathbb{C}\setminus K}}_{P:w_1\to w_2}
\mathrm{length}(P)
\end{equation}
where the infimum is taken over paths $P\subset\mathbb{C}\setminus K$
connecting the points $w$ and $z$.  If $0<\alpha\le 1$, an analytic
function $f:\mathbb{C}\setminus K\to\mathbb{C}$ is said to be
\emph{uniformly H\"older-$\alpha$ continuous} if
\begin{equation}
\sup_{w_1,w_2\in\mathbb{C}\setminus K}\frac{|f(w_1)-f(w_2)|}
{d_K(w_1,w_2)^\alpha}<\infty.
\end{equation}
This definition generalizes in the obvious way to vector-valued or matrix-valued
functions, and also may be analogously defined relative to open domains
$U\setminus K$, $U\subset\mathbb{C}$, $U$ open.

\section{Formulation of the Inverse Problem for Fluxon Condensates}
\label{section-formulation}
Let a function $G:\mathbb{R}\to\mathbb{R}$ satisfying the conditions
set out in the introduction be given, and let a decreasing sequence
$\{\epsilon_N\}_{N=1}^\infty$ be defined by
\eqref{eq:epsilonNgeneral}.  For each integer $N>0$, define numbers
$\lambda_k^0=\lambda_{N,k}^0$ on the positive imaginary axis by the
Bohr-Sommerfeld quantization rule \eqref{eq:BohrSommerfeld} with the
WKB phase integral \eqref{eq:WKBphase} and the conditions
$\epsilon=\epsilon_N$ and $N(\epsilon)=N$.  Regardless of the value of
$N=1,2,3,\dots$, the numbers $\{\lambda_{N,k}^0\}_{k=0}^{N-1}$ will
each have exactly two distinct preimages in the $w$-plane with
$|\arg(-w)|<\pi$ under the map $\lambda=E(w)$ if the following
additional condition is satisfied.
\begin{assume} 
The fraction $\Psi(i/2)/\|G\|_1$ is irrational.  
\label{assume:nodoublegeneral}
\end{assume}
This simply guarantees that none of the numbers
$\{\lambda_{N,k}^0\}_{k=0}^{N-1}$ coincide with $\lambda= i/2$,
the unique critical value of $E(w)$ for $|\arg(-w)|<\pi$, which 
ensures the poles that will appear in the definition of Riemann-Hilbert Problem 
\ref{basic-rhp} below are all simple.  
Rationality of
$\Psi(i/2)/\|G\|_1$ can also be admitted with the same effect
at the cost of passing to a subsequence of values of $N$.

Let a function $Q(w)$ be defined for $|\arg(-w)|<\pi$ as follows:
\begin{equation}
Q(w)=Q(w;x,t):=E(w)x+D(w)t,
\label{eq:Qw}
\end{equation}
and define
\begin{equation}
\Pi_N(w):=\prod_{k=0}^{N-1}\frac{E(w)+\lambda^0_{N,k}}{E(w)-\lambda^0_{N,k}},
\quad |\arg(-w)|<\pi.
\label{eq:PiNw}
\end{equation}
According to Assumption~\ref{assume:nodoublegeneral}, the denominator
of $\Pi_N(w)$ has $2N$ simple zeros in the domain of definition, while
the numerator of $\Pi_N(w)$ is analytic and nonvanishing.  
Let $P_N$ denote the set of poles of $\Pi_N(w)$, $2N$ points consisting
of $2N_\mathsf{B}$ points in complex conjugate pairs on the unit circle $S^1$
and $2N_\mathsf{K}$ points on the negative real axis in involution with 
respect to the map $w\to 1/w$.
We may write $\Pi_N(w)$
equivalently in the form
\begin{equation}
\Pi_N(w)=
\prod_{y\in P_N}\frac{\sqrt{-w}+\sqrt{-y}}{\sqrt{-w}-\sqrt{-y}}.
\label{eq:PiNwagain}
\end{equation}
The basic Riemann-Hilbert problem of inverse scattering to construct
the fluxon condensate corresponding to $G$ is then the following.  We
are following the description of the inverse-scattering problem given
in Appendix A of \cite{BuckinghamM08} but we are exploiting the
symmetry $z\mapsto -z$ to formulate the problem in terms of the
complex variable $w=z^2$, which introduces a jump on $\mathbb{R}_+$.
\begin{rhp}[Basic Problem of Inverse Scattering]
\label{basic-rhp}
Find a $2\times 2$ matrix function $\mathbf{H}(w)=\mathbf{H}_N(w;x,t)$
of the complex variable $w$ with the following properties:
\begin{itemize}
\item[]\textbf{Analyticity:} $\mathbf{H}(w)$ is analytic for $w\in
\mathbb{C}\setminus(P_N\cup\mathbb{R}_+)$.
\item[]\textbf{Jump Condition:} There is a neighborhood $U=U_N$ of $\mathbb{R}_+$ such that $\mathbf{H}(w)$ is uniformly H\"older-$\alpha$
continuous for all $\alpha\in (0,1]$ on $U\setminus\mathbb{R}_+$.  Letting
$\mathbb{R}_+$ be oriented from left to right, the boundary values
taken by $\mathbf{H}(w)$ on $\mathbb{R}_+$ are related by the jump
condition
\begin{equation}
\mathbf{H}_+(\xi)=\sigma_2\mathbf{H}_-(\xi)\sigma_2,\quad \xi\in\vec{\mathbb{R}}_+.
\end{equation}
\item[]\textbf{Singularities:}  Each of the points of $P_N$ is
a simple pole of $\mathbf{H}(w)$.  If $y\in P_N$ with $E(y)=\lambda^0_{N,k}$
for $k=0,\dots,N-1$, then
\begin{equation}
\mathop{\mathrm{Res}}_{w=y}\mathbf{H}(w)=\lim_{w\to y}
\mathbf{H}(w)\begin{bmatrix} 0 & 0 \\
\displaystyle 
(-1)^{k+1}\mathop{\mathrm{Res}}_{w=y}e^{2iQ(w;x,t)/\epsilon_N}
\Pi_N(w) & 0
\end{bmatrix}.
\end{equation}
These amount to one matrix-valued condition on the residue of $\mathbf{H}(w)$
at each of its poles.
\item[]\textbf{Normalization:}  The following normalization condition holds:
\begin{equation}
\lim_{w\to\infty}\mathbf{H}(w)=\mathbb{I},
\label{eq:Hnorm}
\end{equation}
where the limit is uniform with respect to angle for $|\arg(-w)|<\pi$.
\end{itemize}
\label{rhp:basicw}
\end{rhp}
It is an easy application of Liouville's Theorem that any solution of
this problem must satisfy $\det(\mathbf{H}(w))\equiv 1$, from which it
follows that if $\mathbf{H}_1(w)$ and $\mathbf{H}_2(w)$ are any two
solutions, the matrix ratio
$\mathbf{R}(w):=\mathbf{H}_1(w)\mathbf{H}_2(w)^{-1}$ is analytic for
$|\arg(-w)|<\pi$ (has removable singularities at the points of $P_N$
and acquires no additional singularities from inversion of
$\mathbf{H}_2(w)$), is H\"older-$\alpha$ continuous in
$U\setminus\mathbb{R}_+$, and tends to the identity matrix as
$w\to\infty$.  The boundary values taken by $\mathbf{R}(w)$ on
$\vec{\mathbb{R}}_+$ satisfy
$\mathbf{R}_+(\xi)=\sigma_2\mathbf{R}_-(\xi)\sigma_2$.  If we set
\begin{equation}
\mathbf{S}(z):=\begin{cases}\mathbf{R}(z^2),\quad&\Im\{z\}>0,\\
\sigma_2\mathbf{R}(z^2)\sigma_2,\quad &\Im\{z\}<0,
\end{cases}
\label{eq:zw-unfold}
\end{equation}
then it is clear that $\mathbf{S}(z)$ extends to an entire function of
$z\in\mathbb{C}$ with identity asymptotics as $z\to\infty$.  Hence
Liouville's Theorem shows that $\mathbf{S}(z)\equiv\mathbb{I}$ and so
by restriction to $\Im\{z\}>0$, $\mathbf{H}_1(w)\equiv\mathbf{H}_2(w)$
for $|\arg(-w)|<\pi$.  Therefore, solutions to Riemann-Hilbert
Problem~\ref{rhp:basicw} are necessarily unique if they exist.

Given any solution $\mathbf{H}(w)$ of Riemann-Hilbert
Problem~\ref{rhp:basicw}, another is easily generated by setting
$\mathbf{H}^\sharp(w):=\mathbf{H}(w^*)^*$.  By uniqueness it follows
that $\mathbf{H}^\sharp(w)\equiv\mathbf{H}(w)$, or equivalently, that
the unique solution of Riemann-Hilbert Problem~\ref{rhp:basicw}
necessarily satisfies 
\begin{equation}
\mathbf{H}(w^*)=\mathbf{H}(w)^*.
\label{eq:Hsymmetryw}
\end{equation}

If for some $(x,t)\in\mathbb{R}^2$, $\mathbf{H}(w)$ is the unique
solution to Riemann-Hilbert Problem~\ref{rhp:basicw} for all
integer $N\geq N_0$ for 
sufficiently large $N_0$, then a construction similar to
\eqref{eq:zw-unfold} involving the variable $z$ such that $w=z^2$
shows that if it exists the solution $\mathbf{H}(w)$ has convergent
series expansions of the form
\begin{equation}
\mathbf{H}(w)=\sum_{k=0}^\infty\mathbf{H}_N^{0,k}(x,t)(\sqrt{-w})^k,\quad
|w|<r
\label{eq:wzeroexpansion}
\end{equation}
and
\begin{equation}
\mathbf{H}(w)=\mathbb{I}+\sum_{k=1}^\infty\mathbf{H}_N^{\infty,k}(x,t)
(\sqrt{-w})^{-k},\quad |w|>R
\label{eq:winfinityexpansion}
\end{equation}
for suitable numbers $r$ and $R$ independent of $N$.
We will use the notation
\begin{equation}
\mathbf{A}_N(x,t):=\mathbf{H}^{0,0}_N(x,t),\quad\mathbf{B}^0_N(x,t):=
\mathbf{H}^{0,0}_N(x,t)^{-1}\mathbf{H}^{0,1}_N(x,t),\quad
\mathbf{B}^\infty_N(x,t):=\mathbf{H}^{\infty,1}_N(x,t),
\label{eq:ABsdef}
\end{equation}
defining
three matrices
depending parametrically on $(x,t)\in \mathbb{R}^2$ and the integer
$N\ge N_0$.  These matrices necessarily satisfy the
conditions
\begin{equation}
\det(\mathbf{A}_N(x,t))=1,\quad
\mathbf{A}_N(x,t)=\sigma_2 \mathbf{A}_N(x,t)\sigma_2,\quad\text{and}
\quad
\mathbf{A}_N(x,t)=\mathbf{A}_N(x,t)^*,
\label{eq:Hzeroidentities}
\end{equation}
\begin{equation}
\mathrm{tr}(\mathbf{B}_N^\infty(x,t))=0,\quad
\mathbf{B}^\infty_N(x,t)=-\sigma_2\mathbf{B}^\infty_N(x,t)\sigma_2,\quad
\text{and}\quad\mathbf{B}^\infty_N(x,t)=\mathbf{B}^\infty_N(x,t)^*,
\label{eq:Honeidentities}
\end{equation}
and
\begin{equation}
\mathrm{tr}(\mathbf{B}^0_N(x,t))=0,\quad
\mathbf{B}^0_N(x,t)=-\sigma_2\mathbf{B}^0_N(x,t)\sigma_2,\quad
\text{and}\quad\mathbf{B}^0_N(x,t)=\mathbf{B}^0_N(x,t)^*.
\label{eq:Bzeroidentities}
\end{equation}

\begin{definition}[Fluxon condensates]
Given $(x,t)\in\mathbb{R}^2$, suppose that Riemann-Hilbert Problem~\ref{rhp:basicw} has a solution for all integer $N\ge N_0$.  Then, 
the fluxon condensate $\{u_N(x,t)\}_{N=N_0}^\infty$ associated with
the impulse profile $G$ is given (modulo addition of
arbitrary integer multiples of $4\pi$) in terms of
$\mathbf{A}_N(x,t)$ as follows:
\begin{equation}
\cos\left(\frac{1}{2}u_N(x,t)\right)
= A_{N,11}(x,t)
\quad\text{and}\quad
\sin\left(\frac{1}{2}u_N(x,t)\right)=
A_{N,21}(x,t).
\label{eq:cossinuN}
\end{equation}
Note that the identities \eqref{eq:Hzeroidentities} ensure the
the reality of these expressions 
as well as the Pythagorean identity
$\sin(\tfrac{1}{2}u_N(x,t))^2+\cos(\tfrac{1}{2}u_N(x,t))^2=1$.
\label{def:condensate}
\end{definition}
\begin{proposition}
  Suppose that $\Omega\subset\mathbb{R}^2$ is open and 
an integer $N_0>0$ is given such that 
Riemann-Hilbert Problem~\ref{rhp:basicw} has
a solution whenever $(x,t)\in \Omega$ and $N\ge N_0$, and that the fluxon
condensate $\{u_N(x,t)\}_{N=N_0}^\infty$ is defined as above for
$(x,t)\in \Omega$.  Then for each integer $N\ge N_0$, $u=u_N(x,t)$ is an
exact real-valued solution of the sine-Gordon equation \eqref{eq:SG}
with $\epsilon=\epsilon_N$.  Moreover, 
\begin{equation}
\epsilon_N\frac{\partial u_N}{\partial t}(x,t)=B^0_{N,12}(x,t)+B^\infty_{N,12}(x,t).
\label{eq:epsilonut}
\end{equation}
\end{proposition}
\begin{proof}
Consider the matrix $\mathbf{F}(w)$ defined in terms of the solution
$\mathbf{H}(w)$ of Riemann-Hilbert Problem~\ref{rhp:basicw} by 
\begin{equation}
\mathbf{F}(w):=\mathbf{H}(w)e^{-iQ(w;x,t)\sigma_3/\epsilon_N}.
\end{equation}
Aside from an introduced essential singularity at $w=0$, the matrix
$\mathbf{F}(w)$ is analytic exactly where $\mathbf{H}(w)$ is, and has
similar properties where analyticity is violated.  Namely,
$\mathbf{F}(w)$ is H\"older-$\alpha$ continuous in
$U\setminus\mathbb{R}_+$ except at $w=0$ where the exponential factor
has an essential singularity, and as a consequence of the identity
$Q_+(\xi)=-Q_-(\xi)$ for $\xi\in\vec{\mathbb{R}}_+$ one has
that the boundary values of $\mathbf{F}(w)$ satisfy the jump condition
$\mathbf{F}_+(\xi)=\sigma_2\mathbf{F}_-(\xi)\sigma_2$ for
$\xi\in\vec{\mathbb{R}}_+$.  Also, $\mathbf{F}(w)$ has simple poles
at the points of $P_N$, and if $y\in P_N$ with $E(y)=\lambda^0_{N,k}$
then
\begin{equation}
\mathop{\mathrm{Res}}_{w=y}\mathbf{F}(w)=\lim_{w\to y}
\mathbf{F}(w)\begin{bmatrix}0 & 0\\
\displaystyle (-1)^{k+1}\mathop{\mathrm{Res}}_{w=y}\Pi_N(w) & 0
\end{bmatrix}.
\end{equation}
Note that neither the jump nor residue conditions involve
$(x,t)\in\mathbb{R}^2$, from which it follows that the matrices
\begin{equation}
\mathbf{U}(w):=4i\epsilon_N\mathbf{F}_x(w)\mathbf{F}(w)^{-1}=
4i\epsilon_N\mathbf{H}_x(w)\mathbf{H}(w)^{-1} + 4E(w)\mathbf{H}(w)\sigma_3\mathbf{H}(w)^{-1}
\end{equation}
and
\begin{equation}
\mathbf{V}(w):=4i\epsilon_N\mathbf{F}_t(w)\mathbf{F}(w)^{-1}=
4i\epsilon_N\mathbf{H}_t(w)\mathbf{H}(w)^{-1} + 4D(w)\mathbf{H}(w)\sigma_3\mathbf{H}(w)^{-1}
\end{equation}
are functions of $z=i\sqrt{-w}$ that are analytic for
$z\in\mathbb{C}\setminus\{0\}$.  These functions have the following asymptotic
behavior as $w\to 0$ and $w\to\infty$:
\begin{equation}
\begin{split}
\mathbf{U}(w)&=\begin{cases}
\displaystyle 
\frac{1}{\sqrt{-w}}i\mathbf{A}_N(x,t)\sigma_3\mathbf{A}_N(x,t)^{-1} + \bo(1),
\quad & w\to 0,\\
\sqrt{-w}i\sigma_3 + i[\mathbf{B}^\infty_N(x,t),\sigma_3] + \lo(1),\quad &w\to\infty,
\end{cases}\\
\mathbf{V}(w)&=\begin{cases}
\displaystyle
-\frac{1}{\sqrt{-w}}i\mathbf{A}_N(x,t)\sigma_3\mathbf{A}_N(x,t)^{-1} + \bo(1),
\quad & w\to 0,\\
\sqrt{-w}i\sigma_3 + i[\mathbf{B}^\infty_N(x,t),\sigma_3] + \lo(1),\quad &w\to\infty.
\end{cases}
\end{split}
\end{equation}
It then follows from Liouville's Theorem applied
in the $z$-plane that in fact $\mathbf{U}(w)$ and $\mathbf{V}(w)$ are
Laurent polynomials in $z$ of degree $(1,1)$:
\begin{equation}
\begin{split}
\mathbf{U}(w)&=\sqrt{-w}i\sigma_3 + i[\mathbf{B}^\infty_N(x,t),\sigma_3]
+\frac{1}{\sqrt{-w}}i\mathbf{A}_N(x,t)\sigma_3\mathbf{A}_N(x,t)^{-1}\\
\mathbf{V}(w)&=\sqrt{-w}i\sigma_3 + i[\mathbf{B}^\infty_N(x,t),\sigma_3]
-\frac{1}{\sqrt{-w}}i\mathbf{A}_N(x,t)\sigma_3\mathbf{A}_N(x,t)^{-1}.
\end{split}
\end{equation}
According to \eqref{eq:Hzeroidentities}, we may write $\mathbf{A}_N(x,t)$
in the form
\begin{equation}
\mathbf{A}_N(x,t)=\begin{bmatrix} \cos(\phi) & -\sin(\phi)\\
\sin(\phi) &\cos(\phi)
\end{bmatrix}
\label{eq:H0form}
\end{equation}
where $\phi=\phi_N(x,t)$ is a real angle.  Likewise, according to
\eqref{eq:Honeidentities} we may write $\mathbf{B}^\infty_N(x,t)$ in the
form
\begin{equation}
\mathbf{B}^\infty_N(x,t)=\begin{bmatrix}-d & c\\c & d\end{bmatrix}
\end{equation}
where $c=c_N(x,t)$ and $d=d_N(x,t)$ are real-valued fields.  In terms
of these we therefore have
\begin{equation}
\begin{split}
\mathbf{U}(w)&=\sqrt{-w}i\sigma_3 + 2c\sigma_2  +\frac{1}{\sqrt{-w}}
\left(i\cos(2\phi)\sigma_3+i\sin(2\phi)\sigma_1\right)\\
\mathbf{V}(w)&=\sqrt{-w}i\sigma_3 + 2c\sigma_2  -\frac{1}{\sqrt{-w}}
\left(i\cos(2\phi)\sigma_3+i\sin(2\phi)\sigma_1\right).
\end{split}
\end{equation}
Since when considered as a function of $x$ and $t$, the matrix
$\mathbf{F}(w)$ is a simultaneous fundamental solution matrix of the
first-order overdetermined system
\begin{equation}
4i\epsilon_N\mathbf{F}_x = \mathbf{U}\mathbf{F}\quad\text{and}\quad
4i\epsilon_N\mathbf{F}_t = \mathbf{V}\mathbf{F},
\end{equation}
it follows that this system is compatible, implying that the
\emph{zero-curvature condition}
\begin{equation}
4i\epsilon_N\mathbf{U}_t-4i\epsilon_N\mathbf{V}_x +[\mathbf{U},\mathbf{V}]=\mathbf{0}
\end{equation}
holds.  Since
$\{\sigma_1/\sqrt{-w},\sigma_3/\sqrt{-w},\sigma_2\}$
is a linearly independent set, the zero-curvature equation implies the following three equations:
\begin{equation}
\begin{split}
\left(\epsilon_N\phi_t +\epsilon_N\phi_x -c\right)\cos(2\phi)&=0\\
\left(\epsilon_N\phi_t +\epsilon_N\phi_x -c\right)\sin(2\phi)&=0\\
\epsilon_N c_t - \epsilon_N c_x+\frac{1}{2}\sin(2\phi)&=0.
\end{split}
\end{equation}
The first two are together equivalent to the single equation
$\epsilon_N\phi_t+\epsilon_N\phi_x-c=0$, which may be used to eliminate
$c$ from the third equation, yielding
\begin{equation}
\epsilon_N^2\phi_{tt}-\epsilon_N^2\phi_{xx}+\frac{1}{2}\sin(2\phi)=0,
\end{equation}
which is the sine-Gordon equation \eqref{eq:SG} for $u_N=2\phi$ with
$\epsilon=\epsilon_N$.

Obviously we have
\begin{equation}
\begin{split}
\sin\left(\frac{1}{2}u_N\right)&=\sin(\phi)=A_{N,21}(x,t)\\
\cos\left(\frac{1}{2}u_N\right)&=\cos(\phi)=A_{N,11}(x,t)
\end{split}
\end{equation}
in accordance with \eqref{eq:H0form} and \eqref{eq:cossinuN}.
Finally, the identity \eqref{eq:epsilonut} arises from the constant
term in the Laurent expansion (in powers of $\sqrt{-w}$) of the differential 
equation $4i\epsilon_N \mathbf{H}_t(w) + 4D(w)\mathbf{H}(w)\sigma_3=
\mathbf{V}(w)\mathbf{H}(w)$ equivalent to $4i\epsilon_N\mathbf{F}_t(w)=
\mathbf{V}(w)\mathbf{F}(w)$.
\end{proof}

Note that since $P_N$ consists of points $y$ with $|\arg(-y)|<\pi$, it
is a consequence of the definitions \eqref{eq:DE} and \eqref{eq:Qw}
involving the principal branch of the square root that for each fixed
$N$ and $t$ we have
\begin{equation}
\lim_{x\to +\infty}\mathop{\mathrm{Res}}_{w=y \in P_N}
e^{2iQ(w;x,t)/\epsilon_N}\Pi_N(w) = 0.
\end{equation}
From this it can be shown that for each
$w\in\mathbb{C}\setminus (P_N\cup \mathbb{R}_+)$ and for each
$t\in\mathbb{R}$, $\mathbf{H}_N(w;x,t)\to\mathbb{I}$ as $x\to +\infty$.
By contrast, the asymptotic behavior of $\mathbf{H}$ as $x\to -\infty$
is not at all clear, and this is more than merely a technical
difficulty.  The ``preference'' that $\mathbf{H}$ has for large
positive $x$ stems from a choice in defining the spectral theory for
which Riemann-Hilbert Problem~\ref{rhp:basicw} is the inverse-spectral
problem in terms of scattering ``from the right''.  Given this
asymmetry, it should be no surprise if asymptotic analysis of
$\mathbf{H}$ in the semiclassical limit of large $N$ must take a
different route for $-x$ than it does for $x$.  
This asymmetry can be addressed in two different ways:
\begin{itemize}
\item If the initial conditions $F(x):=u(x,0)$ and $G(x):=\epsilon_Nu_t(x,0)$ 
for the sine-Gordon Cauchy problem are both either even
or odd functions then this symmetry is preseved in time and knowledge of the
solution for, say, positive $x$ is sufficient.  The special initial conditions
under consideration in this paper have even symmetry.
Moreover, it can be shown
directly from the conditions of Riemann-Hilbert Problem~\ref{rhp:basicw} that
the functions $\{u_N(x,t)\}_{N=N_0}^\infty$ approximating the solution
to the Cauchy problem are all even functions of $x$.  
\item The inverse problem may be reformulated in terms of scattering ``from
the left''.  
\end{itemize}

In fact it will turn out that for certain values of
$(x,t)\in\mathbb{R}^2$ of interest, neither of the above two
approaches will be of much help.  However, a modification of the
second approach that can be thought of as formulating the inverse
problem in terms of scattering theory ``partly from the left and
partly from the right'' will indeed succeed.  The idea is to select a
subset $\Delta\subset P_N$ of the pole divisor of $\mathbf{H}(w)$ and
set $\nabla:= P_N\setminus\Delta$.  Then define from the solution
$\mathbf{H}_N(w;x,t)$ of Riemann-Hilbert Problem~\ref{rhp:basicw} a
new matrix function given by
\begin{equation}
\mathbf{J}_N(w;x,t):=\mathbf{H}_N(w;x,t)
\left(\prod_{y\in\Delta}\frac{\sqrt{-w}+\sqrt{-y}}{\sqrt{-w}-\sqrt{-y}}
\right)^{-\sigma_3}.
\label{eq:Htildedef}
\end{equation}
It is easy to see that if $\mathbf{H}_N(w;x,t)$ satisfies Riemann-Hilbert
Problem~\ref{rhp:basicw} then $\mathbf{J}_N(w;x,t)$ satisfies the 
following equivalent problem.
\begin{rhp}[Modified Problem of Inverse Scattering]
  Find a $2\times 2$ matrix function
  $\mathbf{J}(w)=\mathbf{J}_N(w;x,t)$ of the complex
  variable $w$ with the following properties:
\begin{itemize}
\item[]\textbf{Analyticity:} $\mathbf{J}(w)$ is analytic for $w\in
\mathbb{C}\setminus(P_N\cup\mathbb{R}_+)$.
\item[]\textbf{Jump Condition:} There is a neighborhood $U=U_N$ of 
$\mathbb{R}_+$ such that $\mathbf{J}(w)$ is uniformly H\"older-$\alpha$
continuous for all $\alpha\in (0,1]$ on $U\setminus\mathbb{R}_+$.  Letting
$\mathbb{R}_+$ be oriented from left to right, the boundary values
taken by $\mathbf{J}(w)$ on $\mathbb{R}_+$ are related by the jump
condition
\begin{equation}
\mathbf{J}_+(\xi)=\sigma_2\mathbf{J}_-(\xi)\sigma_2,\quad \xi\in
\vec{\mathbb{R}}_+.
\end{equation}
\item[]\textbf{Singularities:}  Each of the points of $P_N$ is
a simple pole of $\mathbf{J}(w)$.  
If $y\in \nabla\subset P_N$ with $E(y)=\lambda^0_{N,k}$
for $k=0,\dots,N-1$, then
\begin{equation}
\mathop{\mathrm{Res}}_{w=y}\mathbf{J}(w)=\lim_{w\to y}
\mathbf{J}(w)\begin{bmatrix} 0 & 0 \\
\displaystyle 
(-1)^{k+1}\mathop{\mathrm{Res}}_{w=y}e^{2iQ(w;x,t)/\epsilon_N}\Pi_N(w) & 0
\end{bmatrix},
\label{eq:Jresiduenabla}
\end{equation}
and if $y\in\Delta\subset P_N$ with $E(y)=\lambda^0_{N,k}$ for $k=0,\dots,N-1$,
then
\begin{equation}
\mathop{\mathrm{Res}}_{w=y}\mathbf{J}(w)=\lim_{w\to y}
\mathbf{J}(w)\begin{bmatrix} 0 &  \displaystyle 
(-1)^{k+1}\mathop{\mathrm{Res}}_{w=y}e^{-2iQ(w;x,t)/\epsilon_N}\Pi_N(w)^{-1} \\
0 & 0
\end{bmatrix},
\label{eq:Jresiduedelta}
\end{equation}
where $\Pi_N(w)$ is re-defined as
\begin{equation}
\Pi_N(w):=\prod_{p\in\nabla}\frac{\sqrt{-w}+\sqrt{-p}}{\sqrt{-w}-\sqrt{-p}}
\cdot\prod_{q\in\Delta}\frac{\sqrt{-w}-\sqrt{-q}}{\sqrt{-w}+\sqrt{-q}},
\end{equation}
(note that this definition reduces to the earlier one, see
\eqref{eq:PiNw}
and \eqref{eq:PiNwagain}, 
in the special case when $\Delta=\emptyset$ and
hence $\nabla=P_N$).  These amount to one matrix-valued condition on
the residue of $\mathbf{J}(w)$ at each of its poles.
\item[]\textbf{Normalization:}  The following normalization condition holds:
\begin{equation}
\lim_{w\to\infty}\mathbf{J}(w)=\mathbb{I},
\end{equation}
where the limit is uniform with respect to angle for $|\arg(-w)|<\pi$.
\end{itemize}
\label{rhp:modifiedw}
\end{rhp}
It is clear that in passing from $\mathbf{H}$ to $\mathbf{J}$ the
nature of the residue conditions is changed near those points 
$y\in \Delta\subset P_N$ and is left unchanged near those points
$y\in\nabla=P_N\setminus\Delta$.
The special case of $\Delta=P_N$ and $\nabla=\emptyset$ corresponds to
reformulating the inverse problem in terms of scattering theory ``from
the left''.  Indeed, we see that in this case the exponential
$e^{2iQ(w;x,t)/\epsilon_N}$ has been completely replaced by its
reciprocal, which makes the limit $x\to -\infty$ particularly
transparent; from the conditions of this problem it is easy to prove
that if $\Delta=P_N$ then $\mathbf{J}_N(w;x,t)\to\mathbb{I}$
as $x\to -\infty$ whenever $t\in\mathbb{R}$ and
$w\in\mathbb{C}\setminus (P_N\cup \mathbb{R}_+)$.  More generally,
this calculation suggests that if $x\in\mathbb{R}$ is a value for
which semiclassical asymptotic analysis of $\mathbf{H}(w)$ is
difficult, the problem may be resolved by analyzing the equivalent
matrix $\mathbf{J}(w)$ instead for a particular choice of
$\Delta$.

From now on we will take Riemann-Hilbert Problem~\ref{rhp:modifiedw}
as the basic object of study in the limit $N\to\infty$, where
the set $\Delta$ is to be chosen differently for different
$(x,t)\in\mathbb{R}^2$ to facilitate the asymptotic analysis.  

Having formulated Riemann-Hilbert Problem~\ref{rhp:modifiedw}, we can
now explain how the plots in \S\ref{sec:exactsolutions} were made.  For any
choice of $\Delta\subset P_N$ it is easy to see that $\mathbf{J}(w)$
is a rational function of $z=i\sqrt{-w}$ and it may therefore be
written as a finite partial fraction expansion with simple
denominators and constant term $\mathbb{I}$ to satisfy the
normalization condition.  The matrix coefficients in the expansion are
then determined from the residue conditions \eqref{eq:Jresiduenabla}
and \eqref{eq:Jresiduedelta}, which imply a square system of linear
equations to be solved for the coefficients.  In this way, the
construction of $\sin(\tfrac{1}{2}u_N(x,t))$ and $\cos(\tfrac{1}{2}u_N(x,t))$ 
may be reduced
to a finite-dimensional (of dimension proportional to $N$) linear
algebra problem parametrized explicitly by $(x,t)\in\mathbb{R}^2$.
The linear algebra problem is ill-conditioned when $N$ is large (this
difficulty can be partly ameliorated by judicious choice of
$\Delta$ given $x$ and $t$), but nonetheless implementing this approach
numerically allows one to explore the phenomenology of the
semiclassical limit while avoiding many traditional pitfalls (for
example, stiffness and propagation of errors) of direct numerical
simulation of the Cauchy problem for the sine-Gordon equation
\eqref{eq:SG} when $\epsilon$ is small.  For more details about
implementation of this direct approach to inverse scattering see our
recent paper \cite{BuckinghamM08}.

\section{Elementary Transformations of $\mathbf{J}(w)$}
We now embark upon a sequence of explicit invertible transformations
with the aim of converting Riemann-Hilbert Problem~\ref{rhp:modifiedw}
into an equivalent one that is better suited to asymptotic analysis in
the limit $N\to\infty$.  Throughout the rest of the paper we will use the following notation
for the composition of the WKB phase integral $\Psi$ with the function $E$:
\begin{equation}
\theta_0(w):=\Psi(E(w)).
\label{eq:theta0def}
\end{equation}
\subsection{Choice of $\Delta$}
\label{sec:choiceofDelta}
The set $P_N$ can be decomposed into a disjoint union
$P_N=P_N^\breather \cup P_N^\kink$ with $P_N^\breather\cap
P_N^\kink=\emptyset$.  Here $P_N^\breather$ consists of $N_\breather$
nonreal complex-conjugate pairs of complex numbers on the unit circle $S^1$
in the $w$-plane, while $P_N^\kink$ consists of $N_\kink$ pairs of
reciprocal negative real numbers (\emph{i.e.} of the form $(w,1/w)$
with $w<0$), none equal to $-1$.  Since each conjugate pair of poles
contributes to the fluxon condensate one breather soliton while each
negative pole contributes one kink soliton, the fluxon condensate may
be viewed as a nonlinear superposition of $2N_\kink$ kinks and
$N_\breather$ breathers.  As a consequence of
Assumption~\ref{assume:rotational}, both $N_\kink$ and $N_\breather$ are
proportional to $N$ and as $N\to\infty$, $P_N^\breather$ fills out the
whole unit circle while $P_N^\kink$ fills out a negative interval of
the form $[\mathfrak{a},\mathfrak{b}]$ where:
\begin{equation}
\mathfrak{a}:=-\frac{1}{4}\left(\sqrt{G(0)^2-4}-G(0)\right)^2, \quad
\mathfrak{b}:=-\frac{1}{4}\left(\sqrt{G(0)^2-4}+G(0)\right)^2=\frac{1}{\mathfrak{a}}.
\label{eq:abdef}
\end{equation}
Note that both $\mathfrak{a}$ and $\mathfrak{b}$ are independent of $N$ and that $\mathfrak{a}<-1<\mathfrak{b}<0$.
If Assumption~\ref{assume:rotational} were not satisfied, that is, if 
$G(0)>-2$, then $P_N^\kink$ would be empty and $P_N^\breather$ would 
fill out a proper sub-arc of the unit circle as $N\to\infty$.

We will consider six different configurations for the subset
$\Delta\subset P_N$.  To each fixed number $\tau_\infty\in (\mathfrak{a},-1)\cup
(-1,\mathfrak{b})$ we associate a sequence $\{\tau_N\}_{N=N_0}^\infty$ of real
numbers with limit $\tau_\infty$ by
\begin{equation}
\theta_0(\tau_N)=\pi\epsilon_N\left\lfloor
\frac{\theta_0(\tau_\infty)}{\pi\epsilon_N}
\right\rfloor,\quad N\ge N_0.
\label{eq:transitionpoint}
\end{equation}
(This equation can be solved for $\tau_N$ near 
$\tau_\infty$ for $N\ge N_0$
with $N_0$ sufficiently large by the Implicit Function Theorem since
the only critical point of $\theta_0(w)$ in $(\mathfrak{a},\mathfrak{b})$, namely
$w=-1$, has been excluded.)  According to the Bohr-Sommerfeld
quantization rule \eqref{eq:BohrSommerfeld}, when $N$ is large, $\tau_N$
is approximately halfway between two neighboring points in
$P_N^\kink$.  The six choices of $\Delta$ we consider are the
following:
\begin{itemize}
\item \underline{$\Delta=\emptyset$}. 
 In this case $\mathbf{J}(w)=\mathbf{H}(w)$ and
Riemann-Hilbert Problems~\ref{rhp:basicw} and~\ref{rhp:modifiedw} coincide.
\item \underline{$\Delta=P^{\prec\kink}_{N}$}. 
In this case we choose $\mathfrak{a}<\tau_\infty<-1$ and set 
$\Delta=P_{N}^{\prec\kink}:=P_N^\kink\cap [\mathfrak{a},\tau_N]$.  Thus $\Delta$ is
localized near $w=\mathfrak{a}$.
\item \underline{$\Delta=P^{\kink\succ}_N$}. 
In this case we choose $-1<\tau_\infty<\mathfrak{b}$ and set
$\Delta=P_N^{\kink\succ}:=P_N^\kink\cap [\tau_N,\mathfrak{b}]$.  Thus $\Delta$
is localized near $w=\mathfrak{b}$.
\item \underline{$\nabla=\emptyset$}.
This is complementary to the case when $\Delta=\emptyset$.
\item \underline{$\nabla=P^{\prec\kink}_N$}.
In this case we choose $\mathfrak{a}<\tau_\infty<-1$ and set 
$\nabla=P^{\prec\kink}_N:=P_N^\kink\cap [\mathfrak{a},\tau_N]$.  Thus $\nabla$
is localized near $w=\mathfrak{a}$.
\item \underline{$\nabla=P_N^{\kink\succ}$}. 
In this case we choose $-1<\tau_\infty<\mathfrak{b}$ and set
$\nabla=P_N^{\kink\succ}:=P_N^\kink\cap [\tau_N,\mathfrak{b}]$.  Thus
$\nabla$ is localized near $w=\mathfrak{b}$.
\end{itemize}
We will refer to $\tau_N$ 
(and sometimes also its limit $\tau_\infty$ as $N\to\infty$) as a
\emph{transition point}.  One consequence of the definition 
\eqref{eq:transitionpoint}
in relation to the Bohr-Sommerfeld quantization rule \eqref{eq:BohrSommerfeld}
and the fact that $P_N$ contains an even number of points is the identity
\begin{equation}
\frac{\theta_0(\tau_N)}{\pi\epsilon_N} = \#\Delta \pmod{2}
\label{eq:cardDelta}
\end{equation}
where $\#\Delta$ denotes the number of points in $\Delta$.

For most $(x,t)$, specifically outside of a small neighborhood of $(x,t)=(0,0)$, the only 
cases we will use are $\Delta=\emptyset$ and $\nabla=\emptyset$.  
In fact, $\Delta=\emptyset$ is the only case needed to analyze the 
behavior for small times and positive $x$ bounded away from 
the origin.  It is then possible to use evenness of the 
Cauchy data (Assumption 
\ref{assume:evenness}) to obtain results for negative $x$.  However, we 
include the case $\nabla=\emptyset$ for completeness (and to admit future generalizations of our methodology to non-even Cauchy data). The analysis of the neighborhood of the origin that relies 
on the four additional cases in which neither $\Delta$ nor $\nabla$ is empty is somewhat more complicated and is presented in a self-contained fashion in \S \ref{continuation-from-x=0}.  A reader
not interested in these details can safely ignore all of the cases except for $\Delta=\emptyset$
and $\nabla=\emptyset$ along with references to any transition points, and may also skip most of \S\ref{continuation-from-x=0}.

\subsection{Interpolation of residues.  Removal of poles}
\label{sec:interpolation}
As a consequence of the Bohr-Sommerfeld quantization rule
\eqref{eq:BohrSommerfeld}, whenever $y\in P_N$ with $E(y)=\lambda^0_{N,k}$
then
\begin{equation}
\mp ie^{\pm i\theta_0(y)/\epsilon_N}=(-1)^k,\quad k=0,\dots,N-1.
\label{eq:interpolatew}
\end{equation}
Thus the exponential function on the left-hand side analytically
interpolates the signs $(-1)^k$ at the corresponding poles of
$\mathbf{J}$.  This gives two different ways to combine the sign
$(-1)^{k+1}$ appearing in the nilpotent residue matrices in the
conditions \eqref{eq:Jresiduenabla} and \eqref{eq:Jresiduedelta} with
the residue factor, a fact we can use to formulate an invertible transformation
of $\mathbf{J}$ into another matrix $\mathbf{M}$ that will not have any
point singularities whatsoever.

\begin{figure}[h]
\begin{center}
\includegraphics{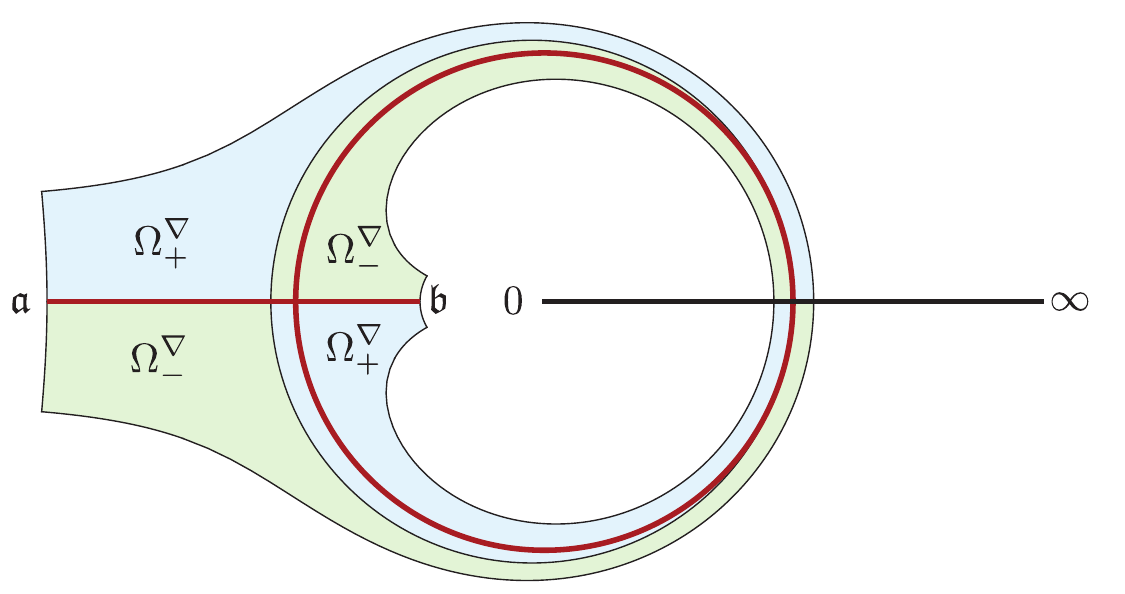}
\end{center}
\caption{\emph{The case of $\Delta=\emptyset$.  Here
    $\Omega_\pm^\Delta=\emptyset$.  The set
    $P_\infty:=[\mathfrak{a},\mathfrak{b}]\cup S^1$ in which $P_N$ accumulates for large $N$ is
    shown in dark red for reference, and the branch cut $\mathbb{R}_+$
    is shown with a black line.}}
\label{fig:Ann_DeltaEmpty}
\end{figure}

\begin{figure}[h]
\begin{center}
\includegraphics{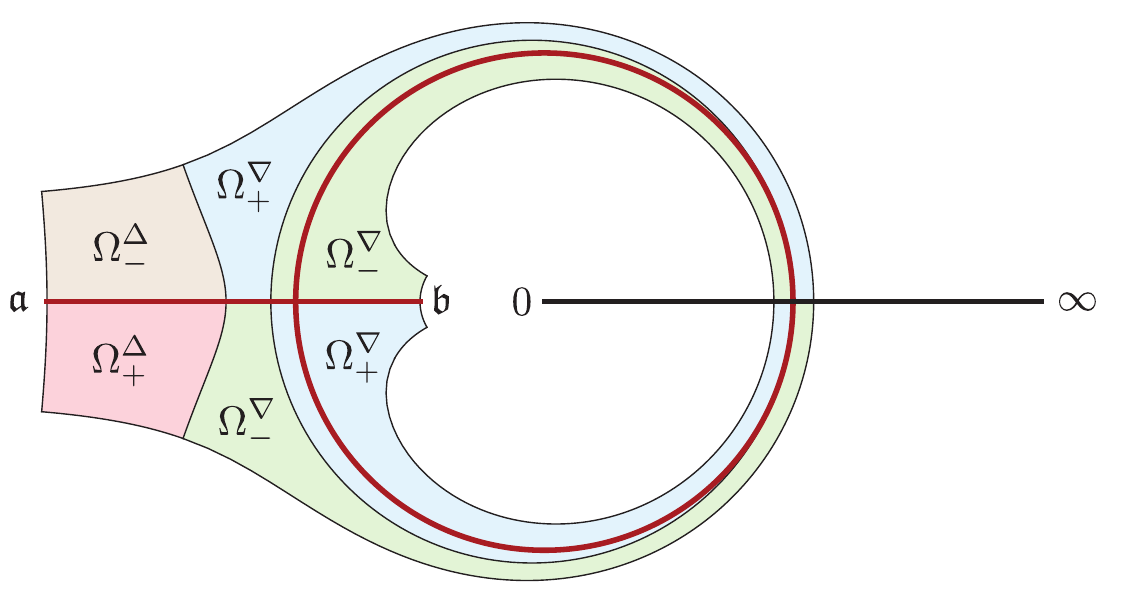}
\end{center}
\caption{\emph{The case of $\Delta=P^{\prec\kink}_N$.  
The regions $\Omega_+^\Delta$
and $\Omega_-^\Delta$ abut the real axis in the interval $(\mathfrak{a},\tau_N)$ where
$\mathfrak{a}<\tau_\infty<\min(w^+,-1)$.}}
\label{fig:Ann_DeltaNearMinusM}
\end{figure}

\begin{figure}[h]
\begin{center}
\includegraphics{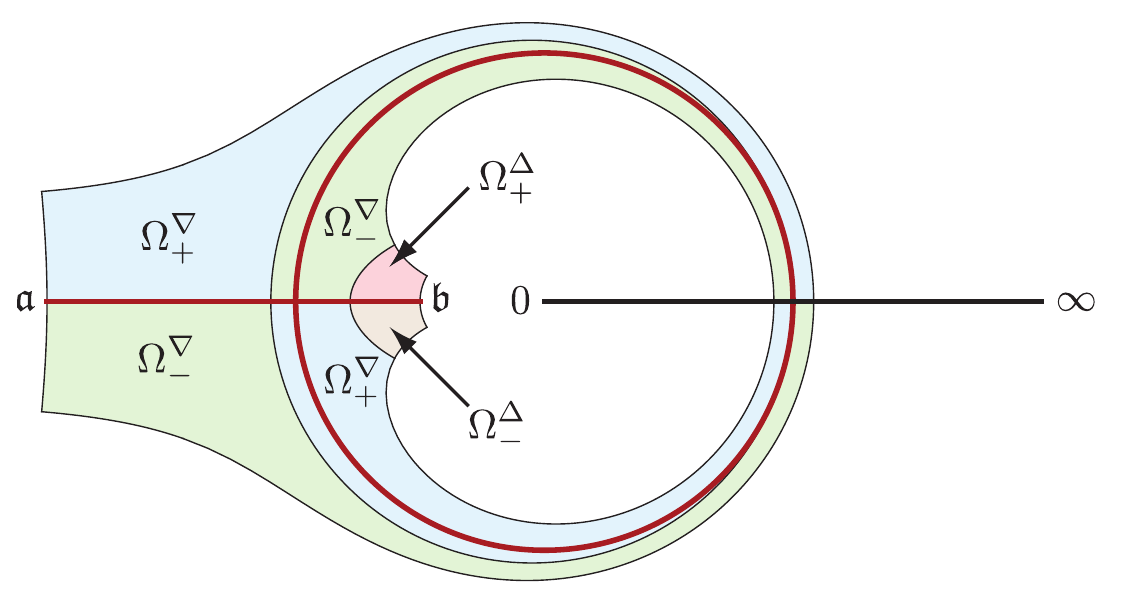}
\end{center}
\caption{\emph{The case of $\Delta=P^{\kink\succ}_N$.  The regions
    $\Omega_+^\Delta$ and $\Omega_-^\Delta$ abut the real axis in the
    interval $(\tau_N,\mathfrak{b})$ where
    $\max(w^+,-1)<\tau_\infty<\mathfrak{b}$.}}
\label{fig:Ann_DeltaNearMinusOneOverM}
\end{figure}

\begin{figure}[h]
\begin{center}
\includegraphics{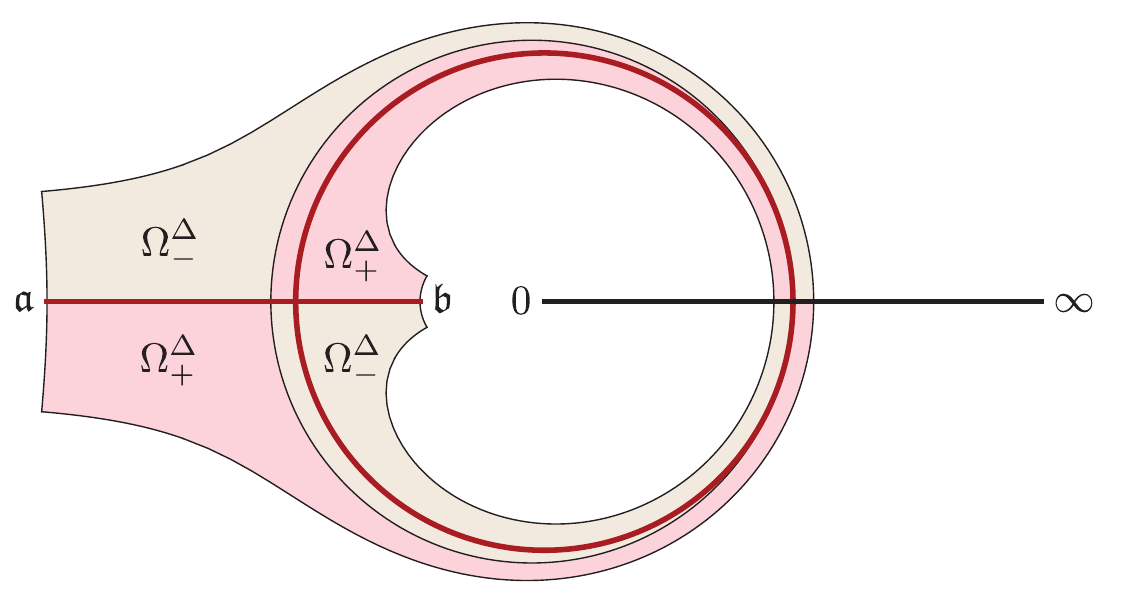}
\end{center}
\caption{\emph{The case of $\nabla=\emptyset$.  Here
    $\Omega_\pm^\nabla=\emptyset$.}}
\label{fig:Ann_NablaEmpty}
\end{figure}

\begin{figure}[h]
\begin{center}
\includegraphics{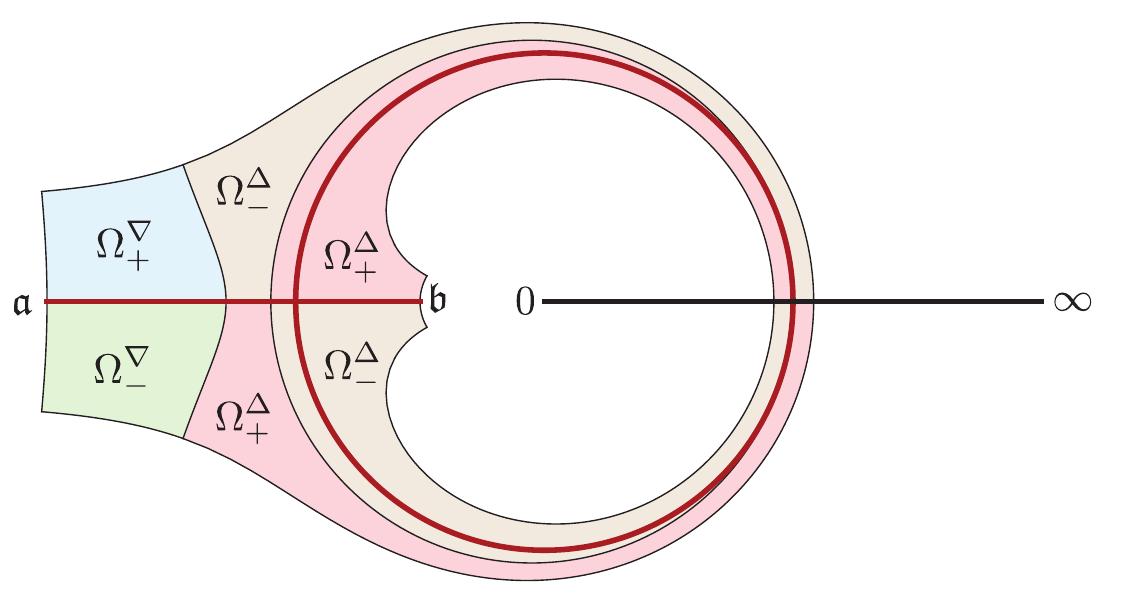}
\end{center}
\caption{\emph{The case of $\nabla=P^{\prec\kink}_N$.  
The regions $\Omega_+^\nabla$
and $\Omega_-^\nabla$ abut the real axis in the interval $(\mathfrak{a},\tau_N)$ where
$\mathfrak{a}<\tau_\infty<\min(w^+,1)$.}}
\label{fig:Ann_NablaNearMinusM}
\end{figure}

\begin{figure}[h]
\begin{center}
\includegraphics{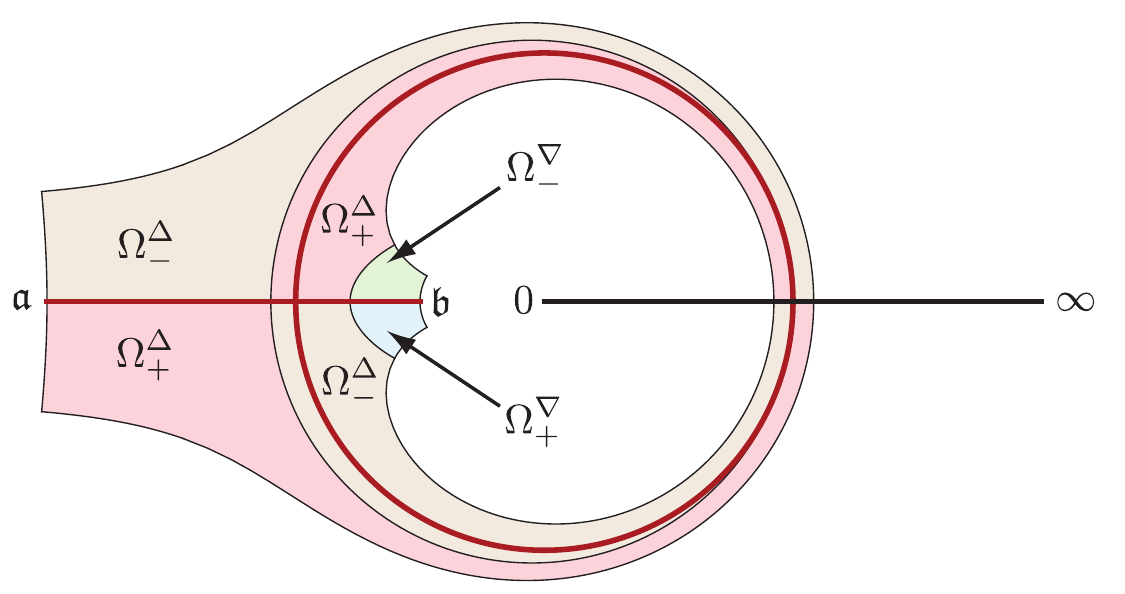}
\end{center}
\caption{\emph{The case of $\nabla=P^{\kink\succ}_N$.  
The regions $\Omega_+^\nabla$
and $\Omega_-^\nabla$ abut the real axis in the interval $(\tau_N,\mathfrak{b})$
where $\max(w^+,1)<\tau_\infty<\mathfrak{b}$.}}
\label{fig:Ann_NablaNearMinusOneOverM}
\end{figure}

We now introduce four disjoint open regions of the $w$-plane with
$|\arg(-w)|<\pi$, denoted $\Omega^\nabla_\pm$ and $\Omega^\Delta_\pm$,
such that $\overline{E(\Omega)}=\overline{D_+\cup D_-}$, where
$\Omega:=\Omega_+^\nabla\cup\Omega_-^\nabla\cup\Omega_+^\Delta\cup\Omega_-^\Delta$,
and where the open rectangles $D_\pm$ are as defined following the
statement of Proposition~\ref{prop:theta0}.  The details of the
definitions of these regions are different depending on which of the
six cases of choice of $\Delta$ we are considering (see
Figures~\ref{fig:Ann_DeltaEmpty}--\ref{fig:Ann_NablaNearMinusOneOverM}),
but the following features are common to all cases:
\begin{itemize}
\item $\overline{\Omega}$
is independent of $N$, $x$, and $t$, and always contains the pole
locus $P_N$ for all $N$.  Moreover, $\overline{\Omega^\nabla}$
contains $\nabla\subset P_N$ while $\overline{\Omega^\Delta}$ contains
$\Delta\subset P_N$, where $\Omega^\nabla:=\Omega_+^\nabla\cup
\Omega_-^\nabla$ and $\Omega^\Delta:=\Omega_+^\Delta\cup
\Omega_-^\Delta$.
This will
allow us to remove the poles from $\mathbf{J}$ by making appropriate
substitutions based on \eqref{eq:interpolatew} in $\Omega^\nabla_\pm$
and $\Omega^\Delta_\pm$.  
\item Schwartz symmetry is present:  $\Omega_-^\nabla=\Omega_+^{\nabla *}$
and $\Omega_-^\Delta=\Omega_+^{\Delta *}$.
\item Either the common boundary of $\Omega_+^\nabla$ with $\Omega_-^\nabla$
or that of $\Omega_+^\Delta$ with $\Omega_-^\Delta$ (but not both) contains
a Schwartz-symmetric closed curve that meets the real axis at $w=1$ and
exactly one other point, $w=w^+\in(\mathfrak{a},\mathfrak{b})$.  This closed curve
is allowed to cross the unit circle at points other than $w=1$, perhaps
nontangentially.
\item
Only the common boundary of $\Omega^\nabla_+$ and $\Omega^\Delta_-$,
or the common boundary of $\Omega^\Delta_+$ and $\Omega^\nabla_-$, 
will depend on $N$ (necessarily through the transition point $\tau_N$
where these curves meet the real axis); the other curves must be independent
of $N$. 
\end{itemize}
For convenience, we assume that the common boundary of $\Omega_+^\nabla$
and $\Omega_-^\Delta$, and also the common boundary of $\Omega^\Delta_+$
and $\Omega^\nabla_-$, are arcs of level curves of $\Im\{E(w)\}$ as
illustrated in 
Figures~\ref{fig:Ann_DeltaEmpty}--\ref{fig:Ann_NablaNearMinusOneOverM}.

In each of the six cases we will now transform the matrix $\mathbf{J}(w)$
into an equivalent matrix $\mathbf{M}(w)=\mathbf{M}_N(w;x,t)$ by the
following explicit formula:
\begin{equation}
\mathbf{M}(w):=\begin{cases}
\displaystyle
\mathbf{J}(w)\begin{bmatrix} 1 & 0 \\
\mp i\Pi_N(w)e^{[2iQ(w;x,t)\pm i\theta_0(w)]/\epsilon_N} & 1\end{bmatrix},
\quad & w\in\Omega_\pm^\nabla,\\\\
\displaystyle
\mathbf{J}(w)\begin{bmatrix}1 & \pm i\Pi_N(w)^{-1}e^{-[2iQ(w;x,t)\pm i
\theta_0(w)]/\epsilon_N}\\0 & 1\end{bmatrix},\quad
&w\in\Omega_\pm^\Delta,\\\\
\mathbf{J}(w),&w\in\mathbb{C}\setminus (\overline{\Omega}\cup\mathbb{R}_+).
\end{cases}
\label{eq:Mdefinew}
\end{equation}
Note that by this transformation, $\mathbf{M}$ inherits from
$\mathbf{J}$ the Schwartz symmetry $\mathbf{M}(w^*)=\mathbf{M}(w)^*$.

It is a consequence of the interpolation formula
\eqref{eq:interpolatew} and the residue conditions
\eqref{eq:Jresiduenabla} and \eqref{eq:Jresiduedelta} that
$\mathbf{M}(w)$ has only removable singularities in and up to the
boundary of each of the regions of its definition, and thus it may be
considered as a sectionally analytic function of $w$ taking continuous
boundary values on but having jump
discontinuities across an oriented contour $\Sigma$ that is
the positively-oriented boundary of
$\Omega_+^\nabla\cup\Omega_+^\Delta$ together with the
negatively-oriented boundary of $\Omega_-^\nabla\cup\Omega_-^\Delta$
and two intervals of the positive real axis, which we take to be both
oriented to the right.  

We introduce notation for various subcontours of $\Sigma$ 
as follows:
\begin{itemize}
\item The part of $\Sigma$ consisting of the common boundaries of
  $\Omega^\nabla_\pm$ (respectively $\Omega^\Delta_\pm$) omitting only
  the segments on the positive real axis we denote by $\Sigma^\nabla$
  (respectively $\Sigma^\Delta$).  
\item
The common boundary of $\Omega_+^\nabla$ and $\Omega_-^\Delta$ will
be denoted $\Sigma^{\nabla\Delta}$.  The common boundary of $\Omega_+^\Delta$
and $\Omega_-^\nabla$ will be denoted $\Sigma^{\Delta\nabla}$.  
\item The common boundary of $\Omega_\pm^\nabla$ (respectively
$\Omega_\pm^\Delta$) and $\overline{\Omega}$ will be denoted
$\Sigma_\pm^\nabla$ (respectively $\Sigma_\pm^\Delta$).  These are
just the non-real arcs of $\Sigma\setminus(\Sigma^\nabla\cup
\Sigma^\Delta\cup\Sigma^{\nabla\Delta}\cup\Sigma^{\Delta\nabla})$.
\item The part of the common boundary of the sets $\Omega_\pm^\nabla$ 
(respectively $\Omega_\pm^\Delta$) 
on the positive real line will be denoted $\Sigma^\nabla_{>0}$
(respectively $\Sigma^\Delta_{>0}$).
\item The contour $\mathbb{R}_+\setminus(\Sigma^\nabla_{>0}\cup
\Sigma^\Delta_{>0})$ will be denoted simply $\Sigma_{>0}$.  
\end{itemize}
The contour $\Sigma$ and its components are illustrated in Figure~\ref{fig:Sigma} for the case of $\Delta=P_N^{\prec\kink}$.
\begin{figure}[h]
\begin{center}
\includegraphics{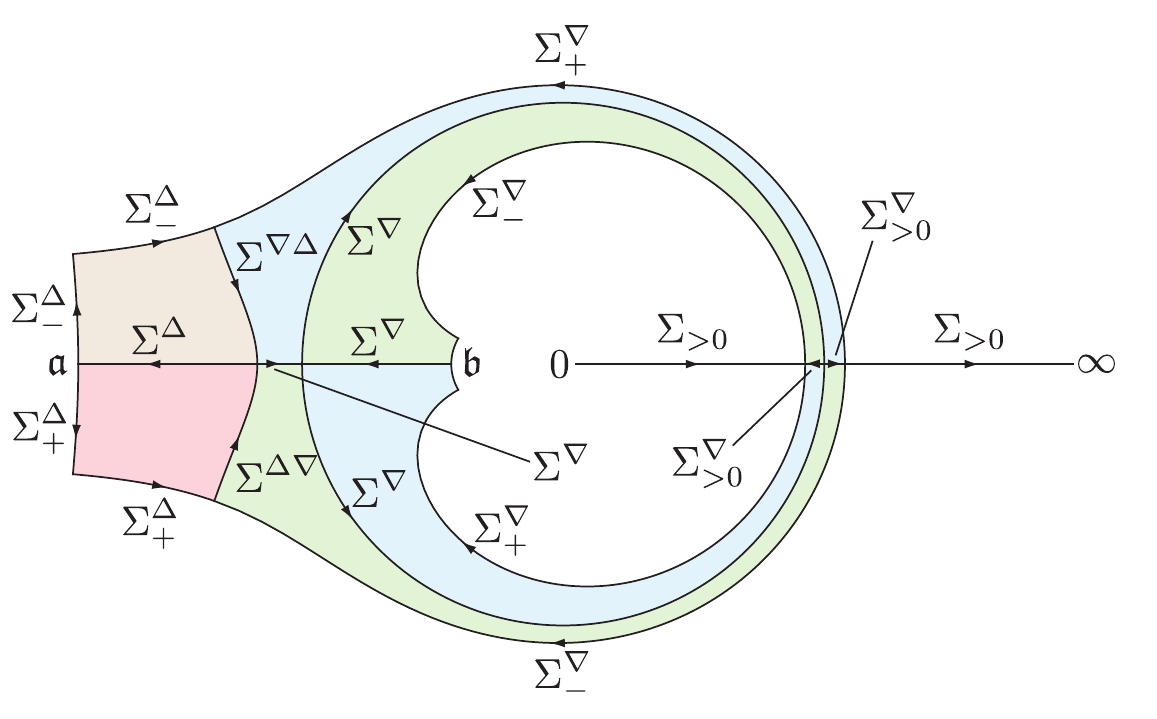}
\end{center}
\caption{\emph{The components of the contour $\Sigma$ in the case $\Delta=P_N^{\prec\kink}$.  Compare with Figure~\ref{fig:Ann_DeltaNearMinusM}.  Since in this case it is $\Omega_\pm^\nabla$
and not $\Omega_\pm^\Delta$ that meet the positive real axis near $w=1$, the contour component $\Sigma^\Delta_{>0}$ does not exist.  It is generally the case that, aside from the two disconnected components of $\Sigma_{>0}$ which
are always taken to be oriented to the right, all contour arcs are oriented with regions labelled ``$+$''
on the left and regions labelled ``$-$'' on the right.}}
\label{fig:Sigma}
\end{figure}
We also make the accumulation set of $P_N$,
 $P_\infty:=[\mathfrak{a},\mathfrak{b}]\cup S^1$, into an oriented contour by using
homotopy to $\Sigma^\nabla\cup\Sigma^\Delta$ fixing the points $w=\mathfrak{a}$,
$w=\tau_N$, and $w=\mathfrak{b}$ to assign orientation.  Note that
$P_\infty$, taken without regard to orientation, is the inverse image
under $\lambda=E(w)$ of the interval $0\le -4i\lambda\le -G(0)$.  The
contour $\Sigma^\nabla\cup \Sigma^\Delta$ followed by $-P_\infty$
forms the boundary of a bounded region that we denote by
$Z\subset\mathbb{C}$.  We write $Z=Z_+\cup Z_-$, where according to
the orientation of $\Sigma^\nabla\cup\Sigma^\Delta\cup (-P_\infty)$,
$\partial Z_+$ is positively oriented while $\partial Z_-$ is
negatively oriented.

\subsection{Analysis of the jump conditions for $\mathbf{M}(w)$}
The jump conditions satisfied by the continuous boundary values 
of $\mathbf{M}(w)$ across the contour $\Sigma$ will involve
the product $\Pi_N(w)$, and to formulate the jump conditions
in a concise way it will be convenient to consider the asymptotic
behavior of $\Pi_N(w)$ as $N\to\infty$.  But first, we introduce some
elementary logarithm-type functions.
For $y\not\in\mathbb{R}_+$, let
\begin{equation}
l(w;y):=\log\left(\frac{\sqrt{-w}+\sqrt{-y}}{\sqrt{-w}-\sqrt{-y}}\right)
\end{equation}
denote the principal branch, obtained by composing the principal
branches of the logarithm and the square roots.  For fixed $y$, this
is a function of $w$ that is single-valued and analytic except along a
piecewise-linear branch cut consisting of a finite line segment from
$w=y$ to $w=0$ along which we have the logarithmic jump condition
$l_+(w;y)-l_-(w;y)=2\pi i$ and the semi-infinite ray
$\mathbb{R}_+$ along which we have $l_+(w;y)+l_-(w;y)=0$.  As
$w\to\infty$ in the domain of analyticity, $l(w;y)\to 0$.  For $y<0$,
we define a new function by setting
\begin{equation}
m^\kink(w;y):=l(w;y)+\begin{cases}-i\pi,\quad & \text{$|w|<1$ and $\Im\{w\}>0$},\\
i\pi,\quad &\text{$|w|<1$ and $\Im\{w\}<0$},\\
0,\quad &\text{$|w|>1$.}
\end{cases}
\end{equation}
If $y\le -1$, then $m^\kink(w;y)$ has a unique analytic continuation to 
the interval
$w\in (-1,0)$ and has branch cuts consisting of the intervals $[y,-1]$ and
$\mathbb{R}_+$ as well as the upper and lower arcs of the unit circle.
If $-1\le y<0$, then $m^\kink(w;y)$ has a unique analytic continuation to the
interval $w\in (y,0)$ and has branch cuts consisting of the intervals $[-1,y]$
and $\mathbb{R}_+$ as well as the upper and lower arcs of the unit circle.
We also define a function $m^{\kink,\Sigma}(w;y)$ for $y<0$ by a completely 
analogous
formula in which the conditions $|w|<1$ and $|w|>1$ are respectively replaced
by the conditions that $w$ lie inside and outside of the region bounded by the
nonreal arcs of $\Sigma^\nabla\cup\Sigma^\Delta$.  Finally, if $y=-e^{i\omega}$
with $-\pi<\omega < \pi$, then we set
\begin{equation}
m^\breather(w;y):=l(w;y)+\begin{cases}-i\pi,\quad &\text{$|w|<1$ and $-\pi<\arg(-w)<\omega$},\\
i\pi,\quad &\text{$|w|<1$ and $\omega<\arg(-w)<\pi$},\\
0,\quad & \text{$|w|>1$.}
\end{cases}
\end{equation}
This function has a unique analytic continuation to the line segment
connecting $w=y$ with $w=0$, and its branch cuts consist of the unit
circle and the interval $\mathbb{R}_+$.

Now, set
\begin{equation}
L^0_N(w):=\mathop{\sum_{y\in\nabla}}_{\Im\{y\}=0}m^\kink(w;y)\epsilon_N -
\mathop{\sum_{y\in\Delta}}_{\Im\{y\}=0}m^\kink(w;y)\epsilon_N +
\mathop{\sum_{y\in\nabla}}_{\Im\{y\}\neq 0}m^\breather(w;y)\epsilon_N -
\mathop{\sum_{y\in\Delta}}_{\Im\{y\}\neq 0}m^\breather(w;y)\epsilon_N.
\end{equation}
(For the various choices of $\Delta$ under consideration, only one of
the latter two sums will be present in any one case.)  This function
is analytic for $w\in\mathbb{C}\setminus (P_\infty\cup\mathbb{R}_+)$
because each summand is.  It also satisfies the identity
$L^0_N(w^*)=L^0_N(w)^*$; indeed this holds term-by-term in the first
two sums, and in each of the last two sums (whichever one is present)
we may group the indices $y$ in complex-conjugate pairs and the
desired Schwartz symmetry holds for each pair.  The boundary values taken
by $L^0_N(w)$ on $\mathbb{R}_+$ satisfy $L^0_{N+}(w)+L^0_{N-}(w)=0$ as this
identity holds term-by-term in each of the sums.  Finally, we have the
identity
\begin{equation}
\Pi_N(w)=e^{L_N^0(w)/\epsilon_N},\quad w\in\mathbb{C}\setminus (P_\infty\cup\mathbb{R}_+).
\end{equation}
This identity would be completely obvious were it not for the $\pm i\pi$
contributions in the definitions of $m^\kink(w;y)$ and
$m^\breather(w;y)$; that it holds in the presence of these contributions
follows from the fact that $P_N=\nabla\cup\Delta$ consists of an even number
of points.

Along any arc of $P_\infty^\nabla$, the part of $P_\infty$ that is the
accumulation set of $\nabla$, we may enumerate the points of $\nabla$
in order $\{\dots,y_{n-1},y_n,y_{n+1},\dots\}$ consistent with the
local orientation of $P_\infty$.  Note that this order is the same as that of
increasing $k$ in the Bohr-Sommerfeld quantization rule
\eqref{eq:BohrSommerfeld}.  Thus, from \eqref{eq:BohrSommerfeld} we
have
\begin{equation}
\theta_0(y_{n+1})-\theta_0(y_{n-1})=2\pi\epsilon_N.
\end{equation}
Expanding the left-hand side around $y_n$ gives the spacing between
consecutive points of $\nabla$ as
\begin{equation}
\Delta y(y_n):=\frac{y_{n+1}-y_{n-1}}{2} =
\frac{\pi\epsilon_N}{\theta_0'(y_n)}
+ \bo(\epsilon_N^3).
\end{equation}
Therefore, for each $w\in \mathbb{C}\setminus (P_\infty\cup\mathbb{R}_+)$,
\begin{multline}
\mathop{\sum_{y\in\nabla}}_{\Im\{y\}=0}m^\kink(w;y)\epsilon_N +
\mathop{\sum_{y\in\nabla}}_{\Im\{y\}\neq 0}m^\breather(w;y)\epsilon_N \\
\begin{split}
{}&= 
\frac{1}{\pi}\mathop{\sum_{y\in\nabla}}_{\Im\{y\}=0}
\theta_0'(y)
m^\kink(w;y)\,\Delta y(y) +
\frac{1}{\pi}\mathop{\sum_{y\in\nabla}}_{\Im\{y\}\neq 0}
\theta_0'(y)
m^\breather(w;y)\,\Delta y(y)\\
{}&= \frac{1}{\pi}\int_{P_\infty^\nabla\cap\mathbb{R}}
\theta_0'(y)
m^\kink(w;y)\,dy + \frac{1}{\pi}\int_{P_\infty^\nabla\cap(\mathbb{C}\setminus\mathbb{R})}
\theta_0'(y)
m^\breather(w;y)\,dy + \bo(\epsilon_N^2),
\end{split}
\end{multline}
where the second-order accuracy comes from the fact that the sum is a
midpoint rule approximation to the integral away from $y=-1$ (the
extreme sample points are asymptotically half as far from the
endpoints of $P_\infty^\nabla$ as they are from their neighbors due to
the $1/2$ in the Bohr-Sommerfeld quantization rule
\eqref{eq:BohrSommerfeld} and the choice \eqref{eq:transitionpoint} of
the transition point $w=\tau_N$) and the fact that $E'(-1)=0$.  
In a similar way, but taking into
account that $P_\infty^\Delta = \overline{P_\infty\setminus
  P_\infty^\nabla}$ is oriented \emph{oppositely} to the direction of
increasing $k$ in \eqref{eq:BohrSommerfeld},
\begin{multline}
-\mathop{\sum_{y\in\Delta}}_{\Im\{y\}=0}m^\kink(w;y)\epsilon_N 
-\mathop{\sum_{y\in\Delta}}_{\Im\{y\}\neq 0}m^\breather(w;y)\epsilon_N \\
{}= 
\frac{1}{\pi}\int_{P_\infty^\Delta\cap\mathbb{R}}
\theta_0'(y)
m^\kink(w;y)\,dy + 
\frac{1}{\pi}\int_{P_\infty^\Delta\cap(\mathbb{C}\setminus\mathbb{R})}
\theta_0'(y)
m^\breather(w;y)\,dy + \bo(\epsilon_N^2).
\end{multline}
The error terms are in fact uniform in $w$ for $w$
bounded away from $P_\infty$, and in particular, for such $w$,
\begin{equation}
L^0_N(w)=L^0(w)+\bo(\epsilon_N^2)\quad\text{and}\quad
\Pi_N(w)e^{-L^0(w)/\epsilon_N}=1+\bo(\epsilon_N),\quad N\to\infty,
\label{eq:PiNDeltaOuterApprox}
\end{equation}
where
\begin{equation}
L^0(w):=\frac{1}{\pi}\int_{P_\infty\cap\mathbb{R}}
\theta_0'(y)
m^\kink(w;y)\,dy + \frac{1}{\pi}\int_{P_\infty\cap(\mathbb{C}\setminus\mathbb{R})}
\theta_0'(y)
m^\breather(w;y)\,dy.
\end{equation}
This ``continuum limit'' of $L^0_N(w)$ is analytic for 
$w\in\mathbb{C}\setminus (P_\infty\cup\mathbb{R}_+)$ and shares the same 
symmetry properties
as its discretization:  $L^0(w^*)=L^0(w)^*$, and $L^0_+(w)+L^0_-(w)=0$
for $w\in\mathbb{R}_+$.

Another useful form of $L^0(w)$ is easily obtained by integrating by
parts.  In the integral over $P_\infty\cap\mathbb{R}$, 
one expects endpoint contributions 
from the two
points $y=\mathfrak{a}$ and $y=\mathfrak{b}$, as well as from the points $y=-1$ and 
$y=\tau_N$ (the latter only if neither $\Delta$ nor $\nabla$ is empty so
that a transition point exists)
where the orientation of $P_\infty$ changes.  
In the integral over $P_\infty\cap(\mathbb{C}\setminus\mathbb{R})$
one expects endpoint contributions from the two points at
$y=1$ on opposite sides of the branch cut $\mathbb{R}_+$
and from the two points at $y=-1$.  
However, since $m^\kink(w;-1)=m^\breather(w;-1)$ the contributions from $y=-1$ 
will cancel between the two integrals.  
Also, in the integral over $P_\infty\cap(\mathbb{C}\setminus\mathbb{R})$
the sum of the contributions from the two endpoints at $y=1$ vanishes 
because $E(y)=0$ for both points and $m^\breather(w;1+)+m^\breather(w;1-)=0$.  
Furthermore, in the integral over $P_\infty\cap\mathbb{R}$ the contributions 
from
the endpoints $y=\mathfrak{a}$ and $y=\mathfrak{b}$ both vanish individually because both
correspond to $E(y)=-iG(0)/4$, the point at which $\theta_0(y)$
vanishes.  On the other hand, the contribution from $y=\tau_N$ generally
survives, but it takes a particularly simple form: $\pm
2\theta_0(\tau_N)m^\kink(w;\tau_N)/\pi$, 
where the ``$+$'' sign (respectively
``$-$'' sign) corresponds to the case when $P_\infty$ is oriented toward
(respectively away from) the transition point $y=\tau_N$.  
But according
to the condition \eqref{eq:transitionpoint} characterizing the
transition point $\tau_N$, this contribution can be written as
$2\epsilon_Nn m^\kink(w;\tau_N)$ where $n\in\mathbb{Z}$, and according to 
\eqref{eq:cardDelta}, $n=\#\Delta\pmod{2}$.  Explicitly
differentiating $m^\kink(w;y)$ and $m^\breather(w;y)$ with respect to $y$ 
to integrate by parts then yields
\begin{equation}
L^0(w)=\frac{\sqrt{-w}}{\pi}\int_{P_\infty}\frac{\theta_0(y)}{\sqrt{-y}}
\frac{dy}{y-w} + 2\epsilon_Nnm^\kink(w;\tau_N).
\end{equation}
This form of $L^0(w)$ allows us to exploit analyticity of
$\theta_0(\cdot)$ on $P_\infty$ to apply Cauchy's Theorem, showing that
\begin{equation}
L^0(w)=L(w)+\begin{cases} \mp 2i\theta_0(w)\pmod{2\pi i\epsilon_N},\quad &
w\in Z_\pm\\
0,\quad &w\in\mathbb{C}\setminus\overline{Z},
\end{cases}
\label{eq:L0Lrelation}
\end{equation}
where
\begin{equation}
L(w):=\frac{\sqrt{-w}}{\pi}\int_{\Sigma^\nabla\cup
\Sigma^\Delta}\frac{\theta_0(y)}{\sqrt{-y}}\frac{dy}{y-w}
+2\epsilon_Nnm^{\kink,\Sigma}(w;\tau_N). 
\label{eq:L0L}
\end{equation}
In this definition of $L(w)$, the function $\theta_0(y)$ in the
integrand denotes the analytic continuation of the function of the
same name from $P_\infty$. The domain of analyticity of $L(w)$ is
therefore $\mathbb{C}\setminus
(\Sigma^\nabla\cup\Sigma^\Delta\cup\mathbb{R}_+)$.  
Like $L^0(w)$, $L(w)$ has Schwartz symmetry ($L(w^*)=L(w)^*$) and satisfies
$L_+(w)+L_-(w)=0$ for $w\in\mathbb{R}_+$.  Furthermore, 
\begin{equation}
L_+(\xi)-L_-(\xi)=2i\theta_0(\xi)+\begin{cases}
2\pi i\epsilon_N\#\Delta \pmod{4\pi i\epsilon_N},\quad &\xi\in
(\Sigma^\nabla\cup\Sigma^\Delta)\cap(\mathbb{C}\setminus\mathbb{R})\\
0\pmod{4\pi i\epsilon_N},\quad &\xi\in(\Sigma^\nabla\cup\Sigma^\Delta)\cap\mathbb{R},
\end{cases}
\label{eq:Lplusminusdiff}
\end{equation}
and moreover the exact relation $L_+(\xi)-L_-(\xi)=2i\theta_0(\xi)$ holds
for those $\xi\in (\Sigma^\nabla\cup\Sigma^\Delta)\cap\mathbb{R}$ 
that do not lie on the branch cut of $m^{\kink,\Sigma}(w;\tau_N)$. 
We define for future reference
\begin{equation}
Y(w):=\Pi_N(w)e^{-L(w)/\epsilon_N},
\end{equation}
a quantity that by \eqref{eq:PiNDeltaOuterApprox} and \eqref{eq:L0L}
is uniformly
$1+\bo(\epsilon_N)$ for $w$ bounded away from $Z$ and $P_\infty$.

Let $\overline{L}(\xi)$ denote the average of boundary values taken by
$L(w)$ on $\Sigma^\nabla\cup\Sigma^\Delta$:
\begin{equation}
\overline{L}(\xi):=\frac{1}{2}\left(L_+(\xi)+L_-(\xi)\right),\quad
\xi\in\Sigma^\nabla\cup\Sigma^\Delta.
\end{equation}
We now introduce two quantities closely related to this average:
\begin{equation}
\varphi^\nabla(\xi):=\overline{L}(\xi)+
\begin{cases}
i\pi n\epsilon_N,\quad &\xi\in\Sigma^\nabla,\quad\Im\{\xi\}>0,\\
0,\quad & \xi\in\Sigma^\nabla,\quad\Im\{\xi\}=0,\\
-i\pi n\epsilon_N,\quad &\xi\in\Sigma^\nabla,\quad\Im\{\xi\}<0,
\end{cases}
\label{eq:varphinabla}
\end{equation}
and
\begin{equation}
\varphi^\Delta(\xi):=\overline{L}(\xi)+
\begin{cases}
i\pi n\epsilon_N,\quad &\xi\in\Sigma^\Delta,\quad\Im\{\xi\}>0,\\
0,\quad & \xi\in\Sigma^\Delta,\quad\Im\{\xi\}=0,\\
-i\pi n\epsilon_N,\quad &\xi\in\Sigma^\Delta,\quad\Im\{\xi\}<0.
\end{cases}
\label{eq:varphiDelta}
\end{equation}
It follows from Proposition~\ref{prop:theta0} that $\varphi^\nabla$ and
$\varphi^\Delta$ are analytic functions on each arc of $\Sigma^\nabla$
and $\Sigma^\Delta$ respectively.  The self-intersection point $\xi=w^+$
lies either in $\Sigma^\nabla$ or $\Sigma^\Delta$, and the additive
constants in the definitions \eqref{eq:varphinabla} and \eqref{eq:varphiDelta}
ensure that for the relevant function 
analyticity extends to a full complex neighborhood of $\xi=w^+$.
Now set
\begin{equation}
T^\nabla(\xi):=2\Pi_N(\xi)\cos\left(\epsilon_N^{-1}\theta_0(\xi)\right)
e^{-\varphi^\nabla(\xi)/\epsilon_N},\quad\xi\in\vec{\Sigma}^\nabla
\label{eq:Tnablanew}
\end{equation}
and
\begin{equation}
T^\Delta(\xi):=2\Pi_N(\xi)^{-1}\cos\left(\epsilon_N^{-1}\theta_0(\xi)\right)
e^{\varphi^\Delta(\xi)/\epsilon_N},\quad\xi\in\vec{\Sigma}^\Delta.
\label{eq:TDeltanew}
\end{equation}
According to the Bohr-Sommerfeld quantization rule
\eqref{eq:BohrSommerfeld}, these functions are 
analytic where defined (all singularities
are removable), and we shall denote the analytic continuations from the
contours of definition by the same symbols.  Moreover, it
can be proved that both $T^\nabla(w)\approx 1$ and
$T^\Delta(w)\approx 1$, a fact whose proof is only slightly more
subtle than the analysis of $Y(w)$ presented above (the proof given in
\cite{BaikKMM07} relates the product $\Pi_N(w)$ to the Gamma
function and applies the reflection identity
$\Gamma(z)\Gamma(1-z)\sin(\pi z)=\pi$ and Stirling's formula).  We
formalize our results concerning the functions $Y$, $T^\nabla$, and
$T^\Delta$ with this proposition:
\begin{proposition}[Baik \emph{et.\@ al.},\@ \cite{BaikKMM07}]
  The function $Y(w)$ is analytic for $w\in\mathbb{C}\setminus
  (\Sigma^\nabla\cup\Sigma^\Delta\cup P_\infty\cup \mathbb{R}_+)$.
  Uniformly on compact sets in the domain of analyticity disjoint from
  the region $Z$ in between $P_\infty$ and $\Sigma^\nabla\cup\Sigma^\Delta$,
\begin{equation}
Y(w)=1+\bo(\epsilon_N).
\label{eq:Yalmostone}
\end{equation}
This estimate fails if $w$ approaches either $\mathfrak{a}$ or $\mathfrak{b}$ from 
within the above region of definition, but nonetheless both $Y(w)$
and $Y(w)^{-1}$ remain bounded if these points are approached nontangentially
to the real axis.
On the other hand, uniformly on compact subsets of $Z_\pm$,
\begin{equation}
Y(w)=e^{\mp 2i\theta_0(w)/\epsilon_N}\left(1+\bo(\epsilon_N)\right).
\label{eq:YZpm}
\end{equation}

The function $T^\nabla(w)$ is analytic for $w\in \overline{\Omega_+^\nabla
\cup\Omega_-^\nabla}\setminus(\mathbb{R}_+\cup\{\mathfrak{a},\mathfrak{b},\tau_N\})$, 
and uniformly
on each compact set in the interior of this region,
\begin{equation}
T^\nabla(w)=1+\bo(\epsilon_N).
\label{eq:Tnablaalmostone}
\end{equation}
This estimate fails if $w$ approaches either $\mathfrak{a}$ or $\mathfrak{b}$ from within
the above region of definition, but nonetheless
$T^\nabla(w)$ remains uniformly bounded near
these points as $\epsilon_N\downarrow 0$.

The function $T^\Delta(w)$ is analytic for
$w\in\overline{\Omega_+^\Delta
  \cup\Omega_-^\Delta}\setminus(\mathbb{R}_+\cup\{\mathfrak{a},\mathfrak{b},\tau_N\})$, and
uniformly on each compact set in the interior of this region,
\begin{equation}
T^\Delta(w)=1+\bo(\epsilon_N).
\label{eq:TDeltaalmostone}
\end{equation}
As above, this estimate fails if 
$w$ approaches either $\mathfrak{a}$ or $\mathfrak{b}$ from within
the region where $T^\Delta(w)$ is defined but still $T^\Delta(w)$
remains bounded.
 
Finally, we have the algebraic relations
\begin{equation}
T^\nabla(w)=Y(w)\left(1+e^{\pm 2i\theta_0(w)/\epsilon_N}\right),\quad
w\in \overline{\Omega^\nabla_\pm}\setminus\mathbb{R},
\label{eq:TnablaoverY}
\end{equation}
and
\begin{equation}
T^\Delta(w)=
Y(w)^{-1}\left(1+e^{\pm 2i\theta_0(w)/\epsilon_N}\right),\quad w\in
\overline{\Omega^\Delta_\mp}\setminus\mathbb{R}.
\label{eq:TDeltatimesY}
\end{equation}
\label{prop:YTnablaDelta}
\end{proposition}
The proof of \eqref{eq:Yalmostone} and \eqref{eq:YZpm} has essentially
been given above, and the algebraic relations \eqref{eq:TnablaoverY}
and \eqref{eq:TDeltatimesY}
follow directly from the definitions of $Y(w)$, $T^\nabla(w)$, and
$T^\Delta(w)$ with the use of the jump condition
\eqref{eq:Lplusminusdiff}.  The asymptotic relations
\eqref{eq:Tnablaalmostone} and \eqref{eq:TDeltaalmostone} are easily
proved from the other results \emph{as long as $w$ is kept bounded
  away from $P_\infty$}.  For example, suppose $w\in \Omega_-^\Delta$, in
which case $T^\Delta(w)=Y(w)^{-1}(1+e^{2i\theta_0(w)/\epsilon_N})$
according to \eqref{eq:TDeltatimesY}, and $w$ lies on the right of
$\Sigma^\Delta$ so $w\not\in Z_+$.  If also 
$w\in E^{-1}(D_+)$, then $w$ lies
on the right of $P_\infty$ so $w\not\in Z_-$, in which case
\eqref{eq:Yalmostone} applies and we obtain \eqref{eq:TDeltaalmostone}
from the inequality $\Im\{\theta_0(w)\}>0$.  On the other hand if
also $w\in E^{-1}(D_-)$, then $w$ lies on the left of $P_\infty$ so $w\in
Z_-$, in which case \eqref{eq:YZpm} applies and we obtain
\eqref{eq:TDeltaalmostone} from the inequality
$\Im\{\theta_0(w)\}<0$.  Note, however, that the case when $w$ is
near $P_\infty$ must be considered separately; see \cite{BaikKMM07} for 
details.

With the notation of $Y(w)$, $T^\nabla(w)$, and $T^\Delta(w)$ established,
we may write down the jump condition satisfied by $\mathbf{M}(w)$ across
the various arcs of $\Sigma$ in a simple form.  Indeed,
we have
\begin{equation}
\mathbf{M}_+(\xi)=\mathbf{M}_-(\xi)\begin{bmatrix}
1 & 0 \\ -iT^\nabla(\xi)e^{[2iQ(\xi)+\varphi^\nabla(\xi)]/\epsilon_N} & 1
\end{bmatrix},
\quad \xi\in\vec{\Sigma}^\nabla,
\end{equation}
\begin{equation}
\mathbf{M}_+(\xi)=\mathbf{M}_-(\xi)\begin{bmatrix}
1 & iT^\Delta(\xi)e^{-[2iQ(\xi)+\varphi^\Delta(\xi)]/\epsilon_N} \\ 0 & 1
\end{bmatrix},
\quad \xi\in\vec{\Sigma}^\Delta,
\end{equation}
\begin{equation}
\mathbf{M}_+(\xi)=\mathbf{M}_-(\xi)\begin{bmatrix}
1+e^{2i\theta_0(\xi)/\epsilon_N} & iY(\xi)^{-1}e^{-[2iQ(\xi)+L(\xi)-
i\theta_0(\xi)]/\epsilon_N}\\
-iY(\xi)e^{[2iQ(\xi)+L(\xi)+i\theta_0(\xi)]/\epsilon_N} & 1\end{bmatrix}, \quad \xi\in \vec{\Sigma}^{\nabla\Delta},
\end{equation}
\begin{equation}
\mathbf{M}_+(\xi)=\mathbf{M}_-(\xi)\begin{bmatrix}
1 & iY(\xi)^{-1}e^{-[2iQ(\xi)+L(\xi)+i\theta_0(\xi)]/\epsilon_N}\\
-iY(\xi)e^{[2iQ(\xi)+L(\xi)-i\theta_0(\xi)]/\epsilon_N} & 
1+e^{-2i\theta_0(\xi)/\epsilon_N}
\end{bmatrix}, \quad \xi\in\vec{\Sigma}^{\Delta\nabla},
\end{equation}
\begin{equation}
\mathbf{M}_+(\xi)=\mathbf{M}_-(\xi)\begin{bmatrix}
1 & 0\\-iY(\xi)e^{[2iQ(\xi)+L(\xi)\pm i\theta_0(\xi)]/\epsilon_N} & 
1\end{bmatrix},
\quad \xi\in \vec{\Sigma}_\pm^\nabla,
\label{eq:MjumpSigmanablapm}
\end{equation}
\begin{equation}
\mathbf{M}_+(\xi)=\mathbf{M}_-(\xi)\begin{bmatrix}
1 & -iY(\xi)^{-1}e^{-[2iQ(\xi)+L(\xi)\mp i\theta_0(\xi)]/\epsilon_N}
\\ 0 & 1\end{bmatrix},
\quad \xi\in\vec{\Sigma}_\pm^\Delta,
\label{eq:MjumpSigmaDeltapm}
\end{equation}
\begin{multline}
\mathbf{M}_+(\xi)=\sigma_2\mathbf{M}_-(\xi)\sigma_2
\begin{bmatrix} 1+e^{i[\theta_{0+}(\xi)-\theta_{0-}(\xi)]/\epsilon_N} &
iY_-(\xi)e^{[2iQ_-(\xi)+L_-(\xi)-i\theta_{0-}(\xi)]/\epsilon_N}\\
-iY_+(\xi)e^{[2iQ_+(\xi)+L_+(\xi)+i\theta_{0+}(\xi)]/\epsilon_N} 
& 1\end{bmatrix},
\\
\xi\in\vec{\Sigma}^\nabla_{>0},
\label{eq:MjumpSigmanablagt0}
\end{multline}
\begin{multline}
\mathbf{M}_+(\xi)=\sigma_2\mathbf{M}_-(\xi)\sigma_2
\begin{bmatrix}
1 & iY_+(\xi)^{-1}e^{-[2iQ_+(\xi)+L_+(\xi)+i\theta_{0+}(\xi)]/\epsilon_N}\\
-iY_-(\xi)^{-1}e^{-[2iQ_-(\xi)+L_-(\xi)-i\theta_{0-}(\xi)]/\epsilon_N}
&
1+e^{-i[\theta_{0+}(\xi)-\theta_{0-}(\xi)]/\epsilon_N}\end{bmatrix},\\
\xi\in\vec{\Sigma}^\Delta_{>0},
\label{eq:MjumpSigmadeltagt0}
\end{multline}
and finally, 
\begin{equation}
\mathbf{M}_+(\xi)=\sigma_2\mathbf{M}_-(\xi)\sigma_2,\quad
\xi\in\vec{\Sigma}_{>0}.
\label{eq:MjumpSigmagt0}
\end{equation}
In deriving the jump conditions \eqref{eq:MjumpSigmanablagt0} and
\eqref{eq:MjumpSigmadeltagt0} we used the fact that for $\xi\in
\vec{\mathbb{R}}_+$, $Q_+(\xi)+Q_-(\xi)\equiv 0$ and
$\Pi_{N+}(\xi)\Pi_{N-}(\xi)\equiv 1$.  Written this
way, the jump relations display the key importance of the exponents
$2iQ(\xi)+L(\xi)\pm i\theta_0(\xi)$,
$2iQ(\xi)+\varphi^\nabla(\xi)$, and $2iQ(\xi)+\varphi^\Delta(\xi)$.  
It remains to introduce a mechanism
to control the corresponding exponentials, and this is the purpose of
the next transformation.

\subsection{Control of exponentials.  The so-called $g$-function}
\label{sec:g-function}
Let $g(w)$ be a scalar function analytic for
$w\in\mathbb{C}\setminus(\Sigma^\nabla\cup\Sigma^\Delta\cup\mathbb{R}_+)$.
In terms of this to-be-determined function, a new unknown can be
obtained from $\mathbf{M}(w)$ as follows:
\begin{equation}
\mathbf{N}(w):=\mathbf{M}(w)e^{-g(w)\sigma_3/\epsilon_N}.
\label{eq:NMw}
\end{equation}
It follows that
\begin{equation}
\mathbf{N}_+(\xi)=\mathbf{N}_-(\xi)\begin{bmatrix}
e^{-i\theta(\xi)/\epsilon_N} & 0\\-iT^\nabla(\xi)e^{\phi(\xi)/\epsilon_N} &
e^{i\theta(\xi)/\epsilon_N}\end{bmatrix},\quad \xi\in\vec{\Sigma}^\nabla,\quad
\Im\{\xi\}=0,
\label{eq:Njumpnabla}
\end{equation}
\begin{equation}
\mathbf{N}_+(\xi)=\mathbf{N}_-(\xi)\begin{bmatrix}
e^{-i\theta(\xi)/\epsilon_N} & 0\\-iT^\nabla(\xi)e^{[\phi(\xi)+i\pi \epsilon_N\#\Delta]/\epsilon_N} &
e^{i\theta(\xi)/\epsilon_N}\end{bmatrix},\quad \xi\in\vec{\Sigma}^\nabla,\quad
\Im\{\xi\}\neq 0,
\label{eq:Njumpnablacomplex}
\end{equation}
and
\begin{equation}
\mathbf{N}_+(\xi)=\mathbf{N}_-(\xi)\begin{bmatrix}
e^{-i\theta(\xi)/\epsilon_N} & iT^\Delta(\xi)e^{-\phi(\xi)/\epsilon_N}\\
0 & e^{i\theta(\xi)/\epsilon_N}\end{bmatrix},\quad
\xi\in\vec{\Sigma}^\Delta,\quad\Im\{\xi\}=0,
\label{eq:NjumpDelta}
\end{equation}
\begin{equation}
\mathbf{N}_+(\xi)=\mathbf{N}_-(\xi)\begin{bmatrix}
e^{-i\theta(\xi)/\epsilon_N} & iT^\Delta(\xi)e^{-[\phi(\xi)+i\pi\epsilon_N\#\Delta]/\epsilon_N}\\
0 & e^{i\theta(\xi)/\epsilon_N}\end{bmatrix},\quad
\xi\in\vec{\Sigma}^\Delta,\quad\Im\{\xi\}\neq 0,
\label{eq:NjumpDeltacomplex}
\end{equation}
where
\begin{equation}
\theta(\xi):=-i(g_+(\xi)-g_-(\xi))\quad\text{and}\quad
\phi(\xi):=2iQ(\xi)+\overline{L}(\xi)-g_+(\xi)-g_-(\xi),\quad
\xi\in\vec{\Sigma}^\nabla\cup\vec{\Sigma}^\Delta.
\label{eq:wthetaphidef}
\end{equation}
These definitions assume that the indicated boundary values of $g$
exist unambiguously (independent of direction of approach).  The key
to the Deift-Zhou steepest-descent method in this context is to choose
$g(w)$ and the non-real arcs of the contour $\Sigma^\nabla$ or
$\Sigma^\Delta$ so that the functions $\theta(\xi)$ and
$\phi(\xi)$ have properties that can be exploited to make the
matrix $\mathbf{N}(w)$ easy to approximate in the limit $\epsilon_N\to
0$.

Consider the analytic matrix functions defined as follows:
\begin{equation}
\mathbf{L}^\nabla(w):=\begin{cases}
\displaystyle 
T^\nabla(w)^{-\sigma_3/2}\begin{bmatrix}
1 & - ie^{-[2iQ(w)+L(w)- i\theta_0(w)-2g(w)]/\epsilon_N}\\
0 & 1\end{bmatrix},\quad & w\in\Omega^\nabla_+,\\\\
\displaystyle T^\nabla(w)^{-\sigma_3/2}\begin{bmatrix}
1 & ie^{-[2iQ(w)+L(w)+i\theta_0(w)-2g(w)]/\epsilon_N}\\
0 & 1\end{bmatrix},\quad & w\in\Omega^\nabla_-,
\end{cases}
\label{eq:Lnabla}
\end{equation}
\begin{equation}
\mathbf{L}^\Delta(w):=\begin{cases}
\displaystyle 
T^\Delta(w)^{\sigma_3/2}\begin{bmatrix}
1 & 0 \\
 ie^{[2iQ(w) + L(w) - i\theta_0(w)-2g(w)]/\epsilon_N} & 1\end{bmatrix},
\quad & w\in\Omega^\Delta_+,\\\\
\displaystyle
T^\Delta(w)^{\sigma_3/2}\begin{bmatrix}
1 & 0 \\
 -ie^{[2iQ(w) + L(w) + i\theta_0(w)-2g(w)]/\epsilon_N} & 1\end{bmatrix},
\quad & w\in\Omega^\Delta_-,
\end{cases}
\label{eq:LDelta}
\end{equation}
where the square roots $T^\nabla(w)^{1/2}$ and $T^\Delta(w)^{1/2}$ are
defined to be the principal branches, so that
in view of \eqref{eq:Tnablaalmostone} and \eqref{eq:TDeltaalmostone} from
Proposition~\ref{prop:YTnablaDelta} we have
$T^\nabla(w)^{1/2}\approx
1$ and $T^\Delta(w)^{1/2}\approx 1$ throughout the domain of definition
of $\mathbf{L}^\nabla(w)$ and $\mathbf{L}^\Delta(w)$ respectively.
\begin{proposition}
The jump conditions for $\mathbf{N}(w)$ on the contours $\Sigma^{\nabla\Delta}$
and $\Sigma^{\Delta\nabla}$ may be written in the form
\begin{equation}
\mathbf{N}_+(\xi)\mathbf{L}^\nabla_+(\xi) = 
\mathbf{N}_-(\xi)\mathbf{L}^\Delta_-(\xi),\quad\xi\in
\vec{\Sigma}^{\nabla\Delta}
\label{eq:OjumpnablaDelta}
\end{equation}
and
\begin{equation}
\mathbf{N}_+(\xi)\mathbf{L}^\Delta_+(\xi)=
\mathbf{N}_-(\xi)\mathbf{L}^\nabla_-(\xi),\quad 
\xi\in\vec{\Sigma}^{\Delta\nabla}.
\label{eq:OjumpDeltanabla}
\end{equation}
\label{prop:nabladelta}
\end{proposition}
\begin{proof}
  Note that on $\vec{\Sigma}^{\nabla\Delta}$ and $\vec{\Sigma}^{\Delta\nabla}$ 
 we
  have $Q_+(\xi)=Q_-(\xi)$, $L_+(\xi)=L_-(\xi)$,
  $g_+(\xi)=g_-(\xi)$, $Y_+(\xi)=Y_-(\xi)$, and
  $\theta_{0+}(\xi)=\theta_{0-}(\xi)$.  The formulae
  \eqref{eq:OjumpnablaDelta} and \eqref{eq:OjumpDeltanabla} then
  follow from taking principal branch square roots factor-by-factor in
the algebraic identities \eqref{eq:TnablaoverY} and
\eqref{eq:TDeltatimesY}, a meaningful step since $\Sigma^{\nabla\Delta}$
and $\Sigma^{\Delta\nabla}$ are disjoint from $Z$.
\end{proof}
\begin{proposition}
  Suppose that 
\begin{equation}
g_+(\xi)+g_-(\xi)=0,\quad
\xi\in\vec{\mathbb{R}}_+.
\label{eq:gplusgminusRplus}
\end{equation}
Then, the jump
  conditions for $\mathbf{N}(w)$ on the contours $\Sigma^\nabla_{>0}$
  and $\Sigma^\Delta_{>0}$ may be written in the form
\begin{multline}
\mathbf{N}_+(\xi)\mathbf{L}_+^\nabla(\xi)=\sigma_2
\mathbf{N}_-(\xi)\mathbf{L}_-^\nabla(\xi)\sigma_2
\begin{bmatrix}
1+A^\nabla(\xi) & B^\nabla(\xi)e^{-[2iQ_+(\xi)+L_+(\xi)-2g_+(\xi)]/\epsilon_N}\\
B^\nabla(\xi)e^{[2iQ_+(\xi)+L_+(\xi)-2g_+(\xi)]/\epsilon_N} & 1+A^\nabla(\xi)
\end{bmatrix},\\\xi\in \vec{\Sigma}^\nabla_{>0},
\label{eq:OjumpSigmanablagt0exact}
\end{multline}
where
\begin{equation}
A^\nabla(\xi)=\bo(\epsilon_N)\quad\text{and}\quad
B^\nabla(\xi)=\bo\left(\epsilon_N\frac{\lambda^2}{\epsilon_N^2}e^{-\alpha\lambda/
\epsilon_N}\right),\quad\lambda=E_+(\xi)>0,
\label{eq:OjumpSigmanablagt0}
\end{equation}
and
\begin{multline}
\mathbf{N}_+(\xi)\mathbf{L}_+^\Delta(\xi)=\sigma_2\mathbf{N}_-(\xi)
\mathbf{L}_-^\Delta(\xi)\sigma_2
\begin{bmatrix}
1+A^\Delta(\xi) & B^\Delta(\xi)e^{-[2iQ_+(\xi)+L_+(\xi)-2g_+(\xi)]/\epsilon_N}\\
B^\Delta(\xi)e^{[2iQ_+(\xi)+L_+(\xi)-2g_+(\xi)]/\epsilon_N} & 1+A^\Delta(\xi)
\end{bmatrix},\\
\xi\in\vec{\Sigma}^\Delta_{>0},
\label{eq:OjumpSigmadeltagt0exact}
\end{multline}
where
\begin{equation}
A^\Delta(\xi)=\bo(\epsilon_N)\quad\text{and}\quad
B^\Delta(\xi)=\bo\left(\epsilon_N\frac{\lambda^2}{\epsilon_N^2}
e^{-\alpha\lambda/\epsilon_N}\right),\quad\lambda=E_-(\xi)>0.
\label{eq:OjumpSigmadeltagt0}
\end{equation}
In both \eqref{eq:OjumpSigmanablagt0} and \eqref{eq:OjumpSigmadeltagt0}
the parameter $\alpha>0$ is defined in Proposition~\ref{prop:theta0}.
Also,
\begin{equation}
\mathbf{N}_+(\xi)=\sigma_2\mathbf{N}_-(\xi)\sigma_2,\quad
\xi\in\vec{\Sigma}_{>0}.
\label{eq:JumpNSigmagt0}
\end{equation}
\label{prop:JumpORplus}
\end{proposition}
\begin{proof}
The relation \eqref{eq:JumpNSigmagt0} follows directly from 
\eqref{eq:MjumpSigmagt0} and \eqref{eq:NMw} taking into account the
given condition \eqref{eq:gplusgminusRplus} on the boundary values of $g$.  
To prove \eqref{eq:OjumpSigmanablagt0} and \eqref{eq:OjumpSigmadeltagt0},
note firstly that from \eqref{eq:DE}, \eqref{eq:Qw}, and \eqref{eq:L0L} one has
both
\begin{equation}
Q_+(\xi)+Q_-(\xi)=0\quad\text{and} \quad
e^{\tfrac{1}{2}[L_+(\xi)+L_-(\xi)]/\epsilon_N}=1, \quad
\xi\in\vec{\mathbb{R}}_+.
\end{equation}  
From the latter it follows also that $Y_+(\xi)Y_-(\xi)=1$ for
$\xi\in\vec{\mathbb{R}}_+$.  Since $\mathbb{R}_+$ is disjoint from $Z$, we
may define $Y(w)^{1/2}$ as the principal branch of the square root for
$w$ just above and below $\mathbb{R}_+$ and it follows from \eqref{eq:Yalmostone} from
Proposition~\ref{prop:YTnablaDelta} that
$Y_+(\xi)^{1/2}Y_-(\xi)^{1/2}=1$ for $\xi\in\vec{\mathbb{R}}_+$ also.
Now from taking principal branch square roots factor-by-factor in
\eqref{eq:TnablaoverY} and \eqref{eq:TDeltatimesY} one can express 
the factors $T^\nabla(w)^{1/2}$ and $T^\Delta(w)^{1/2}$
appearing in \eqref{eq:Lnabla} and \eqref{eq:LDelta} in terms of
$Y(w)^{1/2}$ and the principal branches $(1+e^{\pm
  \theta_0(w)/\epsilon_N})^{1/2}$, which are also well-defined for
$w\in\mathbb{R}_+$.

One may now combine these facts and definitions with the jump conditions
\eqref{eq:MjumpSigmanablagt0} and \eqref{eq:MjumpSigmadeltagt0} to 
see that \eqref{eq:OjumpSigmanablagt0exact} holds 
where for $\xi\in\vec{\Sigma}^\nabla_{>0}$,
\begin{equation}
\begin{split}
A^\nabla(\xi)&:=
\frac{1+e^{i[\theta_{0+}(\xi)-\theta_{0-}(\xi)]/\epsilon_N}}
{(1+e^{2i\theta_{0+}(\xi)/\epsilon_N})^{1/2}
(1+e^{-2i\theta_{0-}(\xi)/\epsilon_N})^{1/2}}-1\\
B^\nabla(\xi)&:=-i\frac{e^{i\theta_{0+}(\xi)/\epsilon_N}-
e^{-i\theta_{0-}(\xi)/\epsilon_N}}
{(1+e^{2i\theta_{0+}(\xi)/\epsilon_N})^{1/2}
(1+e^{-2i\theta_{0-}(\xi)/\epsilon_N})^{1/2}},
\end{split}
\end{equation}
and that \eqref{eq:OjumpSigmadeltagt0exact} holds
where for $\xi\in\vec{\Sigma}^\Delta_{>0}$,
\begin{equation}
\begin{split}
A^\Delta(\xi)&:=
\frac{1+e^{-i[\theta_{0+}(\xi)-\theta_{0-}(\xi)]/\epsilon_N}}
{(1+e^{-2i\theta_{0+}(\xi)/\epsilon_N})^{1/2}
(1+e^{2i\theta_{0-}(\xi)/\epsilon_N})^{1/2}}-1 \\
B^\Delta(\xi)&:=
-i\frac{e^{i\theta_{0-}(\xi)/\epsilon_N}-e^{-i\theta_{0+}(\xi)/\epsilon_N}}
{(1+e^{-2i\theta_{0+}(\xi)/\epsilon_N})^{1/2}
(1+e^{2i\theta_{0-}(\xi)/\epsilon_N})^{1/2}}.
\end{split}
\end{equation} 
The desired estimates on $A^\nabla(\xi)$, $B^\nabla(\xi)$, $A^\Delta(\xi)$,
and $B^\Delta(\xi)$ now follow from 
Proposition~\ref{prop:theta0}, in particular from the Taylor series of
$\Psi(\lambda)$ about $\lambda=0$.  Indeed, given
$\xi\in\mathbb{R}_+$, $E_+(\xi)+E_-(\xi)=0$, and moreover
$E_+(\xi)$ is real and positive for $\xi\in\vec{\Sigma}^\nabla_{>0}$ and
real and negative for $\xi\in\vec{\Sigma}^\Delta_{>0}$ (the change of sign
comes from the reversal of orientation).  Therefore, recalling \eqref{eq:theta0def}, using
\eqref{eq:theta0series0} from Proposition~\ref{prop:theta0} together 
with Assumption~\ref{assume:epsilonNgeneral}, and writing
\begin{equation}
\nu(\lambda^2):=\sum_{n=1}^\infty\beta_n\lambda^{2n},
\end{equation}
we find
\begin{equation}
A^\nabla(\xi)=
\left|1+\frac{e^{-2\alpha\lambda/\epsilon_N}}{1+e^{-2\alpha\lambda/\epsilon_N}}\left[e^{2i\nu(\lambda^2)/\epsilon}-1\right]\right|^{-1}-1 = 
\bo\left(\frac{\lambda^2}{\epsilon_N}e^{-2\alpha\lambda/\epsilon_N}\right)
=  \bo(\epsilon_N),
\quad \lambda=E_+(\xi)>0
\end{equation}
and
\begin{equation}
B^\nabla(\xi)=2(-1)^N\frac{e^{-\alpha\lambda/\epsilon_N}\sin(\epsilon_N^{-1}
\nu(\lambda^2))}{|1+e^{-2\alpha\lambda/\epsilon_N}e^{2i\nu(\lambda^2)/\epsilon_N}|}
=\bo\left(\frac{\lambda^2}{\epsilon_N}e^{-\alpha\lambda/\epsilon_N}\right),
\quad \lambda=E_+(\xi)>0,
\end{equation}
therefore proving \eqref{eq:OjumpSigmanablagt0}.
In exactly the same way one 
obtains the estimates \eqref{eq:OjumpSigmadeltagt0}, thereby completing
the proof.
\end{proof}
Note that in the special case in which $\Psi(\lambda)$ is given by
the formula \eqref{eq:specialtheta0}, the error terms in
\eqref{eq:OjumpSigmanablagt0} and \eqref{eq:OjumpSigmadeltagt0} vanish
identically, a fact that was exploited to simplify the analysis of the semiclassical
limit of the focusing nonlinear Schr\"odinger equation with corresponding
special initial data in \cite{LyngM07}.

\section{Construction of $g(w)$}
\label{sec:construction-of-g}

We proceed under the assumption that the Riemann-Hilbert problem for 
$\mathbf{N}$ will reduce, under appropriate changes of variables, to 
a problem solved using a genus-1 Riemann surface (\textit{i.e.}, an elliptic curve), at least for small $t$ 
and $x$ bounded away from $\pm x_\text{crit}$.  This assumption is 
consistent with the qualitative behavior observed in Figure~\ref{fig:Ap75}, 
in which the solution appears to have one oscillatory phase for these 
values of $x$ and $t$.  In other parts of the space-time plane, the 
solution appears to have more than one oscillatory phase, suggesting that the 
corresponding Riemann-Hilbert problem will be solved using a Riemann surface 
of genus greater than one, which will require modifying the following 
calculations.

\subsection{Two types of ``genus-1'' ansatz for $g(w)$}
Let $\mathfrak{p}$ and $\mathfrak{q}$ be real parameters and consider the quadratic polynomial 
$R(w;\mathfrak{p},\mathfrak{q})^2$ given by
\begin{equation}
R(w;\mathfrak{p},\mathfrak{q})^2 := (w-\mathfrak{p})^2-\mathfrak{q}.
\end{equation}
When working in a region of the space-time plane where more than 
one nonlinear phase is expected, it is necessary to choose $R(w)^2$ to 
be a higher-degree polynomial (with a corresponding increase in the 
number of parameters $\mathfrak{p}_j,\mathfrak{q}_j$ specifying the roots 
of $R(w)^2$).

We must distinguish two cases, which we label as ``\librational'' and 
``\rotational'' as these will correspond ultimately to local asymptotics for
the fluxon condensate $u_N(x,t)$ in terms of periodic 
librational and rotational wavetrains, respectively.
\begin{itemize}
\item[\fbox{\librational}] This case is defined by the inequality $\mathfrak{q}<0$.
  The quadratic $R(w;\mathfrak{p},\mathfrak{q})^2$ has distinct roots forming a
  complex-conjugate pair $w=\mathfrak{p}\pm i\sqrt{-\mathfrak{q}}$. The roots are assumed to
  lie on the nonreal arcs of the contour
  $\Sigma^\nabla\cup\Sigma^\Delta$ (or, rather, given $\mathfrak{p}$ and $\mathfrak{q}$ with $\mathfrak{q}<0$
the regions $\Omega_\pm^\nabla$ and $\Omega_\pm^\Delta$ are assumed to
be positioned so that this holds).  We define a subcontour
$\beta\subset\Sigma^\nabla\cup\Sigma^\Delta$ consisting of the closure of
the arc of $\Sigma$ connecting the two roots of $R(w;\mathfrak{p},\mathfrak{q})^2$ via
$w=1$.  Thus $\beta$ is a simple contour (with no self-intersection
points).  Whenever we are in case \librational, we will assume that
there is no transition point, \textit{i.e.}, either $\Delta=\emptyset$
or $\nabla=\emptyset$.
\item[\fbox{\rotational}] This case is defined by the inequalities $\mathfrak{q}>0$
and 
\begin{equation}
\mathfrak{a}\le \mathfrak{p}-\sqrt{\mathfrak{q}} <\mathfrak{p}+\sqrt{\mathfrak{q}}\le \mathfrak{b}.
\label{eq:wrotationalbounds}
\end{equation}
The distinct real roots $w=\mathfrak{p}\pm\sqrt{\mathfrak{q}}$ therefore lie in the real
interval $[\mathfrak{a},\mathfrak{b}]$ of $\Sigma^\nabla\cup\Sigma^\Delta$.  We then assume further
that the point $w=w^+$ where the nonreal arcs of
$\Sigma^\nabla\cup\Sigma^\Delta$ meet the interval $[\mathfrak{a},\mathfrak{b}]$ lies
between the two roots.  We define a subcontour
$\beta\subset\Sigma^\nabla\cup\Sigma^\Delta$ as the union of the closure
of the nonreal arcs of $\Sigma^\nabla\cup\Sigma^\Delta$ with the real
interval $[\mathfrak{p}-\sqrt{\mathfrak{q}}, \mathfrak{p}+\sqrt{\mathfrak{q}}]$.  Thus $\beta$ is a non-simple
(self-intersecting) contour having a single self-intersection point (a
simple crossing) at $w=w^+$.  Whenever we are in case \rotational,
we will assume that if there is a transition point $w=\tau_N$, it
lies in $\beta$.
\end{itemize}
In both cases, we define
$\gamma:=\overline{(\Sigma^\nabla\cup\Sigma^\Delta)\setminus \beta}$,
and we assume the contours $\beta$ and $\gamma$ inherit the orientation
of $\Sigma$. See Figures~\ref{fig:betagammaL} and \ref{fig:betagammaR} for
illustrations of $\beta$ and $\gamma$ in cases \librational\ and \rotational\ respectively.
\begin{figure}[h]
\begin{center}
\includegraphics{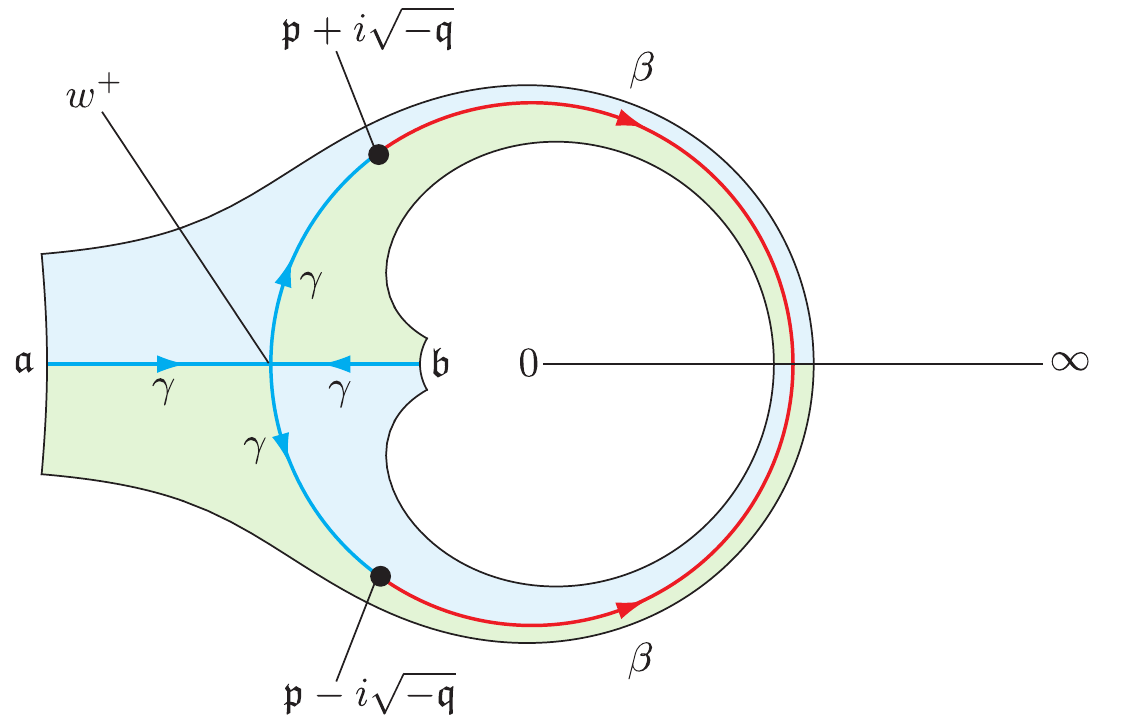}
\end{center}
\caption{\emph{The subcontours $\beta$ (red) and $\gamma$ (blue) for a configuration
of type \librational.  The orientations here correspond to the case $\Delta=\emptyset$ (see Figure~\ref{fig:Ann_DeltaEmpty}).}}
\label{fig:betagammaL}
\end{figure}

\begin{figure}[h]
\begin{center}
\includegraphics{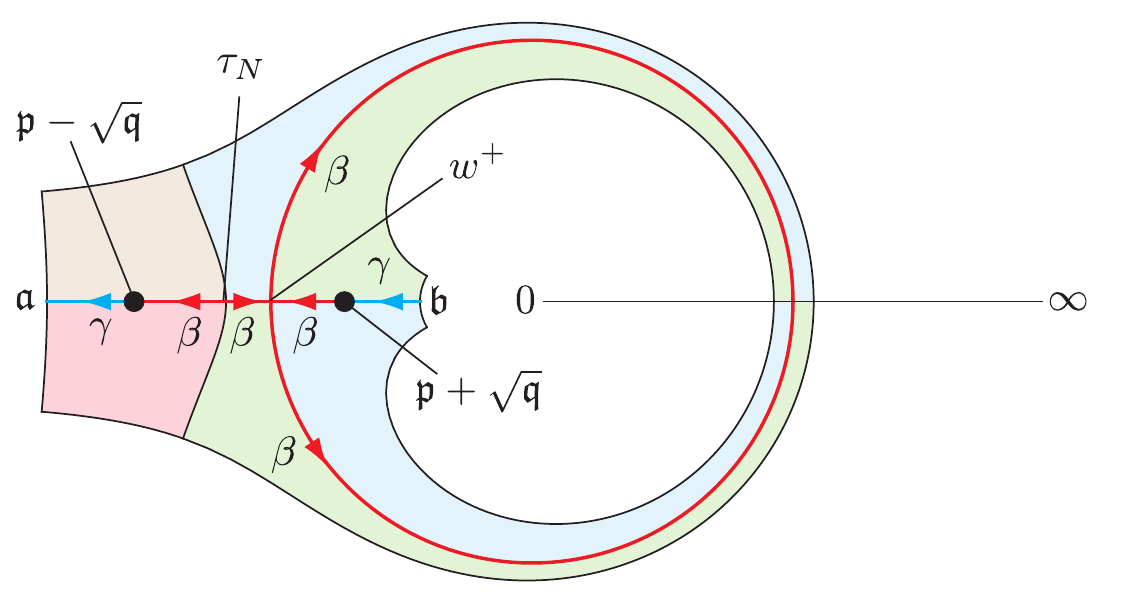}
\end{center}
\caption{\emph{The subcontours $\beta$ (red) and $\gamma$ (blue) for a configuration of type \rotational.  The orientations here correspond to the case $\Delta=P_N^{\prec\kink}$ (see Figure~\ref{fig:Ann_DeltaNearMinusM}).  This illustrates the fact that if a transition point $\tau_N$ is present
it is assumed to lie in the subcontour $\beta$.}}
\label{fig:betagammaR}
\end{figure}

We define an analytic branch $R(w)=R(w;\mathfrak{p},\mathfrak{q})$ of the square root of the
quadratic $R(w;\mathfrak{p},\mathfrak{q})^2$ by taking the branch cut to coincide with $\beta$
and choosing the sign so that $R(w)=w+\bo(1)$ as $w\to\infty$.  Note that
in case \rotational\ as well as the borderline case $v=0$, $R(w)$ is a sectionally
analytic function of $w$ because the branch cut locus $\beta$ has a
self-intersection point at $w=w^+$.

It will be useful below to have available some compact notation for
certain sums of contour integrals.  We therefore define
\begin{equation}
\int_CF(\xi)\,d\xi:=\int_{\partial\Omega_+^\nabla\setminus\Sigma^\nabla}F(\xi)\,d\xi +
\int_{\partial\Omega_-^\nabla\setminus\Sigma^\nabla}F(\xi)\,d\xi +
\int_{\partial\Omega_+^\Delta\setminus\Sigma^\Delta}F(\xi)\,d\xi
+\int_{\partial\Omega_-^\Delta\setminus\Sigma^\Delta}F(\xi)\,d\xi 
\label{eq:intC}
\end{equation}
and we will use this formula in some situations where $F(\xi)$ is to be
understood along the four contours on the right-hand side in the sense
of taking a boundary value from within the region whose boundary is
involved in the integration.  For example, if the regions
$\Omega_\pm^\nabla$ meet the positive real axis in the contour
$\Sigma^\nabla_{>0}$ and if $F$ is a function having a jump
discontinuity across this contour, then the two terms on the
right-hand side of \eqref{eq:intC} involving integration over
$\partial\Omega_\pm^\nabla\setminus \Sigma^\nabla$ require use of two
different boundary values $F_\pm(\xi)$ taken by $F$ on
$\vec{\Sigma}^\nabla_{>0}$.  If $F(\xi)$ depends on a parameter $w$
through a Cauchy factor $(\xi-w)^{-1}$ then $\int_CF(\xi)\,d\xi$
will be analytic in $w$ in a bounded domain we denote as $\Omega^\circ:=
(\Omega\cup\Sigma^\nabla\cup\Sigma^\Delta)\setminus\{\mathfrak{a},\tau_N,\mathfrak{b},1\}$.

In order to construct suitable functions $g(w)$ with which to
transform the matrix $\mathbf{M}(w)$ into $\mathbf{N}(w)$, we will
need to impose certain relations between the real parameters $\mathfrak{p}$ and
$\mathfrak{q}$ and the independent variables $x$ and $t$.  
We will define two functions $M(\mathfrak{p},\mathfrak{q},x,t)$ and $I(\mathfrak{p},\mathfrak{q},x,t)$.  
Then the ``moment condition'' $M(\mathfrak{p},\mathfrak{q},x,t)=0$ and the ``integral condition'' 
$I(\mathfrak{p},\mathfrak{q},x,t)=0$ will be used to study the dependence of the roots of 
$R(w;\mathfrak{p},\mathfrak{q})^2$ on $x$ and $t$.
For fixed $x$ and $t$, the solvability of the equations $M=0$ and 
$I=0$ for $\mathfrak{p}$ and $\mathfrak{q}$ will be a necessary condition 
for asymptotic reduction to a model Riemann-Hilbert problem solvable in terms of elliptic functions.  However, this will not be a sufficient condition.  It 
may happen that for some $(x,t)$ this system is solvable but certain necessary inequalities 
(see Proposition \ref{prop:inequalities}) fail along arcs of the 
subcontour $\gamma$.  In this case, it becomes necessary to introduce new arcs of the subcontour $\beta$ near the points where the inequalities have failed in $\gamma$.  Thus, $\beta$ becomes disconnected, and there are more endpoints (\textit{i.e.}, roots of $R(w)^2$).  To determine these additional parameters it then becomes necessary to also include further
equations $M_j=0$ and $I_j=0$.  Ultimately this will effect an asymptotic reduction to a model Riemann-Hilbert problem solvable in terms of higher-genus hyperelliptic functions.

For the moment, we
think of $\trans{(\mathfrak{p},\mathfrak{q},x,t)}\in\mathbb{R}^4$ as a parameter vector, and
begin by defining a function $M=M(\mathfrak{p},\mathfrak{q},x,t)$ as
\begin{equation}
M:=\frac{x-t}{\sqrt{\mathfrak{p}^2-\mathfrak{q}}}+x+t-\frac{2}{\pi}\int_C
\frac{\theta_0'(\xi)\sqrt{-\xi}\,d\xi}{R(\xi;\mathfrak{p},\mathfrak{q})}.
\label{eq:wM0Cdef}
\end{equation}
Now let $H(w)=H(w;\mathfrak{p},\mathfrak{q},x,t)$ be the function 
defined by
\begin{equation}
H(w):=-\frac{1}{4\sqrt{-w}}\left[\frac{x-t}{w\sqrt{\mathfrak{p}^2-\mathfrak{q}}}+\frac{2}{\pi}
\int_C
\frac{\theta_0'(\xi)\sqrt{-\xi}\,d\xi}{R(\xi;\mathfrak{p},\mathfrak{q})
(\xi-w)}\right],\quad w\in\Omega^\circ.
\label{eq:wgeneralGdef}
\end{equation}
This function is analytic where it is defined (the singularity at
$w=0$ and corresponding branch cut along $\mathbb{R}_+$ are excluded
from $\Omega^\circ$), but it does have jump discontinuities across the
the contours $\Sigma^\nabla_{>0}$ or $\Sigma^\Delta_{>0}$ (depending
on which of the six cases of $\Delta$ we are considering), as well as
across $\Sigma^{\nabla\Delta}$ and $\Sigma^{\Delta\nabla}$ if there
exists a transition point.  Note that $H(w)$ satisfies 
$H(w^*)^*=H(w)$, so that in particular the zeros of $H(w)$ in its
domain of definition either lie on the negative real axis 
or come in complex-conjugate pairs.
We may now define a second function
$I=I(\mathfrak{p},\mathfrak{q},x,t)$ as
\begin{equation}
\begin{split}
I:=&\Re\left\{\int_{\beta\cap\mathbb{C}_+}R_+(\xi)H(\xi)\,d\xi\right\}\\
{}=& \frac{1}{2}\int_{\beta\cap\mathbb{C}_+}R_+(\xi)H(\xi)\,d\xi
+\frac{1}{2}\int_{\beta\cap\mathbb{C}_-}R_-(\xi)H(\xi)\,d\xi.
\end{split}
\label{eq:wgeneralI0def}
\end{equation}
Finally, let
$f(w)=f(w;\mathfrak{p},\mathfrak{q},x,t)$ be given by the formula
\begin{equation}
f(w):=i\frac{dQ}{dw}(w;x,t)+\frac{x-t}{8\sqrt{\mathfrak{p}^2-\mathfrak{q}}}\frac{R(w;\mathfrak{p},\mathfrak{q})}{w\sqrt{-w}}
-\frac{R(w;\mathfrak{p},\mathfrak{q})}{2\pi\sqrt{-w}}
\int_\gamma
\frac{\theta_0'(\xi)\sqrt{-\xi}\,d\xi}{R(\xi;\mathfrak{p},\mathfrak{q})(\xi-w)} +
\frac{1}{2}\frac{dL}{dw}(w).
\label{eq:wgprimedef}
\end{equation}

\begin{proposition}
  Let parameters $\mathfrak{p}$, $\mathfrak{q}$, $x$, and $t$ be given so that the quadratic 
$R(w;\mathfrak{p},\mathfrak{q})^2$ is in case \librational, case \rotational, or the borderline case of $\mathfrak{q}=0$,
 and assume that the regions $\Omega_\pm^\nabla$ and $\Omega_\pm^\Delta$
are chosen so that the contour $\Sigma^\nabla\cup\Sigma^\Delta$ is
consistent with
the roots of the quadratic with well-defined subcontours $\beta$ and
$\gamma$.  Suppose also that the
  moment condition $M=0$ and the integral condition  $I=0$ both hold.  Then an analytic
  function $g(w)$ is defined for $w\in
  \mathbb{C}\setminus (\beta\cup\mathbb{R}_+)$ as follows.
\begin{itemize}
\item[\fbox{\librational}] In this case we set
\begin{equation}
g(w):=\int_0^wf(w')\,dw'
\label{eq:gdeflibrational}
\end{equation}
where the path of integration is arbitrary in $\mathbb{C}\setminus
(\beta\cup\mathbb{R}_+)$.
\item[\fbox{\rotational}] In this case we set
\begin{equation}
g(w):=\begin{cases}\displaystyle \int_0^wf(w')\,dw',\quad & w\in \Upsilon_0,\\\\
\displaystyle\int_\infty^wf(w')\,dw',\quad & w\in\Upsilon_\infty,
\end{cases}
\label{eq:gdefrotational}
\end{equation}
where $\Upsilon_0$ and $\Upsilon_\infty$ are respectively the bounded and unbounded
connected components of $\mathbb{C}\setminus
(\beta\cup\mathbb{R}_+)$,
and in each case the path of
integration is arbitrary in the given domain.
\end{itemize}
(In the borderline case of $\mathfrak{q}=0$ we also use the definition
\eqref{eq:gdefrotational}.)  
The function $g(w)$ so-defined satisfies in all
cases the Schwartz-symmetry condition
\begin{equation}
g(w^*)=g(w)^*
\label{eq:wgsymmetry}
\end{equation}
and is a uniformly H\"older-$\tfrac{1}{2}$ continuous map
$\mathbb{C}\setminus(\beta\cup\mathbb{R}_+)\to\mathbb{C}$.  The
function $\theta:\vec{\Sigma}^\nabla\cup\vec{\Sigma}^\Delta
\to\mathbb{C}$ defined by
\eqref{eq:wthetaphidef} satisfies
\begin{equation}
\Im\{\theta(\xi)\}\equiv 0,\quad\xi\in(\vec{\Sigma}^\nabla\cup\vec{\Sigma}^\Delta)\cap\mathbb{R},
\label{eq:wthetareal}
\end{equation}
\begin{equation}
\theta(\xi)\equiv 0,\quad\xi\in\vec{\gamma},
\label{eq:wthetazerogamma}
\end{equation}
and
\begin{equation}
\frac{d\theta}{d\xi}(\xi)=iR_+(\xi;\mathfrak{p},\mathfrak{q})H(\xi),\quad\xi\in\vec{\beta}.
\label{eq:wthetaprimebeta}
\end{equation}
Also, we have
\begin{equation}
g_+(\xi)+g_-(\xi)=0,\quad \xi\in\vec{\mathbb{R}}_+.
\label{eq:wgsumzeroR}
\end{equation}
The function $\phi:\vec{\Sigma}^\nabla\cup\vec{\Sigma}^\Delta
\to\mathbb{C}$
defined by \eqref{eq:wthetaphidef} satisfies
\begin{equation}
\Im\{\phi(\xi)\}\equiv 0,\quad \xi\in(\vec{\Sigma}^\nabla\cup
\vec{\Sigma}^\Delta)\cap\mathbb{R},
\label{eq:wphireal}
\end{equation}
\begin{equation}
\phi(\xi)\equiv \pm i\Phi,\quad\xi\in\vec{\beta}\cap\mathbb{C}_\pm,
\label{eq:wphiconstbetaCpm}
\end{equation}
for some real number $\Phi$, 
\begin{equation}
\phi(\xi)\equiv 0,\quad \xi\in\vec{\beta}\cap\mathbb{R},
\label{eq:wphizero}
\end{equation}
and
\begin{equation}
\frac{d\phi}{d\xi}(\xi) = R(\xi;\mathfrak{p},\mathfrak{q})H(\xi),\quad\xi\in\vec{\gamma}.
\label{eq:wphiprimegamma}
\end{equation}
Finally, $g$ satisfies the decay condition
\begin{equation}
\lim_{w\to\infty}g(w)=0.
\label{eq:wgdecay}
\end{equation}
\label{prop:wgbasicproperties}
\end{proposition}
\begin{proof}
  Even without the conditions $M=0$ and $I=0$, it is obvious from
  the definition \eqref{eq:wgprimedef} that the function $f(w)$ is
  analytic at least for
  $w\in\mathbb{C}\setminus(\Sigma^\nabla\cup\Sigma^\Delta\cup\mathbb{R}_+)$.
  From the Plemelj formula and the definition \eqref{eq:L0L} of $L(w)$,
  one finds that
\begin{equation}
f_+(\xi)-f_-(\xi)=0,\quad \xi\in\vec{\gamma},
\end{equation}
which, upon taking into account the continuity of the boundary values taken
by $f$ along $\vec{\gamma}$, shows that $f$ is analytic in the larger
domain $\mathbb{C}\setminus(\beta\cup\mathbb{R}_+)$.  

Now $f(w)$ is automatically integrable at $w=0$.  Indeed,
\eqref{eq:L0L} shows that $L(w)$ has a well-defined limiting value as
$w\to 0$ with $|\arg(-w)|<\pi$ (and is in fact uniformly
H\"older-$\tfrac{1}{2}$ for such $w$), the term involving the integral
over $\gamma$ is clearly $\bo((-w)^{-1/2})$, and the sum of the remaining two 
terms
is as well (after canceling a term proportional to
$(-w)^{-3/2}$ between them).  This fact gives sense
to the formulae defining $g(w)$ in case \librational\ and also in the
domain $\Upsilon_0$ in case \rotational.  But the moment condition $M=0$ also makes
$f$ integrable at $w=\infty$:
\begin{equation}
f(w)=\frac{1}{8\sqrt{-w}}\left[
\frac{x-t}{\sqrt{\mathfrak{p}^2-\mathfrak{q}}}+x+t+\frac{4}{\pi}
\int_\gamma
\frac{\theta_0'(\xi)\sqrt{-\xi}\,d\xi}
{R(\xi;\mathfrak{p},\mathfrak{q})}\right] + \bo\left((-w)^{-3/2}\right),\quad w\to\infty.
\label{eq:wfwlarge}
\end{equation}
By a simple contour deformation argument (using the fact that
$R(\xi;\mathfrak{p},\mathfrak{q})$ changes sign across $\vec{\beta}$), we may write the integral
over $\gamma$ as minus one-half of the corresponding integral 
over $C$ (recall \eqref{eq:intC}) and thus identify the term
in brackets as the ``moment'' $M(\mathfrak{p},\mathfrak{q},x,t)$.
Therefore in case \rotational, $g(w)$ is well-defined and analytic for
$w\in\Upsilon_\infty$, and in case \librational\ we observe that $g$ has a
well-defined limiting value as $w\to\infty$.

The Schwartz symmetry \eqref{eq:wgsymmetry} of $g$ now follows
immediately from the definition of $g$ and the corresponding symmetry
$f(z^*)=f(z)^*$ obvious from \eqref{eq:wgprimedef}, and the
H\"older-$\tfrac{1}{2}$ continuity of $g$ can be read off from the
formula for $f=g'$.  Also, since by definition $g$ has no jump across
the contour $\vec{\gamma}$ we see from the definition
\eqref{eq:wthetaphidef} of $\theta$ that \eqref{eq:wthetazerogamma}
holds.

Next, observe that 
\begin{equation}
\begin{split}
f_+(\xi)-f_-(\xi)&=R_+(\xi;\mathfrak{p},\mathfrak{q})
\left[\frac{x-t}{4\sqrt{\mathfrak{p}^2-\mathfrak{q}}\,\xi\sqrt{-\xi}}-\frac{1}{\pi\sqrt{-\xi}}
\int_\gamma
\frac{\theta_0'(\xi')\sqrt{-\xi'}\,d\xi'}{R(\xi';\mathfrak{p},\mathfrak{q})
(\xi'-\xi)}\right]+i\theta_0'(\xi)
\\
&=-R_+(\xi)H(\xi),\quad \xi\in\vec{\beta},
\end{split}
\label{eq:wfdiffbeta}
\end{equation}
where the second line follows from a residue calculation and a contour
deformation like that used in identifying the leading term in
\eqref{eq:wfwlarge} with a multiple of $M(\mathfrak{p},\mathfrak{q},x,t)$.  
Similarly,
\begin{equation}
2i\frac{dQ}{d\xi}(\xi;x,t)+\frac{d\overline{L}}{d\xi}(\xi)
-f_+(\xi)-f_-(\xi) = R(\xi)H(\xi),\quad \xi\in\vec{\gamma}.
\end{equation}
Together with the definitions \eqref{eq:wthetaphidef} of $\theta$
and $\phi$, these two relations establish \eqref{eq:wthetaprimebeta}
and \eqref{eq:wphiprimegamma}.

Simpler calculations show
that 
\begin{equation}
f_+(\xi)+f_-(\xi)=0,\quad \xi\in\vec{\mathbb{R}}_+,
\label{eq:wfsumzeroR}
\end{equation}
integration of which yields \eqref{eq:wgsumzeroR},
and
\begin{equation}
f_+(\xi)+f_-(\xi)=2i\frac{dQ}{d\xi}(\xi;x,t)+\frac{d\overline{L}}{d\xi}(\xi),
\quad \xi\in\vec{\beta}.
\end{equation}
The latter shows that $\phi$ defined by \eqref{eq:wthetaphidef} is
constant in each arc of $\vec{\beta}$.  Since it follows from the
Schwartz symmetry \eqref{eq:wgsymmetry} of $g$ and the jump condition
\eqref{eq:wgsumzeroR} that $g$ takes purely imaginary boundary values
along $\mathbb{R}_+$, an examination of the functions $\theta$ and
$\phi$ along $\beta$ near its intersection point $w=1$ with
$\mathbb{R}_+$ shows that the limiting values of these functions taken
along $\vec{\beta}$ (either of the two arcs) as $\xi\in\vec{\beta}$
tends to $\xi=1$ are, respectively, real and imaginary.  Since
$\phi(\xi^*)=\phi(\xi)^*$, the identities \eqref{eq:wphiconstbetaCpm}
follow immediately.  Now, Schwartz symmetry of $g$
\eqref{eq:wgsymmetry} also shows that both $\theta(\xi)$ and
$\phi(\xi)$ are real-valued for $\xi\in(\vec{\Sigma}^\nabla
\cup\vec{\Sigma}^\Delta)\cap\mathbb{R}$, proving \eqref{eq:wthetareal}
and \eqref{eq:wphireal}.  If we are in case \rotational, so that
$\vec{\beta}$ contains two real arcs, we can now show that $\phi\equiv
0$ on these two arcs, proving \eqref{eq:wphizero}.  Indeed, for
$\xi\in\vec{\beta}\cap\mathbb{R}$, we may use constancy of $\phi$
along the arc to obtain
\begin{equation}
\begin{split}
\Re\{\phi(\xi)\} &= \mathop{\lim_{\xi'\to w^+}}_{\xi'\in\vec{\beta}\cap\mathbb{R}}
\Re\{\phi(\xi')\}\\
&= \Re\{2iQ+\overline{L}\}(w^+)-
\mathop{\lim_{\xi'\to w^+}}_{\xi'\in\vec{\beta}\cap\mathbb{R}}\Re\{g_+(\xi')+g_-(\xi')\},\quad \xi\in\vec{\beta}\cap\mathbb{R}.
\end{split}
\end{equation}
Note that while $2iQ(\xi)+\overline{L}(\xi)$ is not continuous on $\beta$
in a neighborhood of the self-intersection point $\xi=w^+$ due to jump discontinuities in $\overline{L}$ (see \eqref{eq:varphinabla} and
\eqref{eq:varphiDelta} and the discussion just below these definitions), 
its real part is, which makes $\Re\{2iQ +\overline{L}\}$ well-defined
at $\xi=w^+$ regardless of the arc along which this point is
approached, and explains the notation 
$\Re\{2iQ+\overline{L}\}(w^+)$.
But $g_+(\xi')+g_-(\xi')=2g_+(\xi')-i\theta(\xi')$ and $\theta(\xi')$ is real
for $\xi'\in\vec{\beta}\cap\mathbb{R}$, so
\begin{equation}
\Re\{\phi(\xi)\} = \Re\{2iQ+\overline{L}\}(w^+)-
2\mathop{\lim_{\xi'\to w^+}}_{\xi'\in\vec{\beta}\cap\mathbb{R}}\Re\{g_+(\xi')\},
\quad \xi\in\vec{\beta}\cap\mathbb{R}.
\end{equation}
Now if $U$ is a small neighborhood of the self-intersection point
$w=w^+$ of $\beta$, then $g:U\setminus \beta\to \mathbb{C}$ is
uniformly H\"older-$\tfrac{1}{2}$ continuous, so the 
latter limit may be taken along a different arc of $\beta$ with the 
same result:
\begin{equation}
\Re\{\phi(\xi)\} = \Re\{2iQ+\overline{L}\}(w^+)-
2\mathop{\lim_{\xi'\to w^+}}_{\xi'\in\vec{\beta}\cap\mathbb{C}_\pm}\Re\{g_+(\xi')\},
\quad \xi\in\vec{\beta}\cap\mathbb{R},
\end{equation}
where the upper or lower half-plane is used depending on whether the
original arc of $\vec{\beta}\cap\mathbb{R}$ was oriented to the right
or left respectively.  Now we can write this in terms of the limiting
values of $\phi$ and $\theta$ along $\vec{\beta}\cap\mathbb{C}_\pm$:
\begin{equation}
\begin{split}
\Re\{\phi(\xi)\} &=
\mathop{\lim_{\xi'\to w^+}}_{\xi'\in\vec{\beta}\cap\mathbb{C}_\pm}
\Re\{\phi(\xi')+i\theta(\xi')\}\\
&=\mathop{\lim_{\xi'\to w^+}}_{\xi'\in\vec{\beta}\cap\mathbb{C}_\pm}
\Re\{i\theta(\xi')\},\quad\xi\in\vec{\beta}\cap\mathbb{R},
\end{split}
\end{equation}
where the second line follows because $\phi$ is a purely imaginary
constant $\pm i\Phi$ along $\vec{\beta}\cap\mathbb{C}_\pm$.
Finally, since $\theta(\xi)$ has a real limiting value as $\xi\to 1$
with $\xi\in\vec{\beta}\cap\mathbb{C}_\pm$, we obtain that
$\Re\{\phi(\xi)\}\equiv 0$ for $\xi\in\vec{\beta}\cap\mathbb{R}$ as a
consequence of the integral condition $I=0$.  But from \eqref{eq:wphireal},
$\phi(\xi)$ is real for $\xi\in\mathbb{R}$, so \eqref{eq:wphizero} follows.

The decay condition \eqref{eq:wgdecay} is obvious from the definition
in case \rotational, and in case \librational\ it follows from the integral condition
$I=0$.  Indeed, writing $g(\infty)$ in case \librational\ as the the
integral
\begin{equation}
g(\infty)=\int_0^{-\infty}f(\xi)\,d\xi,
\end{equation}
we may split the integral into two equal parts, and in each part we deform
the contour into the right half-plane in opposite directions, bringing
the two contours against $\beta\cup\mathbb{R}_+$.  Using \eqref{eq:wfsumzeroR}
cancels the contributions to $g(\infty)$ coming from integrals along the
upper and lower edges of $\mathbb{R}_+$, leaving only
\begin{equation}
g(\infty)=\pm\frac{1}{2}\int_\beta\left(f_+(\xi)-f_-(\xi)\right)\,d\xi
\end{equation}
where the sign is different depending on whether $\beta\subset\Sigma^\nabla$
or $\beta\subset\Sigma^\Delta$.  But in either case, we now can use
\eqref{eq:wfdiffbeta} and the symmetry $f(\xi^*)=f(\xi)^*$ to see that
the condition $I=0$ guarantees that $g(\infty)=0$ in case \librational.
\end{proof}

If $\mathfrak{p}$ and $\mathfrak{q}$ can be eliminated by means of the equations $M(\mathfrak{p},\mathfrak{q},x,t)=0$ and $I(\mathfrak{p},\mathfrak{q},x,t)=0$,
then we may view $g$ as being a function of $w$ depending parametrically only
on $x$ and $t$.  Now, according to Proposition~\ref{prop:wgbasicproperties} 
there will be a real constant $\Phi$ associated with $g$.  
In the situation that $\mathfrak{p}$ and $\mathfrak{q}$ have been eliminated, we will have
$\Phi=\Phi(x,t)$, and it will be useful to characterize the
dependence of $\Phi(x,t)$ on the remaining parameters $x$ and $t$.

\begin{proposition}
  Suppose that $\mathfrak{p}=\mathfrak{p}(x,t)$ and $\mathfrak{q}=\mathfrak{q}(x,t)$ constitute a differentiable solution
of the equations $M(\mathfrak{p},\mathfrak{q},x,t)=0$ and $I(\mathfrak{p},\mathfrak{q},x,t)=0$, so that $g(w;x,t)=g(w;\mathfrak{p}(x,t),\mathfrak{q}(x,t),x,t)$,
and let the constant
$\Phi=\Phi(x,t)$ be obtained therefrom
as described in
  Proposition~\ref{prop:wgbasicproperties}.  Then $\Phi(x,t)$
is jointly differentiable in $x$ and $t$ with real-valued first-order partial derivatives,
depending on $x$ and $t$ only via $\mathfrak{p}$ and $\mathfrak{q}$, 
given by
\begin{equation}
\frac{\partial\Phi}{\partial x}=\frac{\pi}{4\mathcal{D}}\left[1-\frac{1}{\sqrt{\mathfrak{p}^2-\mathfrak{q}}}\right]\quad\text{and}\quad
\frac{\partial\Phi}{\partial t}=\frac{\pi}{4\mathcal{D}}\left[1+\frac{1}{\sqrt{\mathfrak{p}^2-\mathfrak{q}}}\right],
\label{eq:kappapartials}
\end{equation}
where in case \librational,
\begin{equation}
\mathcal{D}=\frac{K(m_\librational)}{(\mathfrak{p}^2-\mathfrak{q})^{1/4}},\quad
m_\librational:=\frac{1}{2}\left(1-\frac{\mathfrak{p}}{\sqrt{\mathfrak{p}^2-\mathfrak{q}}}\right)\in (0,1),
\label{eq:deltalibrational}
\end{equation}
and in case \rotational,
\begin{equation}
\mathcal{D} = \frac{2K(m_\rotational)}{\sqrt{-\mathfrak{p}+\sqrt{\mathfrak{q}}}+\sqrt{-\mathfrak{p}-\sqrt{\mathfrak{q}}}},
\quad
m_\rotational:=\frac{4\sqrt{\mathfrak{p}^2-\mathfrak{q}}}{(\sqrt{-\mathfrak{p}+\sqrt{\mathfrak{q}}}+\sqrt{-\mathfrak{p}-\sqrt{\mathfrak{q}}})^2}\in (0,1),
\label{eq:deltarotational}
\end{equation}
where $K(\cdot)$ denotes the complete elliptic integral of the first kind defined by \eqref{eq:ellipticKdef}.
In particular, defining a quantity $n_\mathrm{p}$ by
\begin{equation}
n_\mathrm{p}:=-\frac{\displaystyle\frac{\partial\Phi}{\partial x}}
{\displaystyle\frac{\partial\Phi}{\partial t}} = 
\frac{1-\sqrt{\mathfrak{p}^2-\mathfrak{q}}}{1+\sqrt{\mathfrak{p}^2-\mathfrak{q}}},
\label{eq:phasevelocity}
\end{equation}
and noting that $\mathfrak{p}^2-\mathfrak{q}>0$ in both
cases \librational\ and \rotational, we see that $n_\mathrm{p}$ is an algebraic function of $\mathfrak{p}$ and $\mathfrak{q}$ satisfying $|n_\mathrm{p}|<1$.  Also, $0<\mathfrak{p}^2-\mathfrak{q}<1$ implies $n_\mathrm{p}>0$ while $\mathfrak{p}^2-\mathfrak{q}>1$
implies $n_\mathrm{p}<0$.
\label{prop:phasevelocity}
\end{proposition}

\begin{proof}
Consider the functions $X(w)$ and $T(w)$ defined in terms of $g(w)=g(w;x,t)$
as follows:
\begin{equation}
X(w):=\frac{\pi}{\sqrt{-w}R(w)}\frac{\partial g(w)}{\partial x},\quad
T(w):=\frac{\pi}{\sqrt{-w}R(w)}\frac{\partial g(w)}{\partial t}.
\label{eq:XTdefine}
\end{equation}
These functions are both analytic for $w\in\mathbb{C}\setminus \beta$.
Indeed, from \eqref{eq:wgsumzeroR} in
Proposition~\ref{prop:wgbasicproperties} we see that in spite of the
explicit presence of the square root $\sqrt{-w}$ in the definitions, $X(w)$ and
$T(w)$ may be considered to be analytic in a neighborhood of the
positive real axis, and then analyticity for $w\in\mathbb{C}\setminus\beta$
follows from elementary properties of the remaining factors.  It is also
a consequence of \eqref{eq:wgdecay} and the relation $g'(w)=f(w)$ with
$f(w)$ given by \eqref{eq:wgprimedef} that $g(w)=\bo((-w)^{-1/2})$ for large $w$,
and this implies that $X(w)=\bo(w^{-2})$ and $T(w)=\bo(w^{-2})$ as $w\to\infty$.

Differentiation (with respect to $x$ and $t$) of the identities
$\phi(\xi)\equiv \pm i\Phi$ for $\xi\in
\vec{\beta}\cap\mathbb{C}_\pm$ and $\phi(\xi)\equiv 0$ for $\xi\in
\vec{\beta}\cap\mathbb{R}$ using the definition
\eqref{eq:wthetaphidef} of $\phi$ in terms of $g(w)$ yields the
following jump conditions for $X(w)$ and $T(w)$ along $\beta$:
\begin{equation}
\begin{split}
X_+(\xi)-X_-(\xi)&=\frac{\pi}{\sqrt{-\xi}R_+(\xi)}\begin{cases}
\displaystyle 2iE(\xi)\mp i\frac{\partial \Phi}{\partial x},\quad & 
\xi\in \vec{\beta}\cap\mathbb{C}_\pm,\\
2iE(\xi),\quad & \xi\in \vec{\beta}\cap\mathbb{R},
\end{cases}\\
T_+(\xi)-T_-(\xi)&=\frac{\pi}{\sqrt{-\xi}R_+(\xi)}\begin{cases}
\displaystyle 2iD(\xi)\mp i\frac{\partial \Phi}{\partial t},\quad & 
\xi\in \vec{\beta}\cap\mathbb{C}_\pm,\\
2iD(\xi),\quad & \xi\in \vec{\beta}\cap\mathbb{R}.
\end{cases}
\label{eq:XTjumps}
\end{split}
\end{equation}
(Note that $\vec{\beta}\cap\mathbb{R}=\emptyset$ in case \librational.)

Since we know directly from their definitions that the functions 
$X(w)$ and $T(w)$ must be $\bo(w^{-2})$ for large $w$, they are
necessarily given by Cauchy integrals via the Plemelj formula in terms of
the jump data \eqref{eq:XTjumps}:
\begin{equation}
\begin{split}
X(w)&=\int_\beta\frac{E(\xi)\,d\xi}{\sqrt{-\xi}R_+(\xi)(\xi-w)} -
\frac{\partial\Phi}{\partial x}
\left[\frac{1}{2}\int_{\beta\cap\mathbb{C}_+}\frac{d\xi}
{\sqrt{-\xi}R_+(\xi)(\xi-w)} -
\frac{1}{2}\int_{\beta\cap\mathbb{C}_-}\frac{d\xi}{\sqrt{-\xi}R_+(\xi)(\xi-w)}
\right],\\
T(w)&=\int_\beta\frac{D(\xi)\,d\xi}{\sqrt{-\xi}R_+(\xi)(\xi-w)} -
\frac{\partial\Phi}{\partial t}
\left[\frac{1}{2}\int_{\beta\cap\mathbb{C}_+}\frac{d\xi}
{\sqrt{-\xi}R_+(\xi)(\xi-w)} -
\frac{1}{2}\int_{\beta\cap\mathbb{C}_-}\frac{d\xi}{\sqrt{-\xi}R_+(\xi)(\xi-w)}
\right].
\end{split}
\label{eq:XTformulas}
\end{equation}
But while these formulae indeed exhibit decay for large $w$, without
further conditions we will have only $X(w)=\bo(w^{-1})$ and
$T(w)=\bo(w^{-1})$ as $w\to\infty$.  Imposing on these formulae the
more rapid required rate of decay of $\bo(w^{-2})$ as $w\to\infty$ reveals conditions determining
the partial derivatives of
$\Phi$ with respect to $x$ and $t$:
\begin{equation}
\frac{\partial\Phi}{\partial x}=\frac{1}{\mathcal{D}}\int_\beta\frac{E(\xi)}
{\sqrt{-\xi}R_+(\xi)}\,d\xi\quad\text{and}\quad
\frac{\partial\Phi}{\partial t}=\frac{1}{\mathcal{D}}
\int_\beta\frac{D(\xi)}{\sqrt{-\xi}R_+(\xi)}\,d\xi,
\end{equation}
where
\begin{equation}
\mathcal{D}:=\frac{1}{2}\int_{\beta\cap\mathbb{C}_+}
\frac{d\xi}{\sqrt{-\xi}R_+(\xi)}-\frac{1}{2}\int_{\beta\cap\mathbb{C}_-}
\frac{d\xi}{\sqrt{-\xi}R_+(\xi)}.
\label{eq:Denominator}
\end{equation}
Now the fractions $E(w)/\sqrt{-w}$ and $D(w)/\sqrt{-w}$ are
simple rational functions of $w$:
\begin{equation}
\frac{E(w)}{\sqrt{-w}}=\frac{i}{4}\left(1-\frac{1}{w}\right)\quad\text{and}
\quad
\frac{D(w)}{\sqrt{-w}}=\frac{i}{4}\left(1+\frac{1}{w}\right),
\end{equation}
and by elementary contour deformations, 
\begin{equation}
\int_\beta\frac{1\pm \xi^{-1}}{R_+(\xi)}\,d\xi = \frac{1}{2}\oint_0
\frac{1\pm w^{-1}}{R(w)}\,dw -\frac{1}{2}\oint_\infty\frac{1\pm w^{-1}}{R(w)}\,dw,
\end{equation}
where the first integral is over a small positively-oriented circle
surrounding only $w=0$ and the second integral is over a large
positively-oriented circle outside of which $R$ is analytic.  (This
result holds regardless of whether the branch point configuration is
of type \librational\ or of type \rotational.)  Evaluating these integrals by
residues at $w=0$ and $w=\infty$ respectively yields
\begin{equation}
\int_\beta\frac{1\pm \xi^{-1}}{R_+(\xi)}\,d\xi = 
\pm \frac{i\pi}{R(0)}- i\pi.
\end{equation}
Therefore, using $R(0)=-\sqrt{\mathfrak{p}^2-\mathfrak{q}}$ establishes the formulae \eqref{eq:kappapartials}.

It remains to characterize the denominator $\mathcal{D}$ in cases \librational\ and \rotational.
In case \librational, elementary contour deformations show that
\begin{equation}
\mathcal{D}=\frac{1}{2}\int_{0}^{-\infty}\frac{dw}{\sqrt{-w}R(w)} = \frac{1}{2}\int_{-\infty}^0
\frac{dw}{\sqrt{-w((w-\mathfrak{p})^2-\mathfrak{q})}}.
\end{equation}
By the substitution $w\mapsto s$ given by
\begin{equation}
w=\sqrt{\mathfrak{p}^2-\mathfrak{q}}\frac{z-1}{z+1}\quad\text{followed by}\quad z=\pm\sqrt{1-s^2}
\label{eq:wsBsubstKm}
\end{equation}
and comparing with the definition \eqref{eq:ellipticKdef} of the complete elliptic integral of the first kind, we establish \eqref{eq:deltalibrational}.
On the other hand, in case \rotational, by simple contour deformations,
\begin{equation}
  \mathcal{D} = \int_{0}^{\mathfrak{p}+\sqrt{\mathfrak{q}}}\frac{dw}{\sqrt{-w}R(w)}=\int_{u+\sqrt{v}}^0\frac{dw}{\sqrt{-w((w-\mathfrak{p})^2-\mathfrak{q})}}.
\end{equation}
By the substitution $w\mapsto s$ given by
\begin{equation}
w=-\left[\frac{\sqrt{-\mathfrak{p}+\sqrt{\mathfrak{q}}}+\sqrt{-\mathfrak{p}-\sqrt{\mathfrak{q}}}}{2s}-\sqrt{\frac{(\sqrt{-\mathfrak{p}+\sqrt{\mathfrak{q}}}+\sqrt{-\mathfrak{p}-\sqrt{\mathfrak{q}}})^2}{4s^2}-\sqrt{\mathfrak{p}^2-\mathfrak{q}}}
\right]^2, 
\label{eq:wsKsubstKm}
\end{equation}
a bijection mapping the real path from $w=w_1=\mathfrak{p}+\sqrt{\mathfrak{q}}<0$ to $w=0$ onto the real path from $s=1$ to $s=0$,
we then obtain \eqref{eq:deltarotational} by comparing with \eqref{eq:ellipticKdef}.
\end{proof}

\subsection{Finding $g$ when $t=0$}
\begin{proposition}
  Suppose $t=0$, and that $\Delta=\emptyset$ for $x\ge 0$ while
  $\nabla=\emptyset$ for $x\le 0$ (for $x=0$ we may choose either
  $\Delta=\emptyset$ or $\nabla=\emptyset$).  Then the equations $M(\mathfrak{p},\mathfrak{q},x,t)=0$
and $I(\mathfrak{p},\mathfrak{q},x,t)=0$ are satisfied identically if 
\begin{equation}
\mathfrak{p}=\mathfrak{p}(x)=1-\frac{1}{2}G(x)^2,\quad x\in\mathbb{R}
\label{eq:utzero}
\end{equation}
and $\mathfrak{q}=\mathfrak{q}(x)=\mathfrak{p}(x)^2-1$.  
\label{prop:utzero}
\end{proposition}
Note that the proof will show that the equation $M(\mathfrak{p},\mathfrak{p}^2-1,x,0)=0$
has \emph{no real solution} $\mathfrak{p}$ if $x<0$ and $\Delta=\emptyset$.  This
explains the need of introducing in general the set $\Delta\subset P_N$
and shows that scattering theory ``from the right'' is insufficient to capture
the semiclassical asymptotics for all $(x,t)$.
\begin{proof}
Consistently with the condition $\mathfrak{q}=\mathfrak{p}^2-1$ we choose the branch cuts
$\beta$ of $R(w;\mathfrak{p},\mathfrak{p}^2-1)$ to lie on the unit circle (and some intervals
of the negative real axis in case \rotational).  Then it follows that
\begin{equation}
R(w;\mathfrak{p},\mathfrak{p}^2-1)=\sqrt{-w}\hat{R}(E(w);\mathfrak{p}),
\label{eq:RRhat}
\end{equation}
where for $\mathfrak{p}<1$, $\hat{R}(\lambda;\mathfrak{p})^2=-16\lambda^2-2(1-\mathfrak{p})$ with
$\hat{R}(\lambda,\mathfrak{p})$ being analytic away from the branch cut in the
$\lambda$-plane lying along the imaginary axis between the two
imaginary roots of $\hat{R}(\lambda;\mathfrak{p})^2$, and with branch chosen so
that $\hat{R}(\lambda,\mathfrak{p})=4i\lambda + \bo(\lambda^{-1})$ as
$\lambda\to\infty$.  Therefore, with $\mathfrak{q}=\mathfrak{p}^2-1$ and $t=0$ the moment condition
$M=0$ can be written in the form
\begin{equation}
x=\frac{1}{\pi}\int_C
\frac{\theta_0'(\xi)\,d\xi}{\hat{R}(E(\xi);\mathfrak{p})}.
\end{equation}
Recalling \eqref{eq:theta0def}, this suggests taking $\lambda=E(\xi)$ as the variable of integration, yielding
\begin{equation}
x=\frac{2}{\pi}\int_{E(C)}\frac{\Psi'(\lambda)\,d\lambda}{\hat{R}(\lambda;\mathfrak{p})},
\label{eq:M0EC}
\end{equation}
where the factor of $2$ appears because the contours of integration
making up the scheme denoted $C$ contain pairs of distinct points
$(\xi,1/\xi)$ that have the same image in the $\lambda$-plane. Here
$E(C)$ denotes the images of these contours in the $\lambda$-plane
counted with the same multiplicities as in the definition of $C$; this
simply means that in \eqref{eq:M0EC} 
\begin{itemize}
\item
if $\Delta=\emptyset$ we are integrating over two contours
from $\lambda=-iG(0)/4$ on the positive imaginary axis to $\lambda=0$, with
the two contours lying on opposite sides of the branch cut of $\hat{R}(\lambda;\mathfrak{p})$, while
\item
if $\nabla=\emptyset$ we are integrating over the same two contours
as in the case that $\Delta=\emptyset$ but with the orientation reversed.
\end{itemize}
By deforming these contours to the imaginary axis taking into account the change
of sign of $\hat{R}(\lambda;\mathfrak{p})$ across its branch cut, and using the
substitution $\lambda=iv/4$ we therefore obtain
\begin{equation}
x=-\frac{4\sigma}{\pi}\int_{\sqrt{2(1-\mathfrak{p})}}^{-G(0)}\frac{\varphi(v)\,dv}{\sqrt{v^2-2(1-\mathfrak{p})}},\quad
\varphi(v):=\frac{d}{dv}\Psi(iv/4),
\end{equation}
where $\sigma=1$ for $\Delta=\emptyset$ and $\sigma=-1$ for
$\nabla=\emptyset$.  It then follows from Proposition~\ref{prop:AbelInverse}
that
\begin{equation}
x=\sigma G^{-1}(-\sqrt{2(1-\mathfrak{p})})
\end{equation}
which can be solved for $\mathfrak{p}=\mathfrak{p}(x)$ if either $x=0$ or $x$ has the same sign as
$\sigma$, yielding in all of these cases the formula \eqref{eq:utzero}.

If the condition $M(\mathfrak{p},\mathfrak{q},x,0)=0$ is used to eliminate the parameter $x$ from
$H(w;\mathfrak{p},\mathfrak{q},x,0)$, then subject also to the condition $\mathfrak{q}=\mathfrak{p}^2-1$ we obtain
\begin{equation}
R(w;\mathfrak{p},\mathfrak{p}^2-1)H(w;\mathfrak{p},\mathfrak{p}^2-1,x,0)=-\frac{1}{2\pi}\hat{R}(E(w);\mathfrak{p})
\int_C\left(\frac{1}{\xi-w}+\frac{1}{2w}\right)
\frac{\theta_0'(\xi)\,d\xi}{\hat{R}(E(\xi);\mathfrak{p})},
\end{equation}
where $\mathfrak{p}=\mathfrak{p}(x)$ is given by \eqref{eq:utzero}.  But since $C$ is mapped
onto itself, with orientation preserved, by the involution $\xi\mapsto
1/\xi$, we have
\begin{equation}
\int_C\frac{1}{\xi-w}
\frac{\theta_0'(\xi)\,d\xi}{\hat{R}(E(\xi);\mathfrak{p})} = \int_C\frac{1}{\xi^{-1}-w}
\frac{\theta_0'(\xi)\,d\xi}
{\hat{R}(E(\xi);\mathfrak{p})},
\end{equation}
so averaging these two formulae we obtain for $\mathfrak{p}=\mathfrak{p}(x)$ given by 
\eqref{eq:utzero} that
\begin{equation}
R(w;\mathfrak{p},\mathfrak{p}^2-1)H(w;\mathfrak{p},\mathfrak{p}^2-1,x,0)=
\frac{1}{2\pi}\hat{R}(E(w);\mathfrak{p})E(w)\frac{dE}{dw}
\int_C\frac{1}{E(w)^2-E(\xi)^2}
\frac{\theta_0'(\xi)\,d\xi}{\hat{R}(E(\xi);\mathfrak{p})}.
\end{equation}
With the substitution $\lambda=E(\xi)$ this becomes
\begin{equation}
R(w;\mathfrak{p},\mathfrak{p}^2-1)H(w;\mathfrak{p},\mathfrak{p}^2-1,x,0)=
\frac{1}{\pi}\hat{R}(E(w);\mathfrak{p})E(w)\frac{dE}{dw}
\int_{E(C)}\frac{1}{E(w)^2-\lambda^2}
\frac{\Psi'(\lambda)\,d\lambda}{\hat{R}(\lambda;\mathfrak{p})}.
\label{eq:RGexacttzero}
\end{equation}
Note that the integral factor is real if $E(w)$ is imaginary.
Now, in the definition \eqref{eq:wgeneralI0def} of $I(\mathfrak{p},\mathfrak{q},x,t)$, the
integration variable $\xi$ lies on the unit circle in the upper half-plane,
and therefore $E(\xi)$ is imaginary, as is $E'(\xi)\,d\xi$ and the 
boundary value taken by $\hat{R}(E(\xi);\mathfrak{p})$ on the branch cut.
It then follows immediately that the integral condition $I=0$ holds.
\end{proof}

Note that the relation $\mathfrak{q}=\mathfrak{p}^2-1$ is suggested by the fact that when $t=0$ 
the inverse-scattering problem reduces to that for the Zakharov-Shabat system \eqref{eq:ZS}
under the mapping $\lambda=E(w)$.  The condition $\mathfrak{q}=\mathfrak{p}^2-1$ maps the radical
$R(w;\mathfrak{p},\mathfrak{q})$ into the radical $\hat{R}(\lambda;\mathfrak{p})$ whose branching
points are known (for Klaus-Shaw potentials, see \cite{Miller02}) to lie on the imaginary axis in the $\lambda$-plane. 

\begin{proposition}
Suppose that $t=0$, and that $\Delta=\emptyset$ for $x\ge 0$ while
$\nabla=\emptyset$ for $x\le 0$ (for $x=0$ we may choose either
$\Delta=\emptyset$ or $\nabla=\emptyset$).  Assume also that $\mathfrak{q}=\mathfrak{q}(x)=\mathfrak{p}(x)^2-1$
where $\mathfrak{p}=\mathfrak{p}(x)$ is given by \eqref{eq:utzero}, and that the contour
$\beta\cup\gamma$ coincides with the union of the unit circle $|w|=1$
and the real interval $[\mathfrak{a},\mathfrak{b}]$.  Then for each $x\in\mathbb{R}$, 
an analytic function $g:\mathbb{C}\setminus(\beta\cup\mathbb{R}_+)\to\mathbb{C}$
is well-defined by Proposition~\ref{prop:wgbasicproperties}, with
associated functions $\theta:\vec{\beta}\cup\vec{\gamma}\to\mathbb{C}$
and $\phi:\vec{\beta}\cup\vec{\gamma} \to\mathbb{C}$ defined by
\eqref{eq:wthetaphidef}, and the following hold:
\begin{itemize}
\item $\Phi=0$.
\item $\phi(\xi)<0$ for $\xi\in\vec{\gamma}$ if $\Delta=\emptyset$ and
$\phi(\xi)>0$ for $\xi\in\vec{\gamma}$ if $\nabla=\emptyset$.  Moreover,
$\phi(\xi)$ is bounded away from zero for $\xi\in\vec{\gamma}$ except
in a neighborhood of either of the two roots of $R(\xi;\mathfrak{p},\mathfrak{q})^2$ (which
are endpoints of $\vec{\gamma}$).
\item $\theta(\xi)$ and $\theta_0(\xi)-\theta(\xi)$ are both real and 
nondecreasing (nonincreasing) with orientation 
if $\Delta=\emptyset$ (if $\nabla=\emptyset$).  Moreover,
for $\xi\in\vec{\beta}$, 
$\theta'(\xi)$ is bounded away from zero except in neighborhoods of 
the two roots of $R(\xi;\mathfrak{p},\mathfrak{q})^2$ (endpoints of $\vec{\beta}$) and,
in case \rotational, the point $\xi=-1$.
\item
$H(\xi)=H(\xi;\mathfrak{p}(x),\mathfrak{q}(x),x,0)$ is bounded away from zero for $\xi\in
\beta\cup\gamma$ except in a neighborhood of $\xi=-1$ where $H(\xi)$
has a simple zero.  
\end{itemize}
\label{prop:wsolntzeroproperties}
\end{proposition}
Recalling the definition \eqref{eq:xcrit} of $x_\mathrm{crit}$, we note that \eqref{eq:utzero} shows that for $t=0$, $g$ is in case \rotational\ for
$|x|<x_\mathrm{crit}$ and in case \librational\ for $|x|>x_\mathrm{crit}$, while
$|x|=x_\mathrm{crit}$ is the borderline case.
\begin{proof}
According to Proposition~\ref{prop:wgbasicproperties}, 
$g(w)\to 0$ as $w\to\infty$ and the boundary values taken
by $g$ on $\beta\cup\mathbb{R}_+$ are related as follows:
\begin{equation}
g_+(\xi)+g_-(\xi) = \begin{cases} 0,\quad & \xi\in\mathbb{R}_+,\\
2iQ(\xi)+\overline{L}(\xi),\quad &\xi\in\beta\cap\mathbb{R},\\
2iQ(\xi)+\overline{L}(\xi)\mp i\Phi,\quad &\xi\in\beta\cap\mathbb{C}_\pm.
\end{cases}
\label{eq:gplusplusgminus}
\end{equation}
Writing $g(w)=\sqrt{-w}R(w) s (w)$, we see that $ s $ is a function
analytic for $w\in\mathbb{C}\setminus\beta$ satisfying $ s (w)=o(w^{-3/2})$
as $w\to\infty$ (and therefore since $\beta$ is bounded, $ s (w)=\bo(w^{-2})$
as $w\to\infty$).  Moreover, the differences of boundary values of $ s $
on $\beta$ are now determined from \eqref{eq:gplusplusgminus}:
\begin{equation}
 s _+(\xi)- s _-(\xi) = \frac{1}{\sqrt{-\xi}R_+(\xi)}\begin{cases} 
2iQ(\xi)+\overline{L}(\xi),\quad &\xi\in\beta\cap\mathbb{R},\\
2iQ(\xi)+\overline{L}(\xi)\mp i\Phi,\quad &\xi\in\beta\cap\mathbb{C}_\pm.
\end{cases}
\end{equation}
It follows that $ s $ is necessarily given by a Cauchy integral:
\begin{equation}
 s (w)=\frac{1}{2\pi i}\int_\beta\frac{2iQ(\xi)+\overline{L}(\xi)}{\sqrt{-\xi}R_+(\xi)}\frac{d\xi}{\xi-w} - \frac{\Phi}{2\pi}
\left[\int_{\beta\cap\mathbb{C}_+}\frac{1}{\sqrt{-\xi}R_+(\xi)}\frac{d\xi}{\xi-w} - \int_{\beta\cap\mathbb{C}_-}\frac{1}{\sqrt{-\xi}R_+(\xi)}\frac{d\xi}{\xi-w}\right].
\end{equation}
Now this formula provides apparent decay at the rate $ s (w)=\bo(w^{-1})$
as $w\to\infty$, so the true faster rate of decay $ s (w)=\bo(w^{-2})$ implies
that
\begin{equation}
\Phi = \frac{1}{2i\mathcal{D}}\int_\beta\frac{2iQ(\xi)+\overline{L}(\xi)}{\sqrt{-\xi}R_+(\xi)}\,d\xi,
\end{equation}
where $\mathcal{D}$ is given by \eqref{eq:deltalibrational} in case \librational\ and by \eqref{eq:deltarotational}
in case \rotational, and is obviously nonzero in both cases.
But if
$t=0$ and if $\mathfrak{p}=\mathfrak{p}(x)$ and $\mathfrak{q}=\mathfrak{q}(x)$ while the contour $\beta$ is as
specified in the hypotheses, then it is easy to check that $\beta$ is
mapped onto itself preserving orientation by the involution
$\xi\to\xi^{-1}$, while the integrand changes sign under this
involution.  This proves that $\Phi=0$ for all
$x$ at $t=0$.

Now Proposition~\ref{prop:wgbasicproperties}
also asserts that 
\begin{equation}
R(\xi)H(\xi) = \phi'(\xi) = 2iQ'(\xi) + \overline{L}'(\xi) 
-2g'(\xi),\quad \xi\in\vec{\gamma},
\end{equation}
where $R(\xi)=R(\xi;\mathfrak{p}(x),\mathfrak{q}(x))$.  Holding $\xi\in\vec{\gamma}$ fixed,
we differentiate with respect to $x$:
\begin{equation}
\frac{\partial}{\partial x}R(\xi)H(\xi) = 2iE'(\xi)-2\frac{\partial^2 g}{\partial\xi\partial x},\quad \xi\in\vec{\gamma}.
\end{equation}
Recalling the function $X(w)$ defined in the proof of 
Proposition~\ref{prop:phasevelocity} 
by \eqref{eq:XTdefine} and given in
explicit form by \eqref{eq:XTformulas}, we have
\begin{equation}
\frac{\partial}{\partial x}R(\xi)H(\xi)=2iE'(\xi)-
\frac{2}{\pi}\frac{\partial}{\partial \xi}\left[\sqrt{-\xi}R(\xi)X(\xi)\right],
\quad \xi\in\vec{\gamma}.
\end{equation}
But, $\partial \Phi/\partial x=0$, so we may evaluate $X(\xi)$ in
closed form using \eqref{eq:XTformulas} and simple contour
deformations:
\begin{equation}
X(\xi)=\int_\beta\frac{E(\zeta)\,d\zeta}{\sqrt{-\zeta}R_+(\zeta)(\zeta-\xi)}=
\frac{\pi}{4\xi R(\xi)}\left[1-\xi + R(\xi)\right],
\label{eq:Xequals}
\end{equation}
where we have used the fact that $\mathfrak{p}^2-\mathfrak{q}=1$ implies $R(0)=-1$.  Therefore,
recalling the definition \eqref{eq:DE} of $E(w)$ we find simply that
\begin{equation}
\frac{\partial}{\partial x}R(\xi)H(\xi)=\frac{1}{2}\frac{\partial}{\partial\xi}
\left[\frac{R(\xi)}{\sqrt{-\xi}}\right],
\quad \xi\in\vec{\gamma}.
\end{equation}
Using \eqref{eq:RRhat} and the identity
\begin{equation}
\frac{\partial\hat{R}}{\partial\lambda}(\lambda;\mathfrak{p}(x))=-\frac{16\lambda}{\mathfrak{p}'(x)}
\frac{\partial\hat{R}}{\partial x}(\lambda;\mathfrak{p}(x))
\end{equation}
then yields
\begin{equation}
\frac{\partial}{\partial x}R(\xi)H(\xi)=-\frac{8E(\xi)E'(\xi)}{\mathfrak{p}'(x)}
\frac{\partial \hat{R}}{\partial x}(E(\xi);\mathfrak{p}(x)),\quad \xi\in\vec{\gamma}.
\end{equation}
Now with $\sigma=1$ for $\Delta=\emptyset$ and $\sigma=-1$ for
$\nabla=\emptyset$ we integrate from $x_0(\xi):=\sigma
G^{-1}(4iE(\xi))$ to $x$:
\begin{equation}
R(\xi)H(\xi)=-8E(\xi)E'(\xi)\int_{\sigma G^{-1}(4iE(\xi))}^x
\frac{\partial\hat{R}}{\partial y}(E(\xi);\mathfrak{p}(y))\frac{dy}{\mathfrak{p}'(y)},\quad \xi\in\vec{\gamma},
\label{eq:RGgap}
\end{equation}
where we observe that the contribution from the lower limit of integration
vanishes because by definition (see \eqref{eq:wgeneralGdef}) 
$H(\xi;\mathfrak{p}(x_0),\mathfrak{q}(x_0),x_0,0)$ is finite and $R(\xi;\mathfrak{p}(x_0),\mathfrak{q}(x_0))=0$.  If
$w_k$ is either of the two roots of $R(w;\mathfrak{p}(x),\mathfrak{q}(x))^2$, then we have
$\phi(w_k)=\pm i\Phi=0$, so
\begin{equation}
\phi(\xi)=\int_{w_k}^{\xi}R(\zeta)H(\zeta)\,d\zeta = 
-8\int_{E(w_k)}^{E(\xi)}\lambda\int_{\sigma G^{-1}(4i\lambda)}^x
\frac{\partial\hat{R}}{\partial y}(\lambda;\mathfrak{p}(y))\frac{dy}{\mathfrak{p}'(y)}\,
d\lambda,\quad \xi\in\vec{\gamma}
\end{equation}
where we have made the substitution $\lambda=E(\zeta)$.  So, if
$\Delta=\emptyset$ (so $x\ge 0$) then $x\ge \sigma G^{-1}(4i\lambda)$
so $dy>0$ and $\mathfrak{p}'(y)\ge 0$ while $\partial
\hat{R}(\lambda;\mathfrak{p}(y))/\partial y\le 0$ and $\lambda\,d\lambda<0$,
yielding $\phi(\xi)<0$ for $\xi\in\vec{\gamma}$.  On the other
hand, if $\nabla=\emptyset$ (so $x\le 0$) then $x\le \sigma
G^{-1}(4i\lambda)$ so $dy<0$ and $\mathfrak{p}'(y)\le 0$ while $\partial
\hat{R}(\lambda;\mathfrak{p}(y))/\partial y\ge 0$ and $\lambda\,d\lambda<0$,
yielding $\phi(\xi)>0$ for $\xi\in\vec{\gamma}$.
In both cases we easily obtain from \eqref{eq:RGgap} the inequality
\begin{equation}
|H(\xi)|=\frac{8|E(\xi)||E'(\xi)|}{|\xi|^{1/2}|\hat{R}(E(\xi);\mathfrak{p}(x))|}
\left|\int_{\sigma G^{-1}(4iE(\xi))}^x
\frac{\partial \hat{R}}{\partial y}(E(\xi);\mathfrak{p}(y))\frac{dy}{\mathfrak{p}'(y)}\right|
\ge \frac{8|E(\xi)||E'(\xi)|}{|\xi|^{1/2}\sup_{y\in\mathbb{R}}|G(y)G'(y)|},
\quad \xi\in\vec{\gamma}.
\label{eq:Gineqgap}
\end{equation}

Proposition~\ref{prop:wgbasicproperties} asserts further that
\begin{equation}
iR_+(\xi)H(\xi)=\theta'(\xi)=-i\left[g'_+(\xi)
-g'_-(\xi)\right],\quad \xi\in\vec{\beta}.
\end{equation}
Differentiating with respect to $x$ for $\xi\in\vec{\beta}$ fixed gives
\begin{equation}
\frac{\partial}{\partial x}iR_+(\xi)H(\xi)=-i\left[
\frac{\partial^2 g_+}{\partial\xi\partial x}-\frac{\partial^2 g_-}
{\partial\xi\partial x}\right]=
-\frac{i}{\pi}\frac{\partial}{\partial\xi}\left[
\sqrt{-\xi}R_+(\xi)(X_+(\xi)+X_-(\xi))\right],\quad \xi\in\vec{\beta}.
\end{equation}
Substituting from the explicit formula \eqref{eq:Xequals} gives
\begin{equation}
\frac{\partial}{\partial x}iR_+(\xi)H(\xi)=\frac{i}{2}\frac{\partial}
{\partial\xi}\left[\frac{R_+(\xi)}{\sqrt{-\xi}}\right],\quad\xi\in\vec{\beta},
\end{equation}
which can be equivalently written in the form
\begin{equation}
\frac{\partial}{\partial x}iR_+(\xi)H(\xi)=-i\frac{8E(\xi)E'(\xi)}{\mathfrak{p}'(x)}
\frac{\partial\hat{R}_+}{\partial x}(E(\xi);\mathfrak{p}(x)),\quad \xi\in\vec{\beta},
\end{equation}
where $\hat{R}_+(E(\xi);\mathfrak{p}(x))$ refers to the boundary value from the
left as the (vertical) branch cut is traversed by $E(\xi)$ when $\xi$
moves along an oriented arc of $\vec{\beta}$.  Integrating from
$x'=\sigma G^{-1}(4iE(\xi))$ to $x'=x$ then gives
\begin{equation}
\frac{d\theta}{d\xi}(\xi)=iR_+(\xi)H(\xi)=-8iE(\xi)E'(\xi)\int_{\sigma G^{-1}(4iE(\xi))}^x
\frac{\partial\hat{R}_+}{\partial y}(E(\xi);\mathfrak{p}(y))\frac{dy}{\mathfrak{p}'(y)},
\quad \xi\in\vec{\beta}.
\label{eq:iRplusGband}
\end{equation}
Of course $d\theta/d\xi=0$ for $\xi\in\vec{\gamma}$.  This shows that
$d\theta/dv$ is a monotone function of $x$, where $v=-4iE(\xi)$ is a real
parameter for $\beta\cup\gamma$.  The extreme value is attained at $x=0$,
at which point $\gamma$ vanishes and according to 
\eqref{eq:WKBphaserewrite}, for each $\xi\in\vec{\beta}$ the
right-hand side of \eqref{eq:iRplusGband} becomes $\theta_0'(\xi)$.
This proves the desired monotonicity of $\theta_0(\xi)-\theta(\xi)$.

Now $\theta(w_k)=0$ when $w_k$ denotes either of the two roots
of $R(w;\mathfrak{p}(x),\mathfrak{q}(x))$, so
\begin{equation}
\theta(\xi)=i\int_{w_k}^{\xi}R_+(\zeta)H(\zeta)\,d\zeta=
-8i\int_{E(w_k)}^{E(\xi)}\lambda\int_{\sigma G^{-1}(4i\lambda)}^x
\frac{\partial\hat{R}_+}{\partial y}(\lambda;\mathfrak{p}(y))\frac{dy}{\mathfrak{p}'(y)}
\,d\lambda,\quad\xi\in\vec{\beta}.
\end{equation}
From this formula it follows that $\theta(\xi)> 0$ for $\xi\in\vec{\beta}$
regardless of whether $\Delta=\emptyset$ or $\nabla=\emptyset$.  Also,
from \eqref{eq:iRplusGband} we have the inequality
\begin{equation}
|H(\xi)|=\frac{8|E(\xi)||E'(\xi)|}{|\xi|^{1/2}|\hat{R}(E(\xi);\mathfrak{p}(x))|}
\left|\int_{\sigma G^{-1}(4iE(\xi))}^x\frac{\partial\hat{R}_+}{\partial y}
(E(\xi);\mathfrak{p}(y))\frac{dy}{\mathfrak{p}'(y)}\right|\ge 
\frac{8|E(\xi)||E'(\xi)|}{|\xi|^{1/2}\sup_{y\in\mathbb{R}}|G(y)G'(y)|},
\quad \xi\in\vec{\beta}.
\label{eq:Gineqband}
\end{equation}

Combining \eqref{eq:Gineqgap} and \eqref{eq:Gineqband} shows that $H$
is bounded away from zero along $\beta\cup\gamma$, uniformly with
respect to $x$, except possibly in neighborhoods of points
$\xi\in\beta\cup\gamma$ at which either $E(\xi)=0$ (corresponding to
$\xi=1$) or $E'(\xi)=0$ (corresponding to $\xi=-1$).  While $H(w)$
indeed has a simple zero at $w=-1$ as is clear from
\eqref{eq:RGexacttzero}, there is no zero at $w=1$ as we will now
show.  Suppose that $x\in\mathbb{R}$ is fixed, and
$\xi\in\vec{\beta}$.  Then by collapsing the contour $E(C)$ to the
imaginary axis and extracting a residue at $\lambda=E(\xi)$,
\eqref{eq:RGexacttzero} shows that
\begin{equation}
H(\xi;\mathfrak{p},\mathfrak{p}^2-1,x,0)=\frac{E'(\xi)\Psi'(E(\xi))}{i\sqrt{-\xi}
\hat{R}_+(E(\xi);\mathfrak{p})}
-\frac{2 E(\xi)}{\pi\sqrt{-\xi}}\frac{dE}{d\xi}(\xi)
\int_{E(\gamma)}\frac{1}{E(\xi)^2-\lambda^2}
\frac{\Psi'(\lambda)\,d\lambda}{\hat{R}(\lambda;\mathfrak{p})},\quad\xi\in\vec{\beta}.
\end{equation}
Taking the limit along $\beta$ of $\xi\to 1$ (either from the upper or
lower half-plane) shows that $H$ has nonzero limiting values
contributed by the first (residue) term in the above formula.

This completes the proof of the assertion regarding the function
$H$.  The assertions regarding the functions $\phi$ and $\theta$
then follow from this, the formulae $\phi'(\xi)=R(\xi)H(\xi)$
for $\xi\in\vec{\gamma}$ and $\theta'(\xi)=iR_+(\xi)H(\xi)$
for $\xi\in\vec{\beta}$, and the inequalities $\theta(\xi)>0$
for $\xi\in\vec{\beta}$, $\phi(\xi)<0$ for $\xi\in\vec{\gamma}$
if $\Delta=\emptyset$, and $\phi(\xi)>0$ or $\xi\in\vec{\gamma}$
if $\nabla=\emptyset$.
\end{proof}

\subsection{Continuation of $g$ to nonzero $t$}
We begin by establishing some differential identities.  Here we are viewing $M$ and $I$
as functions of the roots of $R(w;\mathfrak{p},\mathfrak{q})$ rather than as functions of $\mathfrak{p}$ and $\mathfrak{q}$ themselves.
\begin{proposition}
  Let $w_k$, $k=0,1$, denote either of the two roots of the quadratic
  $R(w;\mathfrak{p},\mathfrak{q})^2$.  Then
\begin{equation}
\frac{\partial M}{\partial w_k}=2\sqrt{-w_k}H(w_k)
\label{eq:M0partials}
\end{equation}
and
\begin{equation}
\frac{\partial I}{\partial w_k}=-\frac{1}{4}\sqrt{-w_k}H(w_k)
\left[\int_{\beta\cap\mathbb{C}_+}\frac{R_+(\xi)\,d\xi}{\sqrt{-\xi}(\xi-w_k)}
+\int_{\beta\cap\mathbb{C}_-}\frac{R_-(\xi)\,d\xi}{\sqrt{-\xi}(\xi-w_k)}
\right],
\label{eq:I0partials}
\end{equation}
where in each case the partial derivative with respect to $w_k$ is
calculated holding the other root fixed, that is, by the chain rule in
which $u$ and $v$ are expressed in terms of the two roots.
\label{prop:partials}
\end{proposition}
\begin{proof}
Writing $R(w)^2=(w-w_0)(w-w_1)$ it follows easily that 
\begin{equation}
\frac{\partial R}{\partial w_k}=-\frac{1}{2}\frac{R(w)}{w-w_k}
\label{eq:partialR}
\end{equation}
where both $w$ and the other root are held fixed.  The identity
\eqref{eq:M0partials} follows by substituting
$\sqrt{\mathfrak{p}^2-\mathfrak{q}}=\sqrt{-w_0}\sqrt{-w_1}$ into \eqref{eq:wM0Cdef},
differentiating under the integral sign using \eqref{eq:partialR} (the
contours of $C$ are independent of $w_k$), and comparing with the
definition \eqref{eq:wgeneralGdef} of $H$.  In a similar way, one shows
that 
\begin{equation}
\frac{\partial}{\partial w_k}R(\xi)H(\xi)=-\frac{1}{2}\frac{\sqrt{-w_k}}{\sqrt{-\xi}}
\frac{R(\xi)}{\xi-w_k}H(w_k),
\end{equation}
and using this in \eqref{eq:wgeneralI0def} proves \eqref{eq:I0partials}.
\end{proof}
\begin{proposition}
Let $w_0$ and $w_1$ denote the two roots of $R(w)^2$.  The Jacobian
\begin{equation}
\mathscr{J}(w_0,w_1):=\det\begin{bmatrix}
\displaystyle\frac{\partial M}{\partial w_0} &
\displaystyle\frac{\partial M}{\partial w_1}\\
\\
\displaystyle\frac{\partial I}{\partial w_0} &
\displaystyle\frac{\partial I}{\partial w_1}\end{bmatrix}
\label{eq:wJacdef}
\end{equation}
of the map $(w_0,w_1)\mapsto (M,I)$ is 
\begin{equation}
\mathscr{J}(w_0,w_1):=-\mathcal{D} \sqrt{-w_0}\sqrt{-w_1}H(w_0)H(w_1)(w_1-w_0)
\label{eq:wJacformula}
\end{equation}
where $H$ is defined by \eqref{eq:wgeneralGdef} and $\mathcal{D}$ is
defined by \eqref{eq:deltalibrational} in case \librational\ and by \eqref{eq:deltarotational}
in case \rotational.
\label{prop:wJacobian}
\end{proposition}
\begin{proof}
This is an immediate consequence of Proposition~\ref{prop:partials},
the definition of $R(w)$,  and the formula
\eqref{eq:Denominator} for $\mathcal{D}$, which applies in both cases \librational\ and \rotational.
\end{proof}
\subsubsection{Continuation from $x\not\in\{0,\pm x_\mathrm{crit}\}$}
\begin{proposition}
There exist disjoint open neighborhoods $\mathscr{O}_\librational^\pm$
and $\mathscr{O}_\rotational^\pm$ in the $(x,t)$-plane with 
\begin{equation}
\mathscr{O}_\librational^-\cap\mathbb{R}=(-\infty,-x_\mathrm{crit})\quad\text{and}
\quad
\mathscr{O}_\librational^+\cap\mathbb{R}=(x_\mathrm{crit},+\infty)
\end{equation}
and
\begin{equation}
\mathscr{O}_\rotational^-\cap\mathbb{R}=(-x_\mathrm{crit},0)\quad\text{and}
\quad\mathscr{O}_\rotational^+\cap\mathbb{R}=(0,x_\mathrm{crit}),
\end{equation}
such that, with $\mathscr{O}:=
\mathscr{O}_\librational^-\cup\mathscr{O}_\rotational^-\cup
\mathscr{O}_\rotational^+\cup\mathscr{O}_\librational^+$, the following
hold true.
\begin{itemize}
\item
There are differentiable maps $\mathfrak{p}:\mathscr{O}\to\mathbb{R}$
and $\mathfrak{q}:\mathscr{O}\to\mathbb{R}$
uniquely determined by the properties
that
\begin{equation}
\mathfrak{p}(x,0)=1-\frac{1}{2}G(x)^2\quad \text{and}\quad \mathfrak{q}(x,0)=\mathfrak{p}(x,0)^2-1,
\quad (x,0)\in\mathscr{O},
\end{equation}
and
\begin{equation}
M(\mathfrak{p}(x,t),\mathfrak{q}(x,t),x,t)=I(\mathfrak{p}(x,t),\mathfrak{q}(x,t),x,t)=0,\quad 
(x,t)\in\mathscr{O},
\end{equation}
where it is assumed that in the definition of $M$ and $I$,
$\Delta=\emptyset$ for
$(x,t)\in\mathscr{O}_\rotational^+\cup\mathscr{O}_\librational^+$ while
$\nabla=\emptyset$ for
$(x,t)\in\mathscr{O}_\rotational^-\cup\mathscr{O}_\librational^-$.  
\item
The quantity $n_\mathrm{p}(x,t)$ defined for $(x,t)\in\mathscr{O}$ in
terms of $\mathfrak{p}(x,t)$ and $\mathfrak{q}(x,t)$ by \eqref{eq:phasevelocity} satisfies
$n_\mathrm{p}(x,0)=0$ and
\begin{equation}
\frac{\partial n_\mathrm{p}}{\partial t}(x,0)<0,\quad
(x,0)\in \mathscr{O}_\rotational^-\cup \mathscr{O}_\librational^+
\end{equation}
and
\begin{equation}
\frac{\partial n_\mathrm{p}}{\partial t}(x,0)>0,\quad
(x,0)\in \mathscr{O}_\librational^-\cup \mathscr{O}_\rotational^+.
\end{equation}
\item
A connected contour
$\beta\cup\gamma$, consisting of the union of (i) 
a Schwartz-symmetric closed curve
passing through $w=1$ and enclosing the origin with (ii) 
the interval $[\mathfrak{a},\mathfrak{b}]$,
can be chosen so that for each $(x,t)\in\mathscr{O}$
an analytic function
$g:\mathbb{C}\setminus(\beta\cup\mathbb{R}_+)\to\mathbb{C}$ is well-defined
by Proposition~\ref{prop:wgbasicproperties}, with associated functions
$\theta:\vec{\beta}\cup\vec{\gamma}\to\mathbb{C}$ and $\phi:\vec{\beta}
\cup\vec{\gamma}\to\mathbb{C}$ defined by \eqref{eq:wthetaphidef},
and so that the following hold:
\begin{itemize}
\item
The function $\phi$ satisfies 
$\Re\{\phi(\xi)\}<0$ for 
$\xi\in\vec{\gamma}$ if $(x,t)\in\mathscr{O}_\rotational^+\cup
\mathscr{O}_\librational^+$ and $\Re\{\phi(\xi)\}>0$ for $\xi\in\vec{\gamma}$
if $(x,t)\in\mathscr{O}_\rotational^-\cup\mathscr{O}_\librational^-$.  Moreover,
$\Re\{\phi(\xi)\}$ is bounded away from zero for $\xi\in\vec{\gamma}$ except in
a neighborhood of either of the two roots of $R(\xi;\mathfrak{p}(x,t),\mathfrak{q}(x,t))^2$ (which
are endpoints of $\vec{\gamma}$).
\item
The function $\theta(\xi)$ 
is real and nondecreasing (nonincreasing) with orientation
for $\xi\in\vec{\beta}$ if $(x,t)\in\mathscr{O}_\rotational^+\cup\mathscr{O}_\librational^+$ (if $(x,t)\in\mathscr{O}_\rotational^-\cup\mathscr{O}_\librational^-$).  Moreover,
$\theta'(\xi)$ is bounded away from zero except in neighborhoods of the two
roots of $R(\xi;\mathfrak{p}(x,t),\mathfrak{q}(x,t))^2$ (endpoints of $\vec{\beta}$) and,
in case \rotational, a single point $\xi=w^+<0$ that converges to $\xi=-1$
as $t\to 0$.
\item The function 
$H(\xi)=H(\xi;\mathfrak{p}(x,t),\mathfrak{q}(x,t),x,t)$ is bounded away from zero for
$\xi\in\beta\cup\gamma$ except in a neighborhood of $\xi=w^+$ where
$H(\xi)$ has a simple zero.
\end{itemize}
\end{itemize}
\label{prop:tneq0continuegeneral}
\end{proposition}
The notation is meant to suggest the fact that the configuration of the
roots of $R(w;\mathfrak{p}(x,t),\mathfrak{q}(x,t))^2$ is of type \rotational\ for 
$(x,t)\in\mathscr{O}^+_\rotational\cup\mathscr{O}^-_\rotational$ and is of type 
\librational\ for 
$(x,t)\in\mathscr{O}^+_\librational\cup\mathscr{O}^-_\librational$.
We have therefore defined the function $n_\mathrm{p}(x,t)$ appearing in the statements of
Theorems~\ref{thm:librational} and \ref{thm:rotational} in terms of $\mathfrak{p}(x,t)$ and $\mathfrak{q}(x,t)$ for $(x,t)\in\mathscr{O}$
by \eqref{eq:phasevelocity}. We are now also in a position to define the accompanying function $\mathcal{E}(x,t)$ for $(x,t)\in \mathscr{O}$:
\begin{equation}
\mathcal{E}(x,t):= -\frac{\mathfrak{p}(x,t)}{\sqrt{\mathfrak{p}(x,t)^2-\mathfrak{q}(x,t)}},\quad (x,t)\in\mathscr{O}.
\label{eq:Euv}
\end{equation}
We can
now also identify the region $S_\librational$ involved in the statement of Theorem~\ref{thm:librational}
as the union $\mathscr{O}_\librational^+\cup\mathscr{O}_\librational^-$.

\begin{proof}
  The existence of the maps $\mathfrak{p}$ and $\mathfrak{q}$ is a consequence of the
  Implicit Function Theorem.  Indeed, since for $x\neq 0$ the roots $w_0$
and $w_1$ of
  $R(w;\mathfrak{p}(x,0),\mathfrak{q}(x,0))^2$ lie within the domain of analyticity of $H(w)$, Proposition~\ref{prop:partials} shows that $M$ and $I$ are differentiable
with respect to these roots. Furthermore, since under the additional
hypothesis $|x|\neq x_\mathrm{crit}$ the roots of $R(w;\mathfrak{p}(x,0),\mathfrak{q}(x,0))^2$
are distinct and neither is equal to $-1$, it follows from
Propositions~\ref{prop:wsolntzeroproperties} and \ref{prop:wJacobian}
that the Jacobian \eqref{eq:wJacdef} is nonzero when
$(x,0)\in\mathscr{O}$.  Therefore we may solve uniquely for $w_0$ and
$w_1$ in terms of $(x,t)$ from the equations $M=I=0$, and since
$\mathfrak{p}=(w_0+w_1)/2$ and $\mathfrak{q}=(w_0-w_1)^2/4$ we also have $\mathfrak{p}$ and $\mathfrak{q}$.

According to \eqref{eq:phasevelocity} in
Proposition~\ref{prop:phasevelocity}, the fact that $\mathfrak{p}(x,0)^2-\mathfrak{q}(x,0)=1$
implies that $n_\mathrm{p}=0$ when $t=0$.  Also, 
\begin{equation}
\frac{\partial n_\mathrm{p}}{\partial t}= -\frac{1}{\sqrt{\Pi}
(1+\sqrt{\Pi})^2}\frac{\partial \Pi}{\partial t},\quad
\Pi:=\mathfrak{p}^2-\mathfrak{q}=w_0w_1
\end{equation}
so that $\partial n_\mathrm{p}/\partial t$ and $\partial \Pi/\partial t$ have opposite signs.  By
differentiation of the equations $M=I=0$ with respect to $t$
one obtains the system of equations
\begin{equation}
\frac{\partial M}{\partial w_0}\frac{\partial w_0}{\partial t}
+ \frac{\partial M}{\partial w_1}\frac{\partial w_1}{\partial t} +
\frac{\partial M}{\partial t} =0\quad\text{and}\quad
\frac{\partial I}{\partial w_0}\frac{\partial w_0}{\partial t}
+ \frac{\partial I}{\partial w_1}\frac{\partial w_1}{\partial t} +
\frac{\partial I}{\partial t} =0,
\end{equation}
from which follows the identity
\begin{equation}
\frac{\partial\Pi}{\partial t} = \frac{1}{\mathscr{J}(w_0,w_1)}
\left[\left(
w_1\frac{\partial M}{\partial w_1}-
w_0\frac{\partial M}{\partial w_0}\right)\frac{\partial I}{\partial t}
+ \left(w_0\frac{\partial I}{\partial w_0}-
w_1\frac{\partial I}{\partial w_1}\right)\frac{\partial M}{\partial t}
\right].
\label{eq:partialPiidentity}
\end{equation}
Now by noting the explicit $t$ dependence in $M$ and (via the definition
of $H$) in $I$ we have
\begin{equation}
\frac{\partial M}{\partial t} = \frac{\sqrt{\Pi}-1}{\sqrt{\Pi}}\quad
\text{and}\quad
\frac{\partial I}{\partial t} = \frac{1}{4\sqrt{\Pi}}\Re\left\{
\int_{\beta\cap\mathbb{C}_+}\frac{R_+(\xi;\mathfrak{p},\mathfrak{q})}{\xi\sqrt{-\xi}}\,d\xi\right\}.
\label{eq:MtIt}
\end{equation}
Therefore, as $t=0$ implies that $\Pi=1$, we also have $\partial
M/\partial t=0$ when $t=0$, simplifying the identity
\eqref{eq:partialPiidentity}.  Using also
Proposition~\ref{prop:partials} and the fact that when $t=0$,
$R_+(\xi;\mathfrak{p},\mathfrak{q})=\sqrt{-\xi}\hat{R}_+(E(\xi);\mathfrak{p}(x,0))$ gives
\begin{equation}
\left.\frac{\partial\Pi}{\partial t}\right|_{t=0}=
\frac{1}{2\mathscr{J}(w_0,w_1)}
\left[w_1\sqrt{-w_1}H(w_1)-w_0\sqrt{-w_0}
H(w_0)\right]\Re\left\{\int_{\beta\cap\mathbb{C}_+}
\hat{R}_+(E(\xi);\mathfrak{p}(x,0))\,\frac{d\xi}{\xi}\right\}.
\label{eq:dPidt1}
\end{equation}
Using \eqref{eq:RGgap} with $\sigma=\mathrm{sgn}(x)$ and taking the
limit as $\xi$ approaches either root $w_k$ of $R(w;\mathfrak{p},\mathfrak{q})^2$ from
$\gamma$, we have the formula
\begin{equation}
\sqrt{-w_k}H(w_k)=-4\frac{D(w_k)}{w_k}U(x),\quad t=0,
\label{eq:sqrtmwjHwjt0}
\end{equation}
where we have also used the identity $D(w)=2wE'(w)$, and 
where $U(x)$ is defined as
\begin{equation}
U(x):=E(w_k)\mathop{\lim_{\xi\to w_k}}_{\xi\in\gamma}
\left[\frac{1}{\hat{R}(E(\xi);\mathfrak{p}(x,0))}\int_{\mathrm{sgn}(x)G^{-1}(4iE(\xi))}^x
\frac{dy}{\hat{R}(E(\xi);\mathfrak{p}(y,0))}\right],\quad
t=0.
\end{equation}
The quantity $U(x)$ has the same value regardless of whether
$w_k=w_0$ or $w_k=w_1$ because $E(w_0)=E(w_1)$ 
at $t=0$, which
explains why we omit any notational dependence of $U$ on $k$.  
Finally, using \eqref{eq:sqrtmwjHwjt0} together with 
formula \eqref{eq:wJacformula} from Proposition~\ref{prop:wJacobian}
we write
\eqref{eq:dPidt1} in the form
\begin{equation}
\left.\frac{\partial\Pi}{\partial t}\right|_{t=0}=
\frac{1}{8\mathcal{D} D(w_0)D(w_1)U(x)}\frac{D(w_1)-D(w_0)}{w_1-w_0}
\Re\left\{\int_{\beta\cap\mathbb{C}_+}\hat{R}_+(E(\xi);\mathfrak{p}(x,0))\,
\frac{d\xi}{\xi}\right\}.
\label{eq:dPidt2}
\end{equation}
Here we have used the relationship $\Pi=w_0w_1=1$ which is valid 
at $t=0$.

We now determine the phases of the various factors in this formula.  
\begin{itemize}
\item
Since
$E(w_0)=E(w_1)$ is a positive imaginary number, and since
$G^{-1}(4iE(\xi))\to |x|$ as $\xi\to w_k$ with $\xi\in\gamma$, while
$\xi\in\gamma$ implies that $G^{-1}(4iE(\xi))<|x|$, and in the range
of integration $\hat{R}(E(\xi);\mathfrak{p}(y,0))<0$ for $\xi\in\gamma$, it
follows that $U(x)$ is imaginary and has the same sign as $x$.
\item
According to \eqref{eq:deltalibrational} and \eqref{eq:deltarotational} from Proposition~\ref{prop:phasevelocity},
$\mathcal{D}$ is a positive quantity.
\item
By explicit calculation,
\begin{equation}
\frac{D(w_1)-D(w_0)}{w_1-w_0}=-\frac{i}{4}\frac{1}{\sqrt{-w_0}+\sqrt{-w_1}}
\left(1+\frac{1}{\sqrt{\Pi}}\right)=-\frac{i}{2}\frac{1}{\sqrt{-w_0}+\sqrt{-w_1}},\quad t=0,
\end{equation}
a quantity that is negative imaginary.
\item  
A similar direct calculation shows that
\begin{equation}
D(w_0)D(w_1)=-\frac{1}{16\sqrt{\Pi}}(1+w_0)(1+w_1)=-\frac{1}{16}(1+w_0)(1+w_1),\quad t=0,
\end{equation}
and this quantity is positive real for $(x,0)\in\mathscr{O}_\rotational^\pm$ but
is negative real for $(x,0)\in\mathscr{O}_\librational^\pm$.
\item Since $\beta\cap\mathbb{C}_+$ is, for $t=0$, an arc of the unit
circle that we may take (without loss of generality) to be oriented in
the counterclockwise direction, we see that $d\xi/\xi=i\,d\theta$,
a positive imaginary increment, while the boundary value $\hat{R}_+(E(\xi);u(x,0))$ is negative imaginary, and hence
\begin{equation}
\int_{\beta\cap\mathbb{C}_+}\hat{R}_+(E(\xi);\mathfrak{p}(x,0))\,\frac{d\xi}{\xi}\in\mathbb{R}_+.
\end{equation}
\end{itemize}
Combining these phases then yields the sign structure of
$\partial\Pi/\partial t$ at $t=0$ that produces the desired sign structure
for $\partial n_{\mathrm{p}}/\partial t$ at $t=0$.

We now describe how to construct the contours $\beta$ and $\gamma$ to
guarantee all of the corresponding conditions in the statement of the
proposition.  Of course all of these conditions are generalizations
for $t\neq 0$ of corresponding conditions that hold true when $t=0$
according to Proposition~\ref{prop:wsolntzeroproperties} when
$\beta\cup\gamma$ is taken to coincide with the union of the unit
circle with the interval $[\mathfrak{a},\mathfrak{b}]$ so our argument will be a
perturbative one, in which the unit circle is replaced by a suitable
nearby curve.  Firstly, since when $t=0$, the function $H$ is bounded
away from zero on $\beta\cup\gamma$ except at $w=-1$ where $H$ has
a simple root, the same holds (also at $t=0$) throughout $\Omega^\circ$
if the latter is chosen without loss of generality to be close enough
to $\beta\cup\gamma$.  Now since $H(w;\mathfrak{p}(x,t),\mathfrak{q}(x,t),x,t)$ is an analytic
function of $w$ depending continuously on $(x,t)$ near $(x,0)$ and that
satisfies $H(w^*)=H(w)$ it will also be bounded away from zero in $\Omega^\circ$
for $t$ small except near some real point $w=w^+$ close to $w=-1$ where
it has a simple zero.  It is easy to check that $H(w;\mathfrak{p}(x,t),\mathfrak{q}(x,t),x,t)$
has (two different) analytic continuations to  a neighborhood of $w=1$
from $\Omega^\circ$ from the upper and lower half planes, so the limiting values
$H(1_\pm;\mathfrak{p}(x,t),\mathfrak{q}(x,t),x,t)$ are both finite and nonzero.  The contour
$\beta\cap\mathbb{C}_+$ is then obtained by solving the well-posed 
autonomous initial-value problem
\begin{equation}
\frac{d\xi^*}{d\tau}=-iR_+(\xi;\mathfrak{p}(x,t),\mathfrak{q}(x,t))H(\xi;\mathfrak{p}(x,t),\mathfrak{q}(x,t),x,t),\quad
\tau>0,
\quad \xi(0)=1
\label{eq:betavectorfield}
\end{equation}
where we interpret $H(1)$ as $H(1_+)$.  Clearly, $\tau$ parameterizes
a trajectory along which $\Im\{\theta\}$ is constant and $\Re\{\theta\}$
is nonincreasing with parametrization $\tau$, since
\begin{equation}
\frac{d\theta(\xi(\tau))}{d\tau} = \frac{d\theta}{d\xi}\cdot\frac{d\xi}{d\tau}
= \left[iR_+(\xi)H(\xi)\right]\left[-iR_+(\xi)H(\xi)\right]^* = 
-\left|iR_+(\xi)H(\xi)\right|^2\le 0.
\end{equation}
The vector field of \eqref{eq:betavectorfield} varies continuously
with time $t$, and the only critical points in $\Omega^\circ$ are
$\xi=w^+$ and the two distinct roots of $R^2$.  Since an integral
curve of this vector field for $t=0$ connected $\xi=1$ with the root
of $R^2$ in $\mathbb{C}_+$ (in case \librational) or with $\xi=-1$ (in
case \rotational) and since the Melnikov-type integral condition $I=0$ continues to
hold true for $t\neq 0$, the solution of the initial-value problem
\eqref{eq:betavectorfield} is a curve terminating at the perturbed
root of $R$ (in case \librational, in finite $\tau$) or at the point
$\xi=w^+<$ (in case \rotational, in infinite $\tau$).  This arc together
with its Schwartz reflection in $\mathbb{C}_-$ and, in case \rotational, the
real interval connecting the roots of $R^2$, is $\beta$ for $t\neq 0$.
Upon assigning $\vec{\beta}$ its orientation according to whether case
$\Delta=\emptyset$ or $\nabla=\emptyset$ holds, we easily see that
since no critical points of the vector field of
\eqref{eq:betavectorfield} lie in $\vec{\beta}$, the desired strict
inequality for $d\theta/d\xi$ holds along $\vec{\beta}$.  To construct
the contour $\gamma$ cruder methods suffice.  In case \librational,
$\gamma$ is the union of the real interval $[\mathfrak{a},\mathfrak{b}]$ with a
Schwartz-symmetrical arc connecting the two complex-conjugate roots of
$R^2$; we define $\gamma\cap\mathbb{C}_+$ for $t\neq 0$ as the image of
the same for $t=0$ under the linear mapping
\begin{equation}
\xi\mapsto\frac{w_0(x,t)-w^+}{w_0(x,0)+1}\xi + \frac{w_0(x,t)+w^+w_0(x,0)}
{w_0(x,0)+1}
\end{equation}
($w_0(x,t)$ is the root of $R^2$ in $\mathbb{C}_+$) taking
$\xi=w_0(x,0)$ to $w_0(x,t)$ and $\xi=-1$ to $w^+$.  In case
\rotational\ there is nothing to do since $\gamma$ has to be the union of
real intervals $[\mathfrak{a},w_0(x,t)]\cup[w_1(x,t),\mathfrak{b}]$ where
$w_0(x,t)<w_1(x,t)$ are the two roots of $R^2$.  In both cases, an
easy continuity argument together with the fact that $\Re\{\phi\}=0$
at the simple roots of $R^2$ and a local analysis of
$d\phi/d\xi=R(\xi)H(\xi)$ near these roots shows that the desired
inequality for $\Re\{\phi(\xi)\}$ holds on $\vec{\gamma}$.
\end{proof}
It should be noted that elements of this proof actually provide computationally
feasible numerical methods for continuation of $g$ as $x$ and $t$ vary.

\subsubsection{Continuation from $x=0$}
\label{continuation-from-x=0}
When $x=t=0$, the roots of $R(w;\mathfrak{p},\mathfrak{q})^2$ coincide with $\mathfrak{a}$ and $\mathfrak{b}$
and this implies that $M$ and $I$ are not differentiable with respect
to the roots at this point.  To circumvent this difficulty, it now
becomes necessary to exploit more than the two simplest configurations
of $\Delta$ and $\nabla$ first introduced in \S\ref{sec:choiceofDelta};
we must now consider the possibility that neither $\Delta$ nor $\nabla$
is empty.

We begin with several observations concerning the functions
$M(\mathfrak{p},\mathfrak{q},x,t)$ and $I(\mathfrak{p},\mathfrak{q},x,t)$.  Fix some small $\delta>0$ and
consider configurations of type \rotational\ in which the roots
$w_\prec=\mathfrak{p}-\sqrt{\mathfrak{q}}$ and $w_\succ=\mathfrak{p}+\sqrt{\mathfrak{q}}$ of $R(w;\mathfrak{p},\mathfrak{q})^2$ (in this section
this will be more suggestive notation than $w_0$ and $w_1$) 
satisfy $\mathfrak{a}<w_\prec<\mathfrak{a}+\delta<-1<\mathfrak{b}-\delta<w_\succ<\mathfrak{b}$, 
which bounds $\mathfrak{q}>0$
away from zero.  If we further fix two real values
$\tau^\prec_{\infty}$ and $\tau^\succ_\infty$ with
$\mathfrak{a}+\delta<\tau^\prec_\infty<-1<\tau^\succ_\infty<\mathfrak{b}-\delta$, then we
may compare $M$ and $I$ for the various cases listed in
\S\ref{sec:choiceofDelta}, in which we use the transition point
$\tau_\infty=\tau^\prec_\infty$ 
when we have $\Delta=P^{\prec\kink}_N$ or
$\nabla=P^{\prec\kink}_N$, and we use the transition point
$\tau_\infty=\tau^\succ_\infty$ 
when we have $\Delta=P^{\kink\succ}_N$ or
$\nabla=P^{\kink\succ}_N$.

The first observation is that in this situation the functions $(M,I)$
are the same in the case $\Delta=P^{\prec\kink}_N$ as in the case
$\nabla=P^{\kink\succ}_N$, 
and are the same in the case $\nabla=P^{\prec\kink}_N$
as in the case $\Delta=P^{\kink\succ}_N$.  To see this, it is useful
to note that by a simple contour deformation in which the components of $\partial\Omega_\pm^\nabla\setminus\Sigma^\nabla$ and 
$\partial\Omega_\pm^\Delta\setminus\Sigma^\Delta$ are collapsed toward $\beta\cup\gamma$ we may write $M$, originally defined by \eqref{eq:wM0Cdef},
in the form
\begin{equation}
M=\frac{x-t}{\sqrt{\mathfrak{p}^2-\mathfrak{q}}}+x+t+\frac{4}{\pi}\int_\gamma
\frac{\theta_0'(\xi)\sqrt{-\xi}\,d\xi}{R(\xi;\mathfrak{p},\mathfrak{q})}
\label{eq:M0rewrite}
\end{equation}
and we see that the only way that this formula depends on the
choice of $\Delta$ is via the orientation of the contour arcs in $\gamma$,
and so the desired equivalence for $M$ follows because the transition
points lie in the complementary contour $\beta$.  Exactly the same contour
deformations, when applied to the definition of $H(w)$, will result
in the additional contribution of a residue at $\xi=w$; if $w\in\beta$
we have:
\begin{equation}
H(w)=-\frac{1}{4\sqrt{-w}}\left[\frac{x-t}{w\sqrt{\mathfrak{p}^2-\mathfrak{q}}}-\frac{4}{\pi}
\int_\gamma
\frac{\theta_0'(\xi)\sqrt{-\xi}\,d\xi}{R(\xi;\mathfrak{p},\mathfrak{q})(\xi-w)}\right]+\frac{\theta_0'(w)}{iR_+(w;\mathfrak{p},\mathfrak{q})}\,.
\label{eq:Hrewrite}
\end{equation}
Our analysis of the formula \eqref{eq:M0rewrite} applies to all but
the last term.  This last term does indeed distinguish between
$\Delta=P^{\prec\kink}_N$ and $\nabla=P^{\kink\succ}_N$
and between $\Delta=P^{\kink\succ}_N$ and
$\nabla=P^{\prec\kink}_N$ due to a change of orientation of
$\beta\cap \mathbb{C}_+$, which changes the sign of the boundary value
$R_+(w;\mathfrak{p},\mathfrak{q})$.  However, this discrepancy contributes nothing to the integral $I$,
since recalling the definition \eqref{eq:wgeneralI0def} we have
\begin{equation}
\Re\left\{\int_{\beta\cap\mathbb{C}_+}R_+(\xi;\mathfrak{p},\mathfrak{q})\left[\frac{2\theta_0'(\xi)}{iR_+(\xi;\mathfrak{p},\mathfrak{q})}
\right]\,d\xi
\right\}=\pm 2\Im\{\theta_0(w^+)-\Psi(0)\} = 0,
\end{equation}
and for the remaining terms we note that integrating over the (oriented) contour
$\beta\cap\mathbb{C}_+$ against the (oriented) boundary value $R_+(\xi;\mathfrak{p},\mathfrak{q})$
is an orientation-invariant operation.

The second observation is that since $\mathfrak{q}$ is bounded away from zero, we
may consider $M$ and $I$, in any of the six choices of $\Delta$
listed in \S\ref{sec:choiceofDelta}, as well-defined functions of the
roots $w_\prec<w_\succ$.  
These are both analytic functions of $w_\prec$ and $w_\succ$
in the intervals $\mathfrak{a}<w_\prec<\mathfrak{a}+\delta$ and 
$\mathfrak{b}-\delta<w_\succ<\mathfrak{b}$.
Next, recall the part of Proposition~\ref{prop:theta0} guaranteeing that
$\theta_0(w)$ has an analytic continuation from $w>\mathfrak{a}$ and
$w<\mathfrak{b}$ to small neighborhoods of $w=\mathfrak{a}$ and of $w=\mathfrak{b}$
respectively.  Using this fact, we may construct the analytic
continuation $\mathcal{M}_\prec\{(M,I)\}$ of $(M,I)$ with respect
to $w_\prec$ (holding $w_\succ$, $x$, and $t$ fixed) about a small
positively-oriented loop about $w=\mathfrak{a}$ beginning and ending on the real
axis with $\mathfrak{a}<w_\prec<\mathfrak{a}+\delta$.  We may also construct a similar
analytic continuation denoted $\mathcal{M}_\succ\{(M,I)\}$ with
respect to $w_\succ$ (holding $w_\prec$, 
$x$, and $t$ fixed) about a small
positively-oriented loop about $w=\mathfrak{b}$ beginning and ending on the
real axis with $\mathfrak{b}-\delta<w_\succ<\mathfrak{b}$.  

The third and most important observation is that these two monodromy
operations $\mathcal{M}_\prec$ and $\mathcal{M}_\succ$ 
simply result in involutive permutations
among the various cases of choice of $\Delta$ enumerated in
\S\ref{sec:choiceofDelta}.  These relations are shown in 
Figure~\ref{fig:Monodromy}.
\begin{figure}[h]
\begin{center}
\includegraphics{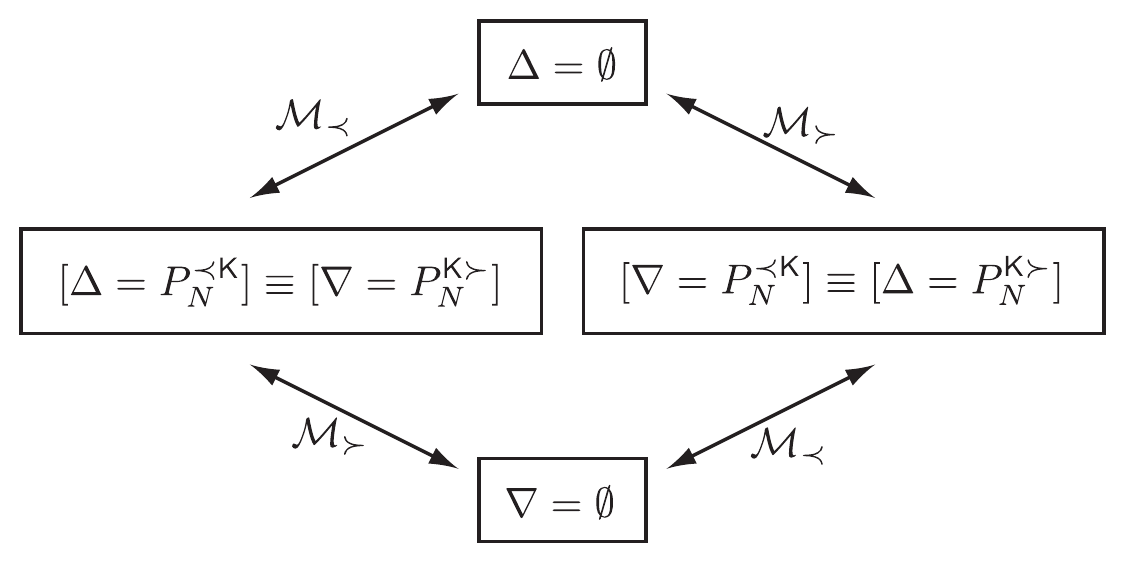}
\end{center}
\caption{\emph{The effect of the two monodromy generators on the four
distinct types of function pairs $(M,I)$.} }
\label{fig:Monodromy}
\end{figure}
These facts are elementary consequences of the formulae \eqref{eq:M0rewrite}
and \eqref{eq:Hrewrite}. Indeed, the analytic continuation 
operations both leave the non-integral terms invariant, and 
$\mathcal{M}_\prec$ (respectively $\mathcal{M}_\succ$) 
may be realized in these formulae by analytic
continuation of the integrand from the real axis and the generalization
of the real segment of $\gamma$ near $w=\mathfrak{a}$ (respectively $w=\mathfrak{b}$)
to a straight-line contour
connecting $w=\mathfrak{a}$ with $w=w_\prec$ (respectively connecting $w=\mathfrak{b}$
with $w=w_\succ$).  This latter operation results in a change of sign
of the square root in the integrand, which is equivalent to the reversal
of orientation of the corresponding segment of $\gamma$, and hence in
the desired permutation of formulae.

Let $\Gamma_\prec$ be the two-sheeted Riemann surface obtained
from taking two copies of the disc $|w-\mathfrak{a}|<\delta$ slit along $w<\mathfrak{a}$
and identified in the usual way.  Similarly, let $\Gamma_\succ$
be the two-sheeted Riemann surface obtained from taking two copies of
the disc $|w-\mathfrak{b}|<\delta$ slit along $w>\mathfrak{b}$ and identified in
the usual way.  The multivalued square roots
$k_\prec:=(w-\mathfrak{a})^{1/2}$ and $k_\succ:=(\mathfrak{b}-w)^{1/2}$
are global analytic coordinates for $\Gamma_\prec$ and
$\Gamma_\succ$ respectively. We now define a function
$\Gamma_\prec\times\Gamma_\succ\times\mathbb{R}^2_{x,t}
\to\mathbb{C}^2$ by setting
\begin{equation}
(\hat{M},\hat{I})(k_\prec,k_\succ,x,t):=
\begin{cases}
\displaystyle
\left.(M,I)(\mathfrak{p},\mathfrak{q},x,t)\right|_{\Delta=\emptyset},
\quad &\Re\{k_\prec\}\ge 0,\quad\Re\{k_\succ\}\ge 0\\
\displaystyle
\left.(M,I)(\mathfrak{p},\mathfrak{q},x,t)\right|_{\Delta=P^{\prec\kink}_N},
\quad &\Re\{k_\prec\}\le 0,\quad\Re\{k_\succ\}\ge 0\\
\displaystyle
\left.(M,I)(\mathfrak{p},\mathfrak{q},x,t)\right|_{\nabla=P^{\prec\kink}_N},
\quad &\Re\{k_\prec\}\ge 0,\quad\Re\{k_\succ\}\le 0\\
\displaystyle
\left.(M,I)(\mathfrak{p},\mathfrak{q},x,t)\right|_{\nabla=\emptyset},
\quad &\Re\{k_\prec\}\le 0,\quad\Re\{k_\succ\}\le 0,
\end{cases}
\label{eq:M0I0RSdef}
\end{equation}
where on the right-hand side, $\mathfrak{p}=(w_\prec+w_\succ)/2$, $\mathfrak{q}=(w_\prec-w_\succ)^2/4$, and
$w_\prec:=\mathfrak{a}+k_\prec^2$ while
$w_\succ:=\mathfrak{b}-k_\succ^2$.  In the second and third
lines we could have equivalently used $\nabla=P^{\kink\succ}_N$
and $\Delta=P^{\kink\succ}_N$ respectively.  The monodromy
arguments above show that $(\hat{M},\hat{I})$ is a pair of
single-valued analytic functions on
$\Gamma_\prec\times\Gamma_\succ$ for each
$(x,t)\in\mathbb{R}^2$, except possibly on the coordinate axes
$k_\prec=0$ or $k_\succ=0$.  But it is easy to see that
the pair $(\hat{M},\hat{I})$ extends continuously to the axes, and
hence any singularities are removable, so $(\hat{M},\hat{I})$ is a
pair of analytic functions defined on the whole product
$\Gamma_\prec\times\Gamma_\succ$.

In fact, given the above discussion, it is not hard to write down explicit formulae for the
functions $\hat{M}$ and $\hat{I}$.  Indeed, starting from the formula \eqref{eq:M0rewrite} for $M$ and taking into
account the change of orientation of the two arcs of $\gamma$ corresponding to changes in
sign of $k_\prec$ and $k_\succ$, we arrive at the formula
\begin{equation}
\begin{split}
\hat{M}(k_\prec,k_\succ,x,t)&=
\frac{x-t}{\sqrt{1-\mathfrak{a}k_\succ^2+\mathfrak{b}k_\prec^2-k_\prec^2k_\succ^2}}+x+t\\
&\qquad{}-\frac{4k_\prec}{\pi}\int_0^1\frac{\theta_0'(\mathfrak{a}+k_\prec^2s)\sqrt{-\mathfrak{a}-k_\prec^2s}\,ds}{\sqrt{1-s}\sqrt{\mathfrak{b}-\mathfrak{a}-k_\succ^2-k_\prec^2s}} +
\frac{4k_\succ}{\pi}\int_0^1\frac{\theta_0'(\mathfrak{b}-k_\succ^2s)\sqrt{-\mathfrak{b}+k_\succ^2s}\,ds}{\sqrt{1-s}\sqrt{\mathfrak{b}-\mathfrak{a}-k_\prec^2-k_\succ^2s}}.
\end{split}
\label{eq:HatM}
\end{equation}
Since the last term of the rewritten formula \eqref{eq:Hrewrite} for $H(w)$ contributes nothing to $I$,
we may omit it and apply the same process as used to arrive at the above formula for $\hat{M}$
to obtain
\begin{equation}
\hat{I}(k_\prec,k_\succ,x,t)=\Re\left\{\int_{\beta\cap\mathbb{C}_+}R_+(w;\mathfrak{p},\mathfrak{q})\hat{H}(w)\,dw\right\}
\label{eq:HatI}
\end{equation}
where $\mathfrak{p}=\tfrac{1}{2}(\mathfrak{a}+\mathfrak{b}+k_\prec^2-k_\succ^2)$ and $\mathfrak{q}=\tfrac{1}{4}(\mathfrak{b}-\mathfrak{a}-k_\prec^2-k_\succ^2)^2$
and where
\begin{equation}
\begin{split}
\hat{H}(w)&:=-\frac{1}{4\sqrt{-w}}\Bigg[
\frac{x-t}{w\sqrt{1-\mathfrak{a}k_\succ^2+\mathfrak{b}k_\prec^2-k_\prec^2k_\succ^2}}\\
&\qquad\qquad{}+\frac{4k_\prec}{\pi}\int_0^1\frac{\theta_0'(\mathfrak{a}+k_\prec^2s)\sqrt{-\mathfrak{a}-k_\prec^2s}\,ds}
{(\mathfrak{a}+k_\prec^2s-w)\sqrt{1-s}\sqrt{\mathfrak{b}-\mathfrak{a}-k_\succ^2-k_\prec^2s}}\\
&\qquad\qquad{}-\frac{4k_\succ}{\pi}\int_0^1\frac{\theta_0'(\mathfrak{b}-k_\succ^2s)\sqrt{-\mathfrak{b}+k_\succ^2s}\,ds}
{(\mathfrak{b}-k_\succ^2s-w)\sqrt{1-s}\sqrt{\mathfrak{b}-\mathfrak{a}-k_\prec^2-k_\succ^2s}}\Bigg].
\end{split}
\label{eq:HatH}
\end{equation}

In particular, it is easy to see that the two functions $\hat{M}$ and
$\hat{I}$ are differentiable with respect to $k_\prec$ and
$k_\succ$ at the origin $k_\prec=k_\succ=0$ where,
essentially, all six cases of $(M,I)$ coincide with the same
value, corresponding to the configuration $w_\prec=\mathfrak{a}$ and
$w_\succ=\mathfrak{b}$ that occurs when $x=t=0$.  Therefore, we may
compute the Jacobian of $(\hat{M},\hat{I})$ with respect to the
global analytic coordinates $(k_\prec,k_\succ)$ of the
manifold $\Gamma_\prec\times\Gamma_\succ$ (holding $x$
and $t$ fixed).  By analogy with \eqref{eq:wJacdef} we define
\begin{equation}
\hat{\mathscr{J}}(k_\prec,k_\succ):=\det
\begin{bmatrix}\displaystyle\frac{\partial\hat{M}}{\partial k_\prec} &
\displaystyle\frac{\partial\hat{M}}{\partial k_\succ}\\\\
\displaystyle\frac{\partial\hat{I}}{\partial k_\prec} &
\displaystyle\frac{\partial\hat{I}}{\partial k_\succ}
\end{bmatrix}.
\label{eq:wJactilde}
\end{equation}
\begin{proposition}
At the origin $k_\prec=k_\succ=0$, the Jacobian is
\begin{equation}
\hat{\mathscr{J}}(0,0)=-\frac{4\mathcal{D}_0}{\pi^2}(G(0)^2-4)\left[
\left.\frac{d}{dv}\Psi(iv/4)\right|_{v=-G(0)}\right]^2,
\label{eq:Jhatformula}
\end{equation}
where $\mathcal{D}_0$ denotes the positive 
quantity $\mathcal{D}$ defined by \eqref{eq:deltarotational}
in the case that the roots $\mathfrak{p}\pm\sqrt{\mathfrak{q}}$ of $R(w)^2$ are taken to be $\mathfrak{a}$ and $\mathfrak{b}$.  
Hence, $\hat{\mathscr{J}}(0,0)<0$ via Proposition~\ref{prop:theta0} 
and Assumption~\ref{assume:rotational}, and
$\hat{\mathscr{J}}(0,0)$ is also independent of $x$ and $t$.
\label{prop:sqrtjacobian}
\end{proposition}
\begin{proof}
We compute the Jacobian $\hat{\mathscr{J}}(0,0)$ by calculating the
partial derivatives of $\hat{\mathscr{J}}(k_\prec,k_\succ)$
for $k_\prec>0$ and $k_\succ>0$ and then letting 
$k_\prec\downarrow 0$ and $k_\succ\downarrow 0$.  In this situation we can use
the formula $(\hat{M},\hat{I})=\left.(M,I)\right|_{\Delta=\emptyset}$
with the right-hand side evaluated at $w_\prec=\mathfrak{a}+k_\prec^2$
and $w_\succ=\mathfrak{b}-k_\succ^2$.  But then, by 
Proposition~\ref{prop:wJacobian} we have
\begin{equation}
\begin{split}
\hat{\mathscr{J}}(k_\prec,k_\succ)&=
\det
\left.\begin{bmatrix}\displaystyle\frac{\partial M}{\partial w_\prec} &
\displaystyle\frac{\partial M}{\partial w_\succ}\\\\
\displaystyle\frac{\partial I}{\partial w_\prec} &
\displaystyle\frac{\partial I}{\partial w_\succ}
\end{bmatrix}\right|_{\Delta=\emptyset}
\cdot\det\begin{bmatrix}
\displaystyle\frac{\partial w_\prec}{\partial k_\prec} &
\displaystyle\frac{\partial w_\prec}{\partial k_\succ}\\\\
\displaystyle\frac{\partial w_\succ}{\partial k_\prec} &
\displaystyle\frac{\partial w_\succ}{\partial k_\succ}
\end{bmatrix}\\
&= -4k_\prec k_\succ\left.\mathscr{J}(w_\prec,w_\succ)
\right|_{\Delta=\emptyset}\\
&=4\mathcal{D}\sqrt{-\mathfrak{a}-k_\prec^2}\sqrt{-\mathfrak{b}+k_\succ^2}
(\mathfrak{b}-\mathfrak{a}-k_\prec^2-k_\succ^2)\\
&\quad\quad\quad{}\cdot 
k_\prec H_{\Delta=\emptyset}(\mathfrak{a}+k_\prec^2)\cdot
k_\succ H_{\Delta=\emptyset}(\mathfrak{b}-k_\succ^2),
\end{split}
\end{equation}
where the notation $H_{\Delta=\emptyset}(w)$ specifies that the formula 
\eqref{eq:wgeneralGdef} is to be interpreted in the case $\Delta=\emptyset$.
Taking the limit $k_\prec\downarrow 0$ and $k_\succ\downarrow 0$
and recalling that $\mathfrak{ab}=1$ gives
\begin{equation}
\hat{\mathscr{J}}(0,0)=4\mathcal{D}_0 (\mathfrak{b}-\mathfrak{a})\cdot
\lim_{k_\prec,k_\succ\downarrow 0} k_\prec H_{\Delta=\emptyset}
(\mathfrak{a}+k_\prec^2)\cdot
\lim_{k_\prec,k_\succ\downarrow 0} k_\succ H_{\Delta=\emptyset}
(\mathfrak{b}-k_\succ^2).
\label{eq:Jhatintermediate}
\end{equation}
By applying elementary contour deformations to \eqref{eq:wgeneralGdef} 
we see that for any sufficiently small $d>0$,
\begin{equation}
\lim_{k_\prec,k_\succ\downarrow 0} k_\prec H_{\Delta=\emptyset}
(\mathfrak{a}+k_\prec^2) = -\frac{1}{\pi}\lim_{k_\prec,k_\succ\downarrow 0}
\frac{k_\prec}{\sqrt{-\mathfrak{a}-k_\prec^2}}\int_{\mathfrak{a}-d}^{\mathfrak{a}}
\frac{\theta_0'(\xi)\sqrt{-\xi}\,d\xi}{R(\xi)(\xi-\mathfrak{a}-k_\prec^2)}
\end{equation}
and
\begin{equation}
\lim_{k_\prec,k_\succ\downarrow 0}k_\succ H_{\Delta=\emptyset}
(\mathfrak{b}-k_\succ^2) = -\frac{1}{\pi}
\lim_{k_\prec,k_\succ\downarrow 0}\frac{k_\succ}{\sqrt{-\mathfrak{b}+k_\succ^2}}
\int_{\mathfrak{b}+d}^{\mathfrak{b}}
\frac{\theta_0'(\xi)\sqrt{-\xi}\,d\xi}{R(\xi)(\xi-\mathfrak{b}+k_\succ^2)}.
\end{equation}
Now,
\begin{equation}
\mathfrak{a}-d<\xi<\mathfrak{a}\;\;\text{implies}\;\;R(\xi)(\xi-\mathfrak{a}-k_\prec^2)=
(\mathfrak{a}+k_\prec^2-\xi)^{3/2}\sqrt{\mathfrak{b}-k_\succ^2-\xi}
\end{equation}
and
\begin{equation}
\mathfrak{b}<\xi<\mathfrak{b}+d\;\;\text{implies}\;\;
R(\xi)(\xi-\mathfrak{b}+k_\succ^2)=-(\xi-\mathfrak{b}+k_\succ^2)^{3/2}
\sqrt{\xi-\mathfrak{a}-k_\prec^2}
\end{equation}
where in both cases we mean the positive $3/2$ power on the right-hand side.
Therefore, by the substitution $\xi=\mathfrak{a}+k_\prec^2\zeta$ and a dominated
convergence argument,
\begin{equation}
\lim_{k_\prec,k_\succ\downarrow 0}
\frac{k_\prec}{\sqrt{-\mathfrak{a}-k_\prec^2}}\int_{\mathfrak{a}-d}^{\mathfrak{a}}
\frac{\theta_0'(\xi)\sqrt{-\xi}\,d\xi}
{R(\xi)(\xi-\mathfrak{a}-k_\prec^2)}=
-\frac{2\theta_0'(\mathfrak{a})}{\sqrt{\mathfrak{b}-\mathfrak{a}}}.
\end{equation}
Similarly, but now using the substitution $\xi=\mathfrak{b}-k_\succ^2\zeta$,
\begin{equation}
\lim_{k_\prec,k_\succ\downarrow 0}
\frac{k_\succ}{\sqrt{-\mathfrak{b}+k_\succ^2}}\int_{\mathfrak{b}+d}^{\mathfrak{b}}
\frac{\theta_0'(\xi)\sqrt{-\xi}\,d\xi}
{R(\xi)(\xi-\mathfrak{b}+k_\succ^2)}=-\frac{2\theta_0'(\mathfrak{b})}{\sqrt{\mathfrak{b}-\mathfrak{a}}}.
\end{equation}
Therefore,  \eqref{eq:Jhatintermediate} becomes
\begin{equation}
\hat{\mathscr{J}}(0,0)=\frac{16\mathcal{D}_0}{\pi^2}
\theta_0'(\mathfrak{a})\theta_0'(\mathfrak{b}).
\end{equation}
Since $E(\mathfrak{a})=E(\mathfrak{b})=-iG(0)/4$, 
\begin{equation}
\hat{\mathscr{J}}(0,0)=\frac{16\mathcal{D}_0}{\pi^2}
\Psi'(-iG(0)/4)^2
E'(\mathfrak{a})E'(\mathfrak{b}).
\end{equation}
Next, since $E'(\mathfrak{a})E'(\mathfrak{b})=-(\mathfrak{a}+\mathfrak{b}+2)/64$ and $\lambda=iv/4$,
\begin{equation}
\hat{\mathscr{J}}(0,0)=\frac{4\mathcal{D}_0}{\pi^2}\left[\left.
\frac{d}{dv}\Psi(iv/4)\right|_{v=-G(0)}\right]^2(\mathfrak{a}+\mathfrak{b}+2).
\end{equation}
Finally, recalling the definitions \eqref{eq:abdef} of $a$ and $b$,  we obtain \eqref{eq:Jhatformula}.
\end{proof}

\begin{proposition}
\label{prop:inequalities}
There exists an open neighborhood $\mathscr{O}_\rotational^0$ of the origin $(0,0)$
in the $(x,t)$-plane such that the following hold true.
\begin{itemize}
\item
There are differentiable maps $k_\prec:\mathscr{O}_\rotational^0\to\mathbb{R}$
and $k_\succ:\mathscr{O}_\rotational^0\to\mathbb{R}$ uniquely characterized by 
the properties that  $k_\prec(0,0)=k_\succ(0,0)=0$ and 
\begin{equation}
\hat{M}(k_\prec(x,t),k_\succ(x,t),x,t)=
\hat{I}(k_\prec(x,t),k_\succ(x,t),x,t)=0,\quad (x,t)\in
\mathscr{O}_\rotational^0.
\end{equation}
Via the identification $w_\prec=\mathfrak{a}+k_\prec^2$ and 
$w_\succ=\mathfrak{b}-k_\succ^2$ and the relations 
$\mathfrak{p}=(w_\prec+w_\succ)/2$ and $\mathfrak{q}=(w_\prec-w_\succ)^2/4$,
we obtain a solution of the equations $M(\mathfrak{p}(x,t),\mathfrak{q}(x,t),x,t)=0$
and $I(\mathfrak{p}(x,t),\mathfrak{q}(x,t),x,t)=0$ corresponding to a configuration with
\begin{itemize}
\item
$\Delta=\emptyset$ when $k_\prec\ge 0$ and $k_\succ\ge 0$,
\item
$\Delta=P_N^{\prec\kink}$ or $\nabla=P_N^{\kink\succ}$ when
$k_\prec\le 0$ and $k_\succ\ge 0$,
\item
$\nabla=P_N^{\prec\kink}$ or $\Delta=P_N^{\kink\succ}$
when $k_\prec\ge 0$ and $k_\succ\le 0$,
\item
$\nabla=\emptyset$ when $k_\prec\le 0$ and $k_\succ\le 0$.
\end{itemize}
Also, the functions $k_\prec(x,t)$ and $k_\succ(x,t)$ 
satisfy
\begin{equation}
\frac{\partial k_\prec}{\partial t}(0,0)<0 \quad\text{and}\quad
\frac{\partial k_\succ}{\partial t}(0,0)>0
\label{eq:dkdts}
\end{equation}
and
\begin{equation}
\frac{\partial k_\prec}{\partial x}(0,0)>0\quad\text{and}\quad
\frac{\partial k_\succ}{\partial x}(0,0)>0.
\label{eq:dkdxs}
\end{equation}
\item Let $\Delta$ and $\nabla$ be chosen according to the signs of
$k_\prec(x,t)$ and $k_\succ(x,t)$ as above.  Then, there
is a Schwartz-symmetric closed curve transversely
  intersecting the real axis (only) at $w=1$ and $w=w^+$ (see below) 
such that with $\beta$ chosen as the
  union of this curve with the interval
  $[w_\prec(x,t),w_\succ(x,t)]$ and $\gamma$ chosen as
  the union of closed intervals $[\mathfrak{a},w_\prec(x,t)]$ and
  $[w_\succ(x,t),\mathfrak{b}]$ (recall that the local orientation of these
  contours depends on choice of $\Delta$),
there is for each
  $(x,t)\in\mathscr{O}_\rotational^0$ an analytic function
  $g:\mathbb{C}\setminus (\beta\cup\mathbb{R}_+)\to\mathbb{C}$
  well-defined by Proposition~\ref{prop:wgbasicproperties}, 
with associated
  functions $\theta:\vec{\beta}\cup\vec{\gamma}\to\mathbb{C}$ and
  $\phi:\vec{\beta}\cup\vec{\gamma}\to\mathbb{C}$ defined by
  \eqref{eq:wthetaphidef}, so that the following hold:
\begin{itemize}
\item
The function $\phi$ satisfies $\phi(\xi)<0$ for $\xi\in\vec{\gamma}\cap\Sigma^\nabla$ and $\phi(\xi)>0$ for $\xi\in\vec{\gamma}\cap\Sigma^\Delta$.
Moreover, $\phi(\xi)$ is bounded away from zero for $\xi\in\vec{\gamma}$
except in neighborhoods of $w_\prec$ and $w_\succ$ (which are
endpoints of $\vec{\gamma}$).
\item The function $\theta(\xi)$ is real and nondecreasing
  (nonincreasing) with orientation for
  $\xi\in\vec{\beta}\cap\Sigma^\nabla$ (for
  $\xi\in\vec{\beta}\cap\Sigma^\Delta$).  Moreover, $\theta'(\xi)$ is
  bounded away from zero except in neighborhoods of
  $\xi=w_\prec$ and $\xi=w_\succ$ (endpoints of $\beta$)
  and the simple root $\xi=w^+$ of $H$ converging to $\xi=-1$ as
  $(x,t)\to (0,0)$.
\item The function $H(\xi)=H(\xi;\mathfrak{p}(x,t),\mathfrak{q}(x,t),x,t)$ 
is bounded away from zero for $\xi\in\beta\cup\gamma$ except in a
neighborhood of $\xi=w^+$, a point converging to $\xi=-1$ as
$(x,t)\to (0,0)$, where $H(\xi)$ has a simple zero.
\end{itemize}
\end{itemize}
\label{prop:origin}
\end{proposition}
\begin{proof}
  By the Implicit Function Theorem, it follows from
  Proposition~\ref{prop:sqrtjacobian} that the equations $\hat{M}=0$
  and $\hat{I}=0$ may be solved uniquely near $(x,t)=(0,0)$ and
  $(k_\prec,k_\succ)=(0,0)$ for differentiable functions
  $k_\prec(x,t)$ and $k_\succ(x,t)$, which are both
  easily seen to be real-valued for real
  $(x,t)\in\mathscr{O}_\rotational^0$.  Sign changes in these two functions
  correspond to sheet exchanges on the Riemann surfaces
  $\Gamma_\prec$ and $\Gamma_\succ$, so from the definition
\eqref{eq:M0I0RSdef} we confirm that we are in fact solving
$M=I=0$ in various cases of choice of $\Delta$.  The inequalities
\eqref{eq:dkdts} follow from implicit differentiation of the equations
$\hat{M}=0$ and $\hat{I}=0$ with respect to $t$ at $(x,t)=(0,0)$
and $(k_\prec,k_\succ)=(0,0)$, and arguments similar to
those used in the proof of Proposition~\ref{prop:sqrtjacobian} to
compute the limiting values of various partial derivatives.  The
inequalities \eqref{eq:dkdxs} can be obtained similarly, but also may
be understood in the context of
Proposition~\ref{prop:tneq0continuegeneral} since increasing
(decreasing) $x$ with $t=0$ leads to an overlap of
$\mathscr{O}^0_\rotational$ with $\mathscr{O}^+_\rotational$ (with $\mathscr{O}^-_\rotational$)
and in the latter we have $\Delta=\emptyset$ ($\nabla=\emptyset$)
corresponding in the present context to $k_\prec$ and $k_\succ$
both positive (both negative).

The perturbation theory of the simple root $\xi=w^+$ of $H$ near
$\xi=-1$ and the proof of existence of the contour
$\beta\cap\mathbb{C}_+$ connecting $\xi=1$ with $\xi=w^+$ along
which $\theta(\xi)$ is real and monotone
both work exactly the same as in the proof
of Proposition~\ref{prop:tneq0continuegeneral}, although we should
point out that given $(x,t)\in\mathscr{O}_\rotational^0$, the contour
$\beta\cap\mathbb{C}_+$ will generally be a different curve, and
$w^+$ a different negative real value, for different allowed
choices of $\Delta$ (this situation is relevant if and only if
$k_\prec(x,t)k_\succ(x,t)\le 0$).  
The desired monotonicity of $\theta(\xi)$ along the intervals
$(w_\prec,w^+)$ and $(w^+,w_\succ)$ also follows
by simple perturbation arguments from $t=0$.  

Of course when $k_\prec(x,t)k_\succ(x,t)=0$ one or both
intervals of $\gamma$ have collapsed to points, so it remains to show
that $\phi$, necessarily real for
$\xi\in\vec{\gamma}\subset\mathbb{R}_-$, actually satisfies the desired
inequalities in $\gamma$ when the degenerate configuration
$\gamma=\{\mathfrak{a},\mathfrak{b}\}$ present at $(x,t)=(0,0)$ is unfolded.  This
will follow from an analysis of $H(w)$ valid when $w$ is near either
$\mathfrak{a}$ or $\mathfrak{b}$ and $(x,t)$ is near $(0,0)$.  For some small $d>0$
fixed, 
\begin{equation}
H(w)=-\frac{\sigma_\prec}{\pi\sqrt{-w}}\int_{\mathfrak{a}-d}^{\mathfrak{a}}
\frac{\theta_0'(\xi)\sqrt{-\xi}\,d\xi}{R(\xi;\mathfrak{p},\mathfrak{q})(\xi-w)}
+ \bo(1),\quad w\downarrow \mathfrak{a},
\label{eq:HwnearminusM}
\end{equation}
where $\sigma_\prec=1$ if $\Delta=\emptyset$,
$\Delta=P_N^{\kink\succ}$, or $\nabla=P_N^{\prec\kink}$
and $\sigma_\prec=-1$ if $\nabla=\emptyset$,
$\nabla=P_N^{\kink\succ}$, or
$\Delta=P_N^{\prec\kink}$, and where the error term $\bo(1)$ is uniform
for $(x,t)\in\mathscr{O}_\rotational^0$.  By choosing $d$ small enough, we have
from Proposition~\ref{prop:theta0} and $E(\mathfrak{a})=-iG(0)/4$ that 
\begin{equation}
\theta_0'(\xi)
\ge \frac{1}{2}
\theta_0'(\mathfrak{a})
 >0,\quad \mathfrak{a}-d\le\xi\le \mathfrak{a},
\end{equation}
Also, since $\xi$ lies to the left of both $w_\prec$ and $w_\succ$
where $R<0$,
\begin{equation}
-\frac{\sqrt{-\xi}}{R(\xi;\mathfrak{p},\mathfrak{q})}\ge -\frac{\sqrt{-\mathfrak{a}}}{R(\xi;\mathfrak{p},\mathfrak{q})}
\ge \frac{\sqrt{-\mathfrak{a}}}{\sqrt{w_\prec-\mathfrak{a}+d}\sqrt{w_\succ-\mathfrak{a}+d}}
\ge \frac{\sqrt{-\mathfrak{a}}}{\mathfrak{b}-\mathfrak{a}+d}>0
\end{equation}
holds in the same interval $\mathfrak{a}-d\le\xi\le \mathfrak{a}$.  
So from \eqref{eq:HwnearminusM} we have
\begin{equation}
-\sigma_\prec H(w)\ge \frac{\theta_0'(\mathfrak{a})\sqrt{-\mathfrak{a}}|\log(w-\mathfrak{a})|}{2\pi (\mathfrak{b}-\mathfrak{a}+d)\sqrt{-w}}
 +\bo(1)>0,\quad
w\downarrow \mathfrak{a}.
\label{eq:Hout}
\end{equation}
Similarly, 
\begin{equation}
H(w)=-\frac{\sigma_\succ}{\pi\sqrt{-w}}
\int_{\mathfrak{b}}^{\mathfrak{b}+d}
\frac{\theta_0'(\xi)\sqrt{-\xi}\,d\xi}{R(\xi;\mathfrak{p},\mathfrak{q})(\xi-w)}+\bo(1),\quad
w\uparrow \mathfrak{b},
\label{eq:HwnearminusoneoverM}
\end{equation}
where $\sigma_\succ=1$ if $\nabla=\emptyset$,
$\nabla=P_N^{\prec\kink}$, or $\Delta=P_N^{\kink\succ}$
and $\sigma_\succ=-1$ if $\Delta=\emptyset$, 
$\Delta=P_N^{\prec\kink}$, or $\nabla=P_N^{\kink\succ}$,
and where the error term $\bo(1)$ is uniform for $(x,t)\in\mathscr{O}_\rotational^0$.
Over the interval of integration we have from Proposition~\ref{prop:theta0}
and the fact that $E(\mathfrak{b})=-iG(0)/4$ that
\begin{equation}
-\theta_0'(\xi)
\ge -\frac{1}{2}
\theta_0'(\mathfrak{b})
>0,\quad
\mathfrak{b}\le\xi\le \mathfrak{b}+d,
\end{equation}
and since now $\xi$ lies to the right of both $w_\prec$ and 
$w_\succ$ in a region where again $R<0$,
\begin{equation}
-\frac{\sqrt{-\xi}}{R(\xi;\mathfrak{p},\mathfrak{q})}\ge 
-\frac{\sqrt{-\mathfrak{b}-d}}{R(\xi;\mathfrak{p},\mathfrak{q})}\ge
\frac{\sqrt{-\mathfrak{b}-d}}{\sqrt{d+\mathfrak{b}-w_\prec}
\sqrt{d+\mathfrak{b}-w_\succ}}\ge\frac{\sqrt{-\mathfrak{b}-d}}{\mathfrak{b}-\mathfrak{a}+d}>0
\end{equation}
also holds for 
$\mathfrak{b}\le\xi\le \mathfrak{b}+d$, so
\begin{equation}
-\sigma_\succ H(w)\ge -\frac{\theta_0'(\mathfrak{b})\sqrt{-\mathfrak{b}-d}|\log(\mathfrak{b}-w)|}{2\pi (\mathfrak{b}-\mathfrak{a}+d)\sqrt{-w}}
 +
\bo(1),\quad w\uparrow \mathfrak{b}.
\label{eq:Hin}
\end{equation}
The inequality \eqref{eq:Hout} shows that $\sigma_\prec H(w)$ is
large and negative for $w$ to the right of $w=\mathfrak{a}$, while
\eqref{eq:Hin} shows that $\sigma_\succ H(w)$ is large and
negative for $w$ to the left of $w=\mathfrak{b}$, with both statements
holding uniformly for $(x,t)\in\mathscr{O}_\rotational^0$.  Therefore, since
$R(\xi;\mathfrak{p},\mathfrak{q})<0$ both for $\mathfrak{a}\le\xi<w_\prec$ and for
$w_\succ<\xi\le \mathfrak{b}$, we learn that
$\phi'(\xi)=H(\xi)R(\xi;\mathfrak{p},\mathfrak{q})$ has the same sign as
$\sigma_\prec$ for $\mathfrak{a}\le\xi<w_\prec$ and has the same
sign as $\sigma_\succ$ for $w_\succ<\xi\le \mathfrak{b}$.
Since $\phi(\xi)=0$ both when $\xi=w_\prec$ and when
$\xi=w_\succ$, we obtain the desired inequalities on $\phi$ for
$\xi\in\vec{\gamma}$.
\end{proof}

Note that via the construction of $g$ for each
$(x,t)\in\mathscr{O}_\rotational^0$ and
Proposition~\ref{prop:wgbasicproperties},
Proposition~\ref{prop:origin} effectively guarantees the existence of
a real constant $\Phi=\Phi(x,t)$ that might at first sight
appear to depend upon various artificial details of the choice of
$\Delta$.  However, according to Proposition~\ref{prop:phasevelocity},
the partial derivatives $\partial\Phi/\partial x$ and
$\partial\Phi/\partial t$ are necessarily given in terms of the
functions $u:\mathscr{O}_\rotational^0\to\mathbb{R}$ and
$v:\mathscr{O}_\rotational^0\to\mathbb{R}$, and this makes $\Phi(x,t)$
well-defined for $(x,t)\in\mathscr{O}_\rotational^0$ up to a constant.  The
constant is then determined by the fact that $\Phi(0,0)=0$ as
guaranteed by Proposition~\ref{prop:wsolntzeroproperties} by taking
limits from nonzero $x$ at $t=0$.  Therefore,
$\Phi:\mathscr{O}_\rotational^0\to\mathbb{R}$ is indeed a well-defined
differentiable function that also agrees with corresponding functions
defined in terms of $g$ in $\mathscr{O}_\rotational^\pm$ where these domains
overlap.

The functions $u:\mathscr{O}_\rotational^0\to\mathbb{R}$ and $v:\mathscr{O}_\rotational^0\to\mathbb{R}$
agree with those previously defined in the overlap region $\mathscr{O}_\rotational^0\cap(\mathscr{O}_\rotational^+\cup\mathscr{O}_\rotational^-)$, and this allows us to extend the definitions \eqref{eq:phasevelocity} of 
$n_\mathrm{p}(x,t)$ and \eqref{eq:Euv} of $\mathcal{E}(x,t)$ in a consistent way to the full union
$\mathscr{O}_\rotational^0\cup\mathscr{O}_\rotational^+\cup\mathscr{O}_\rotational^-$.  The region $S_\rotational$
involved in the formulation of Theorem~\ref{thm:rotational} is exactly this union with two curves $t=t_\pm(x)$
omitted.  
The curve $t=t_+(x)$ is obtained by solving the equation $k_\prec(x,t)=0$ for $t$, which is possible
near the origin according to the inequalities \eqref{eq:dkdts}--\eqref{eq:dkdxs}, and the signs of $t-t_+(x)$ and $k_\prec(x,t)$
are opposite.  Similarly the curve $t=t_-(x)$ is obtained from the equation
$k_\succ(x,t)=0$, and the signs of $t-t_-(x)$ and $k_\succ(x,t)$ coincide.
Therefore along the curve $t=t_+(x)$, we have $w_\prec=\mathfrak{p}-\sqrt{\mathfrak{q}}=\mathfrak{a}$, while along the curve
$t=t_-(x)$, we have $w_\succ=\mathfrak{p}+\sqrt{\mathfrak{q}}=\mathfrak{b}$.  The inequalities \eqref{eq:rotationalineq1}--\eqref{eq:rotationalineq2}
are a consequence of the inequalities $w_\prec=\mathfrak{a}+k_\prec^2>\mathfrak{a}$ and $w_\succ=\mathfrak{b}-k_\succ^2<\mathfrak{b}$ that hold for $(x,t)\in S_\rotational$ because $k_\prec(x,t)$ and $k_\succ(x,t)$ are real.

\subsection{Solution of the Whitham modulation equations}
\label{sec:Whitham}
We are now in a position to describe how the functions $\mathfrak{p}=\mathfrak{p}(x,t)$ and
$\mathfrak{q}=\mathfrak{q}(x,t)$ relate to the formally-derived Whitham modulation equations described in the introduction.
\begin{proposition}  Let $\mathfrak{p}=\mathfrak{p}(x,t)$ and $\mathfrak{q}=\mathfrak{q}(x,t)$
be the unique solutions of the equations $M=I=0$ matching the given initial data as described in
Proposition~\ref{prop:tneq0continuegeneral} (or Proposition~\ref{prop:origin} for $(x,t)$ near $(0,0)$).
Let $n_\mathrm{p}=n_\mathrm{p}(x,t)$ and  $\mathcal{E}=\mathcal{E}(x,t)$ be calculated explicitly
from $\mathfrak{p}$ and $\mathfrak{q}$ via \eqref{eq:phasevelocity} and \eqref{eq:Euv} respectively.
Then these fields satisfy the Whitham modulation equations \eqref{eq:Whithamsystem_rewrite} for
superluminal rotational waves (when $\mathfrak{p}$ and $\mathfrak{q}$ correspond to case \rotational, that is, for $(x,t)\in S_\rotational$)
and for superluminal librational waves (when $\mathfrak{p}$ and $\mathfrak{q}$ correspond to case \librational, that is, for $(x,t)\in S_\librational$).
\label{prop:Whitham}
\end{proposition}
The rest of this section is devoted to the proof of this proposition.
In fact, the roots $w_0$ and $w_1$ of the quadratic $R(w;\mathfrak{p},\mathfrak{q})^2$, when expressed in terms of $n_\mathrm{p}$ and $\mathcal{E}$, will turn out to be \emph{Riemann invariants} for the system \eqref{eq:Whithamsystem_rewrite}.  

\subsubsection{Diagonalization of the Whitham system}
We begin by writing the Whitham system \eqref{eq:Whithamsystem_rewrite} in terms of standard
complete elliptic integrals by evaluating $J(\mathcal{E})$ given by \eqref{eq:Jdefine} in the two cases of 
modulated superluminal wavetrains of librational and rotational types.  The condition of superluminality ($\omega^2>k^2$) implies that for librational wavetrains we have $J(\mathcal{E})=I_\librational(-\mathcal{E})$ and for rotational wavetrains we have $J(\mathcal{E})=I_\rotational(-\mathcal{E})$.

For superluminal rotational waves, we recall the definition \eqref{eq:Irotational} and obtain:
\begin{equation}
J(\mathcal{E})=I_\rotational(-\mathcal{E})=\frac{1}{\pi\sqrt{2}}\int_{-\pi}^\pi\sqrt{\cos(\phi)+\mathcal{E}}\,d\phi=
\frac{\sqrt{2}}{\pi}\int_0^\pi\sqrt{\cos(\phi)+\mathcal{E}}\,d\phi,
\end{equation}
for $\mathcal{E}>1$.  By the substitution $\cos(\phi)=1-2z^2$, this becomes simply
\begin{equation}
J(\mathcal{E})=\frac{4}{\pi}\sqrt{\frac{1+\mathcal{E}}{2}}E\left(\frac{2}{1+\mathcal{E}}\right)\quad
\text{for superluminal rotational waves,}
\label{eq:Jrot}
\end{equation}
where $E(m)$ is the standard complete elliptic integral of the second kind:
\begin{equation}
E(m):=\int_0^1\frac{\sqrt{1-mz^2}}{\sqrt{1-z^2}}\,dz,\quad 0<m<1.
\label{eq:elliipticEdef}
\end{equation}

For superluminal librational waves, we recall instead the definition \eqref{eq:Ilibrational} and obtain:
\begin{equation}
J(\mathcal{E})=I_\librational(-\mathcal{E}) = \frac{\sqrt{2}}{\pi}\int_{-\cos^{-1}(-\mathcal{E})}^{\cos^{-1}(-\mathcal{E})}\sqrt{\cos(\phi)+\mathcal{E}}\,d\phi = \frac{2\sqrt{2}}{\pi}\int_0^{\cos^{-1}(-\mathcal{E})}
\sqrt{\cos(\phi)+\mathcal{E}}\,d\phi,
\end{equation}
for $-1<\mathcal{E}<1$.  By the substitution $\cos(\phi)=1-(1+\mathcal{E})z^2$, this becomes simply
\begin{equation}
J(\mathcal{E}) = \frac{8}{\pi}E\left(\frac{1+\mathcal{E}}{2}\right) -\frac{4}{\pi}(1-\mathcal{E})
K\left(\frac{1+\mathcal{E}}{2}\right)\quad\text{for superluminal librational waves,}
\label{eq:Jlib}
\end{equation}
where $K(m)$ is the standard complete elliptic integral of the first kind defined for $0<m<1$ by
\eqref{eq:ellipticKdef}.

Now recall the differential identities \cite{Akhiezer}
\begin{equation}
K'(m)=\frac{E(m)-(1-m)K(m)}{2m(1-m)},\quad E'(m)=\frac{E(m)-K(m)}{2m}.
\label{eq:KprimeEprime}
\end{equation}
Since in each case $J(\mathcal{E})$ is a linear combination of $K(m)$ and $E(m)$ and $m$ is a 
function of $\mathcal{E}$, it is clear that $K(m)$ and $E(m)$ can be eliminated between $J(\mathcal{E})$, $J'(\mathcal{E})$, and $J''(\mathcal{E})$ to yield a linear second-order differential equation satisfied
by $J(\mathcal{E})$.  We shall now show that this equation is the same in both cases.
In the case of superluminal rotational waves, the above formula for $E'(m)$ yields
\begin{equation}
J'(\mathcal{E})=\frac{1}{\pi}\sqrt{\frac{2}{1+\mathcal{E}}}K\left(\frac{2}{1+\mathcal{E}}\right)
\quad\text{for superluminal rotational waves,}
\label{eq:Jprimerot}
\end{equation}
and then differentiating again, eliminating $K$ and $E$, and comparing with \eqref{eq:Jrot} we obtain
\begin{equation}
J''(\mathcal{E}) = \frac{1}{4(1-\mathcal{E}^2)}J(\mathcal{E}).
\label{eq:JODE}
\end{equation}
Similarly, in the case of superluminal librational waves, we have
\begin{equation}
J'(\mathcal{E})=\frac{2}{\pi}K\left(\frac{1+\mathcal{E}}{2}\right)\quad\text{for superluminal
librational waves,}
\label{eq:Jprimelib}
\end{equation}
so after one further differentiation we eliminate $K$ and $E$ and compare with \eqref{eq:Jlib} to arrive again at exactly the same second-order linear differential equation \eqref{eq:JODE}.

Given two complex-conjugate or real and negative variables $w_0$ and $w_1$, we express $\mathcal{E}$ and $n_\mathrm{p}$ in terms of these by the relations
\begin{equation}
\mathcal{E}=-\frac{w_0+w_1}{2\sqrt{w_0w_1}},\quad n_\mathrm{p}=\frac{1-\sqrt{w_0w_1}}{1+\sqrt{w_0w_1}}.
\label{eq:Enpw0w1}
\end{equation}
It is obvious that if $w_0$ and $w_1$ are real and negative, then $\mathcal{E}>1$, and this corresponds
to the case of modulated superluminal rotational waves.  It is also clear that if $w_0$ and $w_1$ are complex conjugates, then $-1<\mathcal{E}<1$, and this corresponds to the case of modulated superluminal librational waves.
The Jacobian  of the map $(w_0,w_1)\mapsto (n_\mathrm{p},\mathcal{E})$ is 
\begin{equation}
\mathbf{S}:=\begin{bmatrix}\partial n_\mathrm{p}/\partial w_0 & \partial n_\mathrm{p}/\partial w_1\\
\partial\mathcal{E}/\partial w_0 & \partial\mathcal{E}/\partial w_1\end{bmatrix} = 
\begin{bmatrix}\displaystyle -\frac{w_1}{(1+\sqrt{w_0w_1})^2\sqrt{w_0w_1}} &
\displaystyle -\frac{w_0}{(1+\sqrt{w_0w_1})^2\sqrt{w_0w_1}}\\
\displaystyle\frac{w_1-w_0}{4w_0\sqrt{w_0w_1}} &\displaystyle
\frac{w_0-w_1}{4w_1\sqrt{w_0w_1}}\end{bmatrix}.
\end{equation}
Therefore, if we write the Whitham system in the form
\begin{equation}
\frac{\partial}{\partial t}\begin{bmatrix}n_\mathrm{p}\\\mathcal{E}\end{bmatrix}+
\mathbf{A}(n_\mathrm{p},\mathcal{E})\frac{\partial}{\partial x}\begin{bmatrix}
n_\mathrm{p}\\\mathcal{E}\end{bmatrix}=\mathbf{0},
\end{equation}
where $\mathbf{A}(n_\mathrm{p},\mathcal{E})$ is the coefficient matrix appearing in equation \eqref{eq:Whithamsystem_rewrite} (and of course we can equally well express this explicitly
in terms of $w_0$ and $w_1$ using \eqref{eq:Enpw0w1})
then we may equivalently write this as a system of equations for $w_0$ and $w_1$ as
\begin{equation}
\frac{\partial}{\partial t}\begin{bmatrix}w_0\\w_1\end{bmatrix}+\mathbf{S}^{-1}\mathbf{A}\mathbf{S}
\frac{\partial}{\partial x}\begin{bmatrix}w_0\\w_1\end{bmatrix}=\mathbf{0}.
\end{equation}
It is now a direct calculation to show that, \emph{as a consequence of the differential equation \eqref{eq:JODE}}, $\mathbf{S}^{-1}\mathbf{A}\mathbf{S}$ is a diagonal matrix, that is, the columns of
the Jacobian $\mathbf{J}$ are independent eigenvectors of $\mathbf{A}$.  Moreover, once 
$J''(\mathcal{E})$ is eliminated using \eqref{eq:JODE}, the denominator $\mathcal{N}(n_\mathrm{p},\mathcal{E})$
defined by \eqref{eq:calNdef} factors as a difference of squares, 
and the numerator of each of the diagonal elements of $\mathbf{S}^{-1}\mathbf{A}\mathbf{S}$
is divisible by exactly one of the two factors in $\mathcal{N}(n_\mathrm{p},\mathcal{E})$.  It follows that  the Whitham
system can be written in the diagonal form
\begin{equation}
\frac{\partial w_j}{\partial t} + c_j(w_0,w_1)\frac{\partial w_j}{\partial x}=0,\quad j=0,1
\end{equation}
where the characteristic velocities (eigenvalues of $\mathbf{A}$) are given by
\begin{equation}
c_j:=\frac{\sqrt{w_0w_1}(1+\sqrt{w_0w_1})J(\mathcal{E})+(w_j-w_{1-j})(1-\sqrt{w_0w_1})J'(\mathcal{E})}{\sqrt{w_0w_1}(1-\sqrt{w_0w_1})J(\mathcal{E})+(w_j-w_{1-j})(1+\sqrt{w_0w_1})J'(\mathcal{E})}.
\label{eq:cj}
\end{equation}
In other words, the variables $w_0$ and $w_1$ are Riemann invariants for the Whitham system \eqref{eq:Whithamsystem_rewrite}.

\subsubsection{Diagonal quasilinear system satisfied by the roots of $R(w)^2$}
We will now show that if $w_0(x,t)$ and $w_1(x,t)$ are obtained by solving the moment and integral conditions $M=I=0$, then they also satisfy a system of partial differential equations in Riemann-invariant form.  Implicitly differentiating $M$ and $I$ with respect to $x$ and $t$ gives
\begin{equation}
M_{w_0}w_{0,(x,t)} +M_{w_1}w_{1,(x,t)} + M_{(x,t)}=0\quad\text{and}\quad
I_{w_0}w_{0,(x,t)}+I_{w_1}w_{1,(x,t)} + I_{(x,t)}=0.
\end{equation}
Solving for the partial derivatives of $w_j$ with respect to $x$ and $t$ assuming that the Jacobian
determinant of $(M,I)$ with respect to $(w_0,w_1)$ is nonzero, we can easily confirm the identities
\begin{equation}
\frac{\partial w_j}{\partial t}+\hat{c}_j\frac{\partial w_j}{\partial x}=0,\quad j=0,1
\end{equation}
where
\begin{equation}
\hat{c}_j:=-\frac{I_{w_{1-j}}M_t - M_{w_{1-j}}I_t}{I_{w_{1-j}}M_x-M_{w_{1-j}}I_x},\quad j=0,1.
\label{eq:lambdaj}
\end{equation}

We will now show that while the various partial derivatives appearing in \eqref{eq:lambdaj} contain
explicit dependence on $x$ and $t$, as well as dependence on initial data through the function
$\theta_0(w)$, the combinations $\hat{c}_j$ are in fact functions of $w_0$ and $w_1$ alone.
Indeed, recalling \eqref{eq:MtIt}, we have
\begin{equation}
M_t = \frac{\sqrt{w_0w_1}-1}{\sqrt{w_0w_1}}\quad\text{and}\quad
I_t = \frac{1}{8\sqrt{w_0w_1}}\left[\int_{\beta\cap\mathbb{C}_+}\frac{R_+(\xi)\,d\xi}{\xi\sqrt{-\xi}} +
\int_{\beta\cap\mathbb{C}_-}\frac{R_-(\xi)\,d\xi}{\xi\sqrt{-\xi}}\right],
\end{equation}
and by similar explicit calculations using the definitions of $M$ and $I$ (the latter in terms of the function $H$) we obtain
\begin{equation}
M_x = \frac{\sqrt{w_0w_1}+1}{\sqrt{w_0w_1}}\quad\text{and}\quad
I_x = -\frac{1}{8\sqrt{w_0w_1}}\left[\int_{\beta\cap\mathbb{C}_+}\frac{R_+(\xi)\,d\xi}{\xi\sqrt{-\xi}} +
\int_{\beta\cap\mathbb{C}_-}\frac{R_-(\xi)\,d\xi}{\xi\sqrt{-\xi}}\right].
\end{equation}
These partial derivatives are obviously functions of $w_0$ and $w_1$ (independent of $x$ and $t$, and
also not depending on the initial data).  The partial derivatives of $M$ and $I$ with respect to $w_{1-j}$
are obtained from equations \eqref{eq:M0partials} and \eqref{eq:I0partials} respectively in the statement of Proposition~\ref{prop:partials}, and while these depend also
on $x$, $t$, and $\theta_0$ via a common factor of $H(w_{1-j})$, this factor will cancel out of the 
expression for $\hat{c}_j$.  The result is that
\begin{equation}
\hat{c}_j = \frac{A-(1-\sqrt{w_0w_1})B_j}{A-(1+\sqrt{w_0w_1})B_j},
\end{equation}
where
\begin{equation}
A:=\int_{\beta\cap\mathbb{C}_+}\frac{R_+(\xi)\,d\xi}{\xi\sqrt{-\xi}} +\int_{\beta\cap\mathbb{C}_-}
\frac{R_-(\xi)\,d\xi}{\xi\sqrt{-\xi}},
\end{equation}
and
\begin{equation}
\begin{split}
B_j:=&\int_{\beta\cap\mathbb{C}_+}\frac{R_+(\xi)\,d\xi}{\sqrt{-\xi}(\xi-w_{1-j})} +\int_{\beta\cap\mathbb{C}_-}\frac{R_-(\xi)\,d\xi}{\sqrt{-\xi}(\xi-w_{1-j})}\\
=&\int_{\beta\cap\mathbb{C}_+}\frac{\xi-w_j}{R_+(\xi)\sqrt{-\xi}}\,d\xi +\int_{\beta\cap\mathbb{C}_-}
\frac{\xi-w_j}{R_-(\xi)\sqrt{-\xi}}\,d\xi.
\end{split}
\end{equation}
It is now obvious that $\hat{c}_j=\hat{c}_j(w_0,w_1)$ is an explicit function of $w_0$ and $w_1$ only with no dependence on initial data.

We may put $\hat{c}_j$ into a form that admits a comparison with the characteristic velocities $c_j$
of the Whitham system \eqref{eq:Whithamsystem_rewrite} by noting that regardless of whether the
radical $R$ is in case \rotational\ or in case \librational\, and regardless of whether in the former case
there may exist transition points,  we may integrate by parts to write $A$ in the form
\begin{equation}
A=\int_{\beta\cap\mathbb{C}_+}\frac{2\xi-w_0-w_1}{R_+(\xi)\sqrt{-\xi}}\,d\xi +
\int_{\beta\cap\mathbb{C}_-}\frac{2\xi-w_0-w_1}{R_-(\xi)\sqrt{-\xi}}\,d\xi.
\end{equation}
Therefore, some elementary algebraic manipulations show that $\hat{c}_j$ may be written in
the form
\begin{equation}
\hat{c}_j = \frac{\sqrt{w_0w_1}(1+\sqrt{w_0w_1})U + (w_j-w_{1-j})(1-\sqrt{w_0w_1})V}
{\sqrt{w_0w_1}(1-\sqrt{w_0w_1})U +(w_j-w_{1-j})(1+\sqrt{w_0w_1})V}
\label{eq:lambdaj_rewrite}
\end{equation}
where
\begin{equation}
U:=\frac{A}{\sqrt{w_0w_1}}=\frac{1}{\sqrt{w_0w_1}}\left[\int_{\beta\cap\mathbb{C}_+}
\frac{2\xi-w_0-w_1}{\sqrt{-\xi}R_+(\xi)}\,d\xi +\int_{\beta\cap\mathbb{C}_-}\frac{2\xi-w_0-w_1}{\sqrt{-\xi}R_-(\xi)}\,d\xi\right]
\end{equation}
and
\begin{equation}
V:=\int_{\beta\cap\mathbb{C}_+}\frac{d\xi}{\sqrt{-\xi}R_+(\xi)}+\int_{\beta\cap\mathbb{C}_-}
\frac{d\xi}{\sqrt{-\xi}R_-(\xi)}.
\end{equation}
To complete the proof of Proposition~\ref{prop:Whitham}, it therefore remains to show that the functions
$\hat{c}_j(w_0,w_1)$ coincide with the Whitham characteristic velocities $c_j(w_0,w_1)$, a calculation that is different in cases \rotational\ and \librational.

\subsubsection{Equivalence of $c_j$ and $\hat{c}_j$ in case \rotational}
Comparing \eqref{eq:lambdaj_rewrite} with \eqref{eq:cj} we observe that to prove that $c_j=\hat{c}_j$ holds in case \rotational\ it is sufficient to show
that there is some nonvanishing quantity $C(w_0,w_1)$ that is symmetrical in its two variables
such that the following identities hold:
\begin{equation}
CU = J(\mathcal{E})=\frac{4}{\pi}\sqrt{\frac{1+\mathcal{E}}{2}}E\left(\frac{2}{1+\mathcal{E}}\right)
\label{eq:UidentityR}
\end{equation}
and
\begin{equation}
CV =J'(\mathcal{E})=\frac{1}{\pi}\sqrt{\frac{2}{1+\mathcal{E}}}K\left(\frac{2}{1+\mathcal{E}}\right)
\label{eq:VidentityR}
\end{equation}
where $\mathcal{E}=\mathcal{E}(w_0,w_1)$ is defined by \eqref{eq:Enpw0w1}.

First, we evaluate $V$ in case \rotational.  By elementary contour deformations, we have
\begin{equation}
V=2\int_0^{w_1}\frac{d\xi}{\sqrt{-\xi}R(\xi)} = 2\int_{w_1}^0\frac{dw}{\sqrt{-w(w-w_0)(w-w_1)}},
\end{equation}
where $w_0<w_1<0$.  By the substitution \eqref{eq:wsKsubstKm} we then obtain (after some
nontrivial algebra)
\begin{equation}
V=\frac{2}{(w_0w_1)^{1/4}}\sqrt{\frac{2}{1+\mathcal{E}}}K\left(\frac{2}{1+\mathcal{E}}\right).
\end{equation}
Therefore, the identity \eqref{eq:VidentityR} holds if
\begin{equation}
C(w_0,w_1)=\frac{(w_0w_1)^{1/4}}{2\pi}.
\end{equation}
With $C$ so determined, we carry out the same contour deformations and make exactly the
same substitution \eqref{eq:wsKsubstKm} to obtain
\begin{equation}
\begin{split}
CU&=\frac{1}{\pi(w_0w_1)^{1/4}}\int_0^{w_1}\frac{2\xi-w_0-w_1}{\sqrt{-\xi}R(\xi)}\,d\xi \\
&= 
\frac{1}{\pi(w_0w_1)^{1/4}}\int_{w_1}^0\frac{2w-w_0-w_1}{\sqrt{-w(w-w_0)(w-w_1)}}\,dw\\
&=\frac{4}{\pi}\sqrt{\frac{1+\mathcal{E}}{2}}\int_0^1\frac{\sqrt{1-ms^2}-(1-s^2)}{\sqrt{1-s^2}\sqrt{1-ms^2}}
\,\frac{ds}{s^2},\quad m:=\frac{2}{1+\mathcal{E}}\in (0,1).
\end{split}
\label{eq:CU}
\end{equation}
In the final integral, the numerator of the integrand cannot be broken up without introducing a nonintegrable
singularity at $s=0$.  Now let $R_1(s)^2=1-s^2$ and $R_2(s)^2=1-ms^2$, take the branch
cut of $R_1(s)$ to be the real interval $[-1,1]$ with $R_1\sim is$ as $s\to\infty$, and take the 
branch cut of $R_2(s)$ to be the union $(-\infty,-1/\sqrt{m}]\cup[1/\sqrt{m},+\infty)$ with $R_2(0)=1$.
Then by elementary contour deformations,
\begin{equation}
\int_0^1\frac{\sqrt{1-ms^2}-(1-s^2)}{\sqrt{1-s^2}\sqrt{1-ms^2}}\,\frac{ds}{s^2} = 
\frac{1}{2}\int_{-1}^1\frac{\sqrt{1-ms^2}-(1-s^2)}{\sqrt{1-s^2}\sqrt{1-ms^2}}\,\frac{ds}{s^2} = 
\frac{1}{4}\oint_L\frac{1-s^2-R_2(s)}{R_1(s)R_2(s)}\,\frac{ds}{s^2}
\end{equation}
where $L$ is a positively-oriented loop surrounding the branch cut of $R_1(s)$ but lying in
the domain of analyticity of $R_2(s)$.  In this formulation the integral can indeed be broken into
two separate integrals as $s=0$ is no longer on the path of integration.  Therefore,
\begin{equation}
\int_0^1\frac{\sqrt{1-ms^2}-(1-s^2)}{\sqrt{1-s^2}\sqrt{1-ms^2}}\frac{ds}{s^2} = 
\frac{1}{4}\oint_L\frac{R_1(s)}{R_2(s)}\,\frac{ds}{s^2} -\frac{1}{4}\oint_L\frac{ds}{sR_1(s)}.
\end{equation}
Since $sR_1(s)$ is analytic and nonvanishing outside of $L$, and since $sR_1(s)\sim is^2$
as $s\to\infty$, the second integral vanishes identically.  By another contour deformation and 
the substitution $x=1/(s\sqrt{m})$, we therefore
obtain
\begin{equation}
\int_0^1\frac{\sqrt{1-ms^2}-(1-s^2)}{\sqrt{1-s^2}\sqrt{1-ms^2}}\frac{ds}{s^2} = 
\int_{1/\sqrt{m}}^{+\infty}\frac{\sqrt{s^2-1}}{\sqrt{ms^2-1}}\,\frac{ds}{s^2} = 
\int_0^1\frac{\sqrt{1-mx^2}}{\sqrt{1-x^2}}\,dx = E(m).
\end{equation}
Using this result in \eqref{eq:CU} confirms the identity \eqref{eq:UidentityR} and completes the
proof that $c_j=\hat{c}_j$ in case \rotational.
\subsubsection{Equivalence of $c_j$ and $\hat{c}_j$ in case \librational}
To prove that $c_j=\hat{c}_j$ in case \librational, it is sufficient to show that for some symmetrical nonvanishing function $C(w_0,w_1)$ we
have the identities
\begin{equation}
CU=J(\mathcal{E}) = \frac{8}{\pi}E\left(\frac{1+\mathcal{E}}{2}\right)-\frac{4}{\pi}(1-\mathcal{E})
K\left(\frac{1+\mathcal{E}}{2}\right)
\label{eq:UidentityL}
\end{equation}
and
\begin{equation}
CV=J'(\mathcal{E})=\frac{2}{\pi}K\left(\frac{1+\mathcal{E}}{2}\right).
\label{eq:VidentityL}
\end{equation}
To prove these identities we consider, for $p=0,1$, the integral
\begin{equation}
I_p:=\int_{\beta\cap\mathbb{C}_+}\frac{(2\xi-w_0-w_1)^p\,d\xi}{\sqrt{-\xi}R_+(\xi)} +
\int_{\beta\cap\mathbb{C}_-}\frac{(2\xi-w_0-w_1)^p\,d\xi}{\sqrt{-\xi}R_-(\xi)}  = 
\frac{1}{2i}\oint_L\frac{(2\xi-w_0-w_1)^p\,d\xi}{\sqrt{\xi}R(\xi)}
\end{equation}
where the second equality follows from a simple contour deformation, and $L$ is a positively-oriented
loop surrounding the branch cut of $R$, which we recall for case \librational\ is an arc connecting
$w_0:=re^{i\theta}$ and $w_1=re^{-i\theta}$ with $r>0$ and $0<\theta<\pi$, passing through the
real axis only at $\xi=1$.  We make the substitution
\begin{equation}
\xi = r\frac{1+z}{1-z}
\end{equation}
which maps the points $\xi=re^{\pm i\theta}$ to $z=\pm i\tan(\theta/2)$ and maps $\xi=0$ to $z=-1$
and $\xi=\infty$ to $z=1$.  Taking care with the sign of $R(\xi(z))$ we find that $I_p$ becomes
\begin{equation}
I_p=\frac{1}{i\sqrt{r}\cos(\theta/2)}\int_{-i\tan(\theta/2)}^{+i\tan(\theta/2)}\left(\frac{2r}{1-z}\right)^p
\frac{[(1-\cos(\theta))+(1+\cos(\theta))z]^p\,dz}{\sqrt{1-z^2}\sqrt{z^2+\tan(\theta/2)^2}}.
\end{equation}
In the case $p=0$, set $z=ix\tan(\theta/2)$ to obtain
\begin{equation}
I_0=\frac{1}{\sqrt{r}\cos(\theta/2)}\int_{-1}^{+1}\frac{dx}{\sqrt{1-x^2}\sqrt{1+\tan(\theta/2)^2x^2}}
=\frac{2}{\sqrt{r}\cos(\theta/2)}\int_0^1\frac{dx}{\sqrt{1-x^2}\sqrt{1+\tan(\theta/2)^2x^2}}.
\end{equation}
Then, set $t=\sqrt{1-x^2}$ to obtain
\begin{equation}
I_0=\frac{2}{\sqrt{r}} K(\sin(\theta/2)^2) = \frac{2}{(w_0w_1)^{1/4}}K\left(\frac{1+\mathcal{E}}{2}\right).
\label{eq:I0formula}
\end{equation}
To evaluate $I_1$, first multiply the numerator and denominator of the integrand by $1+z$ to 
find
\begin{equation}
\begin{split}
I_1&=\frac{2\sqrt{r}}{i\cos(\theta/2)}\int_{-i\tan(\theta/2)}^{+i\tan(\theta/2)}
\frac{(1-\cos(\theta))+2z+(1-\cos(\theta))z^2}{(1-z^2)\sqrt{1-z^2}\sqrt{z^2+\tan(\theta/2)^2}}\,dz\\
&=\frac{2\sqrt{r}}{i\cos(\theta/2)}\int_{-i\tan(\theta/2)}^{+i\tan(\theta/2)}
\frac{(1-\cos(\theta))+(1-\cos(\theta))z^2}{(1-z^2)\sqrt{1-z^2}\sqrt{z^2+\tan(\theta/2)^2}}\,dz
\end{split}
\end{equation}
where we have used even/odd parity to simplify the result.  Since $(1-\cos(\theta))+(1+\cos(\theta))z^2=
2-(1+\cos(\theta))(1-z^2)$, we get
\begin{equation}
\begin{split}
I_1&=\frac{4\sqrt{r}}{i\cos(\theta/2)}\int_{-i\tan(\theta/2)}^{+i\tan(\theta/2)}\frac{dz}{(1-z^2)^{3/2}\sqrt{z^2+\tan(\theta/2)^2}}-2r(1+\cos(\theta))I_0\\
&=\frac{8\sqrt{r}}{i\cos(\theta/2)}\int_{0}^{+i\tan(\theta/2)}\frac{dz}{(1-z^2)^{3/2}\sqrt{z^2+\tan(\theta/2)^2}}-4(w_0w_1)^{1/4}(1-\mathcal{E})K\left(\frac{1+\mathcal{E}}{2}\right).
\end{split}
\end{equation}
Again using the substitution $z=ix\tan(\theta/2)$ followed by $t=\sqrt{1-x^2}$ we arrive at
\begin{equation}
I_1=8\sqrt{r}\cos(\theta/2)^2\int_0^1\frac{dt}{\sqrt{1-t^2}(1-mt^2)^{3/2}}-4(w_0w_1)^{1/4}(1-\mathcal{E})
K\left(\frac{1+\mathcal{E}}{2}\right),\quad m:=\frac{1+\mathcal{E}}{2}.
\end{equation}
Finally, making the monotone decreasing substitution $t\mapsto s$ given by $(1-mt^2)(1-ms^2)=1-m$, we obtain
\begin{equation}
I_1=8(w_0w_1)^{1/4}E\left(\frac{1+\mathcal{E}}{2}\right)-4(w_0w_1)^{1/4}(1-\mathcal{E})K\left(\frac{1+\mathcal{E}}{2}\right).
\label{eq:I1formula}
\end{equation}
Since $U=I_1/\sqrt{w_0w_1}$ and $V=I_0$, the formulae \eqref{eq:I0formula} and \eqref{eq:I1formula} confirm the identities \eqref{eq:UidentityL}
and \eqref{eq:VidentityL} upon making the choice of $C(w_0,w_1)=(w_0w_1)^{1/4}/\pi$.  This completes the 
proof that $c_j=\hat{c}_j$ in case \librational.

\section{Use of $g(w)$}
\subsection{Opening a lens}
\label{sec:opening-a-lens}
Recall the matrix unknown $\mathbf{N}(w)$ defined by \eqref{eq:NMw},
and the matrix functions $\mathbf{L}^\nabla(w)$ and
$\mathbf{L}^\Delta(w)$ defined by \eqref{eq:Lnabla} and
\eqref{eq:LDelta} respectively.  The most crucial step in the
Deift-Zhou steepest descent method is the opening of a ``lens'' about
the contour $\beta$, which in the current context takes the form of an
explicit substitution leading to a new unknown $\mathbf{O}(w)$ as
follows.  Let the lens $\Lambda$ be an open subset of $\Omega^\circ$
containing $\vec{\beta}$ as indicated in various configurations
(depending on case \librational\ or case \rotational\, and in the latter case
whether or where a transition point exists) in
Figures~\ref{fig:BcaseOcontour}--\ref{fig:KcaseOMinusOOM}.  The precise
shape of the lens is not important, but rather that it lies close
to $\vec{\beta}$ without coinciding except at the endpoints, that
it fully abuts the real segments $\vec{\Sigma}^\nabla_{>0}$ or
$\vec{\Sigma}^\Delta_{>0}$, and that if a transition point $\tau_N$ is
present in case \rotational\ then $\Lambda$ fully abuts the contours
$\vec{\Sigma}^{\nabla\Delta}$ and $\vec{\Sigma}^{\Delta\nabla}$.

Now set
\begin{equation}
\mathbf{O}(w):=\begin{cases}
\mathbf{N}(w)\mathbf{L}^\nabla(w),\quad &w\in\Omega^\nabla\cap\Lambda\\
\mathbf{N}(w)\mathbf{L}^\Delta(w),\quad &w\in\Omega^\Delta\cap\Lambda\\
\mathbf{N}(w),\quad & w\in\mathbb{C}\setminus (\overline{\Lambda}\cup\partial\Omega\cup\mathbb{R}_+).
\end{cases}
\label{eq:wOdef}
\end{equation}
The reason for including the contours $\vec{\Sigma}^{\nabla\Delta}$
and $\vec{\Sigma}^{\Delta\nabla}$ in the interior of $\Lambda$ 
should a transition point exist is now clear:  
according to Proposition~\ref{prop:nabladelta}, $\mathbf{O}(w)$ can be
analytically continued to the contours $\vec{\Sigma}^{\nabla\Delta}$
and $\vec{\Sigma}^{\Delta\nabla}$, so $\mathbf{O}(w)$ no longer has any 
jump discontinuity across any contour emanating into the complex plane
from a transition point.  
The contour of discontinuity for $\mathbf{O}(w)$ is therefore $\Sigma_\mathbf{O}:=
(\Sigma\setminus(\Sigma^{\nabla\Delta}\cup\Sigma^{\Delta\nabla}))\cup\partial\Lambda$,
as is illustrated
with solid curves in Figures~\ref{fig:BcaseOcontour}--\ref{fig:KcaseOMinusOOM}.
\begin{figure}[h]
\begin{center}
\includegraphics{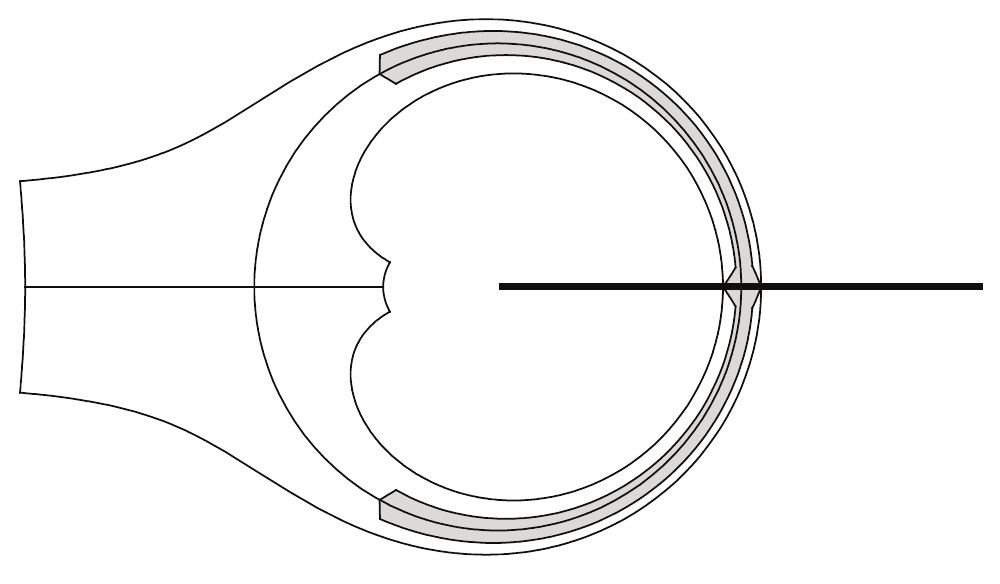}
\end{center}
\caption{\emph{The contour $\Sigma_\mathbf{O}$ of discontinuity of the sectionally analytic
function $\mathbf{O}(w)$ in case \librational\ assuming either
$\Delta=\emptyset$ or $\nabla=\emptyset$.  The lens $\Lambda$ is shaded.}}
\label{fig:BcaseOcontour}
\end{figure}
\begin{figure}[h]
\begin{center}
\includegraphics{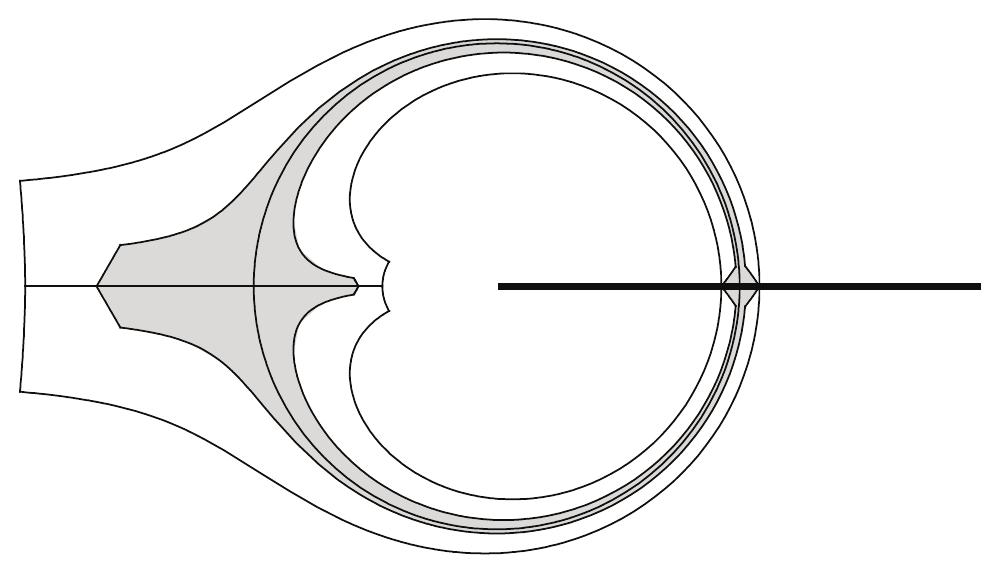}
\end{center}
\caption{\emph{The contour $\Sigma_\mathbf{O}$ of discontinuity of the sectionally analytic
function $\mathbf{O}(w)$ in case \rotational\ with either $\Delta=\emptyset$
or $\nabla=\emptyset$.  The lens $\Lambda$ is shaded.}}
\label{fig:KcaseOEmpty}
\end{figure}
\begin{figure}[h]
\begin{center}
\includegraphics{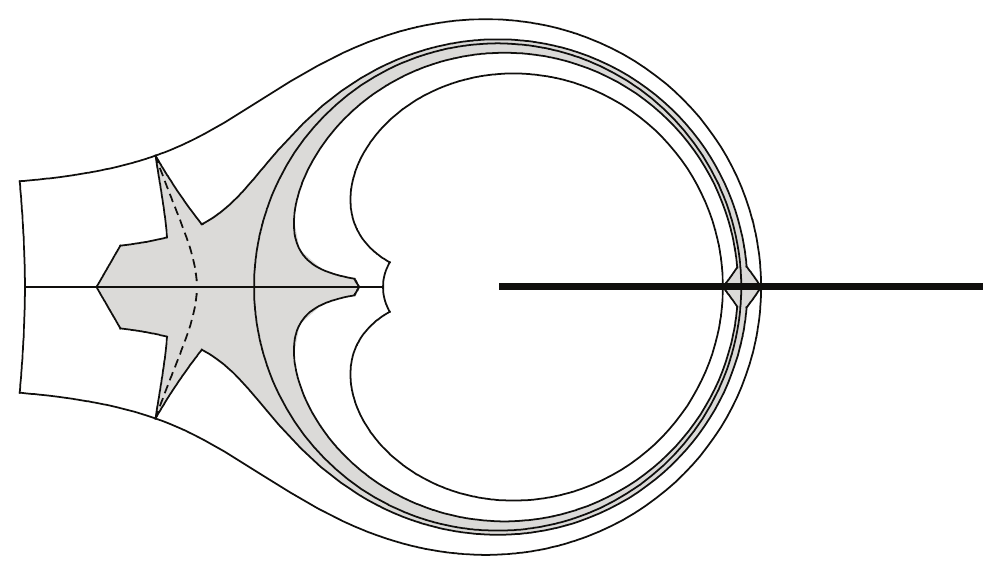}
\end{center}
\caption{\emph{The contour $\Sigma_\mathbf{O}$ of discontinuity of the sectionally
    analytic function $\mathbf{O}(w)$ in case \rotational\ with either
    $\Delta=P_N^{\prec\rotational}$ or $\nabla=P_N^{\prec\rotational}$.  
The lens $\Lambda$ is
    shaded, and the dashed curves emanating from the transition point
    $\tau_N\in\beta$ do not belong to $\Sigma_\mathbf{O}$ according to
    Proposition~\ref{prop:nabladelta}.}}
\label{fig:KcaseOMinusM}
\end{figure}
\begin{figure}[h]
\begin{center}
\includegraphics{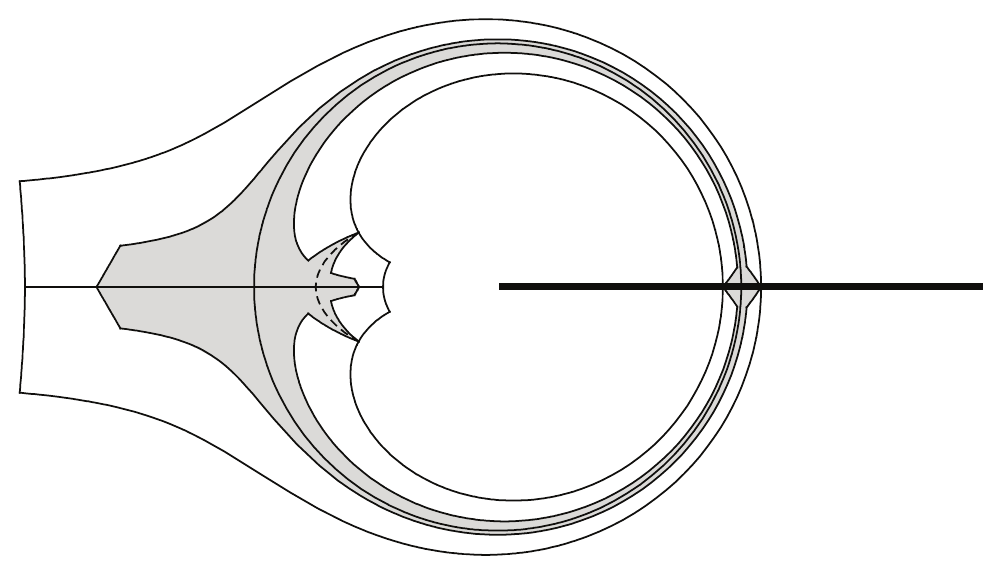}
\end{center}
\caption{\emph{The contour $\Sigma_\mathbf{O}$ of discontinuity of the sectionally
    analytic function $\mathbf{O}(w)$ in case \rotational\ with either
    $\Delta=P_N^{\rotational\succ}$ or $\nabla=P_N^{\rotational\succ}$.  The lens
    $\Lambda$ is shaded, and the dashed curves emanating from the
    transition point $\tau_N\in\beta$ do not belong to $\Sigma_\mathbf{O}$ 
    according to Proposition~\ref{prop:nabladelta}.}}
\label{fig:KcaseOMinusOOM}
\end{figure}

Recall that on the arcs $\vec{\beta}\cap\mathbb{C}_\pm$ we have
$\phi(\xi)\equiv \pm i\Phi$ whereas on $\vec{\beta}\cap\mathbb{R}$
we have $\phi(\xi)\equiv 0$. Set
\begin{equation}
\Phi_\Delta=\Phi_\Delta(x,t):=\Phi(x,t) + \pi\epsilon_N\#\Delta,
\end{equation}
where recall $\#\Delta$ is defined in \eqref{eq:cardDelta}.  
Then by direct calculation
using the definition \eqref{eq:Lnabla}, the jump condition \eqref{eq:Njumpnabla}
for $\mathbf{N}(w)$, the definitions \eqref{eq:wthetaphidef} of $\theta$
and $\phi$ in terms of boundary values of $g$, and the jump condition
\eqref{eq:Lplusminusdiff} for $L$,
 we see that $\mathbf{O}(w)$
satisfies the jump conditions
\begin{equation}  
\mathbf{O}_+(\xi)=\mathbf{O}_-(\xi)
\begin{bmatrix} 0 & -ie^{\mp i\Phi_\Delta/\epsilon_N}\\
-ie^{\pm i\Phi_\Delta/\epsilon_N} & 0\end{bmatrix},\quad\xi\in\vec{\beta}\cap\Sigma^\nabla\cap\mathbb{C}_\pm
\label{eq:Ojumpbetanablacomplex}
\end{equation}
and
\begin{equation}
\mathbf{O}_+(\xi)=\mathbf{O}_-(\xi)
\begin{bmatrix} 0 & -i\\
-i & 0\end{bmatrix},\quad\xi\in\vec{\beta}\cap\Sigma^\nabla\cap\mathbb{R}.
\label{eq:Ojumpbetanablareal}
\end{equation}
Similarly, from \eqref{eq:LDelta} and \eqref{eq:NjumpDelta} we obtain
\begin{equation}
\mathbf{O}_+(\xi)=\mathbf{O}_-(\xi)
\begin{bmatrix}0 & ie^{\mp i\Phi_\Delta/\epsilon_N}\\
ie^{\pm i\Phi_\Delta/\epsilon_N} & 0\end{bmatrix},\quad
\xi\in\vec{\beta}\cap\Sigma^\Delta\cap\mathbb{C}_\pm
\label{eq:OjumpbetaDeltacomplex}
\end{equation}
and
\begin{equation}
\mathbf{O}_+(\xi)=\mathbf{O}_-(\xi)
\begin{bmatrix}0 & i\\
i & 0\end{bmatrix},\quad
\xi\in\vec{\beta}\cap\Sigma^\Delta\cap\mathbb{R}.
\label{eq:OjumpbetaDeltareal}
\end{equation}
Therefore, the jump matrix for $\mathbf{O}(w)$ along $\vec{\beta}$
is piecewise constant with respect to $\xi$ (but has nontrivial dependence
on the parameters $x$ and $t$ via
$\Phi=\Phi(x,t)\in\mathbb{R}$).  
Moreover, we observe that
the jump matrices in \eqref{eq:Ojumpbetanablacomplex} and
\eqref{eq:Ojumpbetanablareal} are respectively inverse to those in
\eqref{eq:OjumpbetaDeltacomplex} and \eqref{eq:OjumpbetaDeltareal},
while the contours $\vec{\beta}\cap\Sigma^\nabla$
and $\vec{\beta}\cap\Sigma^\Delta$ are sub-arcs of $\vec{\beta}$ that
are oppositely oriented.  This shows that 
\eqref{eq:Ojumpbetanablacomplex}--\eqref{eq:Ojumpbetanablareal}
are describing \emph{the same} jump
conditions as are
\eqref{eq:OjumpbetaDeltacomplex}--\eqref{eq:OjumpbetaDeltareal}.  

Next, we consider the jump conditions relating the boundary values
taken by $\mathbf{O}(w)$ from the upper and lower half-planes along
the positive real axis, recalling Proposition~\ref{prop:JumpORplus}.
Clearly, from \eqref{eq:JumpNSigmagt0} and the fact that the boundary
values of $\mathbf{N}(w)$ and of $\mathbf{O}(w)$ agree for
$\xi\in\vec{\Sigma}_{>0}$, we have
\begin{equation}
\mathbf{O}_+(\xi)=\sigma_2\mathbf{O}_-(\xi)\sigma_2,\quad\xi\in\vec{\Sigma}_{>0}.
\label{eq:ORplusexactjump}
\end{equation}
To analyze the jump
conditions for $\mathbf{O}(w)$ on the contours $\Sigma^\nabla_{>0}$ or
$\Sigma^\Delta_{>0}$ (only one or the other of which is present in any
of the six choices of $\Delta$ under consideration), we need to
calculate the exponent $2iQ_+(\xi)+L_+(\xi)-2g_+(\xi)$ appearing in
\eqref{eq:OjumpSigmanablagt0exact} and
\eqref{eq:OjumpSigmadeltagt0exact}.  Of course this exponent has an
analytic continuation into $\Omega^\nabla_+$ or $\Omega^\Delta_+$ as
$2iQ(w)+L(w)-2g(w)$, and taking the boundary value on $\vec{\beta}$
near $\xi=1$ from the relevant domain we learn that
$2iQ_+(\xi)+L_+(\xi)-2g_+(\xi)$ is the analytic continuation through
the adjacent domain $\Omega_+^\nabla$ or $\Omega_+^\Delta$ of the
function $\phi(\xi)+i\theta_0(\xi)-i\theta(\xi) +
i\pi\epsilon_N\#\Delta \pmod{2\pi i\epsilon_N}$ defined on
$\vec{\beta}$.  Since $\phi(\xi)\equiv \pm i\Phi\in i\mathbb{R}$
along $\vec{\beta}$, we then have
\begin{equation}
B^\nabla(\xi)e^{\pm [2iQ_+(\xi)+L_+(\xi)-2g_+(\xi)]/\epsilon_N} =
\bo\left(\epsilon_N\frac{\lambda^2}{\epsilon_N^2}e^{-\alpha\lambda/\epsilon_N}
e^{|\Re\{i\theta(\xi)-i\theta_0(\xi)\}|/\epsilon_N}\right),\quad
\lambda=E_+(\xi)>0,\quad\xi\in\vec{\Sigma}^\nabla_{>0},
\end{equation}
and
\begin{equation}
B^\Delta(\xi)e^{\pm [2iQ_+(\xi)+L_+(\xi)-2g_+(\xi)]/\epsilon_N} =
\bo\left(\epsilon_N\frac{\lambda^2}{\epsilon_N^2}e^{-\alpha\lambda/\epsilon_N}
e^{|\Re\{i\theta(\xi)-i\theta_0(\xi)\}|/\epsilon_N}\right),\quad
\lambda=E_-(\xi)>0,\quad\xi\in\vec{\Sigma}^\Delta_{>0},
\end{equation}
where in both cases the aforementioned analytic continuation of
$\theta(\xi)-\theta_0(\xi)$ from $\vec{\beta}$ is implied.  Due to
Proposition~\ref{prop:theta0}, this difference has an analytic
continuation from $\vec{\beta}$ to a neighborhood of $\xi=1$ (note that
this will be a different analytic function depending upon which of the
two arcs of $\beta$ meeting at $\xi=1$ is involved), and 
furthermore, from the integral condition $I=0$, we will have $\Re\{i\theta(1)-i\theta_0(1)\}=0$.  
Using the formula
\eqref{eq:wthetaprimebeta}, we can easily obtain that
\begin{equation}
i\theta'(\xi)
-i\theta_0'(\xi)=
\frac{R_+(\xi;\mathfrak{p},\mathfrak{q})}{4\sqrt{-\xi}}\left(\frac{x-t}{\xi\sqrt{\mathfrak{p}^2-\mathfrak{q}}}-
\frac{4}{\pi}\int_\gamma
\frac{\theta_0'(\zeta)\sqrt{-\zeta}\,d\zeta}{R(\zeta;\mathfrak{p},\mathfrak{q})(\zeta-\xi)}\right),
\quad\xi\in\vec{\beta},
\end{equation}
so letting $\xi\to 1$ along $\beta$ yields a purely imaginary quantity
in the limit.  It follows that $\Re\{i\theta(\xi)-i\theta_0(\xi)\} =
\bo(\lambda^2)$ as $\lambda=E_\pm(\xi)\to 0$ for $\xi$ real.  Therefore,
by choosing the width parameter $\delta_1$ of the rectangles $D_\pm$ defined after the statement of
Proposition~\ref{prop:theta0} to be sufficiently small we will have
$|\Re\{i\theta(\xi)-i\theta_0(\xi)\}|\le \alpha \lambda/2$, and then from
Proposition~\ref{prop:JumpORplus} we will have that 
\begin{equation}
\mathbf{O}_+(\xi)=\sigma_2\mathbf{O}_-(\xi)\sigma_2\left(\mathbb{I}+\bo(\epsilon_N)
\right),\quad \xi\in\vec{\Sigma}^\nabla_{>0}\cup\vec{\Sigma}^\Delta_{>0}.
\label{eq:ORplusapproxjump}
\end{equation}

Now suppose that $\xi\in\vec{\gamma}$.  Then from \eqref{eq:wOdef}
we have $\mathbf{O}_\pm(\xi)=\mathbf{N}_\pm(\xi)$, so to evaluate the
jump conditions satisfied by $\mathbf{O}(w)$ we may recall
the jump conditions for $\mathbf{N}(w)$ in the form
\eqref{eq:Njumpnabla}--\eqref{eq:NjumpDeltacomplex}.
According to Proposition~\ref{prop:YTnablaDelta}, the factors $T^\nabla(\xi)$
and $T^\Delta(\xi)$ are uniformly bounded on $\gamma\cap\Sigma^\nabla$
and $\gamma\cap\Sigma^\Delta$ respectively.  Furthermore, by 
Proposition~\ref{prop:wgbasicproperties}, we have $\theta(\xi)\equiv 0$
for $\xi\in\gamma$.  Finally, Proposition~\ref{prop:tneq0continuegeneral}
or Proposition~\ref{prop:origin} guarantees that $\Re\{\phi(\xi)\}$
is strictly negative for $\xi\in\gamma\cap\Sigma^\nabla$ 
and strictly positive for $\xi\in\gamma\cap\Sigma^\Delta$ as long as
$\xi$ is bounded away from the band endpoints.  We conclude that
$\mathbf{O}_+(\xi)=\mathbf{O}_-(\xi)(\mathbb{I}+\text{exponentially small in $\epsilon_N$})$
holds uniformly for $\xi\in\vec{\gamma}$ bounded away from both band endpoints.

Consider next the jump of $\mathbf{O}(w)$ across the boundary of the lens 
$\Lambda$.  We assume that the arcs of $\partial\Lambda$ inherit orientation
from the arcs of the band $\vec{\beta}$ that they enclose.  
The matrix $\mathbf{N}(w)$ is continuous across $\partial \Lambda$,
so from \eqref{eq:wOdef}, 
\begin{equation}
\mathbf{O}_+(\xi)=\begin{cases}
\mathbf{O}_-(\xi)\mathbf{L}^\nabla(\xi)^{\mp 1},
&\quad\xi\in\partial\Lambda\cap\Omega_\pm^\nabla\\
\mathbf{O}_-(\xi)\mathbf{L}^\Delta(\xi)^{\mp 1},
&\quad\xi\in\partial\Lambda\cap\Omega_\pm^\Delta.
\end{cases}
\end{equation}
Referring to the definitions \eqref{eq:Lnabla}--\eqref{eq:LDelta}, 
we see that these may be written in the form
\begin{equation}
\begin{split}
\mathbf{O}_+(\xi)&=\mathbf{O}_-(\xi)
\left(\mathbb{I} + \bo(T^\nabla(\xi)^{1/2}-1) +
\bo(T^\nabla(\xi)^{-1/2}-1) + \bo(e^{-[2iQ(\xi)+L(\xi)\mp i\theta_0(\xi)
-2g(\xi)]/\epsilon_N})\right),\\
&\qquad\qquad\qquad \xi\in\partial\Lambda\cap
\Omega_\pm^\nabla,
\end{split}
\label{eq:lensnablaawayI}
\end{equation}
and 
\begin{equation}
\begin{split}
\mathbf{O}_+(\xi)&=\mathbf{O}_-(\xi)
\left(\mathbb{I}+\bo(T^\Delta(\xi)^{1/2}-1)+
\bo(T^\Delta(\xi)^{-1/2}-1) +\bo(e^{[2iQ(\xi)+L(\xi)\mp i\theta_0(\xi)-2g(\xi)]/\epsilon_N})\right),\\
&\qquad\qquad\qquad\xi\in\partial\Lambda\cap\Omega_\pm^\Delta,
\end{split}
\label{eq:lensDeltaawayI}
\end{equation}
assuming that in each case the three error terms on the right-hand
side are bounded.  If as $\epsilon_N\downarrow 0$ the point $\xi$
remains bounded away from the singular points $\mathfrak{a}$ and $\mathfrak{b}$ of
$T^\nabla(\cdot)$ and $T^\Delta(\cdot)$, then according to
Proposition~\ref{prop:YTnablaDelta} the first two error terms in each
case are $\bo(\epsilon_N)$.  The exponent $2iQ(\xi)+L(\xi)\mp
i\theta_0(\xi)-2g(\xi)$ is, modulo $i\pi\epsilon_N$, 
the analytic continuation from $\beta$
to $\Omega_\pm^\nabla$ of $\phi(\xi)\mp i\theta(\xi)$.  Since
according to Proposition~\ref{prop:wgbasicproperties} $\phi(\xi)$ is
an imaginary constant in $\beta$, and since according to
Proposition~\ref{prop:tneq0continuegeneral} or \ref{prop:origin}
$\theta(\xi)$ is analytic, real, and increasing along parts of
$\vec{\beta}$ in $\Sigma^\nabla$ with derivative bounded away from
zero away from the band endpoints, it follows from the Cauchy-Riemann
equations that $\Re\{2iQ(\xi)+L(\xi)\mp i\theta_0(\xi)-2g(\xi)\}$
is strictly positive for $\xi\in\partial\Lambda\cap\Omega_\pm^\nabla$
bounded away from the band endpoints.  Therefore the final error term
in \eqref{eq:lensnablaawayI} is exponentially small as $\epsilon_N\downarrow 0$
uniformly for $\xi$ bounded away from band endpoints.  Completely analogous
reasoning yields the same conclusion for the final error term in 
\eqref{eq:lensDeltaawayI}.  We conclude that 
$\mathbf{O}_+(\xi)=\mathbf{O}_-(\xi)(\mathbb{I}+\bo(\epsilon_N))$ holds 
for $\xi\in\partial\Lambda$ as long as $\xi$ is bounded away from the two
band endpoints (which, given the shape of the lens $\Lambda$, also implies
that $\xi$ is bounded away from $\mathfrak{a}$ and $\mathfrak{b}$).

Finally let us consider the jump conditions satisfied by $\mathbf{O}(w)$
across the contours $\Sigma^\nabla_\pm$ and $\Sigma^\Delta_\pm$, that is the
contours that make up the boundary of the whole region $\overline{\Omega}$.
For $\xi$ on any of these curves we have 
$\mathbf{O}_\pm(\xi)=\mathbf{N}_\pm(\xi)$.  Therefore, to calculate the
jump matrices on these curves we need to combine the definition \eqref{eq:NMw}
of $\mathbf{N}(w)$ in terms of $\mathbf{M}(w)$ with the jump conditions
\eqref{eq:MjumpSigmanablapm}--\eqref{eq:MjumpSigmaDeltapm}.  Since
$g(w)$ is analytic on $\Sigma^\nabla_\pm$ and $\Sigma^\Delta_\pm$, we obtain
the jump conditions
\begin{equation}
\mathbf{O}_+(\xi)=\mathbf{O}_-(\xi)\begin{bmatrix}
1 & 0\\
-iY(\xi)e^{[2iQ(\xi)+L(\xi)\pm i\theta_0(\xi)-2g(\xi)]/\epsilon_N} & 1
\end{bmatrix},\quad \xi\in\vec{\Sigma}_\pm^\nabla,
\label{eq:OjumpSigmanablapm}
\end{equation}
\begin{equation}
\mathbf{O}_+(\xi)=\mathbf{O}_-(\xi)\begin{bmatrix}
1 & -iY(\xi)^{-1}e^{-[2iQ(\xi)+L(\xi)\mp i\theta_0(\xi)-2g(\xi)]/\epsilon_N}\\
0 & 1\end{bmatrix},\quad \xi\in\vec{\Sigma}_\pm^\Delta.
\label{eq:OjumpSigmaDeltapm}
\end{equation}
The exponent $2iQ(\xi)+L(\xi)\pm i\theta_0(\xi)-2g(\xi)$ occurring
in \eqref{eq:OjumpSigmanablapm} is, modulo $i\pi\epsilon_N$, the analytic
continuation into the domain $\Omega^\nabla_\pm$ from the contour 
$\Sigma^\nabla$ of the function
$\phi(\xi)\mp i\theta(\xi)\pm 2i\theta_0(\xi)$.  If $t=0$, then $\phi(\xi)$
and $\theta_0(\xi)$ are real, and $\phi(\xi)\le 0$, for $\xi\in\Sigma^\nabla$.
According to Proposition~\ref{prop:wsolntzeroproperties}, 
$\theta_0(\xi)-\theta(\xi)$ is real and nondecreasing
for $\xi\in\vec{\Sigma}^\nabla$.  Also, by Proposition~\ref{prop:theta0},
$\theta_0(\xi)$ is strictly increasing with derivative bounded below
by a positive constant for $\xi\in\vec{\Sigma}^\nabla$, and hence the same
is true of $2\theta_0(\xi)-\theta(\xi)$.  It follows by a Cauchy-Riemann
argument that if the width $2\delta_1$ of the 
rectangle $\overline{D_+\cup D_-}$ that is the closure of the image of $\Omega$ under $E(\cdot)$ is sufficiently small,
then when $t=0$, 
$\Re\{2iQ(\xi)+L(\xi)\pm i\theta_0(\xi)-2g(\xi)\}$ will be strictly
negative on all parts of $\Sigma_\pm^\nabla$ with the possible exception of
points near $\mathfrak{a}$ and $\mathfrak{b}$ where the contours 
$\Sigma_\pm^\nabla$ meet $\Sigma^\nabla$.
But if the band endpoints are bounded away from $\mathfrak{a}$ and $\mathfrak{b}$, that is
(according to Proposition~\ref{prop:utzero}) 
if $x$ is bounded away from zero, then $\phi(\mathfrak{a})$ and $\phi(\mathfrak{b})$ will be
strictly negative by Proposition~\ref{prop:wsolntzeroproperties}, so the 
strict inequality on the exponent in \eqref{eq:OjumpSigmanablapm} 
persists right down to the real axis when $t=0$.  Since for $x$ away from $x=0$
the inequality is strict uniformly on $\Sigma_\pm^\nabla$, it also persists
uniformly on $\Sigma_\pm^\nabla$ 
for sufficiently small $t\neq 0$.  (Assuming $|t|$ sufficiently small also
ensures that the nonreal arcs of $\Sigma^\nabla$ or $\Sigma^\Delta$ are contained
within $\Omega$.)  Completely analogous arguments also show that the exponential
factor occurring in \eqref{eq:OjumpSigmaDeltapm} is exponentially small
uniformly on $\Sigma_\pm^\Delta$ if $x$ is bounded away from zero and 
$t$ is small enough.  If $x$ approaches zero, then the exponential decay
only fails near $\xi=\mathfrak{a}$ and $\xi=\mathfrak{b}$. Now, the factors $Y(\xi)$ and 
$Y(\xi)^{-1}$ are uniformly bounded on $\Sigma_\pm^\nabla$ and $\Sigma_\pm^\Delta$
respectively, according to Proposition~\ref{prop:YTnablaDelta}.
We conclude that for $t$ sufficiently small, 
$\mathbf{O}_+(\xi)=\mathbf{O}_-(\xi)(\mathbb{I}+\text{exponentially small})$ holds uniformly for $\xi\in\Sigma_\pm^\nabla$ and $\xi\in\Sigma_\pm^\Delta$ except near $\mathfrak{a}$ and $\mathfrak{b}$ when $x$ is also small.

We formalize these observations concerning the jump conditions satisfied by
$\mathbf{O}(w)$ in the following proposition.
\begin{proposition}
Suppose that the point $(x,t)$ lies 
in one of the domains $\mathscr{O}_\librational^\pm$,
$\mathscr{O}_\rotational^\pm$ (see Proposition~\ref{prop:tneq0continuegeneral}),
or $\mathscr{O}_\rotational^0$ (see Proposition~\ref{prop:origin}),
and $|t|$ is sufficiently small.  Let $U_0$  and $U_1$ 
be discs of small radius independent
of $x$, $t$, and $\epsilon_N$ centered at the band endpoints $w_0(x,t)$ and $w_1(x,t)$ (see Figures 
\ref{fig:BcaseEcontour}--\ref{fig:KcaseEMinusOOM}).
Then for $\xi\in\vec{\beta}$, $\mathbf{O}(w)$ satisfies exactly the
piecewise constant jump conditions \eqref{eq:Ojumpbetanablacomplex}--\eqref{eq:OjumpbetaDeltareal}.
For $\xi\in\mathbb{R}_+$, $\mathbf{O}(w)$ satisfies exactly the simple jump
condition \eqref{eq:ORplusexactjump} except in a small interval near $\xi=1$
where the $\bo(\epsilon_N)$ approximate version
\eqref{eq:ORplusapproxjump} of this condition holds.  Finally
uniformly for $\xi\in\Sigma_\mathbf{O}\setminus (\beta\cup\mathbb{R}_+\cup U_0\cup U_1)$,
we have simply 
$\mathbf{O}_+(\xi)=\mathbf{O}_-(\xi)(\mathbb{I}+\bo(\epsilon_N))$.
\label{prop:OjumpsAway}
\end{proposition}

\subsection{Construction of a global parametrix}
\label{sec:global-parametrix}
A \emph{global parametrix} for $\mathbf{O}(w)$ is a
sectionally-analytic matrix function $\dot{\mathbf{O}}(w)$ designed to
satisfy the jump conditions of $\mathbf{O}(w)$ for $\xi\in\Sigma_\mathbf{O}\cap
(\beta\cup\mathbb{R}_+\cup U_0\cup U_1)$, with the only modification
being that we use the jump condition \eqref{eq:ORplusexactjump} on all
of $\mathbb{R}_+$ rather than omitting a small interval near $\xi=1$
where $\mathbf{O}(w)$ satisfies the approximate relation
\eqref{eq:ORplusapproxjump}.  The global parametrix
$\dot{\mathbf{O}}(w)$ will be analytic (have no jump discontinuity)
for $\xi\in\Sigma_\mathbf{O}\setminus(\beta\cup\mathbb{R}_+\cup U_0\cup U_1)$,
that is, on the arcs of the jump contour $\Sigma_\mathbf{O}$ where
Proposition~\ref{prop:OjumpsAway} guarantees that
$\mathbf{O}_+(\xi)=\mathbf{O}_-(\xi) (\mathbb{I}+\bo(\epsilon_N))$.
It is standard to construct $\dot{\mathbf{O}}(w)$ by patching together
(i) an \emph{outer parametrix} denoted $\dot{\mathbf{O}}^\mathrm{out}(w)$
expected to be a valid approximation of $\mathbf{O}(w)$ 
away from the two roots of $R(w;\mathfrak{p},\mathfrak{q})^2$ and (ii) two \emph{inner parametrices}
denoted $\dot{\mathbf{O}}^\mathrm{in}_0(w)$ and 
$\dot{\mathbf{O}}^\mathrm{in}_1(w)$ expected to be valid approximations of
$\mathbf{O}(w)$ for $w\in U_0$ and $w\in U_1$ respectively:
\begin{equation}
\dot{\mathbf{O}}(w):=\begin{cases}
\dot{\mathbf{O}}_k^\mathrm{in}(w),\quad &
w\in U_k,\quad k=0,1,\\
\dot{\mathbf{O}}^\mathrm{out}(w),\quad & \text{otherwise}.
\end{cases}
\label{eq:hatOdef}
\end{equation}
It is also an important part of the construction of $\dot{\mathbf{O}}(w)$
that the outer and inner parametrices are well-matched on the boundaries
$\partial U_k$ of the discs.

Our construction of $\dot{\mathbf{O}}(w)$ in this section rests essentially upon the following two facts
that hold true for $(x,t)\in S_\librational\cup S_\rotational$:
\begin{itemize}
\item The roots of $R(w;\mathfrak{p},\mathfrak{q})^2$ are distinct.  Error terms will become uncontrolled if the roots
are allowed to coalesce, which occurs only on the common boundary of $S_\librational$ and
$S_\rotational$ consisting of the two points $x=\pm x_\mathrm{crit}$.  In a forthcoming paper \cite{BuckinghamMseparatrix} we will give a complete analysis
of the interesting dynamics that occur for $x\approx \pm x_\mathrm{crit}$ and $|t|$ small.
\item No root of $R(w;\mathfrak{p},\mathfrak{q})^2$ coincides with $\mathfrak{a}$ or $\mathfrak{b}$.  Again, error estimates will fail
if one of the roots approaches either $\mathfrak{a}$ or $\mathfrak{b}$, as occurs along the two curves $t=t_\pm(x)$
excluded from $S_\rotational$ along which the choice of $\Delta$ must be changed as explained in
Proposition~\ref{prop:origin}.
Unlike the coalescence of two roots of $R(w;\mathfrak{p},\mathfrak{q})^2$, which as mentioned above leads to interesting
new dynamics, the exclusion of the curves $t=t_\pm(x)$ appears to us to be serving a merely technical purpose, allowing us to avoid more complicated inner parametrices.  As mentioned in \S\ref{sec:results} there
appears to be no exceptional behavior near the omitted curves in the $(x,t)$-plane.
\end{itemize}

The outer parametrix must be considered separately 
in the two cases (\librational\ and \rotational)
because the contour geometry is qualitatively different.  
Taking into account just the jump conditions for $\mathbf{O}(w)$
on the band $\beta$ and the positive real axis $\mathbb{R}_+$, the
outer parametrix $\dot{\mathbf{O}}^\mathrm{out}(w)$ is 
to be determined as a solution of one or the other of the following
two Riemann-Hilbert problems.  The outer parametrix will depend parametrically
on the fast phase $\nu:=\Phi_\Delta/\epsilon_N$ as well as on the geometry of
the contour $\beta$, although both of these dependencies are suppressed in our notation.

\begin{rhp}[Outer parametrix, case \librational]
Let $\nu$ be a real number, and let
$w_0\in\mathbb{C}_+$. Let $\vec{\beta}_+$ denote an oriented arc in
$\mathbb{C}_+$ from $w=1$ to $w=w_0$ and let $\vec{\beta}_-$ denote
the complex-conjugated arc from $w=1$ to $w=w_1=w_0^*$. Find a matrix
$\dot{\mathbf{O}}^\mathrm{out}(w)$ with the following properties.
\begin{itemize}
\item[]\textbf{Analyticity:} $\dot{\mathbf{O}}^\mathrm{out}(w)$ is an analytic
  function of $w$ for $w\in\mathbb{C}\setminus 
(\beta_+\cup\beta_-\cup\mathbb{R}_+)$ and  H\"older-$\gamma$
  continuous for any $\gamma\le 1$ with the exception of arbitrarily
  small neighborhoods of the points $w=w_0$ and  $w=w_1=w_0^*$.
  In the neighborhood $U_k$ of $w_k$,
  the elements of $\dot{\mathbf{O}}^\mathrm{out}(w)$ are bounded by a multiple
  of $|w-w_k|^{-1/4}$.
\item[]\textbf{Jump condition:} The boundary values taken by
  $\dot{\mathbf{O}}^\mathrm{out}(w)$ along $\vec{\beta}_\pm$ and
  $\vec{\mathbb{R}}_+$ (the latter oriented from the origin to $+\infty$)
  satisfy the following jump conditions:
\begin{equation}
\dot{\mathbf{O}}^\mathrm{out}_+(\xi)=\dot{\mathbf{O}}^\mathrm{out}_-(\xi)
i\sigma_1e^{\pm i\nu\sigma_3},\quad \xi\in\vec{\beta}_\pm,
\end{equation}
and
\begin{equation}
\dot{\mathbf{O}}^\mathrm{out}_+(\xi)=\sigma_2\dot{\mathbf{O}}^\mathrm{out}_-(\xi)
\sigma_2,\quad\xi\in\vec{\mathbb{R}}_+.
\end{equation}
\item[]\textbf{Normalization:}  The following normalization condition holds:
\begin{equation}
\lim_{w\to\infty}\dot{\mathbf{O}}^\mathrm{out}(w)=\mathbb{I}.
\end{equation}
\end{itemize}
\label{rhp:wOdotlibrational}
\end{rhp}
The jump conditions satisfied by $\dot{\mathbf{O}}^\mathrm{out}(w)$ in
case \librational\ are summarized in Figure~\ref{fig:wOdotB}.
\begin{figure}[h]
\begin{center}
\includegraphics{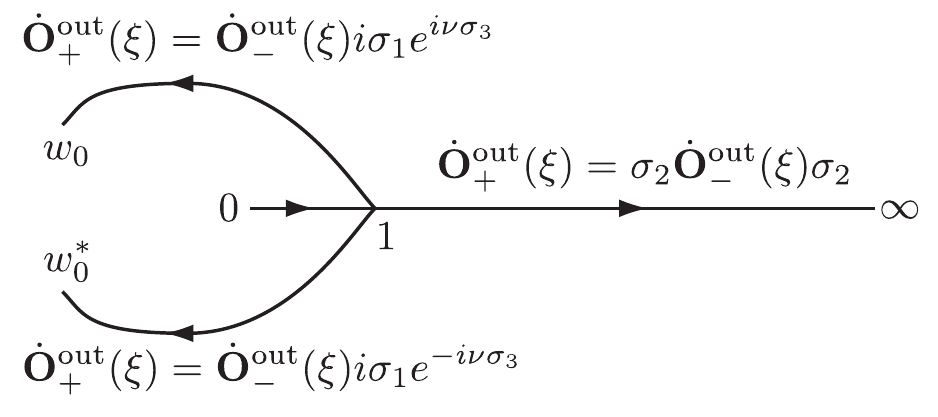}
\end{center}
\caption{\emph{The jump conditions satisfied by the matrix 
$\dot{\mathbf{O}}^\mathrm{out}(w)$
that we will show is normalized in the stronger sense that
$\dot{\mathbf{O}}^\mathrm{out}(w)=\mathbb{I}+\bo(|w|^{-1/2})$ as $w\to\infty$.
The only other singularities admitted are $\bo(|w-w_k|^{-1/4})$ near the
points $w_0$ and $w_0^*$.}}
\label{fig:wOdotB}
\end{figure}

\begin{rhp}[Outer parametrix, case \rotational]
  Let $\nu$ be a real number, 
and let $w_0<w^+<w_1<0$ 
  be given.  Let $\vec{\beta}_+$ denote an oriented arc in $\mathbb{C}_+$
from $w=1$ to $w=w^+$, let $\vec{\beta}_-$ denote the 
complex-conjugated arc  in $\mathbb{C}_-$ from $w=1$ to $w=w^+$,
and let $\vec{\beta}_\prec$ and $\vec{\beta}_\succ$ denote
real arcs oriented from $w=w^+$ to $w=w_0$ and $w=w_1$ respectively.  
Find a matrix 
$\dot{\mathbf{O}}^\mathrm{out}(w)$ with the following properties.
\begin{itemize}
\item[]\textbf{Analyticity:} $\dot{\mathbf{O}}^\mathrm{out}(w)$ is an analytic
  function of $w$ for $w\in\mathbb{C}\setminus (\beta_+\cup
\beta_-\cup\beta_\prec\cup\beta_\succ\cup\mathbb{R}_+)$, 
H\"older-$\gamma$
  continuous for any $\gamma\le 1$ with the exception of arbitrarily
  small neighborhoods of the points $w=w_0$ and $w=w_1$.
  In the neighborhood $U_k$ of $w_k$, the elements
  of $\dot{\mathbf{O}}^\mathrm{out}(w)$ are bounded by a multiple
of $|w-w_k|^{-1/4}$.
\item[]\textbf{Jump condition:}  
The boundary values taken by $\dot{\mathbf{O}}^\mathrm{out}(w)$
 along $\vec{\beta}_\pm$, $\vec{\beta}_\prec$, $\vec{\beta}_\succ$, and
$\vec{\mathbb{R}}_+$ (the latter oriented from the origin to $+\infty$) satisfy
the following jump conditions:
\begin{equation}
\dot{\mathbf{O}}^\mathrm{out}_+(\xi)=\dot{\mathbf{O}}^\mathrm{out}_-(\xi)
i\sigma_1e^{\pm i\nu\sigma_3},\quad\xi\in\vec{\beta}_\pm,
\end{equation}
\begin{equation}
\dot{\mathbf{O}}^\mathrm{out}_+(\xi)=\dot{\mathbf{O}}^\mathrm{out}_-(\xi)
i\sigma_1,\quad \xi\in \vec{\beta}_\prec\cup\vec{\beta}_\succ,
\end{equation}
and
\begin{equation}
\dot{\mathbf{O}}^\mathrm{out}_+(\xi)=\sigma_2\dot{\mathbf{O}}^\mathrm{out}_-(\xi)
\sigma_2,\quad\xi\in\vec{\mathbb{R}}_+.
\end{equation}
\item[]\textbf{Normalization:}  The following normalization condition holds:
\begin{equation}
\lim_{w\to\infty}\dot{\mathbf{O}}^\mathrm{out}(w)=\mathbb{I}.
\end{equation}
\end{itemize}
\label{rhp:wOdotrotational}
\end{rhp}
The jump conditions satisfied by $\dot{\mathbf{O}}^\mathrm{out}(w)$ in case \rotational\ are
summarized in Figure~\ref{fig:wOdotK}.
\begin{figure}[h]
\begin{center}
\includegraphics{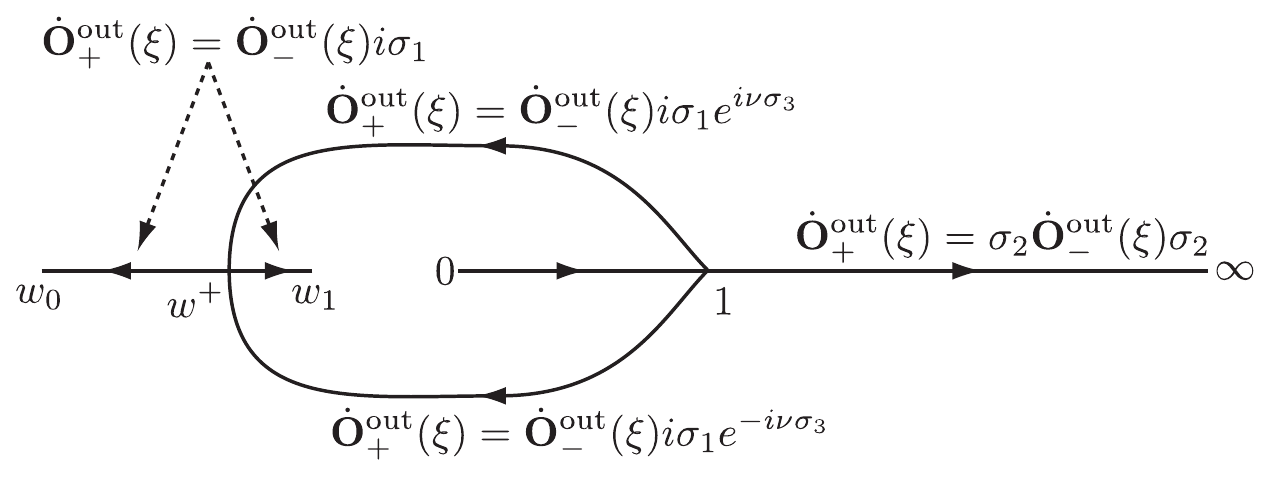}
\end{center}
\caption{\emph{The jump conditions satisfied by the matrix $\dot{\mathbf{O}}^\mathrm{out}(w)$
that we will show is normalized in the stronger sense that
$\dot{\mathbf{O}}^\mathrm{out}(w)=\mathbb{I}+\bo(|w|^{-1/2})$ as $w\to\infty$.
The only other singularities admitted are $\bo(|w-w_k|^{-1/4})$ near the
points $w_0$ and $w_0^*$.}}
\label{fig:wOdotK}
\end{figure}

A standard Liouville argument shows that these two Riemann-Hilbert problems
have at most one solution.  Existence of a unique solution may be accomplished
by explicit construction involving Riemann $\Theta$-functions of
genus one.  For us to be able
to continue the current line of argument it is sufficient to state the following
proposition, whose proof can be found in Appendix~\ref{app:outer}.  Recall that in using the outer
parametrix we will be setting $\nu = \Phi_\Delta/\epsilon_N$.
\begin{proposition}
Riemann-Hilbert Problems~\ref{rhp:wOdotlibrational} and \ref{rhp:wOdotrotational} 
each have a unique solution with the following properties:
\begin{itemize}
\item
$\dot{\mathbf{O}}^\mathrm{out}(w)$ depends continuously on $\nu$ and remains uniformly bounded (despite having no limit) as $\nu\to\infty$
for $w\in\mathbb{C}\setminus (U_0\cup U_1)$.
\item
$|w-w_k|^{1/4}\dot{\mathbf{O}}^\mathrm{out}(w)$ is uniformly bounded as
$\nu\to\infty$ for $w\in U_k$, $k=0,1$.
\item
For all $w$ where $\dot{\mathbf{O}}^\mathrm{out}(w)$ is defined, $\det(\dot{\mathbf{O}}^\mathrm{out}(w))=1$.
\end{itemize}
Whether it solves Riemann-Hilbert Problem~\ref{rhp:wOdotlibrational}
(case \librational) or Riemann-Hilbert Problem~\ref{rhp:wOdotrotational} (case
\rotational), the matrix $\dot{\mathbf{O}}^\mathrm{out}(w)$
has the following asymptotic forms: 
\begin{equation}
\dot{\mathbf{O}}^\mathrm{out}(w)=\dot{\mathbf{O}}^{0,0}+\dot{\mathbf{O}}^{0,1}
\sqrt{-w} + \bo(w),\quad w\to 0
\end{equation}
and
\begin{equation}
\dot{\mathbf{O}}^\mathrm{out}(w)=\mathbb{I} + \frac{\dot{\mathbf{O}}^{\infty,1}}
{\sqrt{-w}} + \bo(w^{-1}),\quad w\to\infty,
\end{equation}
and the matrix coefficients $\dot{\mathbf{O}}^{0,0}$, $\dot{\mathbf{O}}^{0,1}$, and $\dot{\mathbf{O}}^{\infty,1}$ are all uniformly bounded
as $\nu\to\infty$.
\label{prop:outer}
\end{proposition}

In terms of the matrix elements of the coefficients 
$\dot{\mathbf{O}}^{0,0}$, $\dot{\mathbf{O}}^{0,1}$,
and $\dot{\mathbf{O}}^{\infty,1}$ we now define the following quantities:
\begin{equation}
\dot{C}:=(-1)^{\#\Delta}\dot{O}^{0,0}_{11},
\label{eq:dotClibrational}
\end{equation}
\begin{equation}
\dot{S}:=(-1)^{\#\Delta}\dot{O}^{0,0}_{21},
\label{eq:dotSlibrational}
\end{equation}
and
\begin{equation}
\begin{split}
\dot{G}:=&\dot{O}_{12}^{\infty,1} + \left[(\dot{\mathbf{O}}^{0,0})^{-1}\dot{\mathbf{O}}^{0,1}\right]_{12}\\
{}=& \dot{O}^{\infty,1}_{12} + \dot{O}_{22}^{0,0}\dot{O}_{12}^{0,1} - \dot{O}_{12}^{0,0}\dot{O}_{22}^{0,1},
\end{split}
\label{eq:dotvlibrational}
\end{equation}
where the second line follows from the first because according to
Proposition~\ref{prop:outer},
$\det(\dot{\mathbf{O}}^{0,0})=\det(\dot{\mathbf{O}}^\mathrm{out}(0))=1$.  

In Appendix~\ref{app:outer} the following simple formulae for $\dot{C}$,
$\dot{S}$, and $\dot{G}$ are established.
\begin{proposition}
Let $\nu=\Phi_\Delta/\epsilon_N = \Phi(x,t)/\epsilon_N +\pi\#\Delta$, and let
the contour $\beta$ depend on $(x,t)\in\mathbb{R}^2$ as described in
Proposition~\ref{prop:tneq0continuegeneral} or Proposition~\ref{prop:origin}.
If $\dot{\mathbf{O}}^\mathrm{out}(w)$ is the solution of Riemann-Hilbert 
Problem~\ref{rhp:wOdotlibrational} (case \librational), the derived quantities
defined by \eqref{eq:dotClibrational}--\eqref{eq:dotvlibrational} are given
by
\begin{equation}
\begin{split}
\dot{C}=\dot{C}_N(x,t)&=\mathrm{dn}\left(\frac{2\Phi K(m)}{\pi\epsilon_N};m\right),\\
\dot{S}=\dot{S}_N(x,t)&=-\sqrt{m}\mathrm{sn}\left(\frac{2\Phi K(m)}{\pi\epsilon_N};m
\right),\\
\dot{G}=\dot{G}_N(x,t)&=-\frac{4K(m)}{\pi}\frac{\partial\Phi}{\partial t}
\sqrt{m}\mathrm{cn}\left(\frac{2\Phi K(m)}{\pi\epsilon_N};m\right),
\end{split}
\label{eq:CNSNGNLibrational}
\end{equation}
where the elliptic parameter is 
\begin{equation}
m=m_\librational :=\sin(\zeta)^2,\quad\quad0<\zeta:=\frac{1}{2}\arg(w_0)<\frac{\pi}{2},
\end{equation}
which coincides with the function of $\mathcal{E}$ given by
\eqref{eq:mlibrationalE}.
On the other hand, if $\dot{\mathbf{O}}^\mathrm{out}(w)$ is the
solution of Riemann-Hilbert Problem~\ref{rhp:wOdotrotational}, the derived quantities
defined by \eqref{eq:dotClibrational}--\eqref{eq:dotvlibrational} are given by
\begin{equation}
\begin{split}
\dot{C}=\dot{C}_N(x,t)&=\mathrm{cn}\left(\frac{2\Phi K(m)}{\pi\epsilon_N};m\right),\\
\dot{S}=\dot{S}_N(x,t)&=-\mathrm{sn}\left(\frac{2\Phi K(m)}{\pi\epsilon_N};m\right),\\
\dot{G}=\dot{G}_N(x,t)&=-\frac{4K(m)}{\pi}\frac{\partial\Phi}{\partial t}
\mathrm{dn}\left(\frac{2\Phi K(m)}{\pi\epsilon_N};m\right),
\end{split}
\label{eq:CNSNGNRotational}
\end{equation}
where now 
\begin{equation}
m=m_\rotational :=\frac{4\sqrt{w_0w_1}}{(\sqrt{-w_0}+\sqrt{-w_1})^2},
\end{equation}
which coincides with the function of $\mathcal{E}$ given by \eqref{eq:mrotationalE}.
Here $K(\cdot)$ is the complete elliptic integral of the first kind as defined by \eqref{eq:ellipticKdef}, and 
both elliptic parameters correspond to the so-called normal case of 
$0<m<1$.

In both cases, the quantities $\dot{C}_N$, $\dot{S}_N$, and
$\dot{G}_N$ are all periodic in the fast phase variable $\Phi/\epsilon_N$ with period $2\pi$, and the differential relations
\begin{equation}
\epsilon_N\frac{\partial \dot{S}_N}{\partial t}(x,t) = 
\frac{1}{2}\dot{C}_N(x,t)\dot{G}_N(x,t) + \bo(\epsilon_N)\quad\text{and}\quad
\epsilon_N\frac{\partial \dot{C}_N}{\partial t}(x,t) = 
-\frac{1}{2}\dot{S}_N(x,t)\dot{G}_N(x,t) +\bo(\epsilon_N)
\label{eq:CNSNGNrelation}
\end{equation}
(where $m$ is a function of $x$ and $t$ through the roots $w_0=w_0(x,t)$
and $w_1=w_1(x,t)$ of $R(w;\mathfrak{p},\mathfrak{q})^2$) hold uniformly for bounded $(x,t)$.
\label{prop:outerelliptic}
\end{proposition}

Now we describe the inner parametrices $\dot{\mathbf{O}}^\mathrm{in}_k(w)$,
which are constructed in a fairly standard way from Airy functions.
The use of such ``Airy'' parametrices has been
a linchpin of many papers using the Deift-Zhou methodology going back
to the original reference \cite{DeiftZ95}. A reference that
describes a similar construction as in the present case where the functions
$T^\nabla(w)$ and $T^\Delta(w)$ appear in the jump conditions is
\cite{BaikKMM07}. The basic Airy parametrix is the solution of the
following Riemann-Hilbert Problem.
\begin{figure}[h]
\begin{center}
\includegraphics{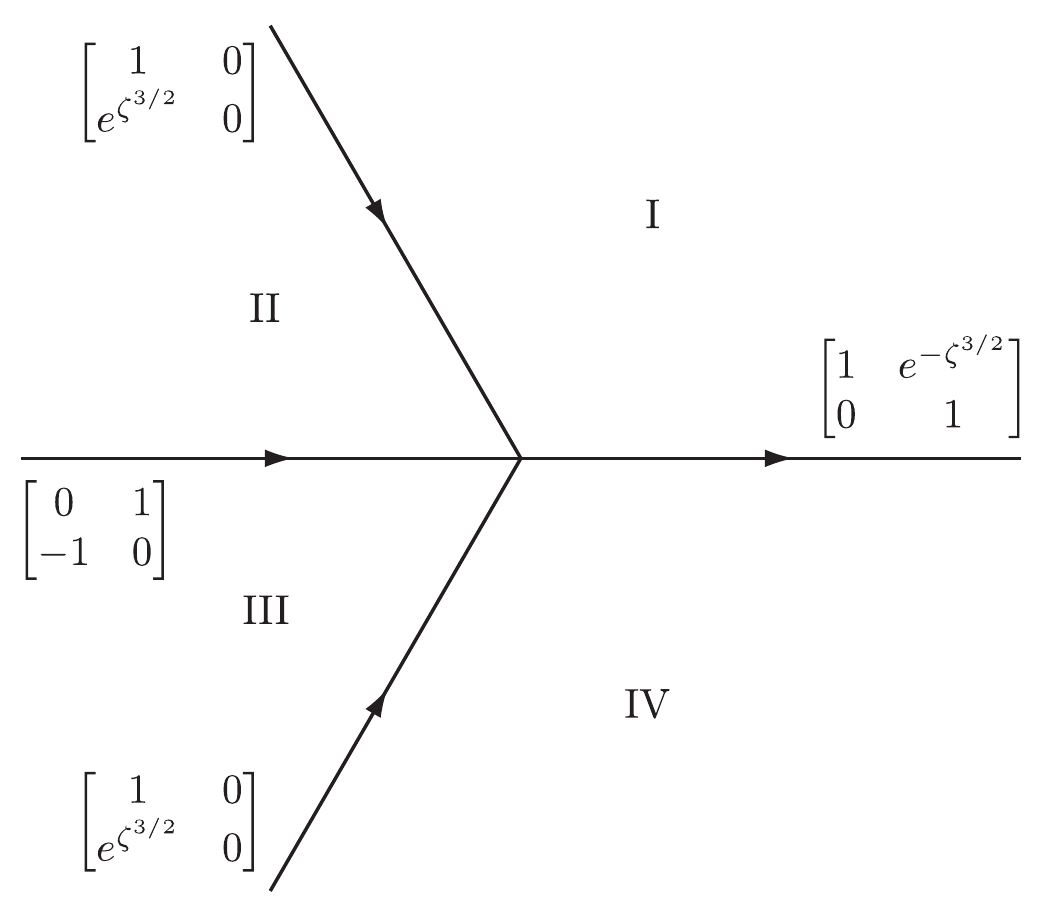}
\end{center}
\caption{\emph{The jump matrix $\mathbf{V}_\mathrm{Airy}(\zeta)$ defined
on the contour $\Sigma_\mathrm{Airy}$.}}
\label{fig:Airycontour}
\end{figure}
\begin{rhp}[Airy parametrix]
Consider the contour $\Sigma_\mathrm{Airy}$ 
illustrated in Figure~\ref{fig:Airycontour} consisting
of four rays with angles $\arg(\zeta)=0$, $\arg(\zeta)=\pm 2\pi/3$, and 
$\arg(-\zeta)=0$.  Find a $2\times 2$ matrix function $\mathbf{Z}(\zeta)$ 
with the following properties:
\begin{itemize}
\item[]\textbf{Analyticity:}  $\mathbf{Z}(\zeta)$ is an analytic function
of $\zeta\in\mathbb{C}\setminus\Sigma_\mathrm{Airy}$ and H\"older-$\gamma$
continuous for any $\gamma\le 1$ in each sector of analyticity. 
\item[]\textbf{Jump condition:} The boundary values taken by
  $\mathbf{Z}(\zeta)$ on the rays of $\Sigma_\mathrm{Airy}$ are
  related by the jump condition
  $\mathbf{Z}_+(\xi)=\mathbf{Z}_-(\xi)\mathbf{V}_\mathrm{Airy}(\xi)$
  for $\xi\in\vec{\Sigma}_\mathrm{Airy}$, where the jump matrix $\mathbf{V}_\mathrm{Airy}(\xi)$ is as defined in Figure~\ref{fig:Airycontour}.
\item[]\textbf{Normalization:}  $\mathbf{Z}(\zeta)$ satisfies the 
normalization condition
\begin{equation}
\lim_{\zeta\to\infty}\mathbf{Z}(\zeta)\mathbf{U}\zeta^{-\sigma_3/4} = \mathbb{I}
\label{eq:Airynorm}
\end{equation}
uniformly with respect to direction, where $\mathbf{U}$ is the unitary matrix
\begin{equation}
\mathbf{U}:=\frac{1}{\sqrt{2}}\begin{bmatrix}e^{-i\pi/4} & e^{i\pi/4}\\
e^{i\pi/4} & e^{-i\pi/4}\end{bmatrix}.
\end{equation}
\end{itemize}
\label{rhp:Airy}
\end{rhp}
It is well-documented \cite{DeiftZ95,BaikKMM07}
that this problem has a unique solution given by the following formulae.
Let $\tau:=(\tfrac{3}{4})^{2/3}\zeta$.  Then
\begin{equation}
\mathbf{Z}(\zeta):=\sqrt{2\pi}\left(\frac{4}{3}\right)^{\sigma_3/6}
\begin{bmatrix}
e^{-3\pi i/4}\mathrm{Ai}'(\tau) & e^{11\pi i/12}\mathrm{Ai}'(\tau e^{-2\pi i/3})\\
e^{-\pi i/4}\mathrm{Ai}(\tau) & e^{\pi i/12}\mathrm{Ai}(\tau e^{-2\pi i/3})
\end{bmatrix}e^{2\tau^{3/2}\sigma_3/3},\quad 0<\arg(\zeta)<\frac{2\pi}{3},
\end{equation}
\begin{equation}
\mathbf{Z}(\zeta):=\sqrt{2\pi}\left(\frac{4}{3}\right)^{\sigma_3/6}
\begin{bmatrix}
e^{-5\pi i/12}\mathrm{Ai}'(\tau e^{2\pi i/3}) & e^{11\pi i/12}
\mathrm{Ai}'(\tau e^{-2\pi i/3})\\
e^{-7\pi i/12}\mathrm{Ai}(\tau e^{2\pi i/3}) &
e^{\pi i/12}\mathrm{Ai}(\tau e^{-2\pi i/3})
\end{bmatrix}e^{2\tau^{3/2}\sigma_3/3},\quad \frac{2\pi}{3}<\arg(\zeta)<\pi,
\end{equation}
\begin{equation}
\mathbf{Z}(\zeta):=\sqrt{2\pi}\left(\frac{4}{3}\right)^{\sigma_3/6}
\begin{bmatrix}
e^{11\pi i/12}\mathrm{Ai}'(\tau e^{-2\pi i/3}) & e^{7\pi i/12}\mathrm{Ai}'(\tau e^{2\pi i/3})\\
e^{\pi i/12}\mathrm{Ai}(\tau e^{-2\pi i/3}) & e^{5\pi i/12}\mathrm{Ai}(\tau e^{2\pi i/3})
\end{bmatrix}e^{2\tau^{3/2}\sigma_3/3},\quad -\pi <\arg(\zeta)<-\frac{2\pi}{3},
\end{equation}
\begin{equation}
\mathbf{Z}(\zeta):=\sqrt{2\pi}\left(\frac{4}{3}\right)^{\sigma_3/6}
\begin{bmatrix} e^{-3\pi i/4}\mathrm{Ai}'(\tau) & e^{7\pi i/12}\mathrm{Ai}'(\tau e^{2\pi i/3})\\
e^{-\pi i/4}\mathrm{Ai}(\tau) & e^{5\pi i/12}\mathrm{Ai}(\tau e^{2\pi i/3})
\end{bmatrix}e^{2\tau^{3/2}\sigma_3/3},\quad -\frac{2\pi}{3}<\arg(\zeta)<0.
\end{equation}
Moreover, it is easy to show from standard asymptotic formulae for the Airy 
function $\mathrm{Ai}(\cdot)$ that the normalization condition 
\eqref{eq:Airynorm}
holds in the stronger sense that
\begin{equation}
\mathbf{Z}(\zeta)\mathbf{U}\zeta^{-\sigma_3/4} = \mathbb{I} +
\begin{bmatrix}\bo(\zeta^{-3/2}) & \bo(\zeta^{-1})\\
\bo(\zeta^{-2}) & \bo(\zeta^{-3/2})\end{bmatrix},\quad
\zeta\to\infty.
\label{eq:Airynormredo}
\end{equation}

To construct $\dot{\mathbf{O}}^\mathrm{in}_k(w)$ from $\mathbf{Z}(\zeta)$,
we must consider two cases, depending upon whether the point $w=w_k$ lies in
$\Sigma^\nabla$ or $\Sigma^\Delta$.  We assume that the corresponding disc
$U_k$ is small enough that all of $U_k\cap(\Sigma^\nabla\cup\Sigma^\Delta)$
lies also in $\Sigma^\nabla$ or $\Sigma^\Delta$ respectively.  

If $w_k\in\Sigma^\nabla$, then in the adjacent contour $\gamma$ the function
$\phi(\xi)$ satisfies $\Re\{\phi(\xi)\}<0$ and $\phi'(\xi)=R(\xi)H(\xi)$
where the quadratic $R(\xi)^2$ has a simple zero at $w_k$ and $H(\xi)$
is analytic and bounded away from zero near $w_k$.  It follows that the function
defined by taking the principal branch in the formula
\begin{equation}
W_\nabla(w):=(\phi(w_k)-\phi(w))^{2/3},\quad w\in\gamma
\end{equation}
(recall that $\phi(w_k)=0$ if $w_k\in\mathbb{R}$, while $\phi(w_k)=\pm
i\Phi$ if $w_k\in\mathbb{C}_\pm$) can be analytically continued
from $\gamma\cap U_k$ to the full neighborhood $U_k$, and (by taking
$U_k$ sufficiently small, but independent of $\epsilon_N$)
$W_\nabla'(w)\neq 0$ for $w\in U_k$.  Thus, $W=W_\nabla(w)$ defines a 
conformal map taking $U_k$ to a neighborhood of the origin in the $W$-plane.
At this point, we choose the parts of the four contour arcs meeting at $w_k$ 
within $U_k$ so that their images under $W_\nabla$ are straight segments
with angles $\arg(W)=0$ (for $W_\nabla(\gamma)$), $\arg(-W)=0$ (for $W_\nabla(\beta)$),
and $\arg(W)=\pm 2\pi/3$ (for $W_\nabla(\partial\Lambda\cap\Omega^\nabla_\mp)$).
To tie the independent variable $\zeta$ of $\mathbf{Z}(\zeta)$ to $w$, we set
\begin{equation}
\zeta:=\frac{W_\nabla(w)}{\epsilon_N^{2/3}},\quad w_k\in\Sigma^\nabla.
\label{eq:wzetanabla}
\end{equation}
Thus, the disc $U_k$ is mapped under $w\mapsto\zeta$ to a neighborhood
of the origin in the $\zeta$-plane whose outer boundary is expanding at
the uniform rate of $\epsilon^{-2/3}_N$ as $\epsilon_N\downarrow 0$.
Set $c_\nabla:=e^{i\pi/4 \pm i\Phi_\Delta/(2\epsilon_N)}$ if $w_k\in \mathbb{C}_\pm$
and $c_\nabla:=e^{i\pi/4}$ if $w_k\in\mathbb{R}$, and let a nonzero piecewise
analytic function 
$d^\nabla(w)$ be defined in $U_k$ as follows:
\begin{equation}
d^\nabla(w):=\begin{cases}1,\quad &|\arg(-W_\nabla(w))|<\pi/3\\
T^\nabla(w)^{1/2},\quad &|\arg(W_\nabla(w))|<2\pi/3.
\end{cases}
\end{equation}
It is then straightforward to verify that if $\mathbf{H}^\nabla(w)$ 
is any matrix
that is holomorphic for $w\in U_k$, with $\zeta$ defined in terms of 
$w$ by \eqref{eq:wzetanabla},
\begin{equation}
\dot{\mathbf{O}}^\mathrm{in}_k(w):=\mathbf{H}^\nabla(w)\mathbf{Z}(\zeta)
(-i\sigma_1)c_\nabla^{\sigma_3}d^\nabla(w)^{\sigma_3},\quad w\in U_k,\quad
w_k\in\Sigma^\nabla,
\label{eq:Odotinnabla}
\end{equation}
is analytic exactly where $\mathbf{O}(w)$ is and satisfies exactly the
same jump conditions as $\mathbf{O}(w)$ does within the neighborhood
$U_k$.  It remains to determine the holomorphic prefactor
$\mathbf{H}^\nabla(w)$, and this is done to achieve accurate matching onto
the outer parametrix $\dot{\mathbf{O}}^\mathrm{out}(w)$ on the disc
boundary $\partial U_k$.
To do this, we now observe that the unimodular matrices 
$\dot{\mathbf{O}}^\mathrm{out}(w)c_\nabla^{-\sigma_3}(i\sigma_1)$
and $W_\nabla(w)^{\sigma_3/4}\mathbf{U}^\dagger$ satisfy the same analyticity and
jump conditions for $w\in U$ and grow at the same rate as $w\to w_k$; 
therefore, their matrix ratio $\mathbf{B}^\nabla(w)$ is unimodular and
analytic in $U_k$, 
and moreover, it is
a consequence of Proposition~\ref{prop:outer} that this ratio is uniformly
bounded as $\epsilon_N\downarrow 0$.  We then set
\begin{equation}
\mathbf{H}^\nabla(w):=\mathbf{B}^\nabla(w)\epsilon_N^{\sigma_3/6},
\quad\text{where}\quad
\mathbf{B}^\nabla(w):=
\left[\dot{\mathbf{O}}^\mathrm{out}(w)c_\nabla^{-\sigma_3}(i\sigma_1)\right]
\cdot\left[W_\nabla(w)^{\sigma_3/4}\mathbf{U}^\dagger\right]^{-1}.
\label{eq:Hnabla}
\end{equation}
Then, it is a direct matter to check
that 
\begin{equation}
\dot{\mathbf{O}}^\mathrm{in}_k(w)\dot{\mathbf{O}}^\mathrm{out}(w)^{-1} = 
\mathbf{C}^\nabla(w)\zeta^{-\sigma_3/4}(\mathbf{Z}(\zeta)\mathbf{U}\zeta^{-\sigma_3/4})\zeta^{\sigma_3/4}\mathbf{C}^\nabla(w)^{-1}\mathbf{D}^\nabla(w),\quad
w\in U_k,\quad w_k\in\Sigma^\nabla,
\end{equation}
where
\begin{equation}
\mathbf{C}^\nabla(w):=\dot{\mathbf{O}}^\mathrm{out}(w)
c_\nabla^{-\sigma_3}(i\sigma_1)
\mathbf{U}\quad\text{and}\quad\mathbf{D}^\nabla(w):=
\dot{\mathbf{O}}^\mathrm{out}(w)d^\nabla(w)^{\sigma_3}\dot{\mathbf{O}}^\mathrm{out}(w)^{-1}.
\end{equation}
Now if $w=\xi\in\partial U_k$, then Proposition~\ref{prop:outer} guarantees that
$\dot{\mathbf{O}}^\mathrm{out}(\xi)$ and its inverse are bounded on
$\partial U_k$ uniformly as $\epsilon_N\downarrow 0$, while 
Proposition~\ref{prop:YTnablaDelta} guarantees that $d^\Delta(w)^{\sigma_3}=
\mathbb{I}+\bo(\epsilon_N)$ holds uniformly for $w\in U_k$, so it follows
that $\mathbf{C}^\nabla(\xi)=\bo(1)$, $\mathbf{C}^\nabla(\xi)^{-1}=\bo(1)$, and 
$\mathbf{D}^\nabla(\xi)=\mathbb{I}+\bo(\epsilon_N)$ are uniform estimates
for $\xi\in\partial U_k$.  Furthermore, since $w=\xi\in\partial U_k$
is equivalent to $\zeta^{-1}=\bo(\epsilon_N^{2/3})$, we may use the
large-$\zeta$ asymptotic formula \eqref{eq:Airynormredo} to obtain
$\zeta^{-\sigma_3/4}(\mathbf{Z}(\zeta)\mathbf{U}\zeta^{-\sigma_3/4})
\zeta^{\sigma_3/4}=\mathbb{I}+\bo(\zeta^{-3/2})=
\mathbb{I}+\bo(\epsilon_N)$ for $w=\xi\in\partial U_k$.
The result of these calculations is the uniform estimate
\begin{equation}
\dot{\mathbf{O}}_k^\mathrm{in}(\xi)\dot{\mathbf{O}}^\mathrm{out}(\xi)^{-1}=
\mathbb{I}+\bo(\epsilon_N),\quad \xi\in\partial U_k,\quad k=0,1.
\label{eq:discmatch}
\end{equation}
This relation shows that our choice of the holomorphic prefactor
$\mathbf{H}^\nabla(w)$ yields an accurate match between the inner and
outer parametrices on the disc boundary $\partial U_k$.

If instead $w_k\in\Sigma^\Delta$, then in the adjacent contour $\gamma$
the function $\phi(\xi)$ satisfies $\Re\{\phi(\xi)\}>0$ so the correct
conformal mapping is defined by taking the principal branch in
\begin{equation}
W_\Delta(w):=(\phi(w)-\phi(w_k))^{2/3},\quad w\in\gamma
\end{equation}
and analytically continuing from $\gamma\cap U_k$ to $U_k$ (using the fact
that $\phi(w)-\phi(w_k)$ behaves as $(w-w_k)^{3/2}$ near $w=w_k$).  Choosing
the contours within $U_k$ so that their images under $W=W_\Delta(w)$ lie along
rays with angles $\arg(W)=0$ (for $W_\Delta(\gamma)$), $\arg(-W)=0$ (for $W_\Delta(\beta)$), and $\arg(W)=\pm 2\pi/3$ (for $W_\Delta(\partial\Lambda\cap\Omega^\Delta_\pm)$), we choose the independent variable $\zeta$ in the Airy parametrix
$\mathbf{Z}(\zeta)$ to be given by 
\begin{equation}
\zeta:=\frac{W_\Delta(w)}{\epsilon_N^{2/3}},\quad w_k\in\Sigma^\Delta.
\label{eq:zetaDeltadef}
\end{equation}
Setting $c_\Delta:=e^{-\pi i/4 \pm i\Phi_\Delta/(2\epsilon_N)}$ if $w_k\in\mathbb{C}_\pm$ and $C_\Delta:=e^{-\pi i/4}$ if $w_k\in\mathbb{R}$, and defining $d^\Delta(w)$
for $w\in U_k$ by
\begin{equation}
d^\Delta(w):=\begin{cases}1,&\quad |\arg(-W_\Delta(w))|<\pi/3\\
T^\Delta(w)^{-1/2},&\quad |\arg(W_\Delta(w))|<2\pi/3,
\end{cases}
\end{equation}
we define the inner parametrix as
\begin{equation}
\dot{\mathbf{O}}^\mathrm{in}_k(w):=\mathbf{H}^\Delta(w)\mathbf{Z}(\zeta)
c_\Delta^{\sigma_3}d^\Delta(w)^{\sigma_3},\quad w\in U_k,\quad w_k\in\Sigma^\Delta,
\label{eq:OdotinDelta}
\end{equation}
where $\zeta$ is a function of $w$ and $\epsilon_N$ by \eqref{eq:zetaDeltadef},
and the holomorphic prefactor $\mathbf{H}^\Delta(w)$ is given by
\begin{equation}
\mathbf{H}^\Delta(w):=\mathbf{B}^\Delta(w)\epsilon_N^{\sigma_3/6},\quad
\text{where}\quad
\mathbf{B}^\Delta(w):=\left[\dot{\mathbf{O}}^\mathrm{out}(w)c_\Delta^{-\sigma_3}\right]\cdot
\left[W_\Delta(w)^{\sigma_3/4}\mathbf{U}^\dagger\right]^{-1}.
\label{eq:HDelta}
\end{equation}
It is again straightforward to confirm that for $w\in U_k$,
$\dot{\mathbf{O}}^\mathrm{in}_k(w)$ is analytic where $\mathbf{O}(w)$
is, and satisfies exactly the same jump conditions as does
$\mathbf{O}(w)$.  Also, 
\begin{equation}
\dot{\mathbf{O}}^\mathrm{in}_k(w)\dot{\mathbf{O}}^\mathrm{out}(w)^{-1}=
\mathbf{C}^\Delta(w)\zeta^{-\sigma_3/4}(\mathbf{Z}(\zeta)\mathbf{U}
\zeta^{-\sigma_3/4})\zeta^{\sigma_3/4}\mathbf{C}^\Delta(w)^{-1}
\mathbf{D}^\Delta(w),\quad w\in U_k,\quad w_k\in\Sigma^\Delta,
\end{equation}
where
\begin{equation}
\mathbf{C}^\Delta(w):=\dot{\mathbf{O}}^\mathrm{out}(w)c_\Delta^{-\sigma_3}\mathbf{U}
\quad\text{and}\quad
\mathbf{D}^\Delta(w):=\dot{\mathbf{O}}^\mathrm{out}(w)d^\Delta(w)^{\sigma_3}
\dot{\mathbf{O}}^\mathrm{out}(w)^{-1}.
\end{equation}
Completely analogous reasoning as in the case that $w_k\in\Sigma^\nabla$ then
shows that the estimate \eqref{eq:discmatch} holds also when 
$w_k\in\Sigma^\Delta$.

We formalize these results in the following proposition.
\begin{proposition}
Suppose that the root $w_k$ of the quadratic $R(w;\mathfrak{p},\mathfrak{q})^2$ is bounded away from
the other root and also from $\mathfrak{a}$ and $\mathfrak{b}$, and let the inner parametrix
$\dot{\mathbf{O}}_k^\mathrm{in}(w)$ be defined for $w$ in a suitably small (but
independent of $\epsilon_N$) neighborhood $U_k$ of $w_k$ by either
\eqref{eq:Odotinnabla}--\eqref{eq:Hnabla} (if $w_k\in\Sigma^\nabla$) or 
\eqref{eq:OdotinDelta}--\eqref{eq:HDelta} 
(if $w_k\in\Sigma^\Delta$).  Then 
$\det(\dot{\mathbf{O}}_k^\mathrm{in}(w))=1$ where
defined, 
$\mathbf{O}(w)\dot{\mathbf{O}}_k^\mathrm{in}(w)^{-1}$ is analytic for
$w\in U_k$, and the mismatch with the outer parametrix on $\partial U_k$
is characterized by the estimate \eqref{eq:discmatch}.
\label{prop:Airy}
\end{proposition}

We emphasize that it is a consequence of the differential identities
\eqref{eq:wthetaprimebeta} and \eqref{eq:wphiprimegamma} along with
the fact that $H$ is bounded away from zero near each of the distinct
roots of $R(w;\mathfrak{p},\mathfrak{q})^2$ that the matrix $\mathbf{Z}(\zeta)$ can be used
to construct the correct inner parametrix for $\mathbf{O}(w)$ in
$U_k$.  Indeed, were the function $\phi(w)-\phi(w_k)$ to vanish to
higher order due to coalescence of roots of $R(w;\mathfrak{p},\mathfrak{q})^2$ or the
presence of a zero of the analytic function $H$, a more exotic inner
parametrix would be required because $W_\nabla(w)$ or $W_\Delta(w)$
would fail to be a proper conformal map.  A more fruitful approach in
such a situation is to investigate the double-scaling limit where
$\phi$ degenerates due to allowing $(x,t)$ to converge to some
critical point at an appropriate rate as $\epsilon_N\downarrow 0$.  We
will carry out such analysis in a forthcoming paper \cite{BuckinghamMseparatrix}
 for the case when the two roots of
$R(w;\mathfrak{p},\mathfrak{q})^2$ coalesce on the real axis when $t=0$ and $x=\pm
x_\mathrm{crit}$.  Modified inner parametrices will
also be required if one or the other roots of $R(w;\mathfrak{p},\mathfrak{q})^2$ ``bounces
off'' of $\mathfrak{a}$ or $\mathfrak{b}$ as explained in
Proposition~\ref{prop:origin}, but this is a far less interesting
special case.

\subsection{The effect of conjugation.  Estimation of the error}
\label{sec:error}
Consider the matrix $\mathbf{E}(w)$ (the \emph{error} in approximating
$\mathbf{O}(w)$ with the global parametrix $\dot{\mathbf{O}}(w)$) defined
by
\begin{equation}
\mathbf{E}(w):=\mathbf{O}(w)\dot{\mathbf{O}}(w)^{-1}
\end{equation}
for all $w\in\mathbb{C}$ where both matrices on the right-hand side are 
well-defined. Since according to Propositions~\ref{prop:outer} and
\ref{prop:Airy} the outer parametrix $\dot{\mathbf{O}}(w)$ defined by
\eqref{eq:hatOdef} is a unimodular sectionally analytic matrix function,
it follows that the same is true of $\mathbf{E}(w)$.  
Also, since (i) the outer parametrix 
$\dot{\mathbf{O}}^\mathrm{out}(w)$ satisfies exactly the same
jump condition as does $\mathbf{O}(w)$ on the arcs of the contour $\vec{\beta}$
and (ii) the inner parametrices $\dot{\mathbf{O}}^\mathrm{in}_k(w)$ satisfy
exactly the same jump conditions as does $\mathbf{O}(w)$ on all contours
within the open discs $U_k$, the error $\mathbf{E}(w)$ can be analytically
continued to all of these contour arcs.  Thus, the jump contour for
$\mathbf{E}(w)$, denoted $\Sigma_\mathbf{E}$, is as illustrated in 
Figures~\ref{fig:BcaseEcontour}--\ref{fig:KcaseEMinusOOM}.
\begin{figure}[h]
\begin{center}
\includegraphics{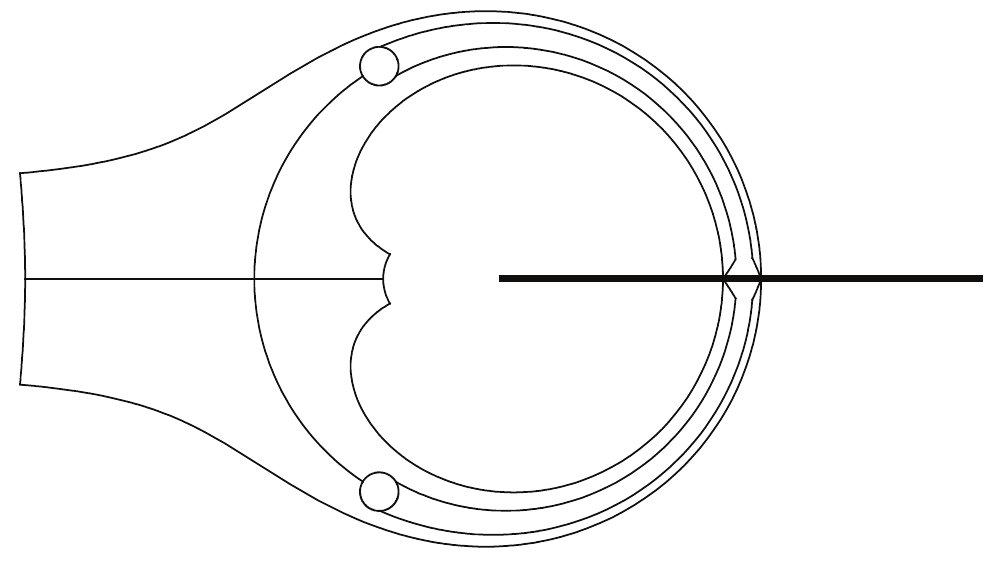}
\end{center}
\caption{\emph{The contour $\Sigma_\mathbf{E}$ of discontinuity of the sectionally analytic
function $\mathbf{E}(w)$ in case \librational.  The circles are the 
boundaries of the discs $U_1$ and $U_2$.}  }
\label{fig:BcaseEcontour}
\end{figure}
\begin{figure}[h]
\begin{center}
\includegraphics{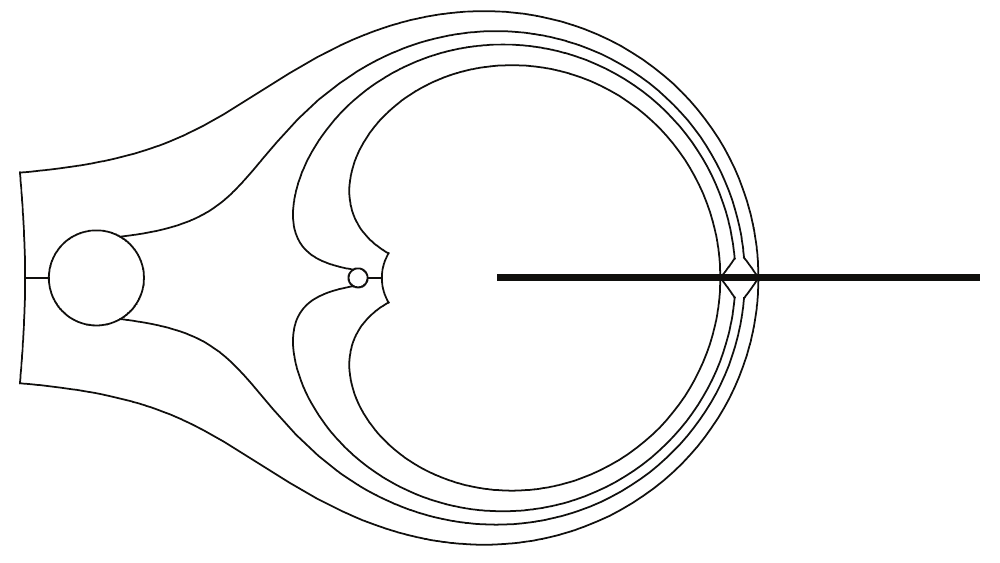}
\end{center}
\caption{\emph{The contour $\Sigma_\mathbf{E}$ of discontinuity of the sectionally analytic
function $\mathbf{E}(w)$ in case \rotational\ with either $\Delta=\emptyset$
or $\nabla=\emptyset$.  The circles are the 
boundaries of the discs $U_1$ and $U_2$.  }}
\label{fig:KcaseEEmpty}
\end{figure}
\begin{figure}[h]
\begin{center}
\includegraphics{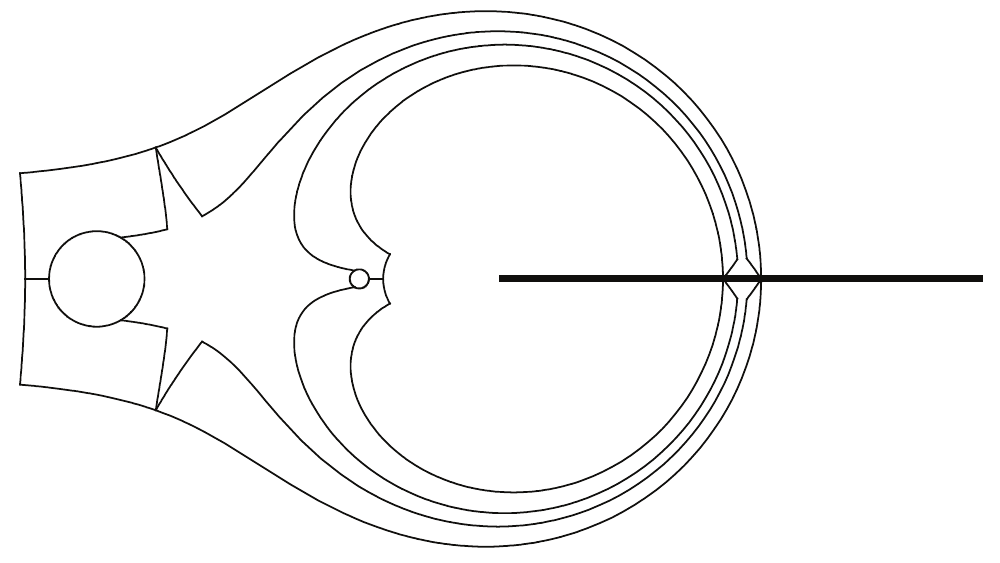}
\end{center}
\caption{\emph{The contour $\Sigma_\mathbf{E}$ of discontinuity of the sectionally analytic
function $\mathbf{E}(w)$ in case \rotational\ with either 
$\Delta=P^{\prec\rotational}_N$
or 
$\nabla=P^{\prec\rotational}_N$.  The circles are the 
boundaries of the discs $U_1$ and $U_2$.   }}
\label{fig:KcaseEMinusM}
\end{figure}
\begin{figure}[h]
\begin{center}
\includegraphics{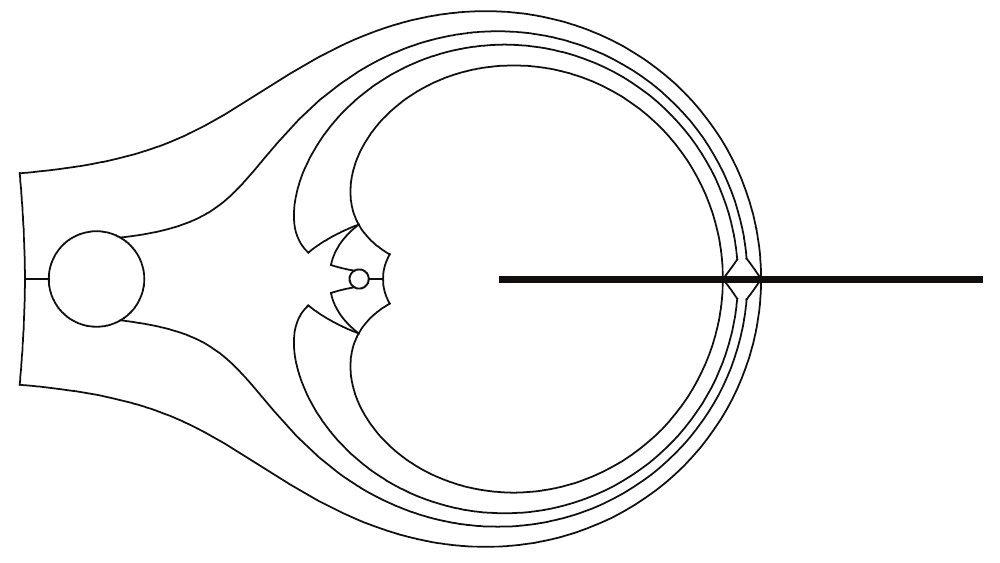}
\end{center}
\caption{\emph{The contour $\Sigma_\mathbf{E}$ of discontinuity of the sectionally analytic
function $\mathbf{E}(w)$ in case \rotational\ with either 
$\Delta=P^{\rotational\succ}_N$
or 
$\nabla=P^{\rotational\succ}_N$.  The circles are the 
boundaries of the discs $U_1$ and $U_2$. }}
\label{fig:KcaseEMinusOOM}
\end{figure}

Now while the global parametrix $\dot{\mathbf{O}}(w)$ is known, the
matrix $\mathbf{O}(w)$ is only characterized by being related via
explicit transformations to $\mathbf{H}(w)$, which in turn is
specified only as the (unknown) solution of Riemann-Hilbert
Problem~\ref{rhp:basicw}.  However, the conditions of Riemann-Hilbert
Problem~\ref{rhp:basicw} imply an equivalent Riemann-Hilbert problem
whose solution must give $\mathbf{E}(w)$.  As both $\mathbf{O}(w)$ and
$\dot{\mathbf{O}}(w)$ tend to the identity matrix as $w\to\infty$ (in
the case of $\mathbf{O}(w)$ this follows from the sequence of explicit
transformations relating it back to $\mathbf{H}(w)$ and the
normalization condition \eqref{eq:Hnorm} for the latter, and in the
case of $\dot{\mathbf{O}}(w)$ this follows from the fact that
$\dot{\mathbf{O}}(w)=\dot{\mathbf{O}}^\mathrm{out}(w)$ for large $|w|$
and from the normalization condition on the outer parametrix as specified
from the conditions of Riemann-Hilbert Problem~\ref{rhp:wOdotlibrational}
or \ref{rhp:wOdotrotational}), we must require that $\mathbf{E}(w)\to\mathbb{I}$
as $w\to\infty$.  To formulate the Riemann-Hilbert problem for the error
it remains to analyze the jump conditions satisfied by $\mathbf{E}(w)$
along the contour $\Sigma_\mathbf{E}$ pictured in 
Figures~\ref{fig:BcaseEcontour}--\ref{fig:KcaseEMinusOOM}.

First consider the jump of $\mathbf{E}(w)$ across the positive real
axis $\xi\in \mathbb{R}_+\subset\Sigma_\mathbf{E}$.  On either side of
$\mathbb{R}_+$ we have $\dot{\mathbf{O}}_\pm(\xi)=
\dot{\mathbf{O}}^\mathrm{out}_\pm(\xi)$, and according to the jump
conditions of Riemann-Hilbert Problem~\ref{rhp:wOdotlibrational} or
\ref{rhp:wOdotrotational} we therefore have the exact relation
$\dot{\mathbf{O}}_+(\xi)=\sigma_2\dot{\mathbf{O}}_-(\xi)\sigma_2$ for
$\xi\in\mathbb{R}_+$.  According to Proposition~\ref{prop:OjumpsAway}
the corresponding boundary values of $\mathbf{O}(w)$ are related by
$\mathbf{O}_+(\xi)=\sigma_2\mathbf{O}_-(\xi)\sigma_2
(\mathbb{I}+\bo(\epsilon_N))$, where the error term is identically
zero except in the interval $J$ where the lens $\Lambda$ abuts the
positive real axis from above and below.  Using the fact (see
Proposition~\ref{prop:outer}) that
$\dot{\mathbf{O}}_-(\xi)=\dot{\mathbf{O}}_-^\mathrm{out}(\xi)$ is,
along with its inverse, uniformly bounded for $\xi>0$, we then see that
$\mathbf{E}_+(\xi)=\sigma_2\mathbf{E}_-(\xi)\sigma_2(\mathbb{I}+\mathbf{X}(\xi))$ holds for $\xi>0$,
where $\mathbf{X}(\xi)=\bo(\epsilon_N)$ for $\xi\in J$ and otherwise
$\mathbf{X}(\xi)\equiv 0$.

Across the disc boundaries $\partial U_k\subset\Sigma_\mathbf{E}$ we
have no discontinuity of $\mathbf{O}(w)$, but $\mathbf{E}(w)$ is
discontinuous because the global parametrix $\dot{\mathbf{O}}(w)$ is
discontinuous due to the mismatch between the outer and inner parametrices.  
If we take the disc boundaries to be oriented in the clockwise direction,
then we have $\mathbf{E}_+(\xi)= \mathbf{E}_-(\xi)
[\dot{\mathbf{O}}^\mathrm{in}_k(\xi)\dot{\mathbf{O}}^\mathrm{out}(\xi)^{-1}]$
for $\xi\in\partial U_k$, and according to
Proposition~\ref{prop:Airy}, this jump condition can be written in the
form
$\mathbf{E}_+(\xi)=\mathbf{E}_-(\xi)(\mathbb{I}+\bo(\epsilon_N))$.

Finally, consider $\xi\in\Sigma_\mathbf{E}\setminus(\mathbb{R}_+\cup\partial U_0\cup\partial U_1)$.  According to Proposition~\ref{prop:OjumpsAway}, we have
$\mathbf{O}_+(\xi)=\mathbf{O}_-(\xi)(\mathbb{I}+\bo(\epsilon_N))$ holding 
uniformly for all such $\xi$.  But the global parametrix has no jump,
and we have $\dot{\mathbf{O}}(\xi)=
\dot{\mathbf{O}}^\mathrm{out}(\xi)$ on these contour arcs, and again
recalling Proposition~\ref{prop:outer} as $\epsilon_N\downarrow 0$
the global parametrix is uniformly bounded here along with its inverse.  
Therefore
$\mathbf{E}_+(\xi)=\mathbf{E}_-(\xi)\dot{\mathbf{O}}^\mathrm{out}(\xi)(\mathbb{I}+\bo(\epsilon_N))
\dot{\mathbf{O}}^\mathrm{out}(\xi)^{-1} = \mathbf{E}_-(\xi)(\mathbb{I}+\bo(\epsilon_N))$ holds uniformly for $\xi\in\Sigma_\mathbf{E}\setminus(\mathbb{R}_+\cup\partial U_0\cup\partial U_1)$.

It follows from these considerations that $\mathbf{E}(w)$ may be characterized
as the solution of the following Riemann-Hilbert problem.
\begin{rhp}[Error]
Seek a $2\times 2$ matrix $\mathbf{E}(w)$ with the following properties:
\begin{itemize}
\item[]\textbf{Analyticity:}  $\mathbf{E}(w)$ is analytic for $w\in\mathbb{C}\setminus\Sigma_\mathbf{E}$ where $\Sigma_\mathbf{E}$ is the contour (independent of
$\epsilon_N$) pictured in various cases in 
Figures~\ref{fig:BcaseEcontour}--\ref{fig:KcaseEMinusOOM}, and in each component
of the domain of analyticity is uniformly H\"older continuous with any
exponent $\gamma\le 1$.
\item[]\textbf{Jump conditions:}  The boundary values taken by $\mathbf{E}(w)$
on $\Sigma_\mathbf{E}$ are related as follows.  For $\xi\in\mathbb{R}_+$,
\begin{equation}
\mathbf{E}_+(\xi)=\sigma_2\mathbf{E}_-(\xi)\sigma_2(\mathbb{I}+\mathbf{X}(\xi))
\label{eq:Ejumptwist}
\end{equation}
where $\mathbf{X}(\xi)=\bo(\epsilon_N)$ for $\xi\in J$ and $\mathbf{X}(\xi)=0$
for $\xi\in\mathbb{R}_+\setminus J$.  For all remaining $\xi\in\Sigma_\mathbf{E}$
we have the uniform estimate
\begin{equation}
\mathbf{E}_+(\xi)=\mathbf{E}_-(\xi)(\mathbb{I}+\bo(\epsilon_N))
\label{eq:Ejumpnormal}
\end{equation}
as $\epsilon_N\downarrow 0$.
\item[]\textbf{Normalization:}  The matrix $\mathbf{E}(w)$ satisfies the condition
\begin{equation}
\mathop{\lim_{w\to\infty}}_{|\arg(-w)|<\pi}\mathbf{E}(w)=\mathbb{I}.
\end{equation}
\end{itemize}
\label{rhp:error}
\end{rhp}

This Riemann-Hilbert problem closely resembles a problem of
``small-norm'' type, except for the form of the jump condition along
the positive real axis (the branch cut of $\sqrt{-w}$).  But if we
consider unfolding the branch cut by setting $w=z^2$ and then defining
a matrix function $\mathbf{F}(z)$ in terms of $\mathbf{E}(w)$ by
\begin{equation}
\mathbf{F}(z):=\begin{cases}\mathbf{E}(z^2),\quad &\Im\{z\}>0\\
\sigma_2\mathbf{E}(z^2)\sigma_2,\quad &\Im\{z\}<0,
\end{cases}
\label{eq:EFtwist}
\end{equation}
then the jump contour $\Sigma_\mathbf{F}$ 
for $\mathbf{F}(z)$ in the $z$-plane includes
two disjoint images of $\Sigma_\mathbf{E}\setminus\mathbb{R}_+$, one
in the upper half-plane and one in the lower half-plane; moreover from
\eqref{eq:Ejumpnormal} the jump of $\mathbf{F}(z)$ across the arcs of
either of these two images is of the form
$\mathbf{F}_+(z)=\mathbf{F}_-(z)(\mathbb{I}+\bo(\epsilon_N))$.  If
$z\in\mathbb{R}$, then it follows from \eqref{eq:Ejumptwist} that
$\mathbf{F}_+(z)=\mathbf{F}_-(z)(\mathbb{I}+\bo(\epsilon_N))$ where
the error term vanishes identically if $z^2\not\in J$.  Thus, $\mathbf{F}(z)$
is the solution of a ``small-norm'' Riemann-Hilbert problem of standard
form.  It is a metatheorem in this subject that such problems have unique 
solutions that are uniformly close to the identity matrix on compact sets
that avoid the jump contour and in a full neighborhood of the point at 
infinity.   Indeed, solving the singular integral equations corresponding 
to such a problem
involves inverting an operator (on, say, $L^2(\Sigma_\mathbf{F})$) that
is a perturbation of the identity of size $\bo(\epsilon_N)$ in operator norm.
Such a problem can of course be solved by iteration, and the resulting 
Neumann series also functions as an asymptotic series as 
$\epsilon_N\downarrow 0$.  This yields a representation of $\mathbf{F}(z)$
in terms of a Cauchy integral:
\begin{equation}
\mathbf{F}(z)=\mathbb{I}+
\frac{1}{2\pi i}\int_{\Sigma_\mathbf{F}}\frac{\mathbf{Y}(\xi)\,d\xi}{\xi-z},
\end{equation}
where $\mathbf{Y}\in L^2(\Sigma_\mathbf{F})$ with 
$\|\mathbf{Y}(\cdot)\|_2=\bo(\epsilon_N)$.  Note that since $\Sigma_\mathbf{F}$
is compact, by Cauchy-Schwarz 
we also have $\mathbf{Y}\in L^1(\Sigma_\mathbf{F})$ with
$\|\mathbf{Y}(\cdot)\|_1=\bo(\epsilon_N)$.  It follows that if $K$ is
a compact subset of $\mathbb{C}$ disjoint from $\Sigma_\mathbf{F}$, then
\begin{equation}
\sup_{z\in K}\|\mathbf{F}(z)-\mathbb{I}\| = \bo(\epsilon_N)
\end{equation}
because the Cauchy kernel $(\xi-z)^{-1}$ is uniformly bounded for 
$\xi\in\Sigma_\mathbf{F}$ and $z\in K$.  Also, if $z$ lies outside of a sufficiently large disc containing $\Sigma_\mathbf{F}$, then the geometric series
$(\xi-z)^{-1}=-(z^{-1}+\xi z^{-2} + \xi^2 z^{-3} +\cdots)$ is uniformly convergent
for $\xi\in\Sigma_\mathbf{F}$, and so we obtain the convergent series expansion
for $\mathbf{F}(z)$ as $z\to\infty$
\begin{equation}
\mathbf{F}(z)=\mathbb{I}-\sum_{n=1}^\infty\frac{1}{2\pi iz^n}
\int_{\Sigma_\mathbf{F}}\mathbf{Y}(\xi)\xi^{n-1}\,d\xi,
\end{equation}
and the coefficient of each negative power of $z$ is $\bo(\epsilon_N)$.
Corresponding results hold for $\mathbf{E}(w)$ by restricting $z$ to the
upper half-plane and using \eqref{eq:EFtwist}.

We have thus shown that the significance of the global parametrix 
$\dot{\mathbf{O}}(w)$ defined by
\eqref{eq:hatOdef} is the following approximation result.
\begin{proposition}
Suppose $(x,t)$ is a point in one of the domains $\mathscr{O}_\librational^\pm$,
$\mathscr{O}_\rotational^\pm$ (see Proposition~\ref{prop:tneq0continuegeneral}), or
$\mathscr{O}_\rotational^0$ (see Proposition~\ref{prop:origin}), and that
$|t|$ is sufficiently small. If also the roots of the quadratic $R(w;\mathfrak{p},\mathfrak{q})^2$
do not coincide, nor does either root equal $\mathfrak{a}$ or $\mathfrak{b}$, then
Riemann-Hilbert Problem~\ref{rhp:basicw} has a unique solution
$\mathbf{H}(w)$, and the matrix $\mathbf{O}(w)$ obtained therefrom
by means of the systematic substitutions $\mathbf{H}\mapsto\mathbf{J}$
(see \eqref{eq:Htildedef}), 
$\mathbf{J}\mapsto\mathbf{M}$ (see \eqref{eq:Mdefinew}), 
$\mathbf{M}\mapsto\mathbf{N}$ (see \eqref{eq:NMw}), and 
$\mathbf{N}\mapsto\mathbf{O}$ (see \eqref{eq:wOdef})
has expansions for large and small $w$ of the form
\begin{equation}
\mathbf{O}(w)=\mathbf{O}_N^{0,0}(x,t) + \mathbf{O}_N^{0,1}(x,t)\sqrt{-w}
+\bo(w),\quad w\to 0
\label{eq:Ozeroexpand}
\end{equation}
and
\begin{equation}
\mathbf{O}(w)=\mathbb{I}+\frac{\mathbf{O}_N^{\infty,1}(x,t)}{\sqrt{-w}}
+\bo(w^{-1}),\quad w\to\infty,
\label{eq:Oinftyexpand}
\end{equation}
and the coefficients satisfy the estimates (see Proposition~\ref{prop:outer})
\begin{equation}
\begin{split}
\mathbf{O}_N^{0,0}(x,t)&=\dot{\mathbf{O}}^{0,0}+\bo(\epsilon_N)\\
\mathbf{O}_N^{0,1}(x,t)&=\dot{\mathbf{O}}^{0,1}+\bo(\epsilon_N)\\
\mathbf{O}_N^{\infty,1}(x,t)&=\dot{\mathbf{O}}^{\infty,1}+\bo(\epsilon_N),
\end{split}
\end{equation}
where the dependence on $N$ and $(x,t)$  on the right-hand side enters through
$\nu=\Phi(x,t)/\epsilon_N +\pi\#\Delta$ and the motion of the contour $\beta$.
The $\bo(\epsilon_N)$ error terms are also uniform
with respect to $(x,t)$ as long as the roots of the quadratic $R(w;\mathfrak{p},\mathfrak{q})^2$
are bounded away from each other and from $\mathfrak{a}$ and $\mathfrak{b}$.
\label{prop:Error}
\end{proposition}

If $w$ lies in a small neighborhood of the origin, or alternatively if $|w|$ 
is sufficiently large, then according to \eqref{eq:wOdef} $\mathbf{O}(w)$ coincides with $\mathbf{N}(w)$,
and according to \eqref{eq:Mdefinew} $\mathbf{M}(w)$ coincides with $\mathbf{J}(w)$, the latter matrix function 
being the solution of Riemann-Hilbert Problem~\ref{rhp:modifiedw}.  Recalling
the relations \eqref{eq:Htildedef} and \eqref{eq:NMw}, we see that for such
$w$,
\begin{equation}
\mathbf{H}(w)=\mathbf{O}(w)e^{g(w)\sigma_3/\epsilon_N}\left(\prod_{y\in\Delta}
\frac{\sqrt{-w}+\sqrt{-y}}{\sqrt{-w}-\sqrt{-y}}\right)^{\sigma_3}.
\end{equation}
Since $g(0)=g(\infty)=0$, the diagonal factor relating $\mathbf{H}(w)$
and $\mathbf{O}(w)$ has the expansions
\begin{equation}
e^{g(w)\sigma_3/\epsilon_N}\left(\prod_{y\in\Delta}\frac{\sqrt{-w}+\sqrt{-y}}
{\sqrt{-w}-\sqrt{-y}}\right)^{\sigma_3} = 
\begin{cases}
(-1)^{\#\Delta} + \mathbf{C}_N^{0,1}(x,t)\sqrt{-w} +\bo(w),\quad & w\to 0\\
1+\mathbf{C}_N^{\infty,1}(x,t)/\sqrt{-w} +\bo(w^{-1}),\quad &w\to\infty
\end{cases}
\end{equation}
for some diagonal matrices $\mathbf{C}_N^{0,1}(x,t)$ and
$\mathbf{C}_N^{\infty,1}(x,t)$.  The matrices $\mathbf{A}_N(x,t)$,
$\mathbf{B}_N^0(x,t)$, and $\mathbf{B}_N^\infty(x,t)$ obtained from
$\mathbf{H}(w)$ via \eqref{eq:wzeroexpansion}--\eqref{eq:ABsdef}
are then written in terms of the expansion coefficients of $\mathbf{O}(w)$
from \eqref{eq:Ozeroexpand}--\eqref{eq:Oinftyexpand} as
\begin{equation}
\begin{split}
\mathbf{A}_N(x,t)&=(-1)^{\#\Delta}\mathbf{O}_N^{0,0}(x,t)\\
\mathbf{B}_N^0(x,t)&=(-1)^{\#\Delta}\mathbf{C}_N^{0,1}(x,t)+
\mathbf{O}_N^{0,0}(x,t)^{-1}\mathbf{O}_N^{0,1}(x,t)\\
\mathbf{B}_N^\infty(x,t)&=\mathbf{O}_N^{\infty,1}(x,t) +\mathbf{C}_N^{\infty,1}(x,t).
\end{split}
\end{equation}
Then using Proposition~\ref{prop:Error} these are expressed asymptotically
in terms of the corresponding expansion coefficients of the outer parametrix
$\dot{\mathbf{O}}^\mathrm{out}(w)$ as follows:
\begin{equation}
\begin{split}
\mathbf{A}_N(x,t)&=(-1)^{\#\Delta}\dot{\mathbf{O}}^{0,0} +\bo(\epsilon_N)\\
\mathbf{B}_N^0(x,t)&=(-1)^{\#\Delta}\mathbf{C}_N^{0,1}(x,t)+
(\dot{\mathbf{O}}^{0,0})^{-1}\dot{\mathbf{O}}^{0,1}+\bo(\epsilon_N)\\
\mathbf{B}_N^\infty(x,t)&=\dot{\mathbf{O}}^{\infty,1}+\mathbf{C}_N^{\infty,1}(x,t) +\bo(\epsilon_N)
\end{split}
\end{equation}
(recall that according to Proposition~\ref{prop:outer} the
coefficients $\dot{\mathbf{O}}^{0,0}$,
$\dot{\mathbf{O}}^{0,1}$, and
$\dot{\mathbf{O}}^{\infty,1}$ are bounded as
$\epsilon_N\downarrow 0$, along with
$(\dot{\mathbf{O}}^{0,0})^{-1}$).  Now we recall the definitions
of the quantities $\cos(\tfrac{1}{2}u_N(x,t))$ and
$\sin(\tfrac{1}{2}u_N(x,t))$ (see \eqref{eq:cossinuN}), and of
$\epsilon_Nu_{N,t}(x,t)$ (see \eqref{eq:epsilonut})
characterizing the fluxon condensate $\{u_N(x,t)\}_{N=N_0}^\infty$ according to
Definition~\ref{def:condensate}, and compare with the definitions
\eqref{eq:dotClibrational}--\eqref{eq:dotvlibrational} to obtain the asymptotic
formulae
\begin{equation}
\cos\left(\frac{1}{2}u_N(x,t)\right)=\dot{C}_N(x,t) + \bo(\epsilon_N)\quad
\text{and}\quad
\sin\left(\frac{1}{2}u_N(x,t)\right)=\dot{S}_N(x,t) +\bo(\epsilon_N)
\label{eq:cossinasymp1}
\end{equation}
and 
\begin{equation}
\epsilon_N\frac{\partial u_N}{\partial t}(x,t) = \dot{G}_N(x,t) +
\bo(\epsilon_N),
\label{eq:epsutasymp1}
\end{equation}
where the asymptotics are valid for the same ranges of $(x,t)$ and with the
same nature of convergence as in the statement of Proposition~\ref{prop:Error}.
The asymptotic formulae \eqref{eq:cossinasymp1} are differentiable with respect
to $t$ (yielding \eqref{eq:epsutasymp1}) according to \eqref{eq:CNSNGNrelation}
from Proposition~\ref{prop:outerelliptic}.

\appendix
\section{Proofs of Propositions Concerning Initial Data}
\label{app:initialdata}
\subsection{Proof of Proposition~\ref{prop:AbelInverse}}
Given a real-valued differentiable function $f(v)$ defined
on $0<v<V$, define
\begin{equation}
I[f](w):=-\frac{4}{\pi}\int_{-w}^V\frac{f'(v)\,dv}{\sqrt{v^2-w^2}}, \quad
-V<w<0
\end{equation}
as the right-hand side of \eqref{eq:inversetransform} with
$f'(v)=\varphi(v)$.  If $f$ is in the range of
\eqref{eq:WKBphaserewrite} for some $G$ satisfying
Assumptions~\ref{assume:KlausShaw} and \ref{assume:evenness}, then we know
that $V=-G(0)$ and
\begin{equation}
f'(v)=-\frac{v}{2}\int_0^{G^{-1}(-v)}\frac{ds}{\sqrt{G(s)^2-v^2}},
\quad 0<v<V=-G(0).
\end{equation}
Therefore in this case we have
\begin{equation}
I[f](w)=\frac{1}{\pi}\int_{-w}^{-G(0)}\int_{0}^{G^{-1}(-v)}
\frac{2v}{\sqrt{v^2-w^2}\sqrt{G(s)^2-v^2}}\,ds\,dv,\quad G(0)<w<0.
\end{equation}
Exchanging the order of integration and setting $\tau=v^2$ yields
\begin{equation}
I[f](w)=\frac{1}{\pi}\int_0^{G^{-1}(w)}\int_{w^2}^{G(s)^2}\frac{d\tau}
{\sqrt{\tau-w^2}\sqrt{G(s)^2-\tau}}\,ds.
\end{equation}
The inner integral evaluates to $\pi$ (independent of $s$ and $w$) so
\begin{equation}
I[f](w)=\int_0^{G^{-1}(w)}\,ds = G^{-1}(w)
\end{equation}
yielding the identity \eqref{eq:inversetransform} as desired.

\subsection{Proof of Proposition~\ref{prop:theta0}}
First note that analyticity and strict monotonicity of $G(x)$
automatically ensures the
analyticity and positivity of $\mathscr{G}$ in the open interval of its definition.
Now assume that $0<v<-G(0)$.  Introducing $m=G(s)^2$ as a change of
  variables in \eqref{eq:WKBphaserewrite} yields
\begin{equation}
\Psi(\lambda)=\frac{1}{2}\int_{v^2}^{G(0)^2}\sqrt{\frac{m-v^2}{G(0)^2-m}}
\mathscr{G}(m)\frac{dm}{m},\quad \lambda=\frac{iv}{4},
\end{equation}
where $\mathscr{G}$ is defined in terms of $G$ by \eqref{eq:hmdef} and satisfies the conditions of Assumption~\ref{assume:h}.
Let $S(m;v)$ be the analytic function of $m$ for 
$m\in\mathbb{C}\setminus [v^2,G(0)^2]$ that satisfies
\begin{equation}
S(m;v)^2 = \frac{m-v^2}{G(0)^2-m}\quad\text{and}\quad
\lim_{\delta\downarrow 0}S(m+i\delta;v)>0\quad\text{for $v^2<m<G(0)^2$}.
\end{equation}
Then, with $L_1$ being the contour loop shown in Figure~\ref{fig:CDcontours},
\begin{equation}
\Psi(\lambda)=\frac{1}{4}\oint_{L_1}S(m;v)\mathscr{G}(m)\frac{dm}{m}.
\end{equation}
\begin{figure}[h]
\begin{center}
\includegraphics{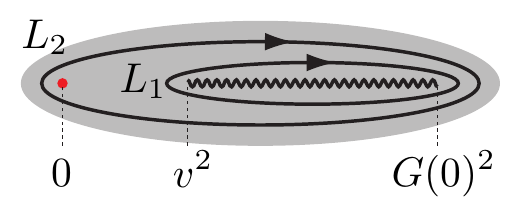}
\end{center}
\caption{\emph{The loop contour $L_1$ surrounds the branch cut
    $[v^2,G(0)^2]$ of $S(m;v)$ while the loop contour $L_2$ also
    encloses the origin.  Both contours lie within the presumed domain
    of analyticity of $\mathscr{G}(m)$ (shaded).}}
\label{fig:CDcontours}
\end{figure}
Since $S$ is also analytic as a function of $v$ when $m$ lies on $L_1$
and $v$ lies in the open region bounded by $L_1$, and since $S$ is
continuous in $m$ for such $v$ and $m$, this formula immediately shows
that $\Psi(\lambda)$ is analytic for each $v\in (0,-G(0))$.  To
analyze the behavior near the endpoints, we proceed as follows.  By a
simple contour deformation,
\begin{equation}
\begin{split}
\Psi(\lambda)&=\frac{i\pi}{2}S(0;v)\mathscr{G}(0) +\frac{1}{4}\oint_{L_2}S(m;v)\mathscr{G}(m)
\frac{dm}{m}\\
&=\frac{\pi v \mathscr{G}(0)}{2G(0)} +\frac{1}{4}\oint_{L_2}S(m;v)\mathscr{G}(m)\frac{dm}{m},
\end{split}
\label{eq:theta0D}
\end{equation}
where we used the fact that $S(0,v)=-iv/G(0)$.  Now $S(m;v)$ has the following
convergent expansions:
\begin{equation}
S(m;v)=S(m;0)\left[1-\frac{v^2}{m}\right]^{1/2}=S(m,0)\left[1-\sum_{k=1}^\infty
\frac{(2k-1)!!}{(2k)!!}\frac{v^{2k}}{m^{k}}\right],\quad |m|>|v|^2,
\end{equation}
\begin{equation}
\begin{split}
S(m;v)&=S(m;-G(0))
\left[1-\frac{-2G(0)(v+G(0))+(v+G(0))^2}{m-G(0)^2}\right]^{1/2}
\\
&=S(m;-G(0))\left[1-\sum_{k=1}^\infty
\frac{(2k-1)!!}{(2k)!!}\left(\frac{-2G(0)(v+G(0))+(v+G(0))^2}{m-G(0)^2}
\right)^k\right]\\
& = S(m;-G(0))\left[1+\frac{G(0)(v+G(0))}{m-G(0)^2}+\cdots
\right],\quad |m-G(0)^2|>|-2G(0)(v+G(0))+(v+G(0))^2|.
\end{split}
\label{eq:Sseriesmg0}
\end{equation}
Note also that $S(m;-G(0))\equiv i$.  Both of these expansions are
uniformly convergent on the contour $L_2$ if $v$ is confined to a
sufficiently small neighborhood of $v=0$ or $v=-G(0)$ respectively,
and hence the integral in \eqref{eq:theta0D} may be calculated
term-by-term.  In the case of the expansion for $v$ small, the result
is a convergent series in even nonnegative powers of $v$ with purely
imaginary coefficients.  In the case of the expansion for $v$ near
$-G(0)$, the result is a convergent series in (generally) all
nonnegative integer powers of $v+G(0)$, again with purely imaginary
coefficients.  Since power series always converge in disks, we now see that
we have constructed the analytic continuation of $\Psi(\lambda)$
valid for $v$ in full complex neighborhoods of $v=0$ and $v=-G(0)$
respectively.  The form of the Taylor series \eqref{eq:theta0series0}
is now clear, but it remains to confirm the positivity of $\alpha$
and the value of $\Psi(0)$.  But from \eqref{eq:theta0D},
\begin{equation}
\alpha = -\frac{2\pi \mathscr{G}(0)}{G(0)}>0
\end{equation}
since $G(0)<0$ and $\mathscr{G}(0)>0$ by hypothesis, and the constant term 
$\Psi(0)=\|G\|_1/4$ is easier to evaluate directly by passing to
the limit $\lambda\to 0$ in the formula \eqref{eq:WKBphaserewrite}
than by working with the series.  Finally, to confirm the relations
\eqref{eq:theta0neartop}, one may pass to the limit $\lambda\to -iG(0)/4$
in \eqref{eq:WKBphaserewrite} to obtain $\Psi(-iG(0)/4)=0$,
and then also from \eqref{eq:theta0D} and \eqref{eq:Sseriesmg0}
\begin{equation}
\left.\frac{d}{dv}\Psi(\lambda)\right|_{v=-G(0)}=\frac{\pi \mathscr{G}(0)}{2G(0)}+\frac{iG(0)}{4}
\oint_{L_2}\frac{\mathscr{G}(m)}{m-G(0)^2}\frac{dm}{m}= \frac{\pi \mathscr{G}(G(0)^2)}{2G(0)}<0
\end{equation}
since $\mathscr{G}(G(0)^2)>0$ by hypothesis, 
where the integral over $L_2$ is evaluated explicitly by residues.

\section{Details of the Outer Parametrix in Cases \librational\ and \rotational}
\label{app:outer}
Since many of the important details of the asymptotic behavior of the
sine-Gordon equation are derived from the appropriate outer parametrix,
we here provide all details of the solution of Riemann-Hilbert Problems~\ref{rhp:wOdotlibrational} and \ref{rhp:wOdotrotational} in terms of Riemann $\Theta$-functions of
genus one.  We also explain how the 
extracted potentials (approximate solutions of sine-Gordon) can be reduced to
a very simple form in terms of Jacobi elliptic functions. 
This appendix contains all details of 
the proofs of Propositions~\ref{prop:outer}
and \ref{prop:outerelliptic}.

\subsection{The outer parametrix in case \librational.  Proof of Proposition~\ref{prop:outer} in this case}
\subsubsection{Solution of Riemann-Hilbert Problem~\ref{rhp:wOdotlibrational} in
terms of Baker-Akhiezer functions}

The first step in solving Riemann-Hilbert Problem~\ref{rhp:wOdotlibrational} 
for $\dot{\mathbf{O}}^\mathrm{out}(w)$ is to introduce a new, 
equivalent, unknown $\mathbf{P}(w)$,
given in terms of $\dot{\mathbf{O}}^\mathrm{out}(w)$ by
\begin{equation}
\mathbf{P}(w):=\begin{cases} \dot{\mathbf{O}}^\mathrm{out}(w),\quad &
\Im\{w\}>0\\
\sigma_2\dot{\mathbf{O}}^\mathrm{out}(w)\sigma_2,\quad & \Im\{w\}<0.
\end{cases}
\label{eq:dotPdefB}
\end{equation}
The equivalent Riemann-Hilbert problem satisfied by $\mathbf{P}(w)$
is described by the scheme shown in Figure~\ref{fig:wPdotB}.
\begin{figure}[h]
\begin{center}
\includegraphics{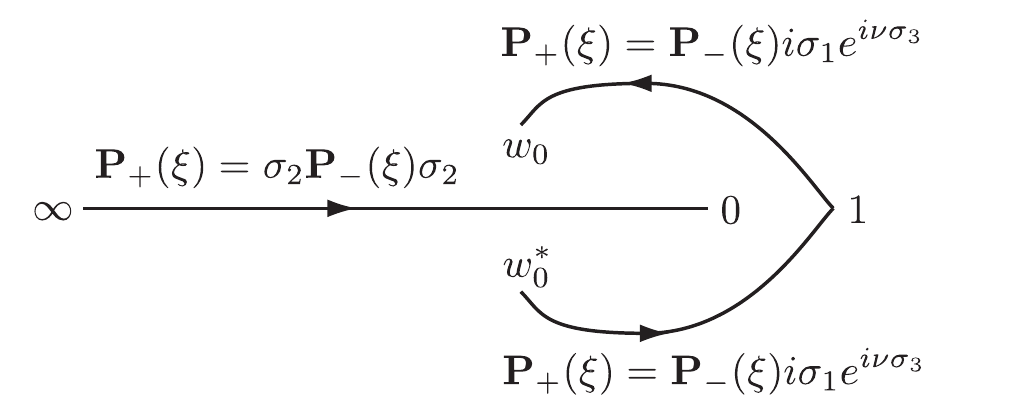}
\end{center}
\caption{\emph{The jump conditions satisfied by the matrix $\mathbf{P}(w)$
normalized as $\mathbf{P}(w)=\mathbb{I}+\bo(|w|^{-1/2})$ as $w\to\infty$.
}}
\label{fig:wPdotB}
\end{figure}
As was the case with $\dot{\mathbf{O}}^\mathrm{out}(w)$, 
the boundary values taken on the
two disjoint components of the jump contour are continuous and bounded with
the exception of the endpoints $w_0$ and $w_0^*$ where inverse fourth roots 
are tolerated.

Next, we remove the real parameter $\nu$ 
from the jump conditions by
defining the scalar function $h(w)$ as follows:
\begin{equation}
h(w):=-\frac{S(w)}{2\pi i}\int_{w_0^*}^{w_0}\frac{ds}{S_+(s)(s-w)},
\label{eq:hdefB}
\end{equation}
where the integration is along the jump contour shown in
Figure~\ref{fig:wPdotB} and where $S(w)^2:=w(w-w_0)(w-w_0^*)$, $S(w)$
is analytic in the complement of the jump contours, and
$S(w)=w^{3/2}(1+\bo(w^{-1}))$ as $w\to\infty$ (principal branch of
$w^{3/2}$).  Note that
\begin{equation}
h(w)=pw^{1/2}+\bo(|w|^{-1/2}),\quad w\to\infty,
\label{eq:hasympB}
\end{equation}
where
\begin{equation}
p:=\frac{1}{2\pi i}\int_{w_0^*}^{w_0}\frac{ds}{S_+(s)}.
\label{eq:h1defB}
\end{equation}
The function defined by \eqref{eq:hdefB} 
satisfies $h_+(\xi)+h_-(\xi)=0$ for $\xi$ on the negative real
axis, and $h_+(\xi)+h_-(\xi)=-1$ for $\xi$ on the contour connecting $w_0^*$
and $w_0$.  It takes continuous and bounded boundary values on the
entire jump contour.  The new unknown we define in place of
$\mathbf{P}(w)$ is then
\begin{equation}
\mathbf{Q}(w):=\mathbf{P}(w)e^{i\nu h(w)\sigma_3}.
\label{eq:dotQdefB}
\end{equation}
Direct calculations then show that the conditions determining
$\mathbf{Q}(w)$ are as indicated in Figure~\ref{fig:wQdotB}.
\begin{figure}[h]
\begin{center}
\includegraphics{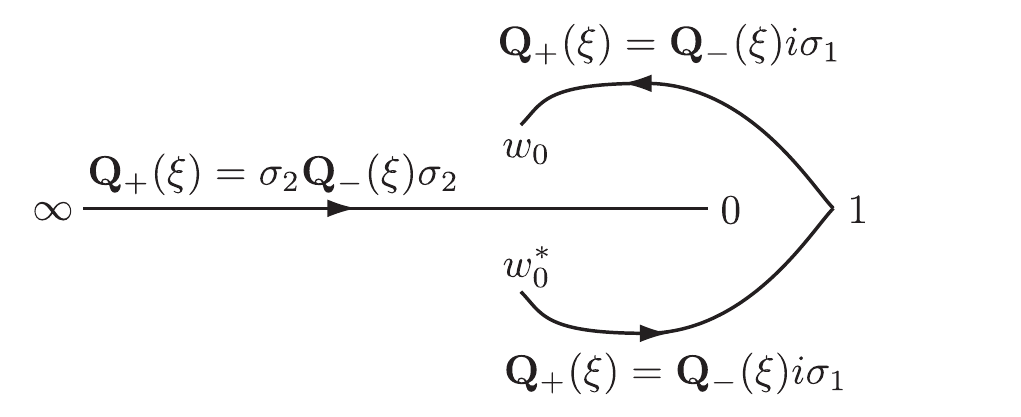}
\end{center}
\caption{\emph{The jump conditions satisfied by the matrix $\mathbf{Q}(w)$
normalized as $\mathbf{Q}(w)e^{-ip\nu w^{1/2}\sigma_3}=
\mathbb{I}+\bo(|w|^{-1/2})$ as $w\to\infty$.
}}
\label{fig:wQdotB}
\end{figure}
Inverse fourth root singularities are again admitted at $w_0$ and $w_0^*$.

The next transformation is undertaken to diagonalize the prefactor of
$\sigma_2$ in the jump condition on the negative real axis.  So, since
\begin{equation}
\sigma_2=\mathbf{V}\sigma_3\mathbf{V}^{-1},\quad\quad
\mathbf{V}:=\frac{1}{\sqrt{2}}\begin{bmatrix}e^{-i\pi/4} & -e^{-i\pi/4}\\
e^{i\pi/4} & e^{i\pi/4}\end{bmatrix},\quad \det(\mathbf{V})=1,\quad
\mathbf{V}^{-1}=\mathbf{V}^\dagger,
\label{eq:sigma2diag}
\end{equation}
the matrix defined by
\begin{equation}
\mathbf{R}(w):=\mathbf{V}^\dagger\mathbf{Q}(w)
\label{eq:dotRdefB}
\end{equation}
satisfies the conditions indicated in Figure~\ref{fig:wRdotB}.
\begin{figure}[h]
\begin{center}
\includegraphics{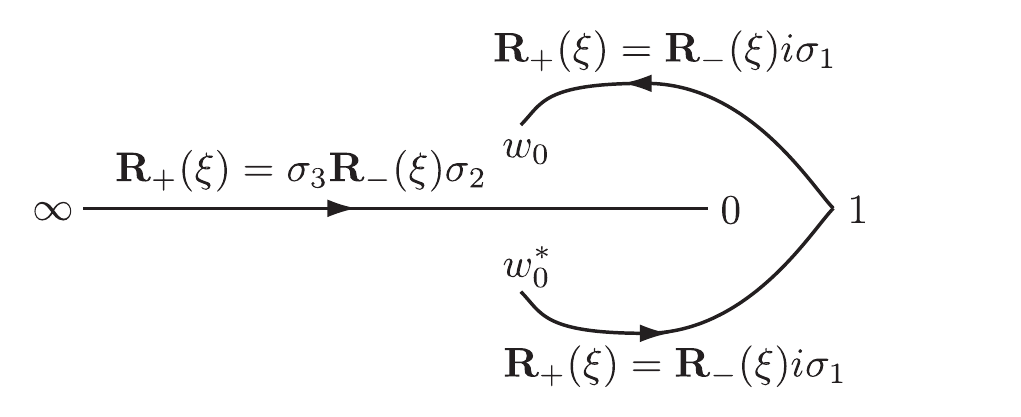}
\end{center}
\caption{\emph{The jump conditions satisfied by the matrix $\mathbf{R}(w)$
normalized as $\mathbf{R}(w)e^{-ip\nu w^{1/2}\sigma_3}=
\mathbf{V}^\dagger+\bo(|w|^{-1/2})$ as $w\to\infty$.
}}
\label{fig:wRdotB}
\end{figure}
Inverse fourth root singularities are once again admitted at $w_0$ and $w_0^*$.

The next transformations aim to convert the post-multiplicative jump
matrices both into the permutation matrix $\sigma_1$.  
To accomplish this on the arc connecting
$w_0^*$ to $w_0$, we introduce the function $q(w)$ satisfying
\begin{equation}
q(w)^4 = \frac{w-w_0}{w-w_0^*}, 
\end{equation}
that is uniquely specified as being analytic except on the contour arc
connecting $w_0^*$ to $w_0$ and satisfying $q(w)=1
+\bo(w^{-1})$ as $w\to\infty$.  This function satisfies
$q_+(w)=iq_-(w)$ for $w$ on the contour arc of discontinuity.
Setting 
\begin{equation}
\mathbf{S}(w):=q(w)^{-1}\mathbf{R}(w),
\label{eq:dotSdefB}
\end{equation}
we find that the new unknown $\mathbf{S}(w)$ corresponds to the
conditions in Figure~\ref{fig:wSdotB}.
\begin{figure}[h]
\begin{center}
\includegraphics{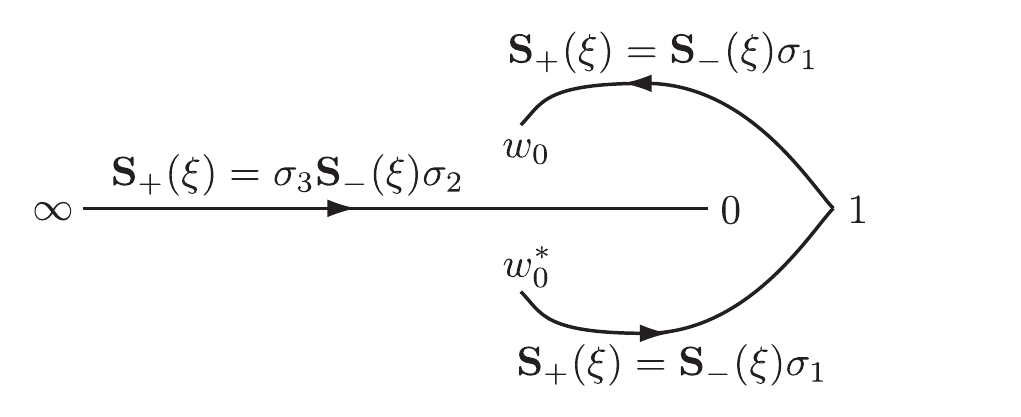}
\end{center}
\caption{\emph{The jump conditions satisfied by the matrix $\mathbf{S}(w)$
normalized as $\mathbf{S}(w)e^{-ip\nu w^{1/2}\sigma_3}=
\mathbf{V}^\dagger+\bo(|w|^{-1/2})$ as $w\to\infty$.  
Unlike $\dot{\mathbf{O}}^\mathrm{out}(w)$,
$\mathbf{P}(w)$, $\mathbf{Q}(w)$, and $\mathbf{R}(w)$,
the matrix $\mathbf{S}(w)$ is required to be bounded in a neighborhood
of $w=w_0^*$ while we allow a stronger singularity
 at $w=w_0$:  $\mathbf{S}(w)=\bo(|w-w_0|^{-1/2})$.
}}
\label{fig:wSdotB}
\end{figure}
Note that multiplication by $q(w)$ changes the nature of the
admissible singularities at $w_0$ and $w_0^*$: $\mathbf{S}(w)$
is now required to be bounded at $w_0^*$, and an inverse square root
singularity is admitted at $w_0$.

To convert the jump condition on the negative real axis to the same form,
we first separate the two rows of $\mathbf{S}(w)$ by writing
\begin{equation}
\mathbf{S}(w)=\begin{bmatrix}\trans{\mathbf{s}_1(w)}\\
\trans{\mathbf{s}_2(w)}
\end{bmatrix}
\end{equation}
and we also write $\mathbf{V}=[\mathbf{v}_1,\mathbf{v}_2]$ to give
notation for the normalized eigenvectors of $\sigma_2$ (see 
\eqref{eq:sigma2diag}).  As $\sigma_3$
is diagonal, the jump conditions for the individual rows
decouple and in each case these conditions only involve matrix
multiplication on the right.  Thus we have
$\trans{\mathbf{s}_{j+}(\xi)}=\trans{\mathbf{s}_{j-}(\xi)}\sigma_1$ for
$j=1,2$ and $\xi$ on the contour arc connecting $w_0^*$ and $w_0$, while
we have
$\trans{\mathbf{s}_{1+}(\xi)}=\trans{\mathbf{s}_{1-}(\xi)}\sigma_2$ and
$\trans{\mathbf{s}_{2+}(\xi)}=\trans{\mathbf{s}_{2-}(\xi)}(-\sigma_2)$ for
$\xi$ on the negative real axis.  Now set
\begin{equation}
\trans{\mathbf{t}_1(w)}:=\sqrt{2}\trans{\mathbf{s}_1(w)}e^{-i\pi
  k(w)\sigma_3/2}\quad \text{and}\quad
\trans{\mathbf{t}_2(w)}:=i\sqrt{2}\trans{\mathbf{s}_2(w)}e^{i\pi
  k(w)\sigma_3/2},\quad\text{where}\quad k(w):=\frac{1}{2}+h(w). 
\end{equation}
The function $k$ is analytic on the complement of the contours and
satisfies $k_+(\xi)+k_-(\xi)=0$ on the arc connecting $w_0^*$ with $w_0$
while $k_+(\xi)+k_-(\xi)=1$ for negative real $\xi$.  Using \eqref{eq:hasympB}
one sees that $k(w)$ has the asymptotic behavior
\begin{equation}
k(w)=pw^{1/2} +\frac{1}{2} + \bo(|w|^{-1/2}),\quad w\to\infty.
\label{eq:kasympB}
\end{equation}
From this information it is easy to see that on all jump contours we have
\begin{equation}
\trans{\mathbf{t}_{j+}(\xi)}=\trans{\mathbf{t}_{j-}(\xi)}\sigma_1,\quad
j=1,2,
\label{eq:Btvecjump}
\end{equation}
and we also have the normalization conditions
\begin{equation}
\trans{\mathbf{t}_j(w)}e^{-ip\varphi_jw^{1/2}\sigma_3}
=[1,1] + \bo(|w|^{-1/2}),\quad
w\to\infty
\end{equation}
where
\begin{equation}
\varphi_1:=\nu-\frac{\pi}{2}\quad\text{and}\quad
\varphi_2:=\nu+\frac{\pi}{2}.
\label{eq:GjB}
\end{equation}
Both $\trans{\mathbf{t}_1(w)}$ and $\trans{\mathbf{t}_2(w)}$ may
become unbounded in the finite $w$-plane only as $w\to w_0$, where all
four scalar components must be $\bo(|w-w_0|^{-1/2})$.

Finally, we implement the involutive jump conditions \eqref{eq:Btvecjump} by viewing
the elements of the row vectors $\trans{\mathbf{t}_j(w)}$ as single-valued scalar
functions on an appropriate Riemann surface.
Let $X$ be the Riemann surface of the equation $y^2=S(w)^2=w(w-w_0)(w-w_0^*)$
compactified at $y=w=\infty$.  We view the finite part of $X$ as two
copies of the $w$-plane (sheets) cut along the branch cuts of the
function $S(w)$ and glued together in the usual way.  We define on
$X\setminus\{\infty,w_0\}$ two scalar functions $t_1(P)$ and $t_2(P)$
as follows:
\begin{equation}
t_j(P):=\begin{cases}[\trans{\mathbf{t}_j(w(P))}]_1,\quad &
P\in \text{sheet 1}\\
[\trans{\mathbf{t}_j(w(P))}]_2,\quad &
P\in \text{sheet 2},
\end{cases}
\end{equation}
where for row vectors $[u_1,u_2]_j:=u_j$,
and $w(P)$ denotes the ``sheet projection'' function.
These definitions are consistent along the cuts where the sheets are
identified precisely because the jump matrices for
$\trans{\mathbf{t}_j(w)}$ have all been reduced to the simple permutation 
(sheet exchange) matrix $\sigma_1$.  The
\emph{Baker-Akhiezer functions} $t_j:X\setminus\{w_0,\infty\}\to\mathbb{C}$ 
are analytic in their domain of definition, which omits just two points of 
$X$.  Since
\begin{equation}
y_0(P):=\begin{cases} (w(P)-w_0)^{1/2},\quad & P\in\text{sheet 1}\\
-(w(P)-w_0)^{1/2},\quad & P\in\text{sheet 2}
\end{cases}
\label{eq:localcoordw0}
\end{equation}
is a holomorphic local coordinate for $X$ near the
branch point $P=w_0$, we see that $t_j(P)$ admits a simple pole at this
point.  Near the point $P=\infty$, $t_j(P)$ has exponential behavior:
\begin{equation}
t_j(P)e^{-ip\varphi_jy_\infty(P)^{-1}}=1 + \bo(y_\infty(P)),\quad P\to\infty
\label{eq:BAasympB}
\end{equation}
where
\begin{equation}
y_\infty(P):=\begin{cases}w(P)^{-1/2},\quad & P\in \text{sheet 1}\\
-w(P)^{-1/2},\quad & P\in\text{sheet 2}
\end{cases}
\label{eq:kinfty}
\end{equation}
is a holomorphic local coordinate for $X$ near the branch point
$P=\infty$.

\subsubsection{Construction of the Baker-Akhiezer functions}
To express $t_j(P)$ in terms of special functions requires a
few ingredients.  Firstly, one chooses a basis of homology cycles on
$X$ consisting of two noncontractible oriented closed paths, $a$ and
$b$, such that $b$ intersects $a$ exactly once, from the right of $a$.
For later convenience, we suppose that neither of these paths passes
through the point $P=\infty$.  The homology cycles we choose are illustrated
in Figure~\ref{fig:wHomologyB}.
\begin{figure}[h]
\begin{center}
\includegraphics{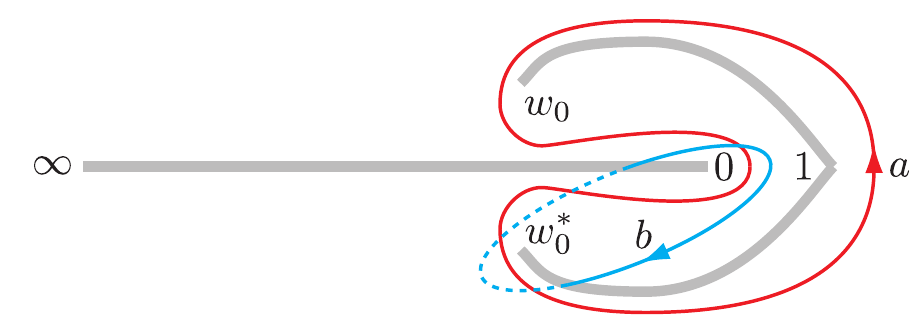}
\end{center}
\caption{\emph{The homology basis on the sheets of $X$.  Thick curves indicate
branch cuts where the two sheets are identified, solid curves are on sheet $1$
and dashed curves are on sheet $2$.
}}
\label{fig:wHomologyB}
\end{figure}
As $X$ is a genus one Riemann surface (elliptic curve), 
there is a one-dimensional space of holomorphic differentials; 
as a basis we choose the unique
holomorphic differential on $X$ of the form
\begin{equation}
\omega(P)=c\frac{dw(P)}{\tilde{S}(P)}
\label{eq:holodiffB}
\end{equation}
where $\tilde{S}(P)$ denotes the lift to $X$ of the function $S(w)$:
\begin{equation}
\tilde{S}(P):=\begin{cases}S(w(P)),\quad &P\in\text{sheet 1}\\
-S(w(P)),\quad &P\in\text{sheet 2},
\end{cases}
\end{equation}
and where the constant $c$ is chosen so that
\begin{equation}
\oint_a\omega(P)=2\pi i.
\label{eq:holodiffnormB}
\end{equation}
Denote by $\mathcal{H}$ the other loop integral of $\omega(P)$:
\begin{equation}
\mathcal{H}:=\oint_b\omega(P)=2\pi i\frac{\displaystyle\oint_b\frac{dw(P)}{\tilde{S}(P)}}
{\displaystyle\oint_a\frac{dw(P)}{\tilde{S}(P)}}.
\label{eq:wHdefineB}
\end{equation}
It is easy to see that although $\mathcal{H}$ is complex, $\Re\{\mathcal{H}\}<0$.
The \emph{Riemann $\Theta$-function} corresponding to $\mathcal{H}$ is the entire function
of $z\in\mathbb{C}$ given by the Fourier series
\begin{equation}
\Theta(z;\mathcal{H}):=\sum_{n=-\infty}^{\infty}e^{\tfrac{1}{2}\mathcal{H}n^2}e^{nz}.
\label{eq:RiemannTheta}
\end{equation}
A simple relabeling of the sum by $n\mapsto -n$ shows that 
\begin{equation}
\Theta(-z;\mathcal{H})=\Theta(z;\mathcal{H}).
\label{eq:Thetaeven}
\end{equation}
The Riemann $\Theta$-function satisfies the \emph{automorphic identities}:
\begin{equation}
\Theta(z+2\pi i;\mathcal{H})=\Theta(z;\mathcal{H})
\label{eq:thetaperiodic}
\end{equation}
and
\begin{equation}
\Theta(z+\mathcal{H};\mathcal{H})=e^{-\frac{1}{2}\mathcal{H}}e^{-z}\Theta(z;\mathcal{H}).
\label{eq:thetamonodromy}
\end{equation}
The function $\Theta(z;\mathcal{H})$ vanishes to first order at all points $z$ of the 
form $z=\mathcal{K}+2\pi
im+\mathcal{H}n$ and nowhere else, where $m$ and $n$ are integers and where
\begin{equation}
\mathcal{K}=\mathcal{K}(\mathcal{H}):=i\pi +\frac{1}{2}\mathcal{H}
\label{eq:RiemannConstantB}
\end{equation}
is the \emph{Riemann constant}.  We choose
as a base point on $X$ the branch point $P_0=w_0$, and then define
the \emph{Abel map} by
\begin{equation}
A(P):=\int_{P_0}^P\omega(P')\quad\pmod { 2\pi im+\mathcal{H}n},\quad m,n\in\mathbb{Z}.
\label{eq:AbelMapB}
\end{equation}
The value of $A(P)$ is not completely determined only because the path is only
determined modulo the cycles $a$ and $b$.  Finally, let $\Omega(P)$ denote
the abelian differential of the second kind with double pole at $P=\infty$:
\begin{equation}
\Omega(P):=\frac{w(P)+C}{2\tilde{S}(P)}\,dw(P)
\label{eq:OmegadefB}
\end{equation}
where the constant $C$ is chosen so that
\begin{equation}
\oint_a\Omega(P)=0.
\label{eq:secondkindnormB}
\end{equation}
Note that by asymptotic expansion of \eqref{eq:OmegadefB}, we have
\begin{equation}
\int_{P_0}^P\Omega(P')=\frac{1}{y_\infty(P)} + \bo(1),\quad P\to\infty.
\label{eq:OmegaasympB}
\end{equation}
Let $\kappa$ denote the other loop integral of $\Omega(P)$:
\begin{equation}
\kappa:=\oint_b\Omega(P).
\label{eq:PhidefB}
\end{equation}
\begin{lemma}[see \cite{Dubrovin81}]
$\kappa=2c$.
\label{lem:dissection}
\end{lemma}
\begin{proof}
Let $\tilde{X}$ denote the canonical dissection of $X$ obtained by
cutting $X$ along the cycles $a$ and $b$.  $\tilde{X}$ is a
simply-connected complex manifold with boundary illustrated in
Figure~\ref{fig:canonicaldissection}.
\begin{figure}[h]
\begin{center}
\includegraphics{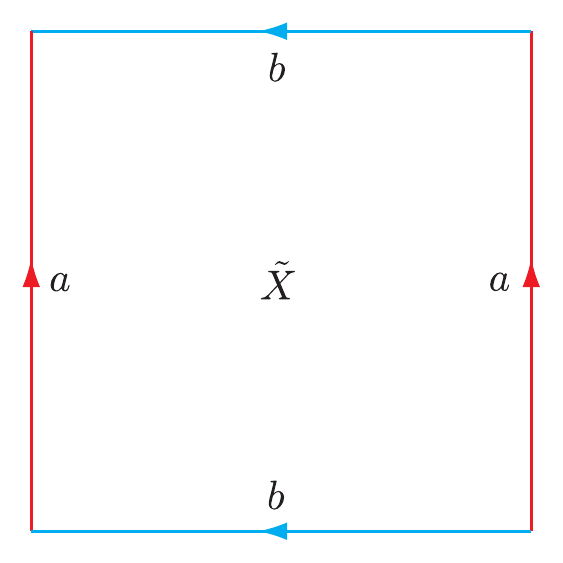}
\end{center}
\caption{\emph{The canonical dissection $\tilde{X}$ of the compact Riemann
surface $X$.  Note that if we omit the upper and right-hand boundaries, then
the points of $X$ and those of $\tilde{X}$ are in one-to-one correspondence.}}
\label{fig:canonicaldissection}
\end{figure}
Consider the meromorphic differential defined on $\tilde{X}$ given by the
following expression:
\begin{equation}
\Gamma(P):=A(P)\Omega(P).
\end{equation}
Here, the Abel mapping is well defined since $\tilde{X}$ is simply
connected.  Let $\partial\tilde{X}$ denote the positively-oriented boundary
of $\tilde{X}$, and consider the integral
\begin{equation}
I:=\oint_{\partial\tilde{X}}\Gamma(P).
\end{equation}
We will evaluate $I$ two different ways.  On the one hand, we may
evaluate $I$ by residues.  The only singularity of $\Gamma(P)$
is a double pole at the point $P=\infty\in\tilde{X}$, and while $\Omega(P)$
has no residue there, $\Gamma(P)$ indeed has one:
\begin{equation}
\begin{split}
\Gamma(P)&=\left[A(\infty) - \int_P^\infty\omega(P')\right]
\left[\left(\frac{1}{2}y_\infty(P)+\bo(y_\infty(P)^3)\right)
\left(-2\frac{dy_\infty(P)}{y_\infty(P)^3}\right)\right]\\
&=\left[A(\infty)-2cy_\infty(P) + \bo(y_\infty(P)^2)\right]
\left[\left(\frac{1}{2}y_\infty(P)+\bo(y_\infty(P)^3)\right)
\left(-2\frac{dy_\infty(P)}{y_\infty(P)^3}\right)\right]\\
&=\left[-\frac{A(\infty)}{y_\infty(P)^2} +\frac{2c}{y_\infty(P)} + \bo(1)
\right]\,dy_\infty(P),\quad P\to\infty,
\end{split}
\end{equation}
so by the Residue Theorem,
\begin{equation}
I=\oint_{\partial\tilde{X}}\Gamma(P) = 4\pi ic.
\label{eq:IresiduesB}
\end{equation}
On the other hand, we may evaluate $I$ directly:
\begin{equation}
I=\oint_{\partial\tilde{X}}\Gamma(P) = \int_a \left[A_\mathrm{right}(P)-A_\mathrm{left}(P)\right]\Omega(P) +
\int_b\left[A_\mathrm{top}(P)-A_\mathrm{bottom}(P)\right]\Omega(P),
\end{equation}
where the subscripts on $A$ indicate where the corresponding points $P$ live
on the diagram of $\tilde{X}$ shown in Figure~\ref{fig:canonicaldissection}.
But clearly,
\begin{equation}
A_\mathrm{right}(P)-A_\mathrm{left}(P) = -\oint_b\omega(P') = -\mathcal{H}\quad
\text{and}\quad
A_\mathrm{top}(P)-A_\mathrm{bottom}(P) = \oint_a\omega(P') = 2\pi i.
\end{equation}
Therefore,
\begin{equation}
I=-\mathcal{H}\oint_a\Omega(P) + 2\pi i\oint_b\Omega(P),
\end{equation}
and due to the choice of the constant $C$, we simply find
\begin{equation}
I=2\pi i\kappa.
\label{eq:IintegrateB}
\end{equation}
Comparing \eqref{eq:IresiduesB} with \eqref{eq:IintegrateB}, we see that
$ \kappa= 2c$, as desired.
\end{proof}
It then follows from \eqref{eq:h1defB}, \eqref{eq:holodiffB}, 
and the normalization condition \eqref{eq:holodiffnormB} that
\begin{equation}
p\kappa = \frac{1}{\pi i}\int_{w_0^*}^{w_0}\frac{c\,ds}{S_+(s)}
= -\frac{1}{2\pi i}\oint_a\omega = -1.
\label{eq:h1PhiB}
\end{equation}
We now give Krichever's formula \cite{Krichever77} 
for the Baker-Akhiezer functions $t_j(P)$:
\begin{equation}
\begin{split}
t_j(P) := &N_j\frac{\Theta(A(P)+\mathcal{K}+ip\kappa\varphi_j;\mathcal{H})}
{\Theta(A(P)+\mathcal{K};\mathcal{H})}
\exp\left(ip\varphi_j\int_{w_0}^P\Omega(P')\right)\\
{}=&N_j\frac{\Theta(A(P)+\mathcal{K}-i\varphi_j;\mathcal{H})}{\Theta(A(P)+\mathcal{K};\mathcal{H})}
\exp\left(ip\varphi_j\int_{w_0}^P\Omega(P')\right),
\end{split}
\label{eq:KricheverFormulaB}
\end{equation}
where $N_j$ is a normalizing constant.  Here, the path in the exponent
is intended to be the same as that in the Abel map $A(P)$.  This
expression is well-defined in spite of the indeterminacy of the path
of integration due to this identification of paths and the two automorphic
identities \eqref{eq:thetaperiodic} and \eqref{eq:thetamonodromy}
satisfied by $\Theta$.  Indeed: 
\begin{itemize}
\item Adding an $a$-cycle to the path does not change
the exponent due to \eqref{eq:secondkindnormB}, and the $\Theta$-functions
are also unchanged due to \eqref{eq:holodiffnormB} and \eqref{eq:thetaperiodic}.
\item
Adding a $b$-cycle to the path adds $-i\varphi_j$ to the exponent according to
\eqref{eq:PhidefB} and \eqref{eq:h1PhiB}, but this is compensated
by the ratio of $\Theta$-functions according to \eqref{eq:wHdefineB}
and \eqref{eq:thetamonodromy}.
\end{itemize}
Moreover, since by our choice of base point
$A(P)$ vanishes to first order in the holomorphic local parameter 
$y_0(P)$ when $P=w_0$,
it is clear that $t_j(P)$ may have a simple pole at this point.  Also,
from the asymptotic behavior \eqref{eq:OmegaasympB} of $\Omega(P)$ it 
is clear that if the constant $N_j$ is chosen correctly $t_j(P)$ 
will have the desired asymptotic behavior \eqref{eq:BAasympB} as $P\to\infty$.  

To compute $N_j$, let us select a path from $P=w_0$ to $P=\infty$ that
lies entirely on sheet 1 of $X$ and along which $\Im\{w(P)\}\ge 0$.  This
unambiguously determines $A(\infty)$ as well as the constant term in
the asymptotic expansion of the exponent.  It is easy to verify for such
a path that
\begin{equation}
A(\infty)=-\frac{1}{2}\oint_b\omega(P)=-\frac{1}{2}\mathcal{H}.
\label{eq:AinfinityB}
\end{equation}
Also, as $P$ tends to $P=\infty$ along such a path (which we take to lie
along the negative real axis for large $w$),
\begin{equation}
\begin{split}
\int_{w_0}^P\Omega(P')&=\frac{1}{2}\oint_a\Omega(P')-\frac{1}{2}\oint_b\Omega(P')
+\int_0^{w(P)}\frac{w+C}{2S_+(w)}\,dw\\
&=-\frac{1}{2}\kappa +\int_0^{w(P)}\frac{w+C}{2S_+(w)}\,dw,
\end{split}
\end{equation}
and
\begin{equation}
\begin{split}
\int_0^{w(P)}\frac{w+C}{2S_+(w)}\,dw &= \int_0^{w(P)}\left[
\frac{w+C}{2S_+(w)}-\frac{1}{2i}(-w)^{-1/2}\right]\,dw + i(-w(P))^{1/2}\\
&=\int_0^{w(P)}\left[
\frac{w+C}{2S_+(w)}-\frac{1}{2i}(-w)^{-1/2}\right]\,dw +\frac{1}{y_\infty(P)}.
\end{split}
\end{equation}
The remaining integrand is integrable at infinity, so by doubling the
contour of integration along the branch cut on the negative real axis
and then closing the contour in the right half-plane we may write
\begin{equation}
\begin{split}
\int_0^{w(P)}\left[
\frac{w+C}{2S_+(w)}-\frac{1}{2i}(-w)^{-1/2}\right]\,dw &=
\int_0^{-\infty}\left[
\frac{w+C}{2S_+(w)}-\frac{1}{2i}(-w)^{-1/2}\right]\,dw +\bo(y_\infty(P))\\
&=-\frac{1}{2}\oint_a\left[\frac{w(P)+C}{2\tilde{S}(P)} -\frac{1}{2\sqrt{w(P)}}
\right]\,dw(P) +\bo(y_\infty(P))\\
&=-\frac{1}{2}\oint_a\Omega(P) +\frac{1}{4}\oint_a\frac{dw(P)}{\sqrt{w(P)}} + 
\bo(y_\infty(P))\\
&=\bo(y_\infty(P)).
\end{split}
\end{equation}
Therefore,
\begin{equation}
\int_{w_0}^P\Omega(P')=\frac{1}{y_\infty(P)}-\frac{1}{2}\kappa+\bo(y_\infty(P)),
\end{equation}
and so we find
\begin{equation}
\lim_{P\to\infty}t_j(P)e^{-i\varphi_jpy_\infty(P)^{-1}}=N_j
\frac{\Theta(i\pi-i\varphi_j;\mathcal{H})}{\Theta(i\pi;\mathcal{H})}e^{i\varphi_j/2},
\end{equation}
It follows that the normalization constant is simply
\begin{equation}
N_j:=\frac{\Theta(i\pi;\mathcal{H})}{\Theta(i\pi-i\varphi_j;\mathcal{H})}e^{-i\varphi_j/2}.
\label{eq:NjdefB}
\end{equation}
This completes the construction of $t_j(P)$, and hence of
$\dot{\mathbf{O}}^\mathrm{out}(w)$.  To obtain an especially useful
formula for $\dot{\mathbf{O}}^\mathrm{out}(w)$, we need another
intermediate result.  The function
\begin{equation}
l(w):=p\int_{w_0}^w\frac{w'+C}{2S(w')}\,dw'
\label{eq:lfuncdefine}
\end{equation}
is well-defined for $w$ in the cut plane (the first sheet of $X$ as illustrated
in Figure~\ref{fig:wHomologyB}) due to the condition \eqref{eq:secondkindnormB},
which makes the integral independent of path.  
\begin{lemma}
$l(w)=k(w):=h(w)+\tfrac{1}{2}$.
\label{lem:lkB}
\end{lemma}
\begin{proof}
Clearly, $l$ and $k$ are analytic in the same domain, and are both
uniformly bounded on bounded subsets of this domain.  Since $S$
changes sign across the cut connecting $w_0$ and $w_0^*$ it is obvious
that along this cut, $l_+(w)+l_-(w)=k_+(w)+k_-(w)=0$.  Now because the
integrand of $l$ is the restriction of the differential $\Omega(P)$ to
the first sheet, it is easy to see that $l(0)=-\tfrac{1}{2}p\kappa$,
so by \eqref{eq:h1PhiB}, in fact $l(0)=\tfrac{1}{2}$.  Again
using the fact that $S$ changes sign across the negative real axis, we
see that along this cut $l_+(w)+l_-(w)=2l(0)=k_+(w)+k_-(w)=1$.  Also,
both $l(w)$ and $k(w)$ have the same leading asymptotic behavior as
$w\to \infty$, namely, $pw^{1/2}+\bo(w^{-1/2})$.  Setting
$m(w):=(l(w)-k(w))/S(w)$ we see that $m(w)$ is a function that is
analytic where $S$ is, with at worst inverse square-root singularities
at $w=0,w_0,w_0^*$.  Moreover, along the open arcs of discontinuity of
$S$, we have $m_+(w)\equiv m_-(w)$, so $m$ extends continuously to these arcs.
It follows that $m(w)$ is necessarily an entire function of $w$.  Then,
since $m(w)=\bo(w^{-2})$ as $w\to\infty$, $m(w)\equiv 0$ by Liouville's Theorem,
and the result follows.
\end{proof}

We now present a formula for $\dot{\mathbf{O}}^\mathrm{out}(w)$.  
For $w$ in the cut plane,
let $P_k(w)$ denote the preimage of $w$ under $w(P)$ on sheet $k$ of $X$.  Then,
by composing the transformations leading from $\dot{\mathbf{O}}^\mathrm{out}(w)$
 to the
Baker-Akhiezer functions $t_j(P)$, we find that for $\Im\{w\}>0$,
\begin{equation}
\dot{\mathbf{O}}^\mathrm{out}(w) = 
\begin{bmatrix}\displaystyle\tfrac{q}{2}
\left(t_1(P_1)e^{-i\varphi_1h}+t_2(P_1)e^{-i\varphi_2h}\right) &\displaystyle
\tfrac{q}{2i}\left(t_1(P_2)e^{i\varphi_1h}-t_2(P_2)e^{i\varphi_2h}\right)\\\\
\displaystyle\tfrac{q}{2i}\left(t_2(P_1)e^{-i\varphi_2h}-t_1(P_1)
e^{-i\varphi_1h}
\right) &
\displaystyle\tfrac{q}{2}\left(t_1(P_2)e^{i\varphi_1h}+t_2(P_2)e^{i\varphi_2h}
\right)
\end{bmatrix}
\label{eq:OdotBup}
\end{equation}
while for $\Im\{w\}<0$,
\begin{equation}
\dot{\mathbf{O}}^\mathrm{out}(w) = 
\begin{bmatrix}
\displaystyle\tfrac{q}{2}\left(t_1(P_2)e^{i\varphi_1h}+t_2(P_2)e^{i\varphi_2h}
\right)
&\displaystyle
\tfrac{q}{2i}\left(t_1(P_1)e^{-i\varphi_1h}-t_2(P_1)e^{-i\varphi_2h}
\right) 
\\\\
\displaystyle
\tfrac{q}{2i}\left(t_2(P_2)e^{i\varphi_2h}-t_1(P_2)e^{i\varphi_1h}\right)
&
\displaystyle\tfrac{q}{2}
\left(t_1(P_1)e^{-i\varphi_1h}+t_2(P_1)e^{-i\varphi_2h}\right) 
\end{bmatrix}.
\label{eq:OdotBdown}
\end{equation}
The dependence on $w$ enters these formulae via 
$q=q(w)$, $P_j=P_j(w)$, and $h=h(w)$.

To simplify $t_j(P_1(w))e^{-i\varphi_jh(w)}$,
we evaluate the formula \eqref{eq:KricheverFormulaB} by selecting 
in both the Abel map and the integral in the exponent 
a path from $w_0$ to $P_1(w)$ lying entirely in the finite $w$-plane 
on sheet $1$ of $X$.  
Therefore, using Lemma~\ref{lem:lkB} and \eqref{eq:NjdefB},
\begin{equation}
\begin{split}
t_j(P_1(w))e^{-i\varphi_jh(w)}&=N_j
\frac{\Theta(A(P_1(w))+\mathcal{K}-i\varphi_j;\mathcal{H})}
{\Theta(A(P_1(w))+\mathcal{K};\mathcal{H})}e^{i\varphi_jl(w)-i\varphi_jh(w)}\\ & = 
\frac{\Theta(i\pi;\mathcal{H})\Theta(A(P_1(w))+\mathcal{K}-i\varphi_j;\mathcal{H})}
{\Theta(i\pi-i\varphi_j;\mathcal{H})\Theta(A(P_1(w))+\mathcal{K};\mathcal{H})}.
\end{split}
\label{eq:tjP1B}
\end{equation}
To evaluate $t_j(P_2(w))e^{i\varphi_jh(w)}$, we proceed similarly, choosing a path
from $w_0$ to $P_2(w)$ lying in the finite $w$-plane on sheet $2$ of $X$
to obtain
\begin{equation}
\begin{split}
t_j(P_2(w))e^{i\varphi_jh(w)}&=N_j\frac{\Theta(A(P_2(w))+\mathcal{K}-i\varphi_j;\mathcal{H})}
{\Theta(A(P_2(w))+\mathcal{K};\mathcal{H})}e^{-i\varphi_jl(w)+i\varphi_jh(w)}\\ &=
\frac{\Theta(i\pi;\mathcal{H})\Theta(A(P_2(w))+\mathcal{K}-i\varphi_j;\mathcal{H})}
{\Theta(i\pi-i\varphi_j;\mathcal{H})\Theta(A(P_2(w))+\mathcal{K};\mathcal{H})}
e^{i\varphi_j}.
\end{split}
\end{equation}
This latter formula can be further simplified with the observation that
$A(P_2(w))=-A(P_1(w))$ because the base point $P_0$ has been chosen as the
branch point $w_0$, so with the use of the relations 
\eqref{eq:Thetaeven}--\eqref{eq:thetamonodromy} and the definition 
\eqref{eq:RiemannConstantB} of $\mathcal{K}$ we obtain ultimately
\begin{equation}
t_j(P_2(w))e^{i\varphi_jh(w)}=
\frac{\Theta(i\pi;\mathcal{H})\Theta(A(P_1(w))+\mathcal{K}+i\varphi_j;\mathcal{H})}
{\Theta(i\pi+i\varphi_j;\mathcal{H})\Theta(A(P_1(w))+\mathcal{K};\mathcal{H})}.
\label{eq:tjP2B}
\end{equation}
In the formulae \eqref{eq:tjP1B} and \eqref{eq:tjP2B}, $A(P_1(w))$ represents 
any value of the integral
\begin{equation}
A(P_1(w)):=\int_{w_0}^w\frac{c\,d\xi}{S(\xi)}.
\label{eq:AP1Librational}
\end{equation}
This essentially completes the proof of Proposition~\ref{prop:outer}
in case \librational.  Indeed, 
the uniform bounds on $\dot{\mathbf{O}}^\mathrm{out}(w)$ 
follow from the formulae \eqref{eq:OdotBup},
\eqref{eq:OdotBdown}, \eqref{eq:tjP1B}, and \eqref{eq:tjP2B} upon 
noting that the only
dependence on $\nu$ occurs via the angles $\varphi_j$ defined 
by \eqref{eq:GjB}, and that $\Theta$ satisfies the periodicity relation
\eqref{eq:thetaperiodic}.  That $\det(\dot{\mathbf{O}}^\mathrm{out}(w))=1$ 
is a consequence
of the fact that the jump matrices have determinant equal to $1$, via a 
standard Liouville argument.

\subsection{The outer parametrix in case \rotational.  Proof of Proposition~\ref{prop:outer} in this case}
\subsubsection{Solution of Riemann-Hilbert Problem~\ref{rhp:wOdotrotational}
in terms of Baker-Akhiezer functions}
The contour of discontinuity of $\dot{\mathbf{O}}^\mathrm{out}(w)$ as illustrated
in Figure~\ref{fig:wOdotK} divides the complementary region into a
bounded component $\Upsilon_0$ and an unbounded component $\Upsilon_\infty$.  The first step
in solving Riemann-Hilbert Problem~\ref{rhp:wOdotrotational} is to make an
explicit substitution to simplify the contour.  We therefore define a
new unknown $\mathbf{P}(w)$ in terms of $\dot{\mathbf{O}}^\mathrm{out}(w)$ by
\begin{equation}
\mathbf{P}(w):=\begin{cases}
\dot{\mathbf{O}}^\mathrm{out}(w),\quad & \text{$\Im\{w\}>0$ and $w\in \Upsilon_\infty$}\\
\dot{\mathbf{O}}^\mathrm{out}(w)i\sigma_1e^{i\nu\sigma_3},\quad
&\text{$\Im\{w\}>0$ and $w\in \Upsilon_0$}\\
\sigma_2\dot{\mathbf{O}}^\mathrm{out}(w)\sigma_2,\quad &\text{$\Im\{w\}<0$ and $w\in \Upsilon_\infty$}\\
\sigma_2\dot{\mathbf{O}}^\mathrm{out}(w)\sigma_3e^{i\nu\sigma_3},\quad
&\text{$\Im\{w\}<0$ and $w\in \Upsilon_0$}.
\end{cases}
\end{equation}
This substitution has the effect of collapsing the contour to the real
axis and removing the jump discontinuity along $\mathbb{R}_+$.  The
matrix $\mathbf{P}(w)$ may be analytically continued to the
domain $\mathbb{C}\setminus (-\infty,0]$, and its boundary values on
the negative real axis are necessarily continuous except at the points
$w_0$ and $w_1$ at which negative one-fourth power singularities are
admitted.  The jump conditions satisfied by $\mathbf{P}(w)$ are
illustrated in Figure~\ref{fig:wPdotK}.
\begin{figure}[h]
\begin{center}
\includegraphics{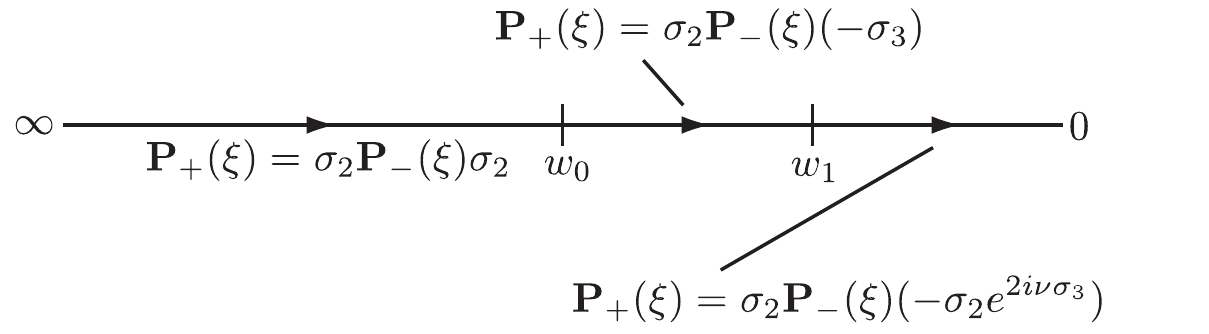}
\end{center}
\caption{\emph{The jump conditions satisfied by the matrix
    $\mathbf{P}(w)$ normalized as
    $\mathbf{P}(w)=\mathbb{I}+\bo(|w|^{-1/2})$ as
    $w\to\infty$.}}
\label{fig:wPdotK}
\end{figure}

Next, we remove the real parameter $\nu$ from the jump
conditions by defining the scalar function $h(w)$ by
\begin{equation}
h(w):=-\frac{S(w)}{\pi i}\int_{w_1}^0\frac{ds}{S_+(s)(s-w)},
\label{eq:hdefineK}
\end{equation}
where $S(w)^2 := w(w-w_0)(w-w_1)$, $S(w)$ is analytic in the complement of
the real intervals $-\infty<w\le w_0$ and $w_1\le w\le 0$, and 
$S(w)=w^{3/2}(1+\bo(w^{-1}))$ as $w\to\infty$ (principal branch of $w^{3/2}$).  
The asymptotic behavior of $h(w)$ as $w\to\infty$ is given by
\begin{equation}
h(w)=pw^{1/2}+\bo(|w|^{-1/2}),\quad w\to\infty,
\end{equation}
where
\begin{equation}
p:=\frac{1}{\pi i}\int_{w_1}^0\frac{ds}{S_+(s)}.
\end{equation}
The function defined by \eqref{eq:hdefineK} is analytic exactly where
$S(w)$ is, and it satisfies $h_+(\xi)+h_-(\xi)=0$
for $-\infty<\xi<w_0$ and $h_+(\xi)+h_-(\xi)=-2$ for $w_1<\xi<0$,
taking continuous and bounded boundary values.  
In place of $\dot{\mathbf{P}}(w)$
we now take the new unknown defined by
\begin{equation}
\mathbf{Q}(w):=\mathbf{P}(w)e^{i\nu h(w)\sigma_3}.
\end{equation}
By direct calculation we obtain the jump conditions determining
$\mathbf{Q}(w)$ as shown in Figure~\ref{fig:wQdotK}.
\begin{figure}[h]
\begin{center}
\includegraphics{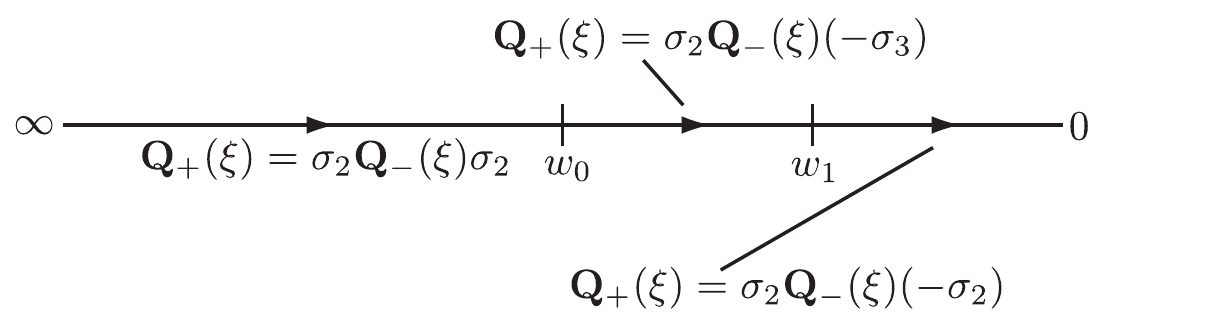}
\end{center}
\caption{\emph{The jump conditions satisfied by the matrix
    $\mathbf{Q}(w)$ normalized as
    $\mathbf{Q}(w)e^{-ip\nu
      w^{1/2}\sigma_3} = \mathbb{I}+\bo(|w|^{-1/2})$ as
    $w\to\infty$.}}
\label{fig:wQdotK}
\end{figure}
The boundary values taken by $\mathbf{Q}(w)$ are continuous except
at the points $w_0$ and $w_1$ at which inverse fourth-root singularities are
admitted.

We next diagonalize the prefactor of $\sigma_2$ with the use of the
eigenvector matrix $\mathbf{V}$ given in \eqref{eq:sigma2diag}.
Setting 
\begin{equation}
\mathbf{R}(w)=\mathbf{V}^\dagger\mathbf{Q}(w),
\end{equation}
one finds that the jump conditions on the negative real axis are reduced
to those shown in Figure~\ref{fig:wRdotK}, 
\begin{figure}[h]
\begin{center}
\includegraphics{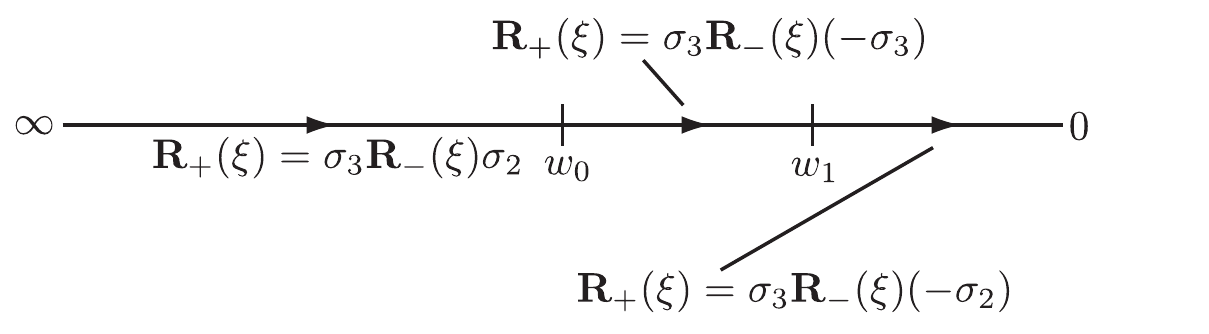}
\end{center}
\caption{\emph{The jump conditions satisfied by the matrix
    $\mathbf{R}(w)$ normalized as
    $\mathbf{R}(w)e^{-ip\nu
      w^{1/2}\sigma_3}=\mathbf{V}^\dagger +
    \bo(|w|^{-1/2})$ as $w\to\infty$.  }}
\label{fig:wRdotK}
\end{figure}
and the only discontinuities of the boundary values are once again inverse
fourth roots admitted at $w_0$ and $w_1$.

Now, let $q(w)$ be the function satisfying
\begin{equation}
q(w)^4=\frac{w-w_0}{w-w_1}
\end{equation}
and the normalization condition
$\lim_{w\to\infty}q(w)=1$, and taken to be analytic for 
$w\in\mathbb{C}\setminus[w_0,w_1]$.
Its boundary values taken on the branch cut $[w_0,w_1]$ are related by
$q_+(\xi)=-iq_-(\xi)$.  Note also that
\begin{equation}
q(0)=\left(\frac{w_0}{w_1}\right)^{1/4}>0.
\end{equation}
Now define a new unknown $\mathbf{S}(w)$ by
\begin{equation}
\mathbf{S}(w):=q(w)^{-1}\mathbf{R}(w).
\end{equation}
Taking into account the type of singularities that 
$\mathbf{R}(w)$ may have near $w=w_0$ and $w=w_1$, we see that
$\mathbf{S}(w)=\bo(|w-w_0|^{-1/2})$ for $w$ near $w_0$, while
$\mathbf{S}(w)$ is bounded near $w=w_1$.  The jump
conditions satisfied by $\mathbf{S}(w)$ on the negative real axis 
are illustrated in Figure~\ref{fig:wSdotK}.
\begin{figure}[h]
\begin{center}
\includegraphics{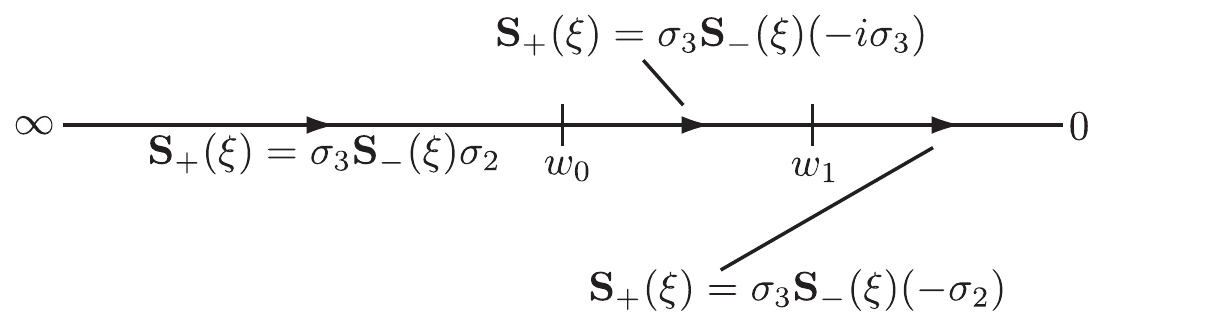}
\end{center}
\caption{\emph{The jump conditions satisfied by the matrix
    $\mathbf{S}(w)$ normalized as
    $\mathbf{S}(w)e^{-ip\nu
      w^{1/2}\sigma_3}=\mathbf{V}^\dagger +
    \bo(|w|^{-1/2})$ as $w\to\infty$.  Unlike $\dot{\mathbf{O}}^\mathrm{out}(w)$,
$\mathbf{P}(w)$, $\mathbf{Q}(w)$, and
$\mathbf{R}(w)$, the matrix $\mathbf{S}(w)$ is required to
be bounded in a neighborhood of $w=w_1$ while we admit a stronger singularity at $w=w_0$:
$\mathbf{S}(w)=\bo(|w-w_0|^{-1/2})$.}}
\label{fig:wSdotK}
\end{figure}

We now separate the rows of the matrix $\mathbf{S}(w)$ by writing
\begin{equation}
\mathbf{S}(w)=\begin{bmatrix}\trans{\mathbf{s}_1(w)}\\
\trans{\mathbf{s}_2(w)}\end{bmatrix},
\end{equation}
and introduce two new row vectors $\trans{\mathbf{t}_1(w)}$
and $\trans{\mathbf{t}_2(w)}$ by setting
\begin{equation}
\trans{\mathbf{t}_1(w)}=\sqrt{2}
\trans{\mathbf{s}_1(w)}e^{-i\pi k(w)\sigma_3/2}\quad
\text{and}\quad
\trans{\mathbf{t}_2(w)}=i\sqrt{2}\trans{\mathbf{s}_2(w)}e^{i\pi k(w)\sigma_3/2}
\end{equation}
where $k(w)$ is the function analytic for $w\in\mathbb{C}\setminus\mathbb{R}_-$
given by
\begin{equation}
k(w):=\frac{1}{2}+h(w)-\frac{S(w)}{2\pi i}\int_{w_0}^{w_1}\frac{ds}{S(s)(s-w)}.
\end{equation}
Note that $k(w)$ is bounded on compact sets in the $w$-plane and
has the asymptotic behavior
\begin{equation}
k(w)=\left[p+\frac{1}{2\pi i}\int_{w_0}^{w_1}\frac{ds}{S(s)}\right]w^{1/2} + \frac{1}{2} 
+ \bo(|w|^{-1/2}),\quad w\to\infty.
\end{equation}
Also, the jump conditions satisfied by $k$ on the negative real axis
are as follows:  $k_+(\xi)+k_-(\xi)=1$ for $\xi<w_0$, $k_+(\xi)-k_-(\xi)=-1$
for $w_0<\xi<w_1$, and $k_+(\xi)+k_-(\xi)=-1$ for $w_1<\xi<0$.

From this information it follows that $\trans{\mathbf{t}_j(w)}$ are analytic
for $w\in\mathbb{C}\setminus ((-\infty,w_0]\cup[w_1,0])$, and they satisfy
the involutive jump conditions
\begin{equation}
\trans{\mathbf{t}_{j+}(\xi)}=\trans{\mathbf{t}_{j-}(\xi)}\sigma_1,\quad
j=1,2,
\end{equation}
for $\xi$ in either of the two intervals of discontinuity, and we also 
have the normalization conditions
\begin{equation}
\trans{\mathbf{t}_j(w)}e^{-ip\varphi_j w^{1/2}\sigma_3}
=[1,1]+\bo(|w|^{-1/2}),\quad w\to\infty
\end{equation}
where $\varphi_1$ and $\varphi_2$ are given by 
\begin{equation}
\varphi_1:=\nu-\frac{\pi}{2}-\frac{1}{4ip}
\int_{w_0}^{w_1}\frac{ds}{S(s)}\quad
\text{and}\quad\varphi_2:=
\nu+\frac{\pi}{2}+
\frac{1}{4ip}\int_{w_0}^{w_1}\frac{ds}{S(s)}.
\end{equation}
Both $\trans{\mathbf{t}_1(w)}$ and $\trans{\mathbf{t}_2(w)}$ may become
unbounded in the finite $w$-plane only as $w\to w_0$, where all four
scalar components must be $\bo(|w-w_0|^{-1/2})$.

We are now in a position to identify the components of the row vectors
$\trans{\mathbf{t}_j(w)}$ as sheet projections of scalar Baker-Akhiezer
functions $t_j(P)$ onto the two sheets of the Riemann surface $X$ of
the equation $y^2=S(w)^2=w(w-w_0)(w-w_1)$ compactified at $y=w=\infty$.
Viewing $X$ as two copies (sheets) of the $w$-plane cut along the intervals
$(-\infty,w_0]$ and $[w_1,0]$ and appropriately glued together,  the
Baker-Akhiezer functions are then defined just as in case \librational:
\begin{equation}
t_j(P):=\begin{cases} [\trans{\mathbf{t}_j(w(P))}]_1,\quad &
P\in\text{sheet $1$}\\
[\trans{\mathbf{t}_j(w(P))}]_2,\quad & P\in\text{sheet $2$}.
\end{cases}
\end{equation}
These functions are analytic on $X$ except at exactly two points:  the branch
point $w=w_0$ at which $t_j(P)$ admits a simple pole (in the holomorphic
local coordinate $y_0(P)$ given by \eqref{eq:localcoordw0}), and the branch
point $w=\infty$ at which $t_j(P)$ has exponential behavior in terms of 
the holomorphic local coordinate $y_\infty(P)$ defined by \eqref{eq:kinfty}:
\begin{equation}
t_j(P)e^{-ip\varphi_jy_\infty(P)^{-1}}=1+\bo(y_\infty(P)),\quad 
P\to\infty.
\label{eq:tjnormK}
\end{equation}

\subsubsection{Construction of the Baker-Akhiezer functions}
To write down Krichever's formula for $t_j(P)$, we define a homology
basis on $X$ as shown in Figure~\ref{fig:rotationalhomology}.
\begin{figure}[h]
\begin{center}
\includegraphics{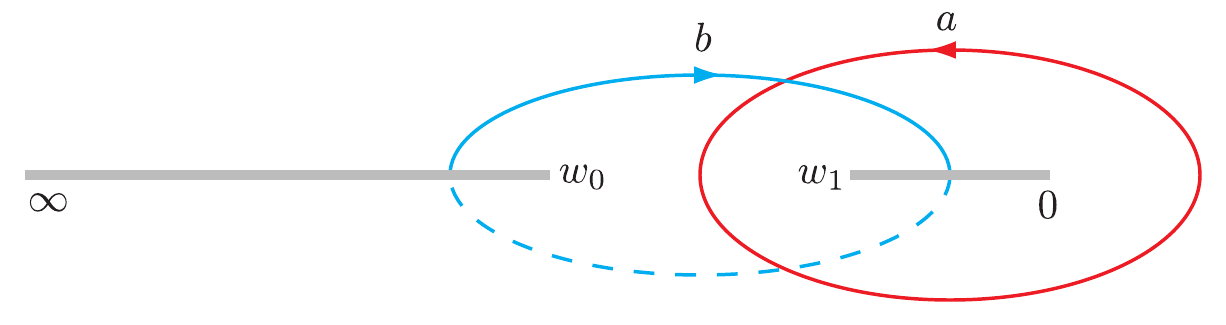}
\end{center}
\caption{\emph{A homology basis for the Riemann surface $X$.  Solid
    curves lie on sheet 1 and dashed curves lie on sheet 2.}}
\label{fig:rotationalhomology}
\end{figure}
With this basis selected, we define a holomorphic differential
$\omega(P)$, corresponding constants $\mathcal{H}$ and $\mathcal{K}$, Riemann
$\Theta$-function $\Theta(z;\mathcal{H})$, the Abel mapping $A(P)$ with base point
$P_0=w_0$, meromorphic differential $\Omega(P)$, and corresponding
constant $\kappa$ by exactly the same sequence of formulae as in case
\librational, namely \eqref{eq:holodiffB}--\eqref{eq:PhidefB}.
Lemma~\ref{lem:dissection} also holds in the current context, with the
implication that
\begin{equation}
p\kappa = \frac{2}{\pi i}\int_{w_1}^0\frac{c\,ds}{S_+(s)}=-\frac{1}{\pi i}
\oint_a\omega(P) = -2.
\label{eq:h1PhiK}
\end{equation}
Also, we see that $\varphi_j$ can be expressed in terms of $\mathcal{H}$ as follows:
\begin{equation}
\varphi_1 = \nu-\frac{\pi}{2}-\frac{i\mathcal{H}}{8}\quad
\text{and}\quad
\varphi_2 = \nu+\frac{\pi}{2}+\frac{i\mathcal{H}}{8}.
\end{equation}
The Krichever formula for the Baker-Akhiezer functions $t_j(P)$ is then
\begin{equation}
\begin{split}
t_j(P) := &N_j\frac{\Theta(A(P)+\mathcal{K}+ip\kappa\varphi_j;\mathcal{H})}
{\Theta(A(P)+\mathcal{K};\mathcal{H})}
\exp\left(ip\varphi_j\int_{w_0}^P\Omega(P')\right)\\
{}=&N_j\frac{\Theta(A(P)+\mathcal{K}-2i\varphi_j;\mathcal{H})}{\Theta(A(P)+\mathcal{K};\mathcal{H})}
\exp\left(ip\varphi_j\int_{w_0}^P\Omega(P')\right),
\end{split}
\label{eq:KricheverFormulaK}
\end{equation}
where $N_j$ is a constant chosen to enforce the normalization condition
\eqref{eq:tjnormK}.

To compute the normalization
constants, choose a path from $P=w_0$ to $P=\infty$ on sheet 1 with
$\Im\{w\}\ge 0$, and obtain
\begin{equation}
A(\infty)=-\frac{1}{2}\oint_a\omega(P)=-\pi i.
\end{equation}
Also, as $P$ tends to $P=\infty$ along such a path (we may take $P$ on sheet 1 
with $w(P)>0$),
\begin{equation}
\begin{split}
\int_{w_0}^P\Omega(P')&=\frac{1}{2}\oint_b\Omega(P') -\frac{1}{2}\oint_a
\Omega(P') + \int_0^{w(P)}\frac{w+C}{2S(w)}\,dw\\
&=\frac{1}{2}\kappa + \int_0^{w(P)}\frac{w+C}{2S(w)}\,dw,
\end{split}
\end{equation}
and
\begin{equation}
\begin{split}
\int_0^{w(P)}\frac{w+C}{2S(w)}\,dw&=\int_0^{w(P)}\left[
\frac{w+C}{2S(w)}-\frac{1}{2}w^{-1/2}\right]\,dw + w(P)^{1/2}\\
&=
\int_0^{w(P)}\left[
\frac{w+C}{2S(w)}-\frac{1}{2}w^{-1/2}\right]\,dw + \frac{1}{y_\infty(P)}.
\end{split}
\end{equation}
The remaining integrand is integrable at infinity, and doubling the
contour of integration and rotating the contours through the upper and
lower half-planes respectively to lie along the negative real axis
gives
\begin{equation}
\begin{split}
\int_0^{w(P)}\left[\frac{w+C}{2S(w)}-\frac{1}{2}w^{-1/2}\right]\,dw&=
\int_0^\infty\left[\frac{w+C}{2S(w)}-\frac{1}{2}w^{-1/2}\right]\,dw +
\bo(y_\infty(P))\\ &=-\frac{1}{2}\oint_b\Omega(P')+\bo(y_\infty(P))\\
&=-\frac{1}{2}\kappa+\bo(y_\infty(P)),
\end{split}
\end{equation}
and therefore
\begin{equation}
\int_{w_0}^P\Omega(P')=\frac{1}{y_\infty(P)} + \bo(y_\infty(P)).
\end{equation}
It follows that
\begin{equation}
\lim_{P\to\infty}t_j(P)e^{-ip\varphi_jy_\infty(P)^{-1}}=N_j
\frac{\Theta(\tfrac{1}{2}\mathcal{H}-2i\varphi_j;\mathcal{H})}
{\Theta(\tfrac{1}{2}\mathcal{H};\mathcal{H})},
\end{equation}
so the normalization constants are given by
\begin{equation}
N_j:=\frac{\Theta(\tfrac{1}{2}\mathcal{H};\mathcal{H})}{\Theta(\tfrac{1}{2}\mathcal{H}-2i\varphi_j;\mathcal{H})},
\end{equation}
completing the construction of the Baker-Akhiezer functions $t_j(P)$, and
hence of $\dot{\mathbf{O}}^\mathrm{out}(w)$.  To obtain a useful expression for 
$\dot{\mathbf{O}}^\mathrm{out}(w)$ we recall the function $l(w)$ defined by the formula
\eqref{eq:lfuncdefine} for $w\in\mathbb{C}\setminus((-\infty,w_0]\cup[w_1,0])$;
the analogue of Lemma~\ref{lem:lkB} in the current context is then 
the following.
\begin{lemma}
$l(w)=h(w)$.
\label{lem:lhK}
\end{lemma}
\begin{proof}
  The proof is virtually the same as that of Lemma~\ref{lem:lkB}; one
  uses the identity \eqref{eq:h1PhiK} to show that $l(w)$ and $h(w)$
  satisfy the same additive jump conditions on the intervals
  $(-\infty,w_0)$ and $(w_1,0)$, and then uses this information to
  establish that the ratio $(l(w)-h(w))/S(w)$ is an entire function of $w$
  that vanishes as $w\to \infty$.
\end{proof}

To write formulae for $\dot{\mathbf{O}}^\mathrm{out}(w)$, it is convenient to introduce
the scalar function
\begin{equation}
n(w):=\frac{\mathcal{H}}{8}h(w)-\frac{S(w)}{4}\int_{w_0}^{w_1}\frac{ds}{S(s)(s-w)}.
\label{eq:nofwdef}
\end{equation}
Then, with $P_j(w)$ denoting the point on sheet $j$ of $X$ corresponding to $w$,
for $\Im\{w\}>0$ and $w\in \Upsilon_\infty$:
\begin{equation}
\dot{\mathbf{O}}^\mathrm{out}(w)=\begin{bmatrix}
\frac{q}{2}\left(t_1(P_1)e^{-i\varphi_1h+n}+t_2(P_1)e^{-i\varphi_2h-n}
\right) &
\frac{q}{2i}\left(t_1(P_2)e^{i\varphi_1h-n}-t_2(P_2)e^{i\varphi_2h+n}\right)
\\\\
\frac{q}{2i}\left(t_2(P_1)e^{-i\varphi_2h-n}-t_1(P_1)e^{-i\varphi_1h+n}\right)
&
\frac{q}{2}\left(t_1(P_2)e^{i\varphi_1h-n}+t_2(P_2)e^{i\varphi_2h+n}\right)
\end{bmatrix},
\label{eq:dotOrotationalformulaUup}
\end{equation}
for $\Im\{w\}>0$ and $w\in \Upsilon_0$:
\begin{equation}
\label{eq:dotOrotationalformulaBup}
\dot{\mathbf{O}}^\mathrm{out}(w)=\begin{bmatrix}
\frac{q}{2}\left(t_1(P_2)e^{i\varphi_1h-n}-t_2(P_2)e^{i\varphi_2h+n}\right)
e^{i\nu}
&-\frac{q}{2i}\left(t_1(P_1)e^{-i\varphi_1h+n}+t_2(P_1)e^{-i\varphi_2h-n}
\right)e^{-i\nu} 
\\\\
-\frac{q}{2i}\left(t_1(P_2)e^{i\varphi_1h-n}+t_2(P_2)e^{i\varphi_2h+n}\right)
e^{i\nu}
&
\frac{q}{2}\left(t_2(P_1)e^{-i\varphi_2h-n}-t_1(P_1)e^{-i\varphi_1h+n}\right)
e^{-i\nu}
\end{bmatrix},
\end{equation}
for $\Im\{w\}<0$ and $w\in \Upsilon_0$:
\begin{equation}
\dot{\mathbf{O}}^\mathrm{out}(w)=\begin{bmatrix}
\frac{q}{2}\left(t_2(P_1)e^{-i\varphi_2h-n}-t_1(P_1)e^{-i\varphi_1h+n}\right)
e^{-i\nu}
&\frac{q}{2i}\left(t_1(P_2)e^{i\varphi_1h-n}+t_2(P_2)e^{i\varphi_2h+n}\right)
e^{i\nu}
\\\\
\frac{q}{2i}\left(t_1(P_1)e^{-i\varphi_1h+n}+t_2(P_1)e^{-i\varphi_2h-n}
\right)e^{-i\nu} 
&
\frac{q}{2}\left(t_1(P_2)e^{i\varphi_1h-n}-t_2(P_2)e^{i\varphi_2h+n}\right)
e^{i\nu}
\end{bmatrix},
\end{equation}
and for $\Im\{w\}<0$ and $w\in \Upsilon_\infty$:
\begin{equation}
\dot{\mathbf{O}}^\mathrm{out}(w)=\begin{bmatrix}
\frac{q}{2}\left(t_1(P_2)e^{i\varphi_1h-n}+t_2(P_2)e^{i\varphi_2h+n}\right)&
-\frac{q}{2i}\left(t_2(P_1)e^{-i\varphi_2h-n}-t_1(P_1)e^{-i\varphi_1h+n}\right)
\\\\
-\frac{q}{2i}\left(t_1(P_2)e^{i\varphi_1h-n}-t_2(P_2)e^{i\varphi_2h+n}\right)
&
\frac{q}{2}\left(t_1(P_1)e^{-i\varphi_1h+n}+t_2(P_1)e^{-i\varphi_2h-n}
\right) 
\end{bmatrix}.
\end{equation}
It is then easy to check from these formulae that the exponentials
$e^{\pm i\varphi_jh(w)}$ are exactly cancelled out by corresponding
exponentials in the Baker-Akhiezer functions; indeed assuming the path
of integration to be confined to the half-plane containing $w$ for
$\Im\{w\}\neq 0$,
\begin{equation}
t_j(P_1(w))e^{-i\varphi_jh(w)}=\frac{\Theta(\tfrac{1}{2}\mathcal{H};\mathcal{H})\Theta(A(P_1(w))+\mathcal{K}-2i\varphi_j;\mathcal{H})}{\Theta(\tfrac{1}{2}\mathcal{H}-2i\varphi_j;\mathcal{H})\Theta(A(P_1(w))+\mathcal{K};\mathcal{H})}
\label{eq:tjP1K}
\end{equation}
and
\begin{equation}
t_j(P_2(w))e^{i\varphi_jh(w)}=
\frac{\Theta(\tfrac{1}{2}\mathcal{H};\mathcal{H})\Theta(A(P_2(w))+\mathcal{K}-2i\varphi_j;\mathcal{H})}{\Theta(\tfrac{1}{2}\mathcal{H}-2i\varphi_j;\mathcal{H})\Theta(A(P_2(w))+\mathcal{K};\mathcal{H})},
\label{eq:tjP2K}
\end{equation}
due to Lemma~\ref{lem:lhK}.
Therefore, 
all remaining dependence
on $\nu$ is either within the arguments of the $\Theta$-functions or in 
the exponential factors $e^{\pm i\nu}$, leading in both
cases to bounded oscillations as $\nu\to \infty$.
This essentially completes the proof of Proposition~\ref{prop:outer} in case
\rotational.

\subsection{Recovery of the potentials.  Proof of Proposition~\ref{prop:outerelliptic}}
Here we obtain explicit formulae for $\dot{C}$, $\dot{S}$,
and $\dot{G}$ as defined from $\dot{\mathbf{O}}^\mathrm{out}(w)$
by \eqref{eq:dotClibrational}--\eqref{eq:dotvlibrational}.

\subsubsection{Formulae for $\dot{C}$, $\dot{S}$, and $\dot{G}$
in case \librational}
To obtain such formulae in case \librational, 
it is enough to evaluate the asymptotic behavior
of $\dot{\mathbf{O}}^\mathrm{out}(w)$ for $\Im\{w\}>0$ using \eqref{eq:OdotBup}.
The value of $A(P_1(w))$ for $w\approx
0$ and $w\approx\infty$ may be determined 
by choosing the path in \eqref{eq:AP1Librational}
to lie in the upper half-plane, avoiding (except of course at the endpoints
$w=w_0$ and $w=0$ or $w=\infty$)
the branch cuts of $S$.  Then it is easy to see that
\begin{equation}
A(P_1(0))=\frac{1}{2}\oint_a\omega(P)-\frac{1}{2}\oint_b\omega(P) = 
\pi i-\frac{1}{2}\mathcal{H},
\end{equation}
and from \eqref{eq:AinfinityB} we have $A(P_1(\infty))=-\mathcal{H}/2$.
Also, using the fact that $w^{1/2}=i\sqrt{-w}$ (principal branches) 
for $\Im\{w\}>0$, we have
\begin{equation}
\begin{split}
A(P_1(w))-A(P_1(0))&=\int_0^w\frac{c\,d\xi}{S(\xi)}\\
&=\int_0^w\frac{c\,d\xi}{-|w_0|\xi^{1/2}(1+\bo(\xi))}\\
&=-\frac{2c}{|w_0|}w^{1/2} + 
\bo(w^{3/2})\\
&=-\frac{2ic}{|w_0|}\sqrt{-w} + \bo(w^{3/2}),\quad w\to 0,\quad\Im\{w\}>0
\end{split}
\end{equation}
and
\begin{equation}
\begin{split}
A(P_1(w))-A(P_1(\infty))&=\int_\infty^w\frac{c\,d\xi}{S(\xi)}\\
&=\int_{\infty}^w\frac{c\,d\xi}{\xi^{3/2}(1+\bo(\xi^{-1}))} \\
&= -2cw^{-1/2} + \bo(w^{-3/2})\\
&=\frac{2ic}{\sqrt{-w}} + \bo(w^{-3/2}),\quad w\to\infty,\quad \Im\{w\}>0.
\end{split}
\end{equation}
The final ingredients needed are the asymptotic formulae for $q(w)$:
\begin{equation}
q(w)=e^{i\zeta} + \bo(w),\quad 0<\zeta:=\frac{1}{2}\arg(w_0)<\frac{\pi}{2},
\quad w\to 0
\end{equation}
and
\begin{equation}
q(w)=1+\bo(w^{-1}),\quad w\to\infty.
\end{equation}

Using the Taylor expansion of the entire function $\Theta(z;\mathcal{H})$ with respect 
to $z$ and the identities 
\eqref{eq:Thetaeven} and \eqref{eq:thetaperiodic}, it then follows from 
substituting \eqref{eq:tjP1B} and \eqref{eq:tjP2B} into \eqref{eq:OdotBup}
and using $\varphi_2-\varphi_1=\pi$ that
\begin{equation}
\dot{O}^{0,0}_{11} = \dot{O}^{0,0}_{22}=
\frac{e^{i\zeta}}{2}\frac{\Theta(i\pi;\mathcal{H})}{\Theta(0;\mathcal{H})}
\left[\frac{\Theta(i\varphi_1;\mathcal{H})}{\Theta(i\varphi_2;\mathcal{H})} +
\frac{\Theta(i\varphi_2;\mathcal{H})}{\Theta(i\varphi_1;\mathcal{H})}\right]
\label{eq:Odot00B11}
\end{equation}
and
\begin{equation}
\dot{O}^{0,0}_{12}= -\dot{O}^{0,0}_{21}=\frac{e^{i\zeta}}{2i}
\frac{\Theta(i\pi;\mathcal{H})}{\Theta(0;\mathcal{H})}
\left[\frac{\Theta(i\varphi_1;\mathcal{H})}{\Theta(i\varphi_2;\mathcal{H})}-
\frac{\Theta(i\varphi_2;\mathcal{H})}{\Theta(i\varphi_1;\mathcal{H})}\right].
\label{eq:Odot00B12}
\end{equation}
Note that the matrix
$\dot{\mathbf{O}}^{0,0}=\dot{\mathbf{O}}^\mathrm{out}(0)$
is invariant under conjugation by
$\sigma_2$, as must be the case since by the conditions of Riemann-Hilbert
Problem~\ref{rhp:wOdotlibrational}, $\dot{\mathbf{O}}^\mathrm{out}(w)$ is to be
H\"older continuous up to $\mathbb{R}_+$.  
These formulae are sufficient to calculate $\dot{C}$ and
$\dot{S}$: substitution into \eqref{eq:dotClibrational} gives
\begin{equation}
\dot{C}=(-1)^{\#\Delta}\frac{e^{i\zeta}}{2}\frac{\Theta(i\pi;\mathcal{H})}{
\Theta(0;\mathcal{H})}\left[\frac{\Theta(i\varphi_2;\mathcal{H})}{\Theta(i\varphi_1;\mathcal{H})}+
\frac{\Theta(i\varphi_1;\mathcal{H})}{\Theta(i\varphi_2;\mathcal{H})}\right]
\label{eq:CNthetaslibrational}
\end{equation}
and substitution into \eqref{eq:dotSlibrational} gives
\begin{equation}
\dot{S}=(-1)^{\#\Delta}\frac{e^{i\zeta}}{2i}
\frac{\Theta(i\pi;\mathcal{H})}{\Theta(0;\mathcal{H})}\left[\frac{\Theta(i\varphi_2;\mathcal{H})}
{\Theta(i\varphi_1;\mathcal{H})}-\frac{\Theta(i\varphi_1;\mathcal{H})}{\Theta(i\varphi_2;\mathcal{H})}
\right].
\label{eq:SNthetaslibrational}
\end{equation}
The higher-order coefficients
that we need to calculate $\dot{G}_N(x,t)$ are:
\begin{equation}
\dot{O}_{22}^{0,1}= -\frac{ice^{i\zeta}}{|w_0|}
\frac{\Theta(i\pi;\mathcal{H})}{\Theta(0;\mathcal{H})}
\left[\frac{\Theta'(i\varphi_1;\mathcal{H})}{\Theta(i\varphi_2;\mathcal{H})}
+\frac{\Theta'(i\varphi_2;\mathcal{H})}{\Theta(i\varphi_1;\mathcal{H})}\right],
\end{equation}
\begin{equation}
\dot{O}_{12}^{0,1} =
-\frac{ce^{i\zeta}}{|w_0|}\frac{\Theta(i\pi;\mathcal{H})}{\Theta(0;\mathcal{H})}
\left[\frac{\Theta'(i\varphi_1;\mathcal{H})}{\Theta(i\varphi_2;\mathcal{H})}
-\frac{\Theta'(i\varphi_2;\mathcal{H})}{\Theta(i\varphi_1;\mathcal{H})}
\right],
\end{equation}
and
\begin{equation}
\dot{O}_{12}^{\infty,1} =c\left[
\frac{\Theta'(i\varphi_2;\mathcal{H})}{\Theta(i\varphi_2;\mathcal{H})} -
\frac{\Theta'(i\varphi_1;\mathcal{H})}{\Theta(i\varphi_1;\mathcal{H})} \right].
\label{eq:Odotinfinity1B12}
\end{equation}
Here by $\Theta'(z;\mathcal{H})$ we mean the partial derivative with respect to $z$ holding $\mathcal{H}$ fixed.
Substituting into
\eqref{eq:dotvlibrational} we have
\begin{equation}
\dot{G}=
c\left[1+\frac{e^{2i\zeta}}{|w_0|}\frac{\Theta(i\pi;\mathcal{H})^2}{\Theta(0;\mathcal{H})^2}
\right]\left[\frac{\Theta'(i\varphi_2;\mathcal{H})}{\Theta(i\varphi_2;\mathcal{H})}-
\frac{\Theta'(i\varphi_1;\mathcal{H})}
{\Theta(i\varphi_1;\mathcal{H})}\right].
\end{equation} 

\subsubsection*{Evaluation of $\mathcal{H}$.  Transformations of $\Theta$ with respect to $\mathcal{H}$}
We now seek to simplify the formulae for $\dot{C}$,
$\dot{S}$, and $\dot{G}$.  Of course one would like to
introduce Jacobi elliptic functions at this stage, but it turns out
that since $\mathcal{H}$ is complex the resulting formulae involve elliptic
parameters $m$ that are not in the so-called \emph{normal case} of
$0<m<1$.  To arrive at the simplest formulae in terms of Jacobi elliptic
functions with a parameter $m\in (0,1)$ we
will therefore first make some transformations of $\mathcal{H}$.  
Elementary contour definitions in the two integrals involved in the definition
\eqref{eq:wHdefineB} show that $\mathcal{H}$ can be written in the form
\begin{equation}
\mathcal{H}=\frac{1}{2}(\mathcal{H}_0 + 2\pi i),
\end{equation}
where $\mathcal{H}_0$ is given by the ratio of integrals
\begin{equation}
\mathcal{H}_0:=-2\pi \frac{\displaystyle\int_0^{+\infty}
\frac{dw}{\sqrt{w(w-w_0)(w-w_0^*)}}}{\displaystyle\int_{-\infty}^0\frac{dw}{\sqrt{-w(w-w_0)(w-w_0^*)}}},
\end{equation}
where the square roots are both positive.
Using the substitution
\begin{equation}
w=-|w_0|\frac{z-1}{z+1}\quad\text{followed by}\quad z=\pm\sqrt{1-s^2}
\end{equation}
and recalling the complete elliptic integral of the first kind defined by \eqref{eq:ellipticKdef} we obtain
\begin{equation}
\int_0^{+\infty}\frac{dw}{\sqrt{w(w-w_0)(w-w_0^*)}}=\frac{2K(\cos(\zeta)^2)}{\sqrt{|w_0|}}.
\end{equation}
Similarly by means of the substitution \eqref{eq:wsBsubstKm}
we find
\begin{equation}
\int_{-\infty}^0\frac{dw}{\sqrt{-w(w-w_0)(w-w_0^*)}}=\frac{2K(\sin(\zeta)^2)}{\sqrt{|w_0|}}.
\end{equation}
Therefore,
\begin{equation}
\mathcal{H}_0=-2\pi\frac{K(m')}{K(m)}
\label{eq:H0m}
\end{equation}
where  the elliptic parameter $m\in (0,1)$ and its complementary parameter $m':=1-m\in (0,1)$ are given by 
\begin{equation}
m:=\sin(\zeta)^2,\quad m':=\cos(\zeta)^2.
\label{eq:mzetalibrational}
\end{equation}
Note that $\mathcal{H}_0$ is a negative real number.
According to the theory of elliptic functions (see \cite{Akhiezer}) it
is Riemann $\Theta$-functions with parameter $\mathcal{H}_0$ that are associated with
Jacobi elliptic functions of elliptic parameter $m\in (0,1)$.  

Adding $2\pi i$ to $\mathcal{H}_0$ and then dividing by two amounts to the composition
of two classical transformations of $\Theta$-functions \cite{Akhiezer}.  
Indeed, the so-called \emph{first principal first-degree transformation}
implies that
\begin{equation}
\Theta(z;\mathcal{H}+2\pi i)=\Theta(z+i\pi;\mathcal{H}),\quad z\in\mathbb{C},\quad\Re\{\mathcal{H}\}<0,
\label{eq:1stprincipal1stdegree}
\end{equation}
an identity that allows us to add $2\pi i$ to $\mathcal{H}_0$.  Also,
the so-called \emph{Gauss transformation} (a second-degree transformation)
implies that
\begin{equation}
\Theta(i\pi;\tfrac{1}{2}\mathcal{H})\Theta(z;\tfrac{1}{2}\mathcal{H})=
\Theta(z+i\pi;\mathcal{H})^2+e^{-z}e^{\mathcal{H}/4}\Theta(z+i\pi-\tfrac{1}{2}\mathcal{H};\mathcal{H})^2,
\quad z\in\mathbb{C},\quad \Re\{\mathcal{H}\}<0,
\label{eq:Gauss1}
\end{equation}
and
\begin{equation}
\Theta(0;\tfrac{1}{2}\mathcal{H})\Theta(z+i\pi;\tfrac{1}{2}\mathcal{H})=
\Theta(z+i\pi;\mathcal{H})^2-e^{-z}e^{\mathcal{H}/4}\Theta(z+i\pi-\tfrac{1}{2}\mathcal{H};\mathcal{H})^2,
\quad z\in\mathbb{C},\quad \Re\{\mathcal{H}\}<0,
\label{eq:Gauss2}
\end{equation}
identities that allow us to divide $\mathcal{H}_0+2\pi i$ by two.  Combining
these together and using the identity \eqref{eq:thetaperiodic} yields
\begin{equation}
\frac{\Theta(i\pi;\mathcal{H})\Theta(z;\mathcal{H})}{\Theta(z+i\pi;\mathcal{H})\Theta(0;\mathcal{H})}=
\frac{\Theta(z;\mathcal{H}_0)^2+ie^{-z}e^{\mathcal{H}_0/4}\Theta(z+i\pi-\tfrac{1}{2}\mathcal{H}_0;\mathcal{H}_0)^2}
{\Theta(z;\mathcal{H}_0)^2-ie^{-z}e^{\mathcal{H}_0/4}\Theta(z+i\pi-\tfrac{1}{2}\mathcal{H}_0;\mathcal{H}_0)^2},\quad
\mathcal{H}=\frac{1}{2}(\mathcal{H}_0+2\pi i).
\label{eq:HH0thetaslibrational}
\end{equation}
The elliptic parameter $m$ 
can be expressed directly in terms of special values of $\Theta$-functions 
with parameter $\mathcal{H}_0$ as follows
\cite{Akhiezer}:
\begin{equation}
\frac{\Theta(i\pi;\mathcal{H}_0)^4}{\Theta(0;\mathcal{H}_0)^4}=1-m\quad\text{and}\quad
e^{\mathcal{H}_0/2}\frac{\Theta(-\tfrac{1}{2}\mathcal{H}_0;\mathcal{H}_0)^4}{\Theta(0;\mathcal{H}_0)^4}=m.
\label{eq:thetasmlibrational}
\end{equation}
Since $\mathcal{H}_0<0$ we see from \eqref{eq:RiemannTheta} 
that all four $\Theta$-function values appearing in this identity
are real, so we may take the positive square root of both sides.

We may now introduce the Jacobi elliptic functions \cite{Akhiezer}
with parameter $m$ related
to $\mathcal{H}_0$ by \eqref{eq:H0m} or equivalently by \eqref{eq:thetasmlibrational}:
\begin{equation}
\mathrm{sn}\left(\frac{K(m)z}{\pi i};m\right):=ie^{-z/2}\frac{\Theta(0;\mathcal{H}_0)
\Theta(z+i\pi-\tfrac{1}{2}\mathcal{H}_0;\mathcal{H}_0)}{\Theta(-\tfrac{1}{2}\mathcal{H}_0;\mathcal{H}_0)
\Theta(z+i\pi;\mathcal{H}_0)},\quad z\in\mathbb{C},
\label{eq:sndef}
\end{equation}
\begin{equation}
\mathrm{cn}\left(\frac{K(m)z}{\pi i};m\right):=e^{-z/2}
\frac{\Theta(i\pi;\mathcal{H}_0)\Theta(z-\tfrac{1}{2}\mathcal{H}_0;\mathcal{H}_0)}{\Theta(-\tfrac{1}{2}\mathcal{H}_0;\mathcal{H}_0)\Theta(z+i\pi;\mathcal{H}_0)},\quad z\in\mathbb{C},
\label{eq:cndef}
\end{equation}
and
\begin{equation}
\mathrm{dn}\left(\frac{K(m)z}{\pi i};m\right):=
\frac{\Theta(i\pi;\mathcal{H}_0)\Theta(z;\mathcal{H}_0)}
{\Theta(0;\mathcal{H}_0)\Theta(z+i\pi;\mathcal{H}_0)},\quad z\in\mathbb{C}.
\label{eq:dndef}
\end{equation}
Setting 
\begin{equation}
Z_j:=\frac{K(m)}{\pi}\varphi_j,
\label{eq:Zjdefine}
\end{equation}
and substituting into \eqref{eq:CNthetaslibrational} and 
\eqref{eq:SNthetaslibrational} firstly from 
\eqref{eq:HH0thetaslibrational}--\eqref{eq:thetasmlibrational} and then
from \eqref{eq:sndef}--\eqref{eq:dndef} we obtain:
\begin{equation}
\dot{C}=
(-1)^{\#\Delta}\frac{e^{i\zeta}}{2}
\left[\frac{\mathrm{dn}(Z_1;m)^2-i\sqrt{mm'}\mathrm{sn}(Z_1;m)^2}
{\mathrm{dn}(Z_1;m)^2+i\sqrt{mm'}\mathrm{sn}(Z_1;m)^2}+
\frac{\mathrm{dn}(Z_2;m)^2-i\sqrt{mm'}\mathrm{sn}(Z_2;m)^2}
{\mathrm{dn}(Z_2;m)^2+i\sqrt{mm'}\mathrm{sn}(Z_2;m)^2}\right]
\end{equation}
\begin{equation}
\dot{S}=
(-1)^{\#\Delta}\frac{e^{i\zeta}}{2i}
\left[\frac{\mathrm{dn}(Z_2;m)^2-i\sqrt{mm'}\mathrm{sn}(Z_2;m)^2}
{\mathrm{dn}(Z_2;m)^2+i\sqrt{mm'}\mathrm{sn}(Z_2;m)^2}-
\frac{\mathrm{dn}(Z_1;m)^2-i\sqrt{mm'}\mathrm{sn}(Z_1;m)^2}
{\mathrm{dn}(Z_1;m)^2+i\sqrt{mm'}\mathrm{sn}(Z_1;m)^2}\right].
\end{equation}
To express $\dot{G}$ in terms of Jacobi elliptic functions, first
we note that using $z=i\pi$ in \eqref{eq:HH0thetaslibrational}, taking
into account the periodicity relation \eqref{eq:thetaperiodic}, and then
using the positive square roots of the identities \eqref{eq:thetasmlibrational}
together with \eqref{eq:mzetalibrational} gives
\begin{equation}
e^{2i\zeta}\frac{\Theta(i\pi;\mathcal{H})^2}{\Theta(0;\mathcal{H})^2}=1.
\end{equation}
Now, by simple contour deformations and the use of the normalization
condition \eqref{eq:holodiffnormB} defining $c$, it follows that
$2c\mathcal{D}=\pi$, where $\mathcal{D}$ is the denominator defined
by \eqref{eq:deltalibrational}.
Therefore, if we allow the branch points $w_0$ and $w_0^*$ to depend on $(x,t)$ via the
moment and integral conditions $M=I=0$, it follows from \eqref{eq:kappapartials} of Proposition~\ref{prop:phasevelocity} that we
may write $\dot{G}$ in the form
\begin{equation}
\dot{G}=2\frac{\partial\Phi}{\partial t}
\left[\frac{\Theta'(i\varphi_2;\mathcal{H})}{\Theta(i\varphi_2;\mathcal{H})}-
\frac{\Theta'(i\varphi_1;\mathcal{H})}{\Theta(i\varphi_1;\mathcal{H})}\right].
\end{equation}
Since $\varphi_2=\varphi_1+\pi$, this can be written as a logarithmic derivative:
\begin{equation}
\dot{G}=-2i\frac{\partial\Phi}{\partial t}
\frac{d}{d\varphi_1}\log\left(\frac{\Theta(i\varphi_1+i\pi;\mathcal{H})}{\Theta(i\varphi_1;\mathcal{H})}\right).
\end{equation}
Applying \eqref{eq:HH0thetaslibrational} and the positive square roots
of \eqref{eq:thetasmlibrational}, we may then substitute from 
\eqref{eq:sndef}--\eqref{eq:dndef} to obtain
\begin{equation}
\dot{G}=-2i\frac{\partial\Phi}{\partial t}
\frac{d}{d\varphi_1}\log\left(\frac{\displaystyle \mathrm{dn}
\left(Z_1;m\right)^2+
i\sqrt{mm'}\mathrm{sn}\left(Z_1;m\right)^2}
{\displaystyle \mathrm{dn}
\left(Z_1;m\right)^2-
i\sqrt{mm'}\mathrm{sn}\left(Z_1;m\right)^2}
\right).
\end{equation}

\subsubsection*{Use of elliptic function identities for a fixed elliptic parameter}
Now, we simplify $\dot{C}$, $\dot{S}$, and $\dot{G}$ 
further by recalling
some identities relating elliptic functions at a fixed value of the
elliptic parameter $m$.  
Firstly, using the Pythagorean identities \cite{Akhiezer}
\begin{equation}
\mathrm{sn}(\cdot;m)^2+\mathrm{cn}(\cdot;m)^2=
\mathrm{dn}(\cdot;m)^2+m\mathrm{sn}(\cdot;m)^2=1
\label{eq:JacobiPythagoras}
\end{equation}
to eliminate $\mathrm{dn}(Z_j;m)^2$ 
and recalling that $e^{i\zeta}=\sqrt{m'}+i\sqrt{m}$, we have
\begin{equation}
\dot{C}=(-1)^{\#\Delta}\frac{e^{i\zeta}}{2}
\left[\frac{1-i\sqrt{m}e^{-i\zeta}\mathrm{sn}(Z_1;m)^2}
{1+i\sqrt{m}e^{i\zeta}\mathrm{sn}(Z_1;m)^2}+
\frac{1-i\sqrt{m}e^{-i\zeta}\mathrm{sn}(Z_2;m)^2}
{1+i\sqrt{m}e^{i\zeta}\mathrm{sn}(Z_2;m)^2}\right],
\end{equation}
\begin{equation}
\dot{S}=(-1)^{\#\Delta}\frac{e^{i\zeta}}{2i}
\left[\frac{1-i\sqrt{m}e^{-i\zeta}\mathrm{sn}(Z_2;m)^2}
{1+i\sqrt{m}e^{i\zeta}\mathrm{sn}(Z_2;m)^2}-
\frac{1-i\sqrt{m}e^{-i\zeta}\mathrm{sn}(Z_1;m)^2}
{1+i\sqrt{m}e^{i\zeta}\mathrm{sn}(Z_1;m)^2}\right],
\end{equation}
and
\begin{equation}
\dot{G}=-2i\frac{\partial\Phi}{\partial t}
\frac{d}{d\varphi_1}\log\left(\frac{1+i\sqrt{m}e^{i\zeta}\mathrm{sn}(Z_1;m)^2}
{1-i\sqrt{m}e^{-i\zeta}\mathrm{sn}(Z_1;m)^2}\right).
\end{equation}
Next, using the double-angle identity
\begin{equation}
\mathrm{sn}(Z;m)^2=\frac{m\mathrm{cn}(2Z;m)-\mathrm{dn}(2Z;m)+m'}
{m\mathrm{cn}(2Z;m)-m\mathrm{dn}(2Z;m)},
\label{eq:doubleangle}
\end{equation}
these can be written in the form
\begin{equation}
\begin{split}
\dot{C}=&(-1)^{\#\Delta}\frac{e^{i\zeta}}{2}
\left[\frac{\sqrt{m}e^{-i\zeta}\mathrm{cn}(2Z_1;m)+i\mathrm{dn}(2Z_1;m)-i
\sqrt{m'}e^{-i\zeta}}{\sqrt{m}e^{i\zeta}\mathrm{cn}(2Z_1;m)-i\mathrm{dn}(2Z_1;m)
+i\sqrt{m'}e^{i\zeta}}\right.\\
&\quad\quad\quad\quad\quad\quad{}+\left.
\frac{\sqrt{m}e^{-i\zeta}\mathrm{cn}(2Z_2;m)+i\mathrm{dn}(2Z_2;m)-i
\sqrt{m'}e^{-i\zeta}}{\sqrt{m}e^{i\zeta}\mathrm{cn}(2Z_2;m)-i\mathrm{dn}(2Z_2;m)
+i\sqrt{m'}e^{i\zeta}}\right],
\end{split}
\end{equation}
\begin{equation}
\begin{split}
\dot{S}=&(-1)^{\#\Delta}\frac{e^{i\zeta}}{2i}
\left[\frac{\sqrt{m}e^{-i\zeta}\mathrm{cn}(2Z_2;m)+i\mathrm{dn}(2Z_2;m)-i
\sqrt{m'}e^{-i\zeta}}{\sqrt{m}e^{i\zeta}\mathrm{cn}(2Z_2;m)-i\mathrm{dn}(2Z_2;m)
+i\sqrt{m'}e^{i\zeta}}\right.\\
&\quad\quad\quad\quad\quad\quad{}-\left.
\frac{\sqrt{m}e^{-i\zeta}\mathrm{cn}(2Z_1;m)+i\mathrm{dn}(2Z_1;m)-i
\sqrt{m'}e^{-i\zeta}}{\sqrt{m}e^{i\zeta}\mathrm{cn}(2Z_1;m)-i\mathrm{dn}(2Z_1;m)
+i\sqrt{m'}e^{i\zeta}}\right],
\end{split}
\end{equation}
and
\begin{equation}
\dot{G}=-2i\frac{\partial\Phi}{\partial t}
\frac{d}{d\varphi_1}\log\left(\frac{\sqrt{m}e^{i\zeta}\mathrm{cn}(2Z_1;m)
-i\mathrm{dn}(2Z_1;m)+i\sqrt{m'}e^{i\zeta}}
{\sqrt{m}e^{-i\zeta}\mathrm{cn}(2Z_1;m)+i\mathrm{dn}(2Z_1;m)-i\sqrt{m'}e^{-i\zeta}}\right).
\end{equation}
Then, since
\begin{equation}
2Z_1=W-K(m)\quad\text{and}\quad 2Z_2=W+K(m)\quad\text{where}\quad
W:=\frac{2\nu K(m)}{\pi},
\end{equation}
the use of the identities
\begin{equation}
\mathrm{cn}(W\pm K(m);m)=\mp\sqrt{m'}\frac{\mathrm{sn}(W;m)}{\mathrm{dn}(W;m)}
\quad\text{and}\quad
\mathrm{dn}(W\pm K(m);m)=\frac{\sqrt{m'}}{\mathrm{dn}(W;m)}
\end{equation}
yields
\begin{equation}
\begin{split}
\dot{C}=& (-1)^{\#\Delta}\frac{e^{i\zeta}}{2}
\left[\frac{\sqrt{m}e^{-i\zeta}\mathrm{sn}(W;m)+i-ie^{-i\zeta}
\mathrm{dn}(W;m)}
{\sqrt{m}e^{i\zeta}\mathrm{sn}(W;m)-i+ie^{i\zeta}
\mathrm{dn}(W;m)}\right.\\
&\quad\quad\quad\quad\quad\quad{}+\left.
\frac{-\sqrt{m}e^{-i\zeta}\mathrm{sn}(W;m)+i-ie^{-i\zeta}
\mathrm{dn}(W;m)}
{-\sqrt{m}e^{i\zeta}\mathrm{sn}(W;m)-i+ie^{i\zeta}
\mathrm{dn}(W;m)}\right],
\end{split}
\end{equation}
\begin{equation}
\begin{split}
\dot{S}=& (-1)^{\#\Delta}\frac{e^{i\zeta}}{2i}
\left[\frac{-\sqrt{m}e^{-i\zeta}\mathrm{sn}(W;m)+i-ie^{-i\zeta}
\mathrm{dn}(W;m)}
{-\sqrt{m}e^{i\zeta}\mathrm{sn}(W;m)-i+ie^{i\zeta}
\mathrm{dn}(W;m)}\right.\\
&\quad\quad\quad\quad\quad\quad{}-\left.
\frac{\sqrt{m}e^{-i\zeta}\mathrm{sn}(W;m)+i-ie^{-i\zeta}
\mathrm{dn}(W;m)}
{\sqrt{m}e^{i\zeta}\mathrm{sn}(W;m)-i+ie^{i\zeta}
\mathrm{dn}(W;m)}\right],
\end{split}
\end{equation}
and
\begin{equation}
\dot{G}=\frac{4K(m)}{i\pi}\frac{\partial\Phi}{\partial t}
\frac{d}{dW}\log\left(\frac{\sqrt{m}\mathrm{sn}(W;m)-ie^{-i\zeta}+i\mathrm{dn}(W;m)}{\sqrt{m}\mathrm{sn}(W;m)+ie^{i\zeta}-i\mathrm{dn}(W;m)}\right).
\end{equation}
Again applying \eqref{eq:JacobiPythagoras}
and $e^{i\zeta}=\sqrt{m'}+i\sqrt{m}$, and the differential identities 
\cite{Akhiezer}
\begin{equation}
\begin{split}
\frac{d}{dW}\mathrm{sn}(W;m)&=\mathrm{cn}(W;m)\mathrm{dn}(W;m)\\
\frac{d}{dW}\mathrm{cn}(W;m)&=-\mathrm{sn}(W;m)\mathrm{dn}(W;m)\\
\frac{d}{dW}\mathrm{dn}(W;m)&=-m\mathrm{cn}(W;m)\mathrm{sn}(W;m),
\end{split}
\label{eq:sncndndiff}
\end{equation}
these become simply
\begin{equation}
\begin{split}
\dot{C}&=(-1)^{\#\Delta}\mathrm{dn}(W;m)\\
\dot{S}&=-(-1)^{\#\Delta}\sqrt{m}\mathrm{sn}(W;m)\\
\dot{G}&=-\frac{4K(m)}{\pi}\frac{\partial\Phi}{\partial t}
\sqrt{m}\mathrm{cn}(W;m).
\end{split}
\label{eq:CSGBpenultimate}
\end{equation}

We have already introduced $(x,t)$-dependence via the conditions $M=I=0$.  If we also
recall that $\nu=\Phi/\epsilon_N + \pi\#\Delta$, then $W$ takes the form
\begin{equation}
W=\frac{2\Phi K(m)}{\pi\epsilon_N}+2\#\Delta K(m).
\end{equation}
But, since $\#\Delta$ is even in case \librational, and since
$\mathrm{sn}(\cdot;m)$, $\mathrm{cn}(\cdot;m)$, and $\mathrm{dn}(\cdot;m)$ are all periodic with period $4K(m)$, the formulae \eqref{eq:CSGBpenultimate} reduce to
the expressions \eqref{eq:CNSNGNLibrational}, as desired.  

This nearly completes the proof of Proposition~\ref{prop:outerelliptic}
in case \librational.  It only remains to confirm the differential relations
\eqref{eq:CNSNGNrelation}.
Note, however, 
that by partial differentiation with respect to $m$ of the relation
\begin{equation}
\int_0^{\mathrm{sn}(u;m)}\frac{dt}{\sqrt{1-t^2}\sqrt{1-mt^2}}=u,
\end{equation}
the use of the Pythagorean identities \eqref{eq:JacobiPythagoras} shows
that
\begin{equation}
\frac{\partial}{\partial m}\mathrm{sn}(u;m)=-\mathrm{cn}(u;m)
\mathrm{dn}(u;m)\int_0^{\mathrm{sn}(u;m)}\frac{\partial}{\partial m}
\left[\frac{1}{\sqrt{1-t^2}\sqrt{1-mt^2}}\right]\,dt.
\end{equation}
Now the integral increases by $4K'(m)$ when $u$ increases by 
$4K(m)$, the fundamental real period of $\mathrm{sn}(\cdot;m)$.  Therefore,
\begin{equation}
\frac{\partial}{\partial m}\mathrm{sn}(u;m)=-\mathrm{cn}(u;m)\mathrm{dn}(u;m)
\left[\frac{K'(m)}{K(m)}u + f(u;m)\right],
\label{eq:dsndm}
\end{equation}
where $f(u+4K(m);m)=f(u;m)$, making $f$ periodic and hence bounded with
respect to $u$.  By partial differentiation of $\dot{S}_N(x,t)$ as given
by the formula \eqref{eq:CNSNGNLibrational} we obtain
\begin{equation}
\begin{split}
\epsilon_N\frac{\partial \dot{S}_N}{\partial t}(x,t)&=
-\frac{\epsilon_N}{2\sqrt{m}}\frac{\partial m}{\partial t}\mathrm{sn}(u;m)
-\epsilon_N\sqrt{m}\frac{\partial m}{\partial t}\frac{\partial}{\partial m}
\mathrm{sn}(u;m)\\
&\quad\quad\quad\quad{}
-\sqrt{m}\frac{\partial}{\partial u}\mathrm{sn}(u;m)
\left[\frac{2K(m)}{\pi}\frac{\partial\Phi}{\partial t} +
\frac{2\Phi K'(m)}{\pi}\frac{\partial m}{\partial t}\right],\quad\quad
u = \frac{2\Phi K(m)}{\pi\epsilon_N}.
\end{split}
\end{equation}
Using \eqref{eq:sncndndiff} and \eqref{eq:dsndm} to evaluate the partial 
derivatives of $\mathrm{sn}(u;m)$ (note that the derivatives in 
\eqref{eq:sncndndiff} are in fact partial derivatives with $m$ held fixed)
this becomes
\begin{equation}
\begin{split}
\epsilon_N\frac{\partial\dot{S}_N}{\partial t}(x,t)&=
-\frac{2K(m)}{\pi}\frac{\partial\Phi}{\partial t}
\sqrt{m}\mathrm{cn}(u;m)\mathrm{dn}(u;m)\\
&\quad\quad\quad\quad{}+\epsilon_N\frac{\partial m}{\partial t}\left[
\sqrt{m}\mathrm{cn}(u;m)
\mathrm{dn}(u;m)f(u;m)
-\frac{1}{2\sqrt{m}}\mathrm{sn}(u;m)\right],
\end{split}
\end{equation}
and the first of the relations \eqref{eq:CNSNGNrelation} then follows
upon comparing with the formulae \eqref{eq:CNSNGNLibrational}.  The
second relation is established in a completely analogous manner.  This
finally completes the proof of Proposition~\ref{prop:outerelliptic} in
case \librational.

\subsubsection{Formulae for $\dot{C}$, $\dot{S}$, and $\dot{G}$ in
case \rotational}
Once again to obtain the asymptotic behavior of $\dot{\mathbf{O}}^\mathrm{out}(w)$
as $w\to 0$ and $w\to\infty$, it is sufficient to assume that
$\Im\{w\}>0$, and therefore to use the formula
\eqref{eq:dotOrotationalformulaUup} to analyze the limit $w\to \infty$ and
the formula \eqref{eq:dotOrotationalformulaBup} to analyze the limit
$w\to\infty$.  By choosing a path in the open upper half-plane, we easily
obtain the following asymptotic formulae for the Abel mapping:
\begin{equation}
A(P_1(w))=\frac{1}{2}H-i\pi + \frac{2ic}{\sqrt{w_0w_1}}\sqrt{-w} + \bo(w^{3/2}),\quad w\to 0,\quad \Im\{w\}>0,
\end{equation}
and
\begin{equation}
A(P_1(w))=-i\pi +\frac{2ic}{\sqrt{-w}} + \bo(w^{-3/2}),\quad w\to\infty,
\quad \Im\{w\}>0,
\end{equation}
and of course $A(P_2(w))=-A(P_1(w))$.  Also, we have the following asymptotic
formulae for $q(w)$:
\begin{equation}
q(w)=\left(\frac{w_0}{w_1}\right)^{1/4} + \bo(w),\quad w\to 0
\end{equation}
and
\begin{equation}
q(w)=1+\bo(w^{-1}),\quad w\to\infty.
\end{equation}
We will also require asymptotic expansions of $n(w)$ defined by 
\eqref{eq:nofwdef} valid for large and small $w$.  Beginning with the exact
formula
\begin{equation}
n(w)=-\frac{S(w)}{4}\left[\frac{\mathcal{H}}{2\pi i}\int_{w_1}^0\frac{ds}{S_+(s)(s-w)}
+\int_{w_0}^{w_1}\frac{ds}{S(s)(s-w)}\right]
\end{equation}
we may use the expansion $S(w)=w^{-3/2}(1+\bo(w^{-1}))$ as $w\to\infty$
and expand $(s-w)^{-1}$ in a geometric series for large $w$ to obtain
\begin{equation}
\begin{split}
n(w)&=\frac{w^{1/2}}{4}(1+\bo(w^{-1}))\left[\frac{\mathcal{H}}{2\pi i}
\int_{w_1}^0\frac{ds}{S_+(s)} +\int_{w_0}^{w_1}\frac{ds}{S(s)} \right.\\
&\quad\quad\quad\quad\quad\quad{}+\left. 
\left(\frac{\mathcal{H}}{2\pi i}\int_{w_1}^0\frac{s\,ds}{S_+(s)} +\int_{w_0}^{w_1}\frac{s\,ds}{S(s)}\right)w^{-1} + \bo(w^{-2})\right],\quad w\to\infty.
\end{split}
\end{equation}
But, by definition of $\mathcal{H}$, we may write this in the form
\begin{equation}
n(w)=\frac{1}{2w^{1/2}}\left[\frac{\mathcal{H}}{2\pi i}\int_{w_1}^0\frac{s+C}{2S_+(s)}\,ds +
\int_{w_0}^{w_1}\frac{s+C}{2S(s)}\,ds + \bo(w^{-1})\right],\quad w\to\infty.
\end{equation}
Identifying the integrals as cycles of the meromorphic differential $\Omega(P)$
then yields
\begin{equation}
\begin{split}
n(w)&=
\frac{1}{2w^{1/2}}\left[\frac{1}{2}\kappa + \bo(w^{-1})\right]\\
&= \frac{c}{2i\sqrt{-w}} + \bo(w^{-3/2}),\quad w\to\infty,\quad\Im\{w\}>0,
\end{split}
\end{equation}
where in the second line we have used Lemma~\ref{lem:dissection}.  On the other
hand, we may write $n(w)$ in the form
\begin{equation}
n(w)=-\frac{S(w)}{4}\left[-\frac{\mathcal{H}}{4\pi i}\oint_a\frac{ds}{S(s)(s-w)} +
\int_{w_0}^{w_1}\frac{ds}{S(s)(s-w)}\right]
\end{equation}
where it is understood that $w$ lies outside of the loop contour $a$
pictured in Figure~\ref{fig:rotationalhomology}.  If we let $w$ cross
the contour $a$ to approach the origin, then we obtain a residue contribution:
\begin{equation}
n(w)=-\frac{\mathcal{H}}{8}-\frac{S(w)}{4}\left[-\frac{\mathcal{H}}{4\pi i}\oint_a
\frac{ds}{S(s)(s-w)} +\int_{w_0}^{w_1}\frac{ds}{S(s)(s-w)}\right]
\end{equation}
where it is now understood that $w$ lies \emph{inside} of $a$.  In
particular, the expression in brackets has a well-defined value at
$w=0$ (it is in fact analytic in a neighborhood of $w=0$).  Since
$S(w)=\sqrt{w_0w_1}w^{1/2}(1+\bo(w))$ as $w\to 0$, we therefore see that
\begin{equation}
n(w)=-\frac{\mathcal{H}}{8}-\frac{\sqrt{w_0w_1}}{8}\left[-\frac{\mathcal{H}}{2\pi i}\oint_a
\frac{ds}{S(s)s} + 2\int_{w_0}^{w_1}\frac{ds}{S(s)s}\right]w^{1/2} + \bo(w^{3/2}),\quad w\to 0.
\end{equation}
The differential identity
\begin{equation}
d(w):=\frac{(w-w_0)(w-w_1)}{S(w)}\quad\implies\quad d'(w)=\frac{w}{2S(w)}-
\frac{w_0w_1}{2wS(w)}
\end{equation}
allows integration by parts, leading to the equivalent expansion
\begin{equation}
n(w)=-\frac{\mathcal{H}}{8}-\frac{1}{8\sqrt{w_0w_1}}\left[-\frac{\mathcal{H}}{2\pi i}\oint_a
\frac{s\,ds}{S(s)} +2\int_{w_0}^{w_1}\frac{s\,ds}{S(s)}\right]w^{1/2} +\bo(w^{3/2}),\quad w\to 0.
\end{equation}
Using the definitions of $\mathcal{H}$ and the meromorphic differential $\Omega(P)$,
this may be written in the form
\begin{equation}
\begin{split}
n(w)&=-\frac{\mathcal{H}}{8} -\frac{1}{4\sqrt{w_0w_1}}\left[-\frac{\mathcal{H}}{2\pi i}\oint_a\Omega(P) +\oint_b\Omega(P)\right]w^{1/2} + \bo(w^{3/2})\\
&= -\frac{\mathcal{H}}{8}-\frac{\kappa}{4\sqrt{w_0w_1}}w^{1/2} + \bo(w^{3/2})\\
&=-\frac{\mathcal{H}}{8}+\frac{c}{2i\sqrt{w_0w_1}}\sqrt{-w} +\bo(w^{3/2}),\quad w\to 0,\quad \Im\{w\}>0,
\end{split}
\end{equation}
where Lemma~\ref{lem:dissection} has again been used.

Applying these expansions for $w$ small with $\Im\{w\}>0$ to the
formula \eqref{eq:dotOrotationalformulaBup} with the help of
\eqref{eq:tjP1K}--\eqref{eq:tjP2K} then yields
\begin{equation}
\dot{O}_{11}^{0,0}=\dot{O}_{22}^{0,0}=
\frac{1}{2}\left(\frac{w_0}{w_1}\right)^{1/4}\frac{\Theta(\tfrac{1}{2}\mathcal{H};\mathcal{H})}
{\Theta(0;\mathcal{H})}e^{\mathcal{H}/8}\left[\frac{\Theta(2i\varphi_1;\mathcal{H})}{\Theta(2i\varphi_2;\mathcal{H})}
e^{i\nu}+
\frac{\Theta(2i\varphi_2;\mathcal{H})}{\Theta(2i\varphi_1;\mathcal{H})}e^{-i\nu}
\right]
\end{equation}
and
\begin{equation}
\dot{O}^{0,0}_{12}=-\dot{O}^{0,0}_{21}=
\frac{1}{2i}\left(\frac{w_0}{w_1}\right)^{1/4}
\frac{\Theta(\tfrac{1}{2}\mathcal{H};\mathcal{H})}{\Theta(0;\mathcal{H})}e^{\mathcal{H}/8}
\left[\frac{\Theta(2i\varphi_1;\mathcal{H})}{\Theta(2i\varphi_2;\mathcal{H})}e^{i\nu}-\frac{\Theta(2i\varphi_2;\mathcal{H})}{\Theta(2i\varphi_1;\mathcal{H})}e^{-i\nu}\right].
\end{equation}
As in case \librational, we observe that the matrix
$\dot{\mathbf{O}}^{0,0}=\dot{\mathbf{O}}^\mathrm{out}(0)$ is invariant under 
conjugation by
$\sigma_2$.  Substitution from these formulae into \eqref{eq:dotClibrational} and
\eqref{eq:dotSlibrational} gives, respectively,
\begin{equation}
\dot{C}=\frac{(-1)^{\#\Delta}}{2}
\left(\frac{w_0}{w_1}\right)^{1/4}\frac{\Theta(\tfrac{1}{2}\mathcal{H};\mathcal{H})}{\Theta(0;\mathcal{H})}
e^{\mathcal{H}/8}\left[\frac{\Theta(2i\varphi_1;\mathcal{H})}{\Theta(2i\varphi_2;\mathcal{H})}e^{i\nu} +\frac{\Theta(2i\varphi_2;\mathcal{H})}{\Theta(2i\varphi_1;\mathcal{H})}e^{-i\nu}\right]
\label{eq:CNthetasrotational}
\end{equation}
and
\begin{equation}
\dot{S}=-\frac{(-1)^{\#\Delta}}{2i}
\left(\frac{w_0}{w_1}\right)^{1/4}\frac{\Theta(\tfrac{1}{2}\mathcal{H};\mathcal{H})}{\Theta(0;\mathcal{H})}
e^{\mathcal{H}/8}\left[\frac{\Theta(2i\varphi_1;\mathcal{H})}{\Theta(2i\varphi_2;\mathcal{H})}e^{i\nu} -\frac{\Theta(2i\varphi_2;\mathcal{H})}{\Theta(2i\varphi_1;\mathcal{H})}e^{-i\nu}\right].
\label{eq:SNthetasrotational}
\end{equation}
The higher-order coefficients we need to calculate $\dot{G}$
are obtained by continuing the expansion for $w$ near the origin to
higher order, and also by expanding \eqref{eq:dotOrotationalformulaUup} as
$w\to\infty$ with $\Im\{w\}>0$ with the help of
\eqref{eq:tjP1K}--\eqref{eq:tjP2K}:
\begin{equation}
\begin{split}
\dot{O}^{0,1}_{22}&=-\frac{ic}{\sqrt{w_0w_1}}\left(\frac{w_0}{w_1}\right)^{1/4}\frac{\Theta(\tfrac{1}{2}\mathcal{H};\mathcal{H})}{\Theta(0;\mathcal{H})}e^{\mathcal{H}/8}
\left[\frac{\Theta'(2i\varphi_1;\mathcal{H})}{\Theta(2i\varphi_2;\mathcal{H})}e^{i\nu}+\frac{\Theta'(2i\varphi_2;\mathcal{H})}{\Theta(2i\varphi_1;\mathcal{H})}e^{-i\nu}\right.\\
&\quad\quad\quad\quad\quad\quad\quad\quad\quad\quad\quad\quad\quad\quad{}+\left.\frac{1}{4}\frac{\Theta(2i\varphi_1;\mathcal{H})}{\Theta(2i\varphi_2;\mathcal{H})}e^{i\nu}-
\frac{1}{4}
\frac{\Theta(2i\varphi_2;\mathcal{H})}{\Theta(2i\varphi_1;\mathcal{H})}e^{-i\nu}
\right],
\end{split}
\end{equation}
\begin{equation}
\begin{split}
\dot{O}^{0,1}_{12}&=-\frac{c}{\sqrt{w_0w_1}}
\left(\frac{w_0}{w_1}\right)^{1/4}\frac{\Theta(\tfrac{1}{2}\mathcal{H};\mathcal{H})}{\Theta(0;\mathcal{H})}
e^{\mathcal{H}/8}\left[\frac{\Theta'(2i\varphi_1;\mathcal{H})}{\Theta(2i\varphi_2;\mathcal{H})}
e^{i\nu}-
\frac{\Theta'(2i\varphi_2;\mathcal{H})}{\Theta(2i\varphi_1;\mathcal{H})}e^{-i\nu}
\right.\\
&\left.\quad\quad\quad\quad\quad\quad\quad\quad\quad\quad\quad\quad\quad\quad{}+\frac{1}{4}
\frac{\Theta(2i\varphi_1;\mathcal{H})}{\Theta(2i\varphi_2;\mathcal{H})}e^{i\nu}
+\frac{1}{4}\frac{\Theta(2i\varphi_2;\mathcal{H})}{\Theta(2i\varphi_1;\mathcal{H})}e^{-i\nu}
\right],
\end{split}
\end{equation}
and
\begin{equation}
\dot{O}^{\infty,1}_{12}=c\left[
\frac{\Theta'(2i\varphi_2;\mathcal{H})}{\Theta(2i\varphi_2;\mathcal{H})}
-\frac{\Theta'(2i\varphi_1;\mathcal{H})}{\Theta(2i\varphi_1;\mathcal{H})}-\frac{1}{2}\right].
\end{equation}
Substituting into \eqref{eq:dotvlibrational} we have
\begin{equation}
\dot{G}=
c\left[1-\frac{1}{w_1}\frac{\Theta(\tfrac{1}{2}\mathcal{H};\mathcal{H})^2e^{\mathcal{H}/4}}{\Theta(0;\mathcal{H})^2}\right]
\left[\frac{\Theta'(2i\varphi_2;\mathcal{H})}{\Theta(2i\varphi_2;\mathcal{H})}-
\frac{\Theta'(2i\varphi_1;\mathcal{H})}{\Theta(2i\varphi_1;\mathcal{H})}-\frac{1}{2}\right].
\label{eq:GNrotational}
\end{equation}

\subsubsection*{Evaluation of $\mathcal{H}$.  Transformations of $\Theta$ with respect to $\mathcal{H}$}
By simple contour deformations (referring to Figure~\ref{fig:rotationalhomology}) and
the definitions of $\tilde{S}(P)$ and of $S(w)$, we see that
\begin{equation}
\mathcal{H} = -2\pi\frac{\displaystyle\int_{w_0}^{w_1}\frac{dw}{\sqrt{w(w-w_0)(w-w_1)}}}
{\displaystyle\int_{w_1}^0\frac{dw}{\sqrt{-w(w-w_0)(w-w_1)}}}
\end{equation}
where in each case the positive square root is meant.  Although this is a
negative real quantity and hence is associated with an elliptic parameter that
is in the normal case, it will be in fact convenient to rewrite $\mathcal{H}$ in the 
form $\mathcal{H}=2\mathcal{H}_0$ and associate $\mathcal{H}_0$ with an elliptic parameter $m$.  To this end,
we recall the substitution \eqref{eq:wsKsubstKm}
and obtain
\begin{equation}
\int_{w_1}^{0}\frac{dw}{\sqrt{-w(w-w_0)(w-w_1)}} = 
\frac{2K(m)}{\sqrt{-w_0}+\sqrt{-w_1}}
\end{equation}
where
\begin{equation}
m:=\frac{4\sqrt{w_0w_1}}{(\sqrt{-w_0}+\sqrt{-w_1})^2}\in (0,1).
\label{eq:modulusK}
\end{equation}
For the denominator of $\mathcal{H}$, we use the globally bijective 
fractional-linear substitution
\begin{equation}
w=-\sqrt{w_0w_1}\frac{(\sqrt{-w_0}+\sqrt{-w_1})-(\sqrt{-w_0}-\sqrt{-w_1})s}
{(\sqrt{-w_0}+\sqrt{-w_1})+(\sqrt{-w_0}-\sqrt{-w_1})s}
\end{equation}
to obtain
\begin{equation}
\int_{w_0}^{w_1}\frac{dw}{\sqrt{w(w-w_0)(w-w_1)}}=
\frac{4K(m')}{\sqrt{-w_0}+\sqrt{-w_1}}
\end{equation}
where $m':=1-m$.  Therefore, it follows that
\begin{equation}
\mathcal{H}=2\mathcal{H}_0,\quad \mathcal{H}_0:=-2\pi\frac{K(m')}{K(m)}
\end{equation}
where the elliptic modulus $m$ is given by \eqref{eq:modulusK}.

Multiplying $\mathcal{H}_0$ by two corresponds to a Gauss transformation (in reverse).
If we try directly to apply \eqref{eq:Gauss1}--\eqref{eq:Gauss2} to express
$\Theta$-functions with parameter $\mathcal{H}$ in terms of those with parameter $\mathcal{H}_0$ we
will need to choose branches of square roots.  This difficulty can, however,
be circumvented by using the formula \cite{Wolfram}
\begin{equation}
\Theta(2z;2\mathcal{H})=\frac{\Theta(z;\mathcal{H})^2 +\Theta(z+i\pi;\mathcal{H})^2}{2\Theta(0;2\mathcal{H})},\quad
z\in\mathbb{C},\quad \Re\{\mathcal{H}\}<0.
\label{eq:Wolfram}
\end{equation}
We may now rewrite 
\eqref{eq:CNthetasrotational} and \eqref{eq:SNthetasrotational} in terms of
$\Theta$-functions with parameter $\mathcal{H}_0$ with the help of \eqref{eq:Wolfram}, 
and then substitute into these formulae the definitions of the Jacobi 
elliptic functions \eqref{eq:sndef}--\eqref{eq:dndef}.  
Using also the positive square roots of
the identities \eqref{eq:thetasmlibrational}, the standard $\Theta$-function
properties \eqref{eq:Thetaeven}--\eqref{eq:thetamonodromy} along with the
fact that $\Theta(i\pi+\tfrac{1}{2}\mathcal{H}_0;\mathcal{H}_0)=0$, and also the identity
\begin{equation}
\left(\frac{w_0}{w_1}\right)^{1/4}=\frac{1+\sqrt{m'}}{\sqrt{m}}
\end{equation}
following from \eqref{eq:modulusK}, we see that 
\begin{equation}
\dot{C}=\frac{(-1)^{\#\Delta}}{2i}\left[\frac{\mathrm{dn}(Z_1;m)^2+\sqrt{m'}}
{\sqrt{m m'}\mathrm{sn}(Z_1;m)^2-\sqrt{m}\mathrm{cn}(Z_1;m)^2}+
\frac{\mathrm{dn}(Z_2;m)^2+\sqrt{m'}}
{\sqrt{m m'}\mathrm{sn}(Z_2;m)^2-\sqrt{m}\mathrm{cn}(Z_2;m)^2}\right]
\end{equation}
and
\begin{equation}
\dot{S}=\frac{(-1)^{\#\Delta}}{2}\left[\frac{\mathrm{dn}(Z_1;m)^2+\sqrt{m'}}
{\sqrt{m m'}\mathrm{sn}(Z_1;m)^2-\sqrt{m}\mathrm{cn}(Z_1;m)^2}-
\frac{\mathrm{dn}(Z_2;m)^2+\sqrt{m'}}
{\sqrt{m m'}\mathrm{sn}(Z_2;m)^2-\sqrt{m}\mathrm{cn}(Z_2;m)^2}\right],
\end{equation}
where $Z_j$ are defined in terms of $\varphi_j$ by \eqref{eq:Zjdefine} (although
in the current case $\varphi_j$ and $m$ have different meanings than they did in
case \librational).

To express $\dot{G}$ given by \eqref{eq:GNrotational} in terms of elliptic
functions, we first note that since $\mathcal{H}=2\mathcal{H}_0$, \eqref{eq:Wolfram} together
with the positive square roots of the identities \eqref{eq:thetasmlibrational},
the definition \eqref{eq:modulusK}, and the fact that 
$\Theta(\tfrac{1}{2}\mathcal{H}_0+i\pi;\mathcal{H}_0)=0$ imply that
\begin{equation}
\frac{1}{w_1}\frac{\Theta(\tfrac{1}{2}\mathcal{H};\mathcal{H})^2e^{\mathcal{H}/4}}{\Theta(0;\mathcal{H})^2}=-\frac{1}{\sqrt{w_0w_1}}.
\end{equation}
Elementary contour deformations show that $c\mathcal{D}=\pi$, where $\mathcal{D}$
is defined for case \rotational\ by \eqref{eq:deltarotational}. If we now let $w_0$ and $w_1$
depend on $(x,t)$ such that the moment and integral conditions $M=I=0$ are satisfied, then
Proposition~\eqref{prop:phasevelocity} applies, and from \eqref{eq:kappapartials} one finds
that
\begin{equation}
c\left[1-\frac{1}{w_1}
\frac{\Theta(\tfrac{1}{2}\mathcal{H};\mathcal{H})^2e^{\mathcal{H}/4}}{\Theta(0;\mathcal{H})^2}\right]=4\frac{\partial\Phi_\Delta}{\partial t} = 4\frac{\partial\Phi}{\partial t}.
\end{equation}
Also, since $\varphi_2=\varphi_1+\pi+i\mathcal{H}/4$, we may write $\dot{G}$
in the form
\begin{equation}
\dot{G}=-2i\frac{\partial\Phi}{\partial t}\frac{d}{d\varphi_1}
\log\left[e^{-i\varphi_1}\frac{\Theta(2i\varphi_1-\tfrac{1}{2}\mathcal{H};\mathcal{H})}
{\Theta(2i\varphi_1;\mathcal{H})}\right].
\end{equation}
Substituting from \eqref{eq:sndef}--\eqref{eq:dndef} after first using
\eqref{eq:Wolfram} and the positive square roots of \eqref{eq:thetasmlibrational}
we arrive at the formula
\begin{equation}
\dot{G}=-2i\frac{\partial\Phi}{\partial t}
\frac{d}{d\varphi_1}\log\left(\frac{\mathrm{cn}(Z_1;m)^2-\sqrt{m'}
\mathrm{sn}(Z_1;m)^2}{\mathrm{dn}(Z_1;m)^2+\sqrt{m'}}\right).
\end{equation}

\subsubsection*{Use of elliptic function identities for a fixed elliptic parameter}  Using
the Pythagorean identities \eqref{eq:JacobiPythagoras} allows us to write
$\dot{C}$ and $\dot{S}$ in terms of $\mathrm{sn}(Z_1;m)^2$ and
$\mathrm{sn}(Z_2;m)^2$ only, and $\dot{G}$ as a logarithmic derivative
of a quantity involving $\mathrm{sn}(Z_1;m)^2$ only.  
Next, we apply the double-angle identity
\eqref{eq:doubleangle} and note that 
\begin{equation}
2Z_1 = W-K(m)+iK(m')\quad\text{and}\quad 2Z_2=W+K(m)-iK(m')\quad\text{where}
\quad W:=\frac{2\nu K(m)}{\pi},
\end{equation}
allowing the use of the identities
\begin{equation}
\mathrm{cn}(W\pm K(m)\mp iK(m');m)=\frac{i\sqrt{m'}}{\sqrt{m}\mathrm{cn}(W;m)}
\quad\text{and}\quad
\mathrm{dn}(W\pm K(m)\mp iK(m');m)=\mp i\sqrt{m'}\frac{\mathrm{sn}(W;m)}{\mathrm{cn}(W;m)}.
\end{equation}
Finally, we again apply \eqref{eq:JacobiPythagoras}, 
and in the case of $\dot{G}$ 
the differential identities \eqref{eq:sncndndiff},
to express $\dot{C}$, $\dot{S}$, and $\dot{G}$ 
simply as:
\begin{equation}
\begin{split}
\dot{C}&=(-1)^{\#\Delta}\mathrm{cn}(W;m)\\
\dot{S}&=-(-1)^{\#\Delta}\mathrm{sn}(W;m)\\
\dot{G}&=-\frac{4K(m)}{\pi}\frac{\partial\Phi}{\partial t}
\mathrm{dn}(W;m).
\end{split}
\label{eq:KCSGpenultimate}
\end{equation}
Finally, inserting the value of $\nu$ as $\nu=\Phi/\epsilon_N + \pi\#\Delta$, the phase
$W$ becomes
\begin{equation}
W=\frac{2\Phi K(m)}{\pi\epsilon_N}+2\#\Delta K(m),
\end{equation}
and since $\mathrm{sn}(u+2K(m);m)=-\mathrm{sn}(u;m)$,
$\mathrm{cn}(u+2K(m);m)=-\mathrm{cn}(u;m)$, while
$\mathrm{dn}(u+2K(m);m)=\mathrm{dn}(u;m)$, we see that even though $\#\Delta$ is not
necessarily even in case \rotational, the formulae \eqref{eq:KCSGpenultimate} indeed have the desired form
\eqref{eq:CNSNGNRotational}.

To prove that \eqref{eq:CNSNGNrelation} holds also in case \rotational,
we differentiate $\dot{S}_N(x,t)$ given by \eqref{eq:CNSNGNRotational}, keeping
in mind that $m=m(x,t)$:
\begin{equation}
\epsilon_N\frac{\partial\dot{S}_N}{\partial t}(x,t)=
-\epsilon_N\frac{\partial m}{\partial t}
\frac{\partial}{\partial m}\mathrm{sn}(u;m) - 
\frac{\partial}{\partial u}\mathrm{sn}(u;m)
\left[\frac{2\Phi K'(m)}{\pi}\frac{\partial m}{\partial t} +
\frac{2K(m)}{\pi}\frac{\partial\Phi}{\partial t}\right],\quad\quad
u=\frac{2\Phi K(m)}{\pi\epsilon_N}.
\end{equation}
Computing the partial derivatives of $\mathrm{sn}(u;m)$ using
\eqref{eq:sncndndiff} and \eqref{eq:dsndm} this becomes
\begin{equation}
\epsilon_N\frac{\partial\dot{S}_N}{\partial t}(x,t)=
-\frac{2K(m)}{\pi}\frac{\partial\Phi}{\partial t}\mathrm{cn}(u;m)
\mathrm{dn}(u,m) +\epsilon_N\frac{\partial m}{\partial t}
\mathrm{cn}(u;m)\mathrm{dn}(u;m)f(u;m)
\end{equation}
from which the first of the formulae 
\eqref{eq:CNSNGNrelation} follows by comparison with
\eqref{eq:CNSNGNRotational} (the second is proved analogously).  
This completes the proof of
Proposition~\ref{prop:outerelliptic} in case \rotational.

\end{document}